\documentclass[11pt]{article}

\usepackage{tikz,tikz-cd}
\usepackage{makecell}

\usepackage{amsbsy}

\usepackage{amsmath,amsthm}
\usepackage{mathrsfs}
\usepackage{amscd,amssymb,amsfonts,verbatim,enumerate}
\usepackage[mathcal]{eucal}
\usepackage{dsfont}
\usepackage{hyperref}

\usepackage{mathpazo}
\usepackage{color,slashed}

\usepackage{authblk} 
\usepackage{booktabs} 
\usepackage{simpler-wick} 
\usepackage[mode=buildnew]{standalone} 
\usepackage{stmaryrd} 
\usepackage{subfig} 
\usepackage[compat=1.1.0]{tikz-feynman} 
\tikzfeynmanset{warn luatex=false}

 \DeclareMathOperator{\End}{End}

\DeclareMathOperator{\Sym}{Sym}

\DeclareMathOperator{\Tr}{Tr}

\newtheorem{conjecture}{Conjecture}
\newtheorem{proposition}{Proposition}
\newtheorem{lemma}{Lemma}

\newcommand{{\vac}}{\left|\emptyset\right\rangle}

\newcommand{\RP}{\mathbb{RP}}
\newcommand{\Dirac}{\slashed{\partial}}

\newcommand{\wbar}{\br{w}} 

\renewcommand{\sl}{\mathfrak{sl}}

\newcommand{\zbar}{\br{z}}
\newcommand{\so}{\mathfrak{so}}

\newcommand{\PV}{\op{PV}}

\newcommand{\dpa}[1]{\frac{\partial}{\partial #1}}

\newcommand{\eps}{\epsilon}
\newcommand{\g}{\mathfrak{g}}

\newcommand{\what}{\widehat}
\newcommand{\tr}{\triangle}

\newcommand{\til}{\widetilde}
\newcommand{\mscr}{\mathscr}

\newcommand{\br}{\overline}

\newcommand{\C}{\mathbb C}

\newcommand{\norm}[1]{\left\| #1 \right\|}
\newcommand{\Oo}{\mscr O}
\newcommand{\Z}{\mathbb Z}

\newcommand{\op}{\operatorname}
\newcommand{\mbf}{\mathbf}
\newcommand{\mbb}{\mathbb}
\newcommand{\mf}{\mathfrak}
\newcommand{\mc}{\mathcal}

\newcommand{\ip}[1]{\left\langle #1 \right\rangle}
\newcommand{\abs}[1]{\left| #1 \right|}

\newcommand{\R}{\mbb R}
\renewcommand{\d}{\mathrm{d}}

\newcommand{\dbar}{\br{\partial}}

\newcommand{\PT}{\mbb{PT}}
\newcommand{\CP}{\mbb{CP}}
\newcommand{\half}{\tfrac{1}{2}}

\newcommand{\Sp}{\op{Sp}}

\renewcommand{\i}{\mathrm{i}}
\renewcommand{\c}{\mathrm{c}}
\renewcommand{\t}{\mathrm{t}}
\renewcommand{\b}{\mathrm{b}}

\newcommand{\p}{\partial}

\newcommand{\be}{\begin{equation}}
\newcommand{\ee}{\end{equation}}

\newcommand{\bea}{\begin{equation}\begin{aligned}}
\newcommand{\eea}{\end{aligned}\end{equation}}

\newcommand{\normord}[1]{:#1:}

\newcommand{\la}{\langle}
\newcommand{\ra}{\rangle}

\newcommand{\cO}{\mathcal{O}}
\newcommand{\cN}{\mathcal{N}}

\newcommand{\J}{\mathds{J}}
\newcommand{\M}{\mathds{M}}

\newcommand{\Jt}{\til{\mathds{J}}}
\newcommand{\Mt}{\til{\mathds{M}}}

\newcommand{\fsp}{\mathfrak{sp}}

\newcommand{\mrm}[1]{\mathrm{#1}}

\newcommand{\Ad}{\mathrm{Ad}}

\newcommand{\done}{{\dot1}}
\newcommand{\dtwo}{{\dot2}}

\newcommand{\da}{{\dot\alpha}}
\newcommand{\db}{{\dot\beta}}
\newcommand{\dc}{{\dot\gamma}}

\newcommand{\lt}{{\widetilde\lambda}}

\newcommand{\sfa}{\mathsf{a}}
\newcommand{\sfb}{\mathsf{b}}
\newcommand{\sfc}{\mathsf{c}}
\newcommand{\sfd}{\mathsf{d}}
\newcommand{\sfe}{\mathsf{e}}

\title{Self-Dual Gauge Theory from the Top Down}

\begin{document}

\author[a]{Roland Bittleston}
\author[a]{Kevin Costello}
\author[a,b]{Keyou Zeng}

\affil[a]{Perimeter Institute for Theoretical Physics, 31 Caroline Street, Waterloo, Ontario, Canada}

\affil[b]{Center of Mathematical Sciences and Applications, Harvard University, Massachusetts 02138,
USA}

\maketitle

\begin{abstract}

{We introduce a family of dualities between certain non-supersymmetric self-dual gauge theories on a large class of $4d$ self-dual asymptotically flat backgrounds, and the large $N$ limit of an independently defined $2d$ chiral defect CFT. Our construction goes via twisted holography for the type I topological string on a Calabi-Yau five-fold which fibres over twistor space. In particular, we show that single-trace operators of the $2d$ defect CFT are in bijection with states of the celestial chiral algebra. We match the operator products of these states with the collinear splitting amplitudes of the self-dual gauge theory up to one-loop. Assigning vacuum expectations to central operators in the boundary theory computes bulk amplitudes on self-dual backgrounds. We are able to extract form factors from these amplitudes, which we use to give a simple closed formula for certain $n$-point two-loop all $+$ amplitudes in $\mathrm{SU}(K) \times \mathrm{SU}(R)$ gauge theory coupled to bifundamental massless fermions.}
		
\end{abstract}

\tableofcontents

\flushbottom


\section{Introduction}

Celestial holography proposes an equivalence between scattering amplitudes in asymptotically flat space-times and two-dimensional CFTs living on the celestial sphere \cite{Strominger:2017zoo}.  In self-dual theories, the celestial CFT is chiral, and scattering amplitudes of self-dual theories can (in good cases) be written as correlators of a celestial chiral algebra. 

In general, celestial CFTs can be non-local.  The great advantage of working with self-dual theories is that, when a certain twistor space anomaly vanishes, the dual CFT is \emph{local}.  This allows one to bootstrap amplitudes and form factors using the classical bootstrap familiar from two-dimensional CFT.  This provides real computational advantage: in \cite{Costello:2023vyy} new two-loop $n$-point amplitudes for massless QCD were computed using this method.  (In \cite{Dixon:2024mzh}, the four-point amplitude computed by chiral algebra methods was checked against the amplitude as computed using standard techniques, and perfect agreement was found when using a mass regulator).

In standard AdS/CFT, both sides of the holographic  duality have an independent definition: a large $N$ CFT is compared with gravity on AdS. A duality like this, where the duality between the two theories can be derived from string theory, is often referred to as a \emph{top-down} duality.  A top-down celestial duality was derived in \cite{Costello:2022jpg,Costello:2023hmi}.  This is a duality between a rather simple $2d$ CFT, and a self-dual theory on a  four-dimensional manifold known as Burns space.

In this paper, we propose a large class of top-down dualities of this form, all of which relate an explicit two-dimensional system, at large $N$, to self-dual QCD on a non-trivial self-dual background.  This class of dualities leads to a number of new explicit computations in gauge theory, including closed-form expressions for tree-level amplitudes in various strong-field self-dual backgrounds, and certain new two-loop amplitude computations.

The chiral algebra on one side of our duality has a nice interpretation in string theory: it is the algebra of BPS operators in a transverse intersection of $N$ $D5$ branes and $K$ $D5'$ branes in the type I string.  The system on the other side of our duality is a subsector of massless QCD, with certain unusual matter content but without any supersymmetry. Our duality therefore connects BPS quantities in ordinary string theory with non-supersymmetric physics in four dimensions. 


\subsection{Self-Dual QCD} 
 
We will start by recalling some background on self-dual gauge theory. Self-dual QCD (SDQCD) has Lagrangian
\be	
    \int_{\R^4}\op{tr}( B \wedge  F(A)_- )  + \int_{\R^4}\la\br{\psi},\Dirac_A\psi\ra_R\,,
\ee
Here $B\in \Omega^2_-(\R^4,\mf{g})$ is a Lagrange multiplier field imposing the self-duality equation $F(A)_- = 0$. The coupling to matter is the standard one. 

We can deform SDQCD by adding the term
\be
	- \frac{1}{2}g_\mrm{YM}^2\op{tr}(B^2) 	 
\ee
Then it is equivalent to ordinary QCD in perturbation theory.  The field $B$ becomes the field strength $F(A)_-$ and corresponds to a negative helicity gluon, whereas $A$ corresponds to the positive helicity gluon.  

SDQCD has $++-$ vertex and $-+$ propagator. The only non-trivial amplitudes in SDQCD are the well-know one-loop all $+$ amplitudes.   One gets more interesting quantities by considering form factors, which are scattering amplitudes in the presence of a local operator insertion. The form factor of the operator $\op{tr}(B^2)$ in SDQCD computes certain QCD amplitudes,  without the momentum-conserving $\delta$-function. This is simply because QCD is obtained by adding this operator to the Lagrangian.  The QCD amplitudes obtained in this way are those which, when written in terms of the SDQCD Lagrangian with the addition of $\op{tr}(B^2)$, involve only one $\op{tr}(B^2)$ vertex.   Simple combinatorics tells us that these terms are:   
\begin{enumerate} 
\item The tree-level MHV (i.e. two $-$) amplitudes.
\item The one-loop one $-$ amplitude.
\item The two-loop all $+$ amplitude. 
\end{enumerate} 

More generally, form factors of the operator $\op{tr}(B^n)$ in SDQCD match form factors of the operator $\op{tr}(F_-^n)$ in QCD up to $n$ loops, where at $l$ loops we have $n-l$ negative helicity external gluons.

The thesis of \cite{Costello:2021bah} is that, when a certain twistor space anomaly vanishes, self-dual QCD furnishes an integrable subsector of full QCD.


\subsection{Twistor Uplifts of SDQCD}

Recall that twistor space is the manifold
\be 
\PT = \Oo(1)^2 \to \CP^1\,.
\ee
As a real manifold, it is diffeomorphic to $\R^4 \times \CP^1$.  Holomorphic theories on $\PT$ give rise to field theories on $\R^4$ by Kaluza-Klein compactification along the $\CP^1$. 

An off-shell version of the Penrose-Ward transformation \cite{Ward:1977ta,Mason:2005zm} tells us that, at tree-level, self-dual gauge theory lifts to the theory on twistor space with fields 
\be
b\in \Omega^{3,1}(\PT,\mf{g} )\,,\qquad a\in\Omega^{0,1}(\PT, \mf{g} )  
\ee
and action
\be
\frac{\i}{2\pi}\int_\PT\op{tr}(b\wedge F^{0,2}(a))\,.
\ee
This is known as holomorphic BF theory.  If $R$ is a real representation of the Lie algebra $\g$, matter is included by adjoining fields
\be
\alpha \in \Pi \Omega^{0,1}(\PT, R \otimes \Oo(-1))\,,\qquad\beta \in \Pi \Omega^{0,1}(\PT, R \otimes \Oo(-3) ) 
\ee
with Lagrangian $\la\beta,\dbar_a\alpha\ra_R$. Here $\Pi$ denotes parity shift.


\subsection{Integrability of SDQCD and the Twistorial Anomaly}

At tree-level, SDQCD is obtained by dimensional reduction from a local holomorphic field theory on twistor space.  As argued in \cite{Costello:2021bah}, any four-dimensional theory which lifts to twistor space in this way is integrable.  Integrability arises in roughly the same way that two-dimensional models which arise from four-dimensional Chern-Simons are integrable \cite{Costello:2017dso,Costello:2018gyb,Costello:2019tri}.  

Holomorphic theories in three complex dimensions can have anomalies \cite{williams2020ren}, and generally, the twistor uplift of SDQCD is anomalous. The anomaly vanishes when the matter representation $R$ and the gauge Lie algebra $\mf{g}$ satisfy
\be
\op{tr}_{\g}(X^4) = \op{tr}_R(X^4)\,. \label{eqn:anomaly} 
\ee
From the four-dimensional perspective, this is an anomaly to integrability, and not a gauge anomaly.  This is because the anomaly can be cancelled by a counter-term which is non-local on twistor space but local on space-time.

This trace identity is enough to guarantee that \emph{all} amplitudes of SDQCD vanish. This is true for any theory which arises from a local holomorphic field theory on the twistor space $\PT = \Oo(1)^2 \to \CP^1$ of flat space.  The point is that, by the half-Fourier transform, four-dimensional states are represented as states which are localized to fibres of the twistor fibration $\PT \to \CP^1$, and locality of the theory on twistor space means that these states do not interact.  

From the four-dimensional perspective, vanishing of these amplitudes is not at all obvious. However, Dixon and Morales \cite{Dixon:2024mzh,Dixon:2024tsb} have shown that vanishing of the all $+$ amplitudes in this situation is a consequence of an infinite number of kinematic identities that were previously studied in \cite{Bjerrum-Bohr:2011jrh}.

Bardeen suggested \cite{Bardeen:1995gk} that the one-loop all $+$ amplitude is an obstruction to integrability of self-dual Yang-Mills (SDYM).  The relation between this amplitude and the twistor space anomaly makes this connection more transparent. 


\subsection{Review of Celestial Chiral Algebras} \label{subsec:ccas}

Given a holomorphic theory on twistor space, one can build a chiral algebra \cite{Costello:2022wso} controlling the form factors of the corresponding four-dimensional self-dual theory.  In this section, we will briefly review the chiral algebra for SDQCD with vanishing twistor space anomaly.

The chiral algebra is described in spinor helicity notation: null momenta $p_{\alpha\da}$ are written in terms of a pair of spinors
\be
p_{\alpha\da} = \lambda_\alpha \lt_\da
\ee
and we set
\be
\lambda_\alpha = (1,z)\,,\qquad z \in \CP^1\,.
\ee
The chiral algebra lives on the $z$ plane.

A key feature of the chiral algebra is that single particle operators placed at $z$ are in bijection with single-particle gauge theory states with momentum $\lambda = (1,z)$, $\lt$ arbitrary. 

\begin{table}
    \centering
    \begin{tabular}{ c c }
    \toprule 
    SDQCD & Chiral algebra \\
    \midrule
    States with momentum $(\lt, \lambda = (1,z))$ & Operators at $z$ \\
    Local operators & Conformal blocks \\
    Form factors & Correlators \\
    Collinear singularities in form factors & OPEs \\
    \bottomrule
    \end{tabular}
    \caption{SDQCD $\leftrightarrow$ chiral algebra dictionary} \label{tab:chiralalgebra}
\end{table}

This means that, for each $\lt$, the chiral algebra has operators 
\be \J_\sfa(\lt;z)\,,\qquad \Jt_\sfa(\lt;z) \ee
corresponding to gluons of positive/negative helicity.  These can be expanded in series
\be
\J_\sfa(\lt;z) = \sum_{m,n\geq0}\frac{\lt_\done^m \lt_\dtwo^n}{m!n!}\J_\sfa[m,n](z) \,.
\ee
Similarly for matter fields and $\Jt$, giving 
\be \til \J_\sfa[m,n](z)\,,\qquad \M_i[m,n](z)\,,\qquad \til \M_i[m,n](z)\,. \ee 
A basis of operators in chiral algebra is given by normally ordered products of derivatives of these states, e.g. 
\be
: \til \J_{\sfa_1} [m_1,n_1] \p_z \J_{\sfa_2}[m_2,n_2] \p_z^3 \M_i[m_3,n_3] \dots  : (z) \,.
\ee

The dictionary between the chiral algebra and SDQCD is illustrated in Table \ref{tab:chiralalgebra}.  Let us explain what we mean by conformal block as it appears in this table. The chiral algebra is very much non-unitary, and most of its states are of negative conformal weight.  This means that the correlators of the chiral algebra are ambiguous. For us, a \emph{conformal block} is a consistent way of defining correlation functions of the chiral algebra, compatible with OPEs and with order of pole or zero at $z = \infty$ dictated by the conformal weight.

\begin{figure}[ht]
\centering
	\subfloat{\includegraphics{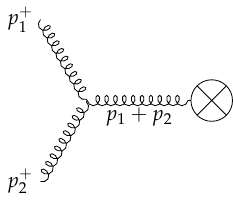}}\hfil
	\subfloat{\includegraphics{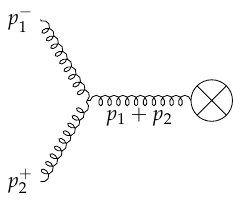}}
\caption{Tree gluon diagrams leading to collinear splitting in form factors}
\label{fig:treesplitting}
\end{figure}

OPEs are given by \emph{collinear singularities} in SDQCD form factors. At tree-level we have
\bea
&\J_\sfa(\lt_1;z_1) \J_\sfb(\lt_2;z_2) \sim \frac{f_{\sfa\sfb}^{~\,~\sfc}}{\ip{12}} \J_\sfc(\lt_1 + \lt_2;z_2)\,, \\
&\Jt_\sfa(\lt_1;z_1) \J_\sfb(\lt_2;z_2) \sim \frac{f_{\sfa\sfb}^{~\,~\sfc}}{\ip{12}} \Jt_\sfc(\lt_1 + \lt_2;z_2)\,.
\eea
where $\ip{12} = z_1 - z_2$. These come from the Feynman diagram in self-dual gauge theory illustrated in Fig. \ref{fig:treesplitting}. There are similar OPEs between $\J$ and $\M$. These OPEs are subject to loop-level corrections \cite{Costello:2022upu}, which we will discuss extensively later.


\subsection{SDQCD from String Theory} \label{subsec:SDQCD-string}

 When the twistor space anomaly vanishes, there is a celestial chiral algebra associated to SDQCD \cite{Costello:2022wso}, whose correlation functions compute form factors.  The goal of this paper is to write this chiral algebra, and deformations of it coming from self-dual backgrounds, as the large $N$ limit of the chiral algebra of a two-dimensional system. Our argument will go via an embedding of SDQCD in string theory.  This  may be of independent interest as it provides a method to connect supersymmetric quantities in string theory with non-supersymmetric amplitudes. 

The twistor uplift of SDQCD is a holomorphic BF theory with matter fields.  It is known \cite{Elliott:2020ecf} that holomorphic BF theory in complex dimension three is the holomorphic twist of $(0,1)$ supersymmetric gauge theory.   Twist here means that we change the spin of the fields in such a way that the theory has a global supercharge $Q$ of square zero, which we then add to the BRST operator.  Holomorphically twisted theories \cite{Johansen:1994aw,Budzik:2023xbr} have the feature that the anti-holomorphic vector fields $\p_{\zbar_i}$ are $Q$-exact; this means that in the twisted theory correlation functions are holomorphic functions of position.  With $(0,1)$ supersymmetry in dimension $6$, holomorphic twists are the only kind that exist \cite{Elliott:2020ecf}.
 
The twistor space anomaly \eqref{eqn:anomaly} is exactly the standard box-diagram anomaly \cite{Green:1984bx, Erler:1993zy} that is extensively studied in supersymmetric gauge theory.  In general, this means that we can engineer an anomaly-free holomorphic theory on twistor space from any anomaly-free theory with $(0,1)$ supersymmetry.  Many such theories involve $(0,1)$ tensor multiplets, introduced to cancel the anomaly by a Green-Schwarz mechanism.   In order to engineer SDQCD using this method, we will focus on $(0,1)$ theories with no tensor multiplets.   

We need to embed such systems in string theory.  There is a standard way to do this: the theory on a $D5$ brane in either type IIB or type I superstring theory has (at least) $(0,1)$ supersymmetry.  The part of the twisting procedure which changes the spins of the fields on the $D$ brane corresponds geometrically to placing the $D5$ brane in a geometry with non-trivial normal bundle.  Supersymmetry requires that the $10d$ geometry is a Calabi-Yau manifold, and the worldvolume of the $D5$ is a complex submanifold of complex dimension $3$.  

In our setting we require the worldvolume of the $D5$ brane to be twistor space $\PT = \Oo(1)^2 \to \CP^1$.  A simple way to build a CY5 containing $\PT$ as a subvariety is to take a rank two holomorphic vector bundle $V$ over $\PT$.  In order for the total space to be a CY5, we need the determinant of $V$ to be the canonical bundle of twistor space, which is $\Oo(-4)$.  It is natural to take $V$ to be
\be
V = \Oo(-1)\oplus\Oo(-3)\,.
\ee
Let us briefly analyze the theory on a $D5$ brane living on the zero section of $V$ in the holomorphic twist of type IIB.  The $(0,1)$ vector multiplet, after twisting, contributes holomorphic BF theory \cite{Elliott:2020ecf}.  The scalar fields living in the normal bundle $V$ contribute adjoint-valued fields in $V$. Since $V = \Oo(-1) \oplus \Oo(-3)$, these become the twistor uplift of ordinary fermions, valued in the adjoint representation.

Therefore, with this choice of $V$, the theory on a $D5$ brane becomes, after twisting and applying the Penrose transform, $\cN=1$ SDSQCD.\footnote{If we made a different choice for $V$, e.g. $V = \Oo(-2) \oplus \Oo(-2)$, we would find anti-commuting bosons instead of fermions.  This is clearly undesirable, and dictates our choice that $V = \Oo(-1) \oplus \Oo(-3)$.}

In general we prefer to avoid supersymmetric models, because many of the amplitudes and form factors we would like to compute vanish in the supersymmetric case.  For this reason we prefer to study the theory on a stack of $D5$ branes in the type I  string on the same Calabi-Yau geometry.  In this case, we have space-filling $D9$ branes with an $\mrm{SO}(32)$ gauge group. The geometry can be equipped with a non-trivial $\mrm{SO}(32)$ bundle, which should be a holomorphic bundle to be compatible with supersymmetry.  There is some freedom to choose which holomorphic $\mrm{SO}(32)$ bundle we use.  We will take it to be
\be
    \Oo(1)^{16} \oplus \Oo(-1)^{16}\,,
\ee
where $\Oo(1)^{16}$ and $\Oo(-1)^{16}$ are both isotropic subbundles.  This will break the $\mrm{SO}(32)$ flavour symmetry of the theory on the $D5$ branes to $\mrm{SL}(16)$.  If we do this, and then keep $V=\Oo(-1) \oplus \Oo(-3)$, we find that the $(0,1)$ theory on the $D5$ brane consists of:
\begin{enumerate} 
\item The $(0,1)$ vector multiplet for gauge group $\Sp(K)$, which after twisting becomes holomorphic BF theory for $\Sp(K)$ and so SDYM for $\Sp(K)$ on space-time.
\item A $(0,1)$  hyper multiplet in $(\Oo(-1)  \oplus \Oo(-3) ) \otimes \wedge^2_0 \mrm{F}$, where $\wedge^2_0 \mrm{F}$ is the exterior square of the fundamental of $\mf{sp}(K)$.\footnote{Here the notation $\wedge^2_{0}\mrm{F}$ means the trace-free exterior square: $\Omega_{IJ}X^{IJ} = 0$, where $\Omega_{IJ}$ is the symplectic form of $\mathfrak{sp}(K)$. In fact, the center degrees of freedom decouple and can be consistently removed.}   After twisting and applying the Penrose transform, this gives us fermions in $\wedge^2_0 \mrm{F}$.
\item A $(0,1)$ hyper multiplet in $(\Oo(-1) \oplus \Oo(-3) ) \otimes \mrm{F} \otimes \C^{16}$, from the $D5-D9$ strings.   On $\R^4$ this gives us fermions in $\mrm{F} \otimes \C^{16}$.  
\end{enumerate}
This particular theory of SDQCD is what we will study in this paper.  We will refer to this theory as the \emph{integrable self-dual $\Sp(K)$ gauge theory}, or simply SDQCD, to avoid spelling out the matter content each time. 

For $K=1$, this theory is quite simple: the  gauge group is $\mrm{SL}(2)$, and the matter content becomes $N_f = 8$. 

Our goal is to build, holographically, the chiral algebra of this SDQCD model, and use it to compute amplitudes and form factors.

To do this, we introduce another set of $N$ $D5$ brane wrapping the locus $\Oo(-1) \oplus \Oo(-3) \to \CP^1$ inside the Calabi-Yau $5$-fold $\Oo(-1)\oplus \Oo(-3) \oplus \Oo(1)^2 \to \CP^1$.   The $K$ $D5'$ and $N$ $D5$ intersect along a copy of $\CP^1$.  The intersection gives a defect in the worldvolume theory of the $N$ $D5$ branes, which has two-dimensional $(0,2)$ supersymmetry. 

As is standard in holography, we expect a correspondence between:
\begin{enumerate}
    \item Large $N$ operators (and their OPEs and correlation functions) on the stack of $N$ $D5$ branes, equipped with the defect from the $D5-D5'$ intersection.
    \item States (and their scattering amplitudes) in the holographic dual theory, consisting of the gravitational theory and the theory on the $K$ $D5'$ branes.
\end{enumerate}
We are interested in the states and amplitudes of the theory on the $K$ $D5'$ branes.  This correspondence tells us that we should be able to understand these quantities from the theory on the $N$ $D5$ branes with its defect.

As in any two-dimensional system with $(0,2)$ supersymmetry, there is a twist of the defect system which leaves a chiral algebra of BPS operators.  We denote this chiral algebra by $\mc{A}^{N,K}_\mrm{Defect}$. It has an explicit description in terms of the $N-K$ bifundamental fermions and the fields on the $N$ $D5$ branes.  It has, in addition, a subalgebra
\be
\mc{A}^N_{D5} \subset \mc{A}^{N,K}_\mrm{Defect} 
\ee
of operators which come from  BPS operators in the system of $N$ $D5$ branes.  Operators in $\mc{A}^N_{D5}$ must have trivial OPE with operators in $\mc{A}^{N,K}_\mrm{Defect}$.  This is because operator products are holomorphic functions of position, and operators in $\mc{A}^N_{D5}$ can be placed at a point in a $3$ complex dimensional space. Hartogs theorem \cite{Hartogs} tells us holomorphic functions can only have singularities in complex codimension $1$.    
 
Holography should imply that
\begin{equation*} 
\lim_{N \to \infty} A^{N}_{D5} \iff \text{ BPS supergravity states in a ten-dimensional geometry.}  
\end{equation*}
Similarly,
\begin{equation*} 
\lim_{N \to \infty} A^{N,K}_\mrm{Defect} \iff \text{ BPS states in the ten-dimensional system with the } K\ D5 \text{ branes} \,.
\end{equation*}

We would like to study the theory on the $D5$ branes decoupled from the bulk string theory. To do this, we need to remove from the algebra $\mc{A}^{N,K}_\mrm{Defect}$ any states which come from $\mc{A}^{N}_{D5}$.  

One can check  that the commutative algebra $\mc{A}^N_{D5}$ admits a unique homomorphism to the complex field $\C$ which preserves all natural symmetries. (This homomorphism sets all states charged under any symmetry to zero).  We let $I_{{\vac}}$ be the kernel of this homomorphism, spanned by charged states.   Then, we can define the algebra which will be related to the celestial chiral algebra by 
\be
    \mc{A}^{N,K} = \mc{A}^{N,K}_\mrm{Defect} / I_{\vac}
\ee
obtained by setting all operators coming from the bulk of the $D5$ worldvolume (except the identity) to zero.  Performing such an operation is only possible because the operators in the ideal $I_{{\vac}}$ have trivial OPE with all other operators.   

This procedure isolates the algebra which, at large $N$, will be holographically dual to the theory on the stack of $D5'$ branes, decoupled from the bulk type I topological string.  This decoupling of the bulk fields may seem rather surprising from the $10d$ perspective, but has been utilized in the physical string where it was termed `rigid holography' \cite{Aharony:2015zea}, and also in the context of twisted holography \cite{Ishtiaque:2018str}.  From the dual perspective, centrality of bulk states ensures that they cannot appear at intermediate steps in a chiral algebra computation, so that bulk exchanges are not contributing to amplitudes.

We arrive at the following conjecture:
\begin{conjecture} \label{conj:largeN}
The large $N$ limit of the open string algebra $\mc{A}^{N,K}$ is the celestial chiral algebra for the self-dual integrable $\Sp(K)$ theory, deformed by the $D5$ brane backreaction. 
\end{conjecture}
Note that a similar holographic duality for intersecting branes in a topological string was proven in \cite{Ishtiaque:2018str}. The details of the brane configuration in that work are rather different: the branes intersect along a line instead of a plane, and the large $N$ algebra was found to be the Yangian of $\mf{gl}(K)$.

Configurations of intersecting $D5$ branes, termed $I$ branes, and their corresponding supergravity backgrounds have been considered elsewhere in the literature \cite{Itzhaki:2005tu,Nunez:2023nnl}. The relationship to our construction is not clear to us and would be interesting to clarify.

In our setting, the backreaction is very simple when we pass to the four-dimensional self-dual QCD. It amounts to introducing a term in the action 
\be
- \frac{N}{4\pi^2}\int_{\R^4}\log\norm{x}\op{tr}(F\wedge F)\,.
\ee
This is a particular varying $\theta$-angle.  We can also access the theory without backreaction, by the simple expedient of using $\mf{osp}(2N|N)$ as the gauge algebra, rather than $\mf{sp}(N)$.  This corresponds to having the same number of $D5$s and anti-$D5$s \cite{Mikhaylov:2014aoa}, and so no source.  

We are able to perform many checks of this conjecture. In particular, single-trace operators in $\mc{A}^{N,K}$ at large $N$ reproduce the spectrum of single-particle states of the self-dual $\Sp(N)$ gauge theory we consider (these are the generators of the celestial chiral algebra).  We are also able to perform many checks of the large $N$ OPEs, matching them with the collinear singularities in self-dual gauge theory.   


\subsection{Turning on Background Fields} \label{subsec:backgrounds}

The above conjecture has an obvious generalization.  We can replace the ideal $I_{{\vac}}$ of bulk states we are annihilating by some other ideal $I_\Psi$.  The choice of $\Psi$ is given by ways of giving a vacuum expectation value (VEV) to bulk operators.  Different choices of $\Psi$ will give rise to different chiral algebras
\be
\mc{A}^{N,K}_\Psi  = \mc{A}^{N,K}_\mrm{Defect} / I_\Psi
\ee
It turns out that these chiral algebras, at large $N$, control collinear singularities of SDQCD on various different self-dual backgrounds, dictated by the choice of $\Psi$.  Possible backgrounds one can engineer are:
\begin{enumerate}
\item A background field $A_0$ for the $\mf{sl}(16)$ flavour symmetry, where $F(A_0)_-$ has a $\delta$-function source at the origin.
\item The Burns space metric \cite{burns1986twistors}.
\item The double cover of Eguchi-Hanson \cite{Eguchi:1978xp}. This solves the self-dual vacuum Einstein equations almost everywhere; with a derivative $\delta$-function source for the anti-self-dual part of the Weyl curvature tensor
\be
W_{\alpha\beta\gamma\eta} = \frac{1}{6}F_{\alpha\beta\gamma\eta}\triangle\delta_{x = 0}\,.
\ee
Here $H$ a tensor symmetric in four spinor indices.
\end{enumerate}
This construction yields many other self-dual backgrounds, since the VEVs can be linearly superimposed. We anticipate that by doing so it will be possible to engineer flavour backgrounds with an arbitrary compact source for $F(A_0)_-$, multi-centred Eguchi-Hanson, Burns space with sources for a self-dual gauge field, etc. A consequence of our analysis is that scattering amplitudes in all such backgrounds are rational, and are given by the correlators of a chiral algebra. 

We can use the chiral algebra to perform many computations of scattering amplitudes in these backgrounds. Some sample formulae will be presented later in the introduction, with detailed computations in the body of the paper.


\subsection{From Physical to Topological Strings}

One issue with the above discussion is that the theory of $D5$ branes in the type  I string is not known to have a clean holographic dual, although the $D1-D5$ system has been extensively studied \cite{Douglas:1996uz}. For our purposes, this turns not to be a major obstacle, because we will work with the \emph{twist} of the type I string, which encodes BPS quantities. In the twisted string theory, backreaction is much more straightforward. 

The familiar procedure of twisting supersymmetric gauge theory can also be applied to supergravity and string theory \cite{Costello:2016mgj,Raghavendran:2021qbh}. Here, the procedure is a little different. In supergravity, supersymmetry is gauged. Gauging supersymmetry introduces ghost fields of opposite parity to the generators of the supersymmetry algebra. In particular, there is a bosonic ghost field $\c_{\alpha }$ corresponding to each supersymmetry $Q_\alpha$. To twist supergravity we give the component of the bosonic ghost field corresponding to our chosen supercharge a VEV.

Twisting string theory is similar: one gives a VEV to the ghost for the space-time supersymmetry. In the RNS formalism, these lie in the Ramond sector.  Giving a VEV to Ramond sector fields is challenging, so twisted string theory is easier to study in the supergravity approximation.

In \cite{Costello:2016mgj}, it was conjectured that the holomorphic twist of type IIB on a Calabi-Yau $5$-fold is the topological $B$-model \cite{Witten:1988ze,Bershadsky:1993cx}.  For open strings, this is a consequence of \cite{Baulieu:2010ch}.  For closed strings, this conjecture was verified in the free limit in \cite{Saberi:2021weg}, using the pure spinor formalism.

A similar conjecture was presented in \cite{Costello:2019jsy} for the holomorphic twist of the type I string.  The twist of the type I string is conjectured to be the type I topological string, which is an orientifold of the familiar topological $B$ model.  The type I topological string on a Calabi-Yau $5$-fold contains $\mrm{SO}(32)$ holomorphic Chern-Simons theory in the open string sector, and in the closed string sector contains half of the fields of Kodaira-Spencer theory.  There is a similar Green-Schwarz mechanism as in the physical theory.  There is also a close relationship \cite{Costello:2021kiv} to the chiral part of the heterotic string.

We discussed how our $\Sp(K)$ self-dual QCD model lifts to holomorphic BF theory with matter on twistor space.  This is precisely the theory on a $D5$ brane in the twisted type I string theory.

Working in the setting of twisted string theory, rather than the physical theory, has several advantages. Firstly, the connection with twistor theory is more direct.  In the Penrose transform one needs to work with holomorphic field theories, and these are what appear in twisted string theory.  

Secondly, the backreaction is rather easy to implement in the twisted setting. Indeed, the field sourced by the $D5$ brane does not involve the Beltrami differential field, and so does not change the complex structure of the Calabi-Yau $5$-fold geometry. All it does is introduce a flux, which we can trace by the Penrose transform to the background axion field mentioned earlier. 


\subsection{Connections with the Witten-Berkovits Twistor String}

Let us briefly mention a version of our story which connects to the Witten-Berkovits twistor string.  In \cite{Witten:2003nn,Berkovits:2004hg} Witten and Berkovits studied a twistor realization of self-dual $\mc{N}=4$ super Yang-Mills theory.  That was obtained as holomorphic Chern-Simons theory on the supermanifold $\PT^{3|4}$, which is the total space of the odd vector bundle $\Pi\Oo(1)^4$ over $\PT$.  It was then argued that $\mc{N}=4$ super Yang-Mills is part of a topological string on twistor space, whose closed string sector gives conformal supergravity.

One can ask how this construction relates to the twisted string theories discussed above.  The relationship is quite simple.  As is very well known in the mathematics community \cite{Coates:2001ewh, beilinson1988koszul},  the topological string on the total space of an odd vector bundle $\Pi V$ on a manifold $Y$ is the same as the topological string on the total space of the even vector bundle $V^\ast$. 

Let us derive this correspondence for the topological $B$-model when the target is an odd vector space $\Pi V$. The category of branes is the category of modules over the exterior algebra $\wedge^\ast V$.  Koszul duality \cite{beilinson1988koszul} gives an equivalence\footnote{There are some minor technicalities in this statement. To be precise, one should use graded modules, or else the completed symmetric algebra.} of derived categories
\be
	\wedge^\ast V-\op{mod} \simeq \Sym^\ast V^\ast -\op{mod}\,.
\ee
On the left hand side we have branes on the total space of $\Pi V$, and on the right hand side we have branes on $V^\ast$. 

This equivalence has the flavour of a $T$-duality.  It takes the space $\wedge^\ast V$, viewed as a module over itself, to the rank one module $\C$ over the Koszul dual algebra  $\Sym^\ast V^\ast$.   In other words, it takes a space filling brane on $\Pi V^\ast$ to a brane supported at the origin in $V$.  

What we have learned is that the Witten-Berkovits twistor string is equivalent to a topological $B$-model on the Calabi-Yau $7$-fold
\be \Oo(-1)^4 \to \PT\,. \ee
The space-filling brane on $\PT^{3 \mid 4}$ supporting the twistor uplift of self-dual $\mc{N}=4$ super Yang-Mills becomes the brane wrapping the zero section
\be \PT\subset\Oo(-1)^4\,. \ee
It is immediate that the fields of this zero section brane match the fields of the space-filling brane on $\PT^{3 \mid 4}$.  In general, in the $B$-model, the space of fields on a brane living on $Y \subset X$ is given by 
\be \op{Ext}^\ast_{\Oo_X}(\Oo_Y,\Oo_Y)\,. \ee
It is a standard (and very old) result in math that if $X$ is the total space of a vector bundle $V \to Y$, then this Ext group quasi-isomorphic to the Dolbeault complex of $Y$ with values in the exterior algebra of $V$.  If $V = \Oo(-1)^4$ over $\PT$, this recovers the field content of the space-filling brane on $\PT^{3 \mid 4}$.

Witten and Berkovits also consider closed string states. The closed string states they consider are only a small part of the full spectrum of the topological string that appears in our analysis.\footnote{The equivalence between the $B$-model on the total space of $V$ and $\Pi V^\ast$ holds both for open and closed strings.  This is a formal consequence of the open string argument, because closed string states are the Hochschild cohomology of the open string algebra.  It can also be seen directly, using the description of closed string states as polyvector fields on the target space.

Polyvector fields on $\Oo(-1)^4 \to \PT$ is a very large space, because all fields have functional dependence on the four extra bosonic directions $\Oo(-1)^4$.  When the Penrose transform is applied, the field components obtained by Taylor expanding in these bosonic directions become an infinite tower of fields of increasing spin on space-time.  From the point of view of $\PT^{3 \mid 4}$, these bosonic directions correspond to the symmetric powers of the odd part of the tangent bundle of $\PT^{3 \mid 4}$ that appear in polyvector fields.

Witten and Berkovits did not find such infinite towers of higher spin fields. The only part of the spectrum they considered correspond to vector fields on $\PT^{3\mid 4}$, whereas polyvector fields necessarily include such towers.}  We will leave a detailed analysis of this discrepancy to future work. 


\subsection{Amplitude Computations} \label{subsec:amplitude-comps}

Before we embark on a detailed analysis of our proposed duality, let us give a summary of the results we are able to obtain using these techniques. These fall into two classes.

First, as explained in Sect. \ref{subsec:backgrounds}, amplitudes in self-dual backgrounds coincide with correlators of the large $N$ chiral algebra in which bulk states are given a vacuum expectation. We have verified that at two-points and to first order in the VEV there is a match between the space-time and dual computations for the bulk states which transform trivially under right-handed rotations. (The two-point amplitudes for states transforming non-trivially can then easily be deduced.)

We then focus on the case of a flavour symmetry background whose anti-self-dual curvature is sourced by
\be F(A_0)_{\alpha\beta} = M_{\alpha\beta}\delta_{x=0}. \ee
Here $M_{\alpha\beta}\in\mrm{End}(\C^{16})$ is a triple of traceless matrices. In general, correlations functions of the chiral algebra for the corresponding bulk VEV are ambiguous. These are eliminated if we take $M_{\alpha\beta} = D\xi_\alpha\xi_\beta$ for $D\in\End(\C^{16})$ traceless and $\xi$ a left-handed reference spinor. In this case, we show that the two-point fermion amplitude is
\bea
& \mc{A}^\text{tree}_\text{flavour}(1^+,2^-) = \frac{1}{\la12\ra}\sum_{m=0}^\infty\frac{(-)^{m+1}}{(m+1)!(m!)^2}\la\xi1\ra^m\la\xi2\ra^{m+2}[12]^mD^{m+1} \\
&= - \frac{\la\xi2\ra^2}{\la12\ra} D\,{}_0F_2\big(1,2;-\la\xi1\ra\la\xi2\ra[12]D\big)
\eea
where the positive helicity fermion has momentum $p_1 = \lt_1\lambda_1$ and the negative helicity fermion has momentum $p_2 = \lt_2\lambda_2$. The two fermions, one gluon amplitude is
\bea
&\mc{A}^\text{tree}_\text{flavour}(1^+,2^+,3^-) \\
&= \frac{1}{\la12\ra\la23\ra}\sum_{a,b,c \ge 0}\frac{(-)^{a+b+c}\la\xi1\ra^{a+b} \la\xi2\ra^{a+c}\la\xi3\ra^{b+c+2}}{a!(a+b)!(b+c)!c!(a + b + c + 1)!} [12]^a[13]^b[23]^cD^{a + b + c + 1}
\eea
here the (positive helicity) gluon has momentum $p_2 = \lt_2\lambda_2$ and the positive and negative helicity fermions have momenta $p_1=\lt_1\lambda_1$ and $p_3 = \lt_3\lambda_3$ receptively. We further show that the two fermions, $n$ gluons amplitude is determined recursively by the background deformed celestial operator product.

We also consider the VEV corresponding to Burns space. Again, we find that the conformal block in this background is ambiguous, though this is resolved by taking the space-time source to be
\be W_{\alpha\beta\gamma\eta} = \frac{1}{6}\xi_\alpha\xi_\beta\xi_\gamma\xi_\eta\delta_{x=0} \ee
for a reference spinor $\xi$. In this case, we are able to recover the two-point gluon amplitude on the curved background, matching the analogous result from \cite{Costello:2022jpg,Costello:2023hmi} adapted to our context. We also obtain novel recurrence relations for $n$-point tree and one-loop gluon amplitudes.

Second, were able to extract certain flat space form factors from these amplitudes in self-dual backgrounds. That is, helicity amplitudes of self-dual gauge theory in the presence of local operators (or integrals thereof). Form factors can be computed using the chiral algebra OPE \cite{Costello:2023vyy}, which we also reproduce in the large $N$ chiral algebra. However, we take a completely different approach in this paper, extracting them from amplitudes in self-dual backgrounds.

From the flavour background we're able to recover a restricted class of form factors in gauge theory with bifundamental matter.  We do this by treating the background flavour symmetry field, which lives in some $\mf{sl}(R) \subset \mf{sl}(16)$, as dynamical, and inserting operators related to those which source the background described above.    If $B_0$ denotes the Lagrange multiplier field of the $\mf{sl}(R)$ gauge theory, we compute the form factor of the operator  $\op{tr}(B_0^2)$. 

Since the operator $\half\op{tr}(B_0^2)$ is the first order deformation from self-dual to full QCD, we can exploit this to find a simple explicit formula for a two-loop partial amplitude for QCD in $\fsp(K)\oplus\mf{sl}(R)$  gauge theory coupled to bifundamental massless fermions (Fig. \ref{fig:twoloop-intro}).  Our computation also holds for $\mf{sl}(K) \oplus \mf{sl}(R)$ gauge theory with bifundamental fermions, and since it is more standard to study $\mrm{SU}(K)$ gauge theories we will describe the result in this context.

We compute the $n$-point two-loop all-plus amplitude where all external gluons are in the $\mf{sl}(K)$ factor and with an internal $\mf{sl}(R)$ gluon exchange.  The trace ordered amplitude at $n$-points is
\bea
&\ip{\frac{1}{2}\int_{\R^4}\op{tr}(B_0^2)\mid \J(\lt_1;z_1) \dots \J(\lt_n;z_n)} \\
&= - \frac{R^2-1}{2R(4\pi)^4}\sum_{1 \le i < j < k < l \le n}\frac{\la ij\ra[jk]\la kl\ra[li] + [ij]\la jk\ra[kl]\la li\ra}{\la12\ra\dots\la n1\ra}\,.
\eea
matching the result appearing in the amplitudes literature at four-points \cite{Bern:2002zk,Dixon:2024mzh}.\footnote{We are very grateful to Anthony Morales for pointing out an error in the original version of this formula.}

\begin{figure}[ht]
    \centering
        \includegraphics{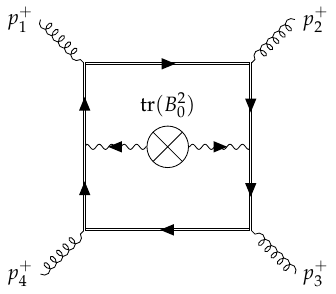}
    \caption{This figure depicts a two-loop diagram with external gluons in $\mf{sp}(K)$ and a gluon exchange in $\mf{sl}(R)$. The $\mf{sl}(R)$ gluons are represented by photon propagators with arrows indicating the flow of helicity. \label{fig:twoloop-intro}}
\end{figure}

From the Burns space amplitudes we're able to recover tree and one-loop form factors for the ANE (average null energy operator) in self-dual gauge theory.  For a null vector $v$, the average null energy operator is the integral 
\be
\int_{t \in \R} T_{vv}(t v)
\ee
of the $vv$ component of the stress tensor along the null ray $v$.  The insertion of the ANE operator amounts to a first order deformation of the flat Lorentzian metric which, in lightcone coordinates $x^{\pm},\ x_i$, take the form $\d x_+^2 \delta_{x_i = 0} \delta_{x_+ = 0}$.  

Fixing some null vector $v$
\bea
&\ip{\int_{t\in\R}\d t\,T_{vv}(tv) \mid \Jt(\lt_1;z_1) \dots \J(\lt_n;z_n)}_\text{tree} = - \frac{2}{\ip{12} \dots \ip{n1}}\Big(\sum_{k=1}^n\ip{1k}[k|v|1\rangle\Big)^2 \\
&\ip{\int_{t\in\R}\d t\,T_{vv}(tv) \mid \J(\lt_1;z_1) \dots \J(\lt_n;z_n)}_\text{one-loop} \\
&= - \frac{1}{\pi^2\ip{12}\dots\ip{n1}}\sum_{1\le i < j \le n}\frac{[ij]}{\ip{ij}}\Big(\sum_{k=1}^n\ip{ik}[k|v|j\ra\Big)^2\,.
\eea

By averaging over the location where we insert the stress tensor, we can reinterpret this result as the scattering of gluons off a single partially off-shell graviton.  In light cone coordinates, the graviton is a first order deformation of the flat metric to a metric of the form $H(x_+,x_i) \d x_+^2$ for an arbitrary function $H$ independent of $x_-$.  A metric like this is called a pp-wave.  The tree and one-loop scattering amplitudes off a graviton of this form are
\bea
&\mc{A}^{\text{tree}} ( 1^-, 2^+, \dots,n^+ \mid H(x) (\d x_+)^2 ) \\
&= - \frac{2}{\ip{12} \dots \ip{n1}}\Big(\sum_{k=1}^n\ip{1k}[k|\d x_+ |1\rangle\Big)^2 \delta_{P\cdot\d x_+ = 0}  \what{H} (P)\,, \\
&\mc{A}^{\text{one-loop}} ( 1^+, 2^+,\dots,n^+ \mid H(x) (\d x_+)^2 ) \\
&= - \frac{1}{\pi^2\ip{12}\dots\ip{n1}}\sum_{1\le i < j \le n}\frac{[ij]}{\ip{ij}}\Big(\sum_{k=1}^n\ip{i k }[k|\d x_+|j\ra\Big)^2 \delta_{P \cdot \d x_+ = 0} \what{H}(P)
\eea
where $\what{H}$ is the Fourier transform of $H$ and $P$ is the total momentum of the gluons. 

While the tree-level identity is independent of the fermionic matter content, the one-loop result relies on our choice of gauge group and matter representation. Although our calculation is performed in Euclidean signature, form factors in twistorial theories of SDQCD are rational functions with poles on the light cone \cite{Costello:2021bah}. They therefore analytically continue to Lorentzian signature.


\section{The Algebra on the Defect}

Let us initiate our analysis of the chiral algebra which at large $N$ will reproduce the celestial chiral algebra of the $\Sp(K)$ SDQCD model.  Recall that we are considering the type I (topological) string on the geometry 
\be
X = \Oo(-1) \oplus \Oo(-3) \to \PT = \Oo(-1) \oplus \Oo(-3) \oplus \Oo(1) \oplus \Oo(1) \to \CP^1.
\ee
We have $D5$ branes on $X$ arranged according to the following table:
\be
\begin{array}{|c|c|c|c|c|c|}
\hline & \Oo(-1) & \Oo(-3) & \Oo(1) & \Oo(1) & \CP^1 \\
\hline N\ D5 & \times & \times & & & \times \\
\hline K\ D5' & & & \times & \times & \times \\
\hline
\end{array}
\ee

The worldvolume of the stack of $N$ $D5$ branes is 
\be 
	 Y = \Oo(-1) \oplus \Oo(-3) \to \CP^1\,.
\ee
The canonical bundle of this variety is $\Oo(2)$. The normal bundle to this variety in the Calabi-Yau $5$-fold is $\Oo(1) \oplus \Oo(1)$. 

The $D5$ branes support a holomorphic BF  gauge theory with $\Sp(N)$ gauge group, with certain matter. The complete field content is
\bea \label{eq:BV_fields}
&\mc{A}\in \Pi \Omega^{0,\ast}(Y,\mf{sp}(N))\,,\qquad \mc{B}\in \Pi \Omega^{0,\ast}(Y,\Oo(2)\otimes\mf{sp}(N))\,, \\
&\Gamma_f\in\Omega^{0,\ast}(Y, \Oo \otimes\mrm{F}_N)\,,\qquad
\til{\Gamma}^f\in\Omega^{0,\ast}(Y, \Oo(2) \otimes  \mrm{F}_N)\,, \\
&\Phi_\da\in\Omega^{0,\ast}(Y,\Oo(1)\otimes\wedge^2_0\mrm{F}_N)\,.
\eea
The field content is written in the BV formalism, including all ghosts and anti-fields. Here $\mrm{F}_N$ is the fundamental representation of $\mf{sp}(N)$. The index $f$ is for the fundamental representation of the $\sl(16)$ flavour symmetry, and $\da$ is a two-component spinor index for the $4d$ space-time Lorentz group.  The field content is only given a $\Z/2$ grading, because the twisted type I string theory has a ghost number anomaly. The symbol $\Pi$ indicates parity change.

The action functional is
\be\label{eq:act_D5} \frac{\i}{2\pi}\int_Y\Tr\big(\mc{B}\wedge F^{0,2}(\mathcal{A})\big) + \frac{1}{2}\la\Phi_\da,\dbar_{\mc{A}}\Phi^\da\ra_{\wedge^2_0\mrm{F}_N} + \la\til{\Gamma}^f,\dbar_{\mc{A}}\Gamma_f\ra_{\mrm{F}_N}\,. \ee
Note that in each term of the action, the integrand is a section of $\Oo(2) = K_Y$. 

We couple this system to a collection of bifundamental fermions
\be \Psi_I\in\Pi\Omega^{0,\ast}(\CP^1,\Oo(-1)\otimes\mrm{F}_N) \ee
where $I$ is an index for the fundamental of $\Sp(K)$.  The fermions are coupled by
\be \label{eq:act_D55'}\frac{\i}{4\pi}\int_{\CP^1} \Omega^{IJ}\la\Psi_I ,\dbar_\mc{A}\Psi_J\ra_{\mrm{F}_N}\,. \ee
The main result of this section is the BRST cohomology of the space of local operators of the $D5$ system with the defect, as $N \to \infty$.  We will use standard techniques to compute the BRST cohomology.  We will find that, after modding out by local operators that come from the bulk of the $D5$ system, the spectrum of operators matches exactly with the spectrum of states in the $4d$ $\op{Sp}(K)$ gauge theory.
 
We can break the BRST operator into two terms: the free part, which is the $\dbar$ operator, and the interacting part. We first compute the BRST cohomology of the free term.

Each field lives in the Dolbeault complex. Since the Dolbeault cohomology is locally given by holomorphic functions, only the components that live in $\Omega^{0,0}$ will be present in the free BRST cohomology of local operators. Further, anti-holomorphic derivatives of these components will be BRST exact. Thus, after taking cohomology with respect to the free BRST operator, we are left with the operators (and their derivatives) listed in Table \ref{tab:free-chomology}.

\begin{table}[ht]
    \centering
    \begin{tabular}{ c c c c c }
	\toprule
    Operator & Parity & Spin & Representation & Allowed derivatives \\
    \midrule
	$\c = \mc{A}^{0}$ & Fermionic & $0$ & $\mf{sp}(N)$ & $\p_{w_i}$, $\p_z$ \\
	$\b = \mc{B}^{0}$ &  Fermionic & $-1$ & $\mf{sp}(N)$ & $\p_{w_i}$, $\p_z$ \\
	$\phi_\da  = \Phi_\da^0$ &  Bosonic  &  $-\half$  &  $\wedge^2_0\mrm{F}_N$ & $\p_{w_i}$, $\p_z$ \\
	$\gamma_f  = \Gamma_f^{0}$ &  Bosonic & $0$& $\mrm{F}_N$ &  $\p_{w_i}$, $\p_z$ \\
	$\til{\gamma}^f  = \til{\Gamma}^{f0}$ &  Bosonic & $-1$ & $\mrm{F}_N$ &  $\p_{w_i}$, $\p_z$ \\
    $\psi_I = \Psi_I^{0}$ & Fermionic & $\half$ &$\mrm{F}_N$ & $\p_z$ \\
    \bottomrule
    \end{tabular}
    \caption{Constituents of local operators after taking cohomology with respect to $\dbar$}
    \label{tab:free-chomology}
\end{table}

Here we use coordinates $z$ on the $\CP^1$ where the defect lives, and $w_1$, $w_2$ on the $\Oo(-1)$, $\Oo(-3)$ fibres.  Differentiating with respect to $w_1$ reduces the spin by $\half$, and differentiating with respect to $w_2$ reduces the spin by $\tfrac{3}{2}$.

The interacting part of the BRST operator is given (at tree-level) by
\be
Q_\mrm{BRST} =  \frac{1}{2}f_{\sfa\sfb}^{~\,~\sfc}\c^\sfa\c^\sfb \dpa{\c^\sfc} + \dots\,. 
\ee
We do not write the contractions of the $\Sp(N)$ indices explicitly.  This acts on a cochain complex built from $\Sp(N)$ invariant words in the fundamental fields and their holomorphic derivatives. 


\subsection{BRST Cohomology at Large \texorpdfstring{$N$}{N}}

The BRST cohomology at tree-level is the Lie algebra cohomology of a certain infinite-dimensional Lie algebra. At infinite $N$, this cohomology can be computed quite easily using standard techniques of homological algebra.  Here we will only state the result, leaving the details to Appendix \ref{app:tree-cohomology}.

The result is presented in Table \ref{tab:celestial-ops} below.  The table lists a  basis\footnote{More precisely, we must take the single-string operators listed in the table and all of their $z$ derivatives to find a basis.} of the single-string part of BRST cohomology involving $\psi$.   Further, we can identify these operators precisely with the operators of the celestial chiral algebra of the integrable $\Sp(K)$ theory. 

\begin{table}[ht] 
\centering
    \begin{tabular}{ c c c c c }
	\toprule
    Celestial operator & Gauge theory operator & Spin & Parity & Representation  \\
    \midrule 
	$\M_{f,I}[k,l]$ & $\gamma_f\phi^{\done(k}\phi^{\dtwo\,l)}\psi_I$ & $\half - \half(k+l)$ & Fermionic & $\C^{16}\otimes\mrm{F}_K$ \\
	$\til{\M}^f_I[k,l]$ & $\til{\gamma}^f\phi^{\done(k}\phi^{\dtwo\,l)}\psi_I$ & $ -\half - \half(k+l)$ & Fermionic & $\C^{16} \otimes \mrm{F}_K$ \\
	$\J_{IJ}[k,l]$ & $\psi_I\phi^{\done(k}\phi^{\dtwo\,l)}\psi_J$ & $1 -\half(k+l)$ & Bosonic & $\mf{sp}(K)$ \\
	$\M_{IJ}[k,l]$ & $\psi_I\phi^{\done(k}\phi^{\dtwo\,l)}\p_{w_1}\c\psi_J$ & $ \half -\half(k+l)$ & Fermionic & $\wedge^2_0\mrm{F}_K$ \\
	$\til{\M}_{IJ}[k,l]$ & $\psi_I\phi^{\done(k}\phi^{\dtwo\,l)}\p_{w_2}\c\psi_J$ & $-\half - \half(k+l)$ & Fermionic & $\wedge^2_0 \mrm{F}_K$  \\
	$\Jt_{IJ}[k,l]$ & $\psi_I\phi^{\done(k}\phi^{\dtwo\,l)}\p_{w_1}\c\p_{w_2}\c\psi_J$ & $-1-\half(k+l)$ & Bosonic & $\mf{sp}(K)$ \\
    \bottomrule 	
    \end{tabular}
    \caption{Single-string states involving $\psi$} 
    \label{tab:celestial-ops}
\end{table}

Here $\phi_\done^{(k}\phi_\dtwo^{l)}$ means the symmetrized product.  Further, in classical BRST cohomology, we can commute $\p_{w_i} \c$ past $\phi_\da$. This is because the BRST variation of $\p_{w_i}\phi_\da$ is $\p_{w_i} [\c,\phi_\da]$, and  because we are taking relative Lie algebra cohomology, we drop any terms which involve a $\c$ without derivatives.

\begin{figure}[ht]
    \centering
        \subfloat{%
        \includegraphics{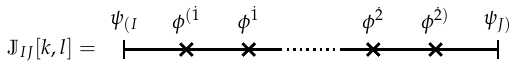}%
        } \\
        \subfloat{%
        \includegraphics{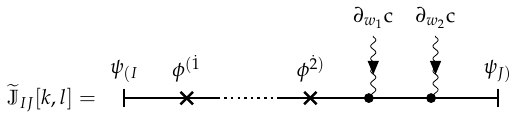}%
        }
    \caption{Diagrammatic representation of the gluon states $\J_{IJ}[k,l]$ and $\Jt_{IJ}[k,l]$} \label{fig:states}
\end{figure}

This identification gives exactly the spectrum of the celestial chiral algebra for the integrable $\Sp(K)$ gauge theory, including states of both helicities. We can represent these states diagrammatically. For example, gluon states of both helicities are illustrated in Fig. \ref{fig:states}.


\subsection{Bulk States}

We have not so far given an explicit description of the bulk states which are removed in our construction.  Here we will enumerate such states at tree-level.  We will not provide the details of the cohomology calculation, which is almost identical to the one presented above and to many other such computations in the literature (see, e.g., \cite{Costello:2018zrm}).  Instead, we will use the general fact that the bulk states on the $N$ $D5$ branes naturally pair with fields of the ten-dimensional string theory, as outlined in \cite{Costello:2019jsy}.

These fields consist of type I Kodaira-Spencer gravity and $\mf{so}(32)$ holomorphic Chern-Simons theory.  The fields of type I Kodaira-Spencer gravity belong to
\be \Pi\PV^{1,\ast}_\p(X)\oplus\Pi\PV^{3,\ast}_\p(X)
\,, \ee
Decomposing $TX = TY\oplus\mc{N}$ where $\mc{N} = \Oo(1)\oplus\Oo(1)\to X$ (and abusing notation slightly by writing $TY$ for the pullback to $X$) we have
\bea \label{eq:decompose-polyvectors}
&\PV^{1,\ast}(X) = \Omega^{0,\ast}(X,TY)\oplus\Omega^{0,\ast}(X,\mc{N})\,, \\
&\PV^{3,\ast}(X) = \Omega^{0,\ast}(X,\wedge^3TY)\oplus\Omega^{0,\ast}(X,\wedge^2TY\otimes\mc{N})\oplus\Omega^{0,\ast}(X,TY\otimes\wedge^2\mc{N})\,.
\eea
Let us fix meromorphic sections $\d z\d w_1\d w_2$ of the canonical bundle on $Y$ and $\d v^\done\d v^\dtwo$ of the canonical bundle on the fibres of $\mc{N}$ such that they wedge together to give $\Omega$.  We can decompose the divergence with respect to the above splitting as $\p = \p_Y + \p_\mc{N}$.  In general the divergence will relate distinct summands of \eqref{eq:decompose-polyvectors}, allowing us to fix the failure of $\p_Y$ closure in one summand by constraining $\p_\mc{N}$ in another.  In this way we can decompose, e.g.,
\be \PV^{1,\ast}_\p(X) = \Omega^{0,\ast}(X,TY)^\prime\oplus\Omega^{0,\ast}_{\p_\mc{N}}(X,\mc{N})\,. \ee
Here $\Omega^{0,\ast}(X,TY)^\prime$ indicates that we've added a term in $\Omega^{0,\ast}(X,\mc{N})$ so that the sum is $\p$-closed.  Elements of $\Omega^{0,\ast}_{\p_\mc{N}}(X,TY\otimes\wedge^2\mc{N})$ require extra care.  $\p_\cN$-closure requires that they be constant along the fibres of $\mc{N}$, and so pushforward to $Y$. Furthermore, since $\wedge^3\mc{N}=0$ we cannot compensate the failure of $\p_Y$-closure and instead must impose it directly. This condition is independent of our choice of meromorphic sections, since $\wedge^2\mc{N} = K_Y$ where again we abuse notation by pulling back the latter to $X$.

Next consider the fields of holomorphic Chern-Simons theory for the chosen $\mf{so}(32)$ bundle, which we took to be $\Oo(1)^{16}\oplus\Oo(-1)^{16}$ for $\Oo(1)^{16}$ and $\Oo(-1)^{16}$ isotropic subbundles.  These decompose into
\be \Pi\Omega^{0,\ast}(X,\mf{sl}(16))\oplus\Pi\Omega^{0,\ast}(X,\Oo(2)\otimes\wedge^2\C^{16})\oplus\Pi\Omega^{0,\ast}(X,\Oo(-2)\otimes\wedge^2\C^{16})\,. \ee

The results are presented in Table \ref{tab:bulk-states}.  The states have been organized according to which bulk fields they couple to. To keep the expressions compact, we use the coordinates $w_1, w_2, w_3$ on the brane, instead of $w_1,w_2,z$.  We also do not include the derivatives on the brane in the table.

\begin{table}[ht]
	\centering
\begin{tabular}{c c c}
    \toprule
    Bulk state & Spin & Corresponding bulk field \\
    \midrule
    $\Tr(\phi^{(\da_1}\dots\phi^{\da_n)})$ & $-\half n$ & $ 
    \Pi\Omega^{0,\ast}_{\p_\cN}(X,\mc{N})$ \\
    $\Tr(\phi^{(\da_1}\dots\phi^{\da_n}\phi^{\db)}\p_{w_i}\phi_\db) + \dots$ & $-\tfrac{3}{2}-\half n,-\tfrac{5}{2}-\half n,-\half n$ & $ 
    \Pi\Omega^{0,\ast}(X,TY)^\prime$ \\
    $\Tr(\phi^{(\da_1}\dots\phi^{\da_n)}\p_{w_{[i}}\c\p_{w_{j]}}\c)$ & $-\tfrac{1}{2}-\half n,\tfrac{1}{2}-\half n,-2-\half n$ & $ 
    \Pi\Omega^{0,\ast}_{\p_\cN}(X,\wedge^2TY\otimes\cN)^\prime$ \\
    $\Tr(\phi^{(\da_1}\dots\phi^{\da_n}\phi^{\db)}\p_{w_{[i}}\phi_\db\p_{w_j}\c\p_{w_{k]}}\c) + \dots$ & $-2-\half n$ & $ 
    \Pi\Omega^{0,\ast}(X,\wedge^3TY)^\prime$ \\
    $\gamma_f\phi^{(\da_1}\dots\phi^{\da_n)}\gamma_g$ & $-\half n$ & $\Pi\Omega^{0,\ast}(X,\Oo(2)\otimes\wedge^2\C^{16})$ \\
    $\til\gamma^f\phi^{(\da_1}\dots\phi^{\da_n)}\til\gamma^g$ & $-2-\half n$ & $\Pi\Omega^{0,\ast}(X,\Oo(-2)\otimes\wedge^2\C^{16})$ \\
    $\gamma_f\phi^{(\da_1}\dots\phi^{\da_n)}\til\gamma^g$ & $-1-\half n$ & $\Pi\Omega^{0,\ast}(X,\mf{sl}(16))$ \\
    \bottomrule
\end{tabular}
\caption{Bulk states} \label{tab:bulk-states}
\end{table}

For two of the above states we've used $+\dots$ to denote terms involving the antifield $b$ and fundamental bosons $\gamma,\til\gamma$. We include complete expressions for both in Appendix \ref{app:bulk-states}. Furthermore, upon coupling to the defect theory the action of the BRST differential on $\b$ is modified. This necessitates further corrections to these two states involving the fundamental fermions $\psi^I$. These can also be found in Appendix \ref{app:bulk-states}.

The $(-1)-$shifted pairings which determine the kinetic term of type I Kodaira-Spencer gravity and $\so(32)$ holomorphic Chern-Simons theory provide an identification between the chiral algebra states appearing in \ref{tab:bulk-states} and bulk fields. We will see in Sect. \ref{sec:backgrounds} that this fixes the background corresponding to a given state.


\section{Backreaction}

The backreaction of the $N$ $D5$ branes affects the theory on the $K$ $D5'$ branes. To determine this, we will first determine the effect of the backreaction on the full type I topological string.

The fields of the type I topological string can be written as polyvector fields, or as forms, using the fact that we can contract a polyvector field into the holomorphic volume form and turn it into a differential form. 

In differential form language, the closed string fields consist of
\be 
	\op{Ker} \p \subset \Pi\Omega^{4,\ast}(X)\oplus\Pi\Omega^{2,\ast}(X)\,.
\ee
Let us denote the bulk closed string fields by
\be \eta\in\Pi\Omega^{2,\ast}(X)\,,\qquad\mu\in\Pi\Omega^{4,\ast}(X)\,. \ee
The kinetic term\footnote{In this version of Kodaira-Spencer theory, we can remove the annoying $\p^{-1}$ operators by replacing $\mu$ by $\p$ of a field in $\Pi\Omega^{3,\ast}(X)$, which is then taken up to a gauge equivalence. We will not do this, however.} of the bulk theory is
\be \frac{\i}{2\pi}\int_X\mu\p^{-1}\dbar\eta\,. \ee
The coupling to the stack of $N$ $D5$ branes is
\be 2N\bigg(\frac{\i}{2\pi}\bigg)\int_{D5}\p^{-1}\mu\,. \ee
Thus, the field sourced by the $N$ $D5$ branes is a $(2,1)$-form $\eta$  satisfying
\be \dbar\eta = 2N \delta_{D5} \ee
where $\delta_{D5}$ indicates a $\delta$-function supported on the location of the $D5$ branes.

An explicit formula for $\eta$ can of course be given, using (for instance) the Bochner-Martinelli kernel. We will not need this, however.

We are interested in the effect of this on the $D5'$ system. The stack of $K$ $D5'$s lives on $\PT$. The worldvolumes of the $D5$ and $D5'$ systems intersect transversely in the zero section $\CP^1 \subset \PT$.  This tells us that, when restricted to $\PT$, the field sourced by the $D5$ branes is a $(2,1)$-form satisfying
\be \dbar \eta = 2N\delta_{\CP^1}\,. \ee
This $(2,1)$-form couples to the fields of the theory on the $D5$ brane as follows. Let
\be a\in\Omega^{0,1}(\PT, \mf{sp}(K)) \ee
be the gauge field on twistor space. 

Then, the coupling is
\be \frac{\i}{2\pi}\int_\PT\eta\op{CS}(a) = \frac{\i}{4\pi}\int_\PT\eta\op{tr}(a\p a)\,. \ee
To make this gauge invariant (away from the $\CP^1$ where $\eta$ is not Dolbeault closed), we also need to modify the action of gauge symmetries on the $b$-field so that $b$ varies by
\be \delta b = [\chi,b] + \eta\p\chi \ee
where $\chi$ is the gauge parameter and we treat $b$ as a $(3,1)$-form.  

The action including the modification to gauge symmetry is conveniently encoded by letting ${\bf a}\in\Omega^{0,\ast}(\PT,\mf{sp}(K))[1]$ be the Dolbeault field including ghosts and anti-fields, and writing
\be \frac{\i}{4\pi}\int_\PT\eta\op{tr}(\bf{a}\p\bf{a})\,. \ee
The dependence of this on the fields in $\Omega^{0,2}$ and $\Omega^{0,0}$ have the effect of modifying the action of gauge symmetries on $b$. 

This expression is gauge invariant on the locus where $\dbar\eta = 0$, which is the complement of the $\CP^1$ where $\eta$ is sourced.  The gauge variation of our Lagrangian including this locus is
\be \frac{\i}{4\pi}\int_\PT\dbar\eta\op{tr}({\bf a}\p{\bf a}) = \frac{\i N}{\pi}\int_{\CP^1}\op{tr}(\chi\p a) \ee
where $\chi = \bf{a}^0$ is the ghost.  This expression is precisely the anomaly we find when we couple the $D5-D5'$ bifundamental fermions on the $\CP^1$.  In other words, the backreaction of the $D5$ system is precisely what is needed to cancel this anomaly by anomaly inflow.

Let us now explain how this effects the $4d$ system.  This was already computed in \cite{Costello:2021bah,Costello:2022wso}: in four-dimensions, introducing an $\eta$ as above means that we add to the SDYM Lagrangian a term with a varying $\theta$-angle like 
\be - \frac{N}{4\pi^2}\int_{\R^4}\log\norm{x}^2\op{tr}(F\wedge F) \label{eq:theta} \ee
where $A$ is the four-dimensional gauge field, and $\op{tr}(F(A)^2)$ is the topological term in the Yang-Mills Lagrangian.

We conclude that the chiral algebra built from the $D5$ system (with its defect) will produce scattering amplitudes in self-dual gauge theory in the presence of the varying $\theta$ angle \eqref{eq:theta}.

For many computational purposes, it is convenient to turn the backreaction off.  In diagrammatic terms this is very simple.  Correlators (or OPEs) can be computed with Feynman diagrams in double line notation.  Faces of such diagrams which have no external states come with a factor of $N$.  To turn off the backreaction we simply throw such diagrams away.

In more conceptual terms, this procedure amounts to working with the super-algebra $\mf{osp}(2N|N)$ instead of $\mf{sp}(N)$.  This superalgebra is the symmetries of the graded vector space $\C^{2 N \mid 2 N}$ preserving a graded-symmetric non-degenerate pairing.  With the superalgebra $\mf{osp}(2M|N)$, the backreaction comes with a coefficient of $2(N - M)$, so setting $M = N$ turns it off. 


\section{Computing OPEs using Feynman Diagrams} \label{sec:how-to-do-OPEs}

The $D5$ brane system, with its defects, appears at first sight to be more complicated than the purely two-dimensional systems discussed in \cite{Costello:2022jpg,Costello:2023hmi}.  However, OPE computations in this system turn out to be not that hard: it turns out that only some very simple Feynman diagrams contribute to OPE computations.

All operators are words in the bulk  fields $\phi_{\dot \alpha}$, $\gamma_f$, $\til{\gamma}^f$, $\c$, and the defect fermions $\psi_I$.  OPEs between operators are computed using Feynman diagrams. For the $\psi$ field this is easy: the propagator is simply $1/z$. The only vertices involving $\psi$ connect it to the descendent $\mc{A}\in\Omega^{0,1}(Y)$ of $\c$, and not to $\c$ itself.  For the bulk fields, things are a little more complicated, and bulk vertices play a role.  However, only a very limited class of Feynman diagrams appear in OPE computations.

Recall that all bulk fields are part of multiplets that live in the Dolbeault complex of the $3$ complex dimensional space-time.  Thus, $\phi^\da$ is the $0$-form component of
\be
\Phi^\da \in \Omega^{0,\ast}(Y, \Oo(1) )\,.
\ee
Working locally in a patch on $Y$ with coordinates $w_1,w_2$, and $w_3 = z$, the propagator for any bulk field is the Bochner-Martinelli kernel
\be
\frac{1}{4\pi^2}\frac{\eps^{ijk}(\wbar_i - \wbar'_i)\d(\wbar_j  - \wbar'_j)\d(\wbar_k - \wbar'_k)}{\norm{w - w'}^6}
\ee
where $w,w'$ refer to coordinates on two copies of the space-time.  The Bochner-Martinelli kernel is the Green's function for the operator $\dbar$ (in a certain gauge).

Thus, the propagator connects a $(0,0)$-form with a $(0,2)$-form, and a $(0,1)$-form with a $(0,1)$-form.  Since the $\phi^\da$ fields are the $(0,0)$-form components, there is no propagator connecting them directly.

\begin{figure}[ht]
\centering
    \subfloat{%
    \includegraphics{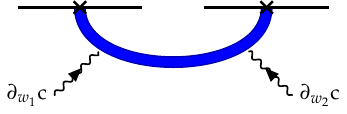}%
    } \hfil
    \subfloat{%
    \includegraphics{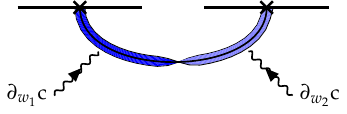}%
    }
    \caption{Planar and non-planar terms in the $\phi-\phi$ propagator, emitting two $\p_{w_i}\c$ fields. There are further contributions from the $\p_{w_i}\c$ fields attaching to the other side of the propagator.} \label{fig:phipropagator}
\end{figure}

However, there is an indirect propagator, depicted in Fig. \ref{fig:phipropagator}, where $\phi^\da$ and $\phi^\db$ are connected by a propagator which has two external lines coming off, labelled by $\c$.  We will refer to this as the $\phi-\phi$ propagator. In fact, this is the only possible propagator between them:
\begin{lemma}\label{lem:prop_phi}
There is only one Feynman diagram that connects two $\phi$ operators, which is the one described above. Similar statements hold for the $\gamma-\til{\gamma}$, $\b-\c$ propagator. There is no Feynman diagram that connects other pairs of fields.
\end{lemma}
\begin{proof}
Consider an arbitrary Feynman diagram with $n$ internal vertices and $k$ internal propagators, in addition to the $2$ propagators connecting the external $\phi$ fields. Each vertex has Dolbeault degree $(0,3)$, while each propagator is a $(0,2)$-form. A valid Feynman must produce $(0,0)$-form on all external legs. Thus, we must satisfy $3n - 4 - 2k = 0$. Suppose the number of $\op{tr}(\mc{B}[\mc{A},\mc{A}])$ vertices is $n_1$ and the number of $\la\Phi^\done,\{\mc{A},\Phi^\dtwo\}\ra_{\wedge^2_0\mrm{F}_N}$ vertices is $n_2$, and suppose there are $k_1$ $\mc{A}-\mc{B}$ propagators and $k_2$ $\Phi-\Phi$ propagators. We have $n_1 + n_2 = n$, $k_1+k_2 = k$. By counting the number of fields and contractions, we have the following inequalities
\be
k_1 \leq n_1,\quad k_1 \leq 2n_1 + n_2,\quad 2k_2+2 \leq 2n_2\,.
\ee
This implies $2k+2\leq2n$. Together with the constraint imposed by the Dolbeault degree, we have the only solution $n = 2$, $k = 1$. This is exactly the Feynman diagram depicted in Fig. \ref{fig:phipropagator}. Similar analysis holds for the $\gamma-\tilde{\gamma}$ or $\b-\c$ Feynman diagrams, where we have the same inequality $2k+2\leq2n$. In all other cases, the inequality becomes $2k+3 \leq 2n$ or $2k+4 \leq 2n$, which has no solution.
\end{proof}

In the appendix, the $\phi-\phi$ propagator is computed explicitly, locally on $Y$ using the gauge associated to the flat metric.  
\begin{proposition}
The $\phi-\phi$ propagator is
\bea
&\phi^{\done pq}(0)\phi^{\dtwo rs}(0,z) = - \frac{1}{\pi^2z}\int_{0\leq x\leq y\leq1}\d x\d y\, \big(\p_{w_1}\c^{[r[p}(0,xz)\p_{w_2}\c^{q]s]}(0,yz) \\
&+ \omega^{[r[p}\p_{w_1}\c^{q]}_{~\,~t}(0,xz)\p_{w_2}\c^{s]t}(0,yz)\big)\,.
\eea
The singular part of this expression is
\be \label{eq:phiOPE-singular}
\phi^{\done pq}(0)\phi^{\dtwo rs}(z,0)\sim - \frac{1}{2\pi^2z}\big(\p_{w_1}\c^{[r[p}\p_{w_2}\c^{q]s]} + \omega^{[r[p}\p_{w_1}\c^{q]}_{~\,t}\p_{w_2}\c^{s]t}\big)(0)\,.
\ee
\end{proposition}
In any OPE computation, the singular terms in the OPE only depend on the field theory locally, for example in the patch with coordinates $z,w_1,w_2$.  We have written down the $\phi-\phi$ OPE as computed locally in the gauge associated to a flat metric, but including the non-singular terms.  This will not be sufficient to compute correlation functions, which will require the full propagator on the curved geometry.  It will, however, be sufficient to compute the singular parts of all OPEs.  The non-singular terms in the propagator can contribute to singular OPEs by combining with singularities coming from other propagators.  

The $\gamma-\til{\gamma}$ propagator has an identical form, with different colour factors, because the $\Gamma$, $\til{\Gamma}$ fields on $Y$ have the same kinetic term and coupling to the gauge field as the $\Phi^\da$ fields do.  
\begin{figure}[ht]
    \centering
        \includegraphics{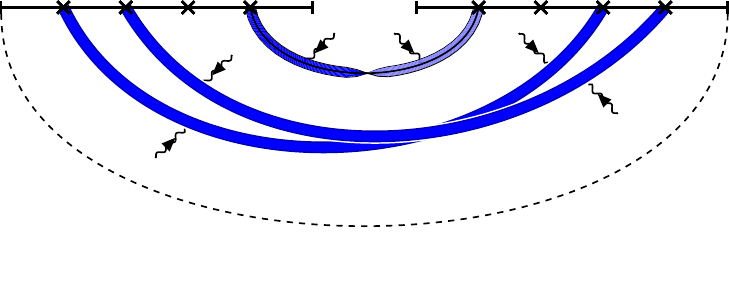}
    \caption{A fairly generic Feynman diagram, including $\phi-\phi$ propagators (blue lines) and a $\psi-\psi$ propagator (dashed line).  Since the gauge group is $\mf{sp}(N)$, the $\phi-\phi$ propagator can be a twisted propagator in double line notation, as drawn.  Each $\phi-\phi$ propagator emits two $\p_{w_i}\c$ fields.} \label{fig:diagramgeneric}
\end{figure}
It turns out that all Feynman diagrams contributing to OPEs are built from Wick contractions involving  the $\phi-\phi$, $\gamma-\til{\gamma}$ and $\psi-\psi$ propagators:  
\begin{lemma} \label{lem:Feyn_rule}
	Consider two single-string operators in the algebra $\mc{A}^{N,K}$, as listed in Table \ref{tab:celestial-ops}.  Then, the most general Feynman diagram that contributes to the OPE between them is obtained by performing some number of the following contractions:
	\begin{enumerate}
		\item Contracting a $\psi$ field in one operator with a $\psi$ field in the other operator, with the propagator $1/z$.
		\item Contracting a $\gamma$ field in one operator with a $\til \gamma$ field in the other operator, using the $\gamma-\til{\gamma}$ propagator described above.
		\item Contracting a $\phi^\done$ field in one operator with a $\phi^\dtwo$ field in another operator, using the $\phi-\phi$ propagator described above.
	\end{enumerate}
\end{lemma}
A generic Feynman diagram built using these rules is illustrated in Fig. \ref{fig:diagramgeneric}. 
\begin{proof} 
	We need to check that no other Feynman diagrams contribute.  Suppose there was a connected Feynman diagram with two vertices given by the operators we are studying. To contribute to the OPE, all external lines in the diagram must be labelled by Dolbeault $0$-forms.  First, note that the operators do not involve $\mc{B}$, only $\psi,\phi,\gamma,\til{\gamma}$ and $\p_{w_i}\c$.  Further, these fields only appear in a vertex together with an $\mc{A}$ field.  The propagator connects an $\mc{A}$ field with an $\mc{B}$ field, and the only vertex involving $\mc{B}$ is $\op{tr}(\mc{B}[\mc{A},\mc{A}])$. Therefore we are in the situation discussed in Lemma \ref{lem:prop_phi}, and where the only possible diagram involves propagators connecting a field in one operator with a field in the other operator, with some trees branch off on which we place $\mc{A}$ fields. Now the argument there applies to show that the only such trees are those which are described in the statement of the lemma.
\end{proof}

\paragraph{Quantum correction to the BRST transformation} Lemma \ref{lem:Feyn_rule} nearly puts us in a situation where we can compute OPEs in the chiral algebra using Wick contractions.  There is one remaining subtlety, however.  The result of the Wick contractions described in the lemma will automatically be a BRST invariant operator. However, it may not be expressed in a simple way in terms of the generators of the BRST cohomology that we are using, as listed in Table \ref{tab:celestial-ops}.  We may therefore have to add on BRST exact terms in order to bring it into this form. 

The only BRST transformation we will really need is that of the $\b$-ghost. At tree-level in the bulk of $Y$ this is given by the standard ADHM-like expression
\be
Q_\mrm{BRST}\b = \frac{1}{2}[\phi^\da,\phi_\da] + \frac{1}{2}\{\til\gamma^f,\gamma_f\}\,.
\ee
In the presence of the defect, however, there is an important correction to this. The corrected BRST transformation is
\bea \label{eq:BRST-full}
&Q_\mrm{BRST}\b^{rs} \\
&= \frac{1}{2}[\phi^\da,\phi_\da]^{rs} + \frac{1}{2}\{\til\gamma^f,\gamma_f\}^{rs} + \frac{1}{8\pi^2}\eps_{ij}\big(\psi_I^{(r}(\p_{w_i}\c\p_{w_j}\c\psi^I)^{s)} + (\p_{w_i}\c\psi_I)^{(r}(\p_{w_j}\c\psi^I)^{s)}\big)\,.
\eea 
Suppressing $\fsp(N)$ indices
\bea \label{eq:Qb}
&Q_\mrm{BRST}\b \\
&= \frac{1}{2}[\phi^\da,\phi_\da] + \frac{1}{2}\{\til\gamma^f,\gamma_f\} + \frac{1}{16\pi^2}\eps_{ij}\big(\psi_I\psi^I\p_{w_i}\c\p_{w_j}\c + \p_{w_i}\c\p_{w_j}\c\psi_I\psi^I + 2\p_{w_i}\c\psi_I\psi^I\p_{w_j}\c\big)\,.
\eea
The additional term involving the $\psi$ fields arises for the following reason. The BRST charge is obtained from the BRST current $J_\mrm{BRST}$, which has two terms: a bulk term and a defect term. These are written as
\be \label{eq:J_BRST} J_\mrm{BRST}^\text{bulk} = \frac{1}{2}\Tr(\b[\c,\c]) + \frac{1}{2}\Tr(\phi_\da[\c,\phi^\da]) + \gamma_f\c\til{\gamma}^f\,,\qquad J_\mrm{BRST}^\text{defect} = \frac{1}{2}\psi_I\c\psi^I\,. \ee
The bulk contribution to the BRST transformation of $\b$  is given by
\be
\frac{1}{2\pi\i}\int_{\norm{w} = 1}\d^3 w\,J_\mrm{BRST}^{\text{bulk},(2)}(w) \b(0)\,.
\ee
Here $J_\mrm{BRST}^{(2)}$ is the operator obtained from descending the BRST charge twice.  Its explicit form is given by replacing the fields $\b,\c,\phi^\da,\gamma_f,\tilde\gamma^g$ in equation \eqref{eq:J_BRST} by their BV counterparts $\mc{A},\mc{B},\Phi^\da,\Gamma_f,\til\Gamma^g$ listed in equation \eqref{eq:BV_fields}, and restricting to the $(0,2)$-form part.  The integral is taken on a $5$-sphere surrounding the location of the operator $\b$.

This integral is computed by performing a Wick contraction in the bulk theory of $\b$ with a copy of the double descendent of $\c$ appearing in the BRST operator.  The bulk propagator is the Bochner-Martinelli kernel, which is built so that its integral on a $5$-sphere against $\d^3 w$ is $1$.  This gives us the classical expression for the BRST transformation of $\b$ seen above.

The contribution of the defect BRST current is computed in a similar way: it is
\be
\frac{1}{2\pi\i}\int_{\abs{z} = 1}\d z\,J_\mrm{BRST}^\text{defect}(z) \b(0)\,. 
\ee
Here, the defect BRST current is given by the coupling between the defect fermions and the bulk gauge fields in \eqref{eq:act_D55'}, and we do not need any descendants. This expression is computed in the same way, using the $\c-\b$ propagator.  Just like with the $\phi-\phi$ propagator, this only produces a non-zero answer when we consider a propagator with two legs coming off in the middle, labelled by $\c$ fields.  Indeed, the $\c-\b$ propagator has an identical form to the $\phi-\phi$ propagator, with only colour indices changed.  The contour integral giving the contribution of the defect BRST current picks up the singular part of the $\b-\c$ propagator, giving rise to the expression \eqref{eq:BRST-full}.

\paragraph{Contributions to the  BRST operator from two Wick contractions}
There are further corrections to the BRST transformation which arise from two Wick contractions of the BRST current.  For the defect BRST current, this will be two $\psi$ contractions, and for the bulk BRST current, it will be two $\gamma$ or $\phi$ contractions.

For example, the contribution of two Wick contractions in the defect BRST current adds a term of the form
\be
Q_\mrm{BRST} :\psi_I^p\psi_{Jq}: = \Omega_{IJ}\p_z\c^p_{~\,q}\,.  
\ee
This occurs in any  word containing the normally ordered product of two $\psi$'s.

In the case of the bulk BRST current, it is a little more complicated.  For example, consider the contribution of double contractions to the BRST transformation of $\phi_{\da} \phi_{\db}$.  There are two propagators connecting the operator with the BRST current $\phi_\da[\c,\phi^\da]$.  The propagators are Dolbeault forms, and we want to end up with a $(0,2)$-form which is integrated on a $5$-sphere against $\d w_1 \d w_2 \d z$.  In a similar way to our discussion above, we find that there must be precisely two $\c$ fields emerging from the propagator, so that one can have contributions to the BRST operator which are schematically of the form
\be
Q_\mrm{BRST} \phi_{\da}\phi_{\db} 
\propto\eps_{\da\db}\eps_{ij} \p_z\c\p_{w_i}\c\p_{w_j}\c\,.
\ee

However, it is not hard to see that neither of these terms contribute to the BRST transformation of our open string states.  Indeed, in all of our open string states, we have symmetrized the $\phi$ fields.  The bulk BRST transformation is anti-symmetric in the $\phi$ fields, so this does not contribute.

The double contraction term in the defect BRST transformation also does not contribute to the BRST transformation of open string states. This is for  the very simple reason that a double contraction will automatically yield a closed string state, and we are setting them to zero.

Later, we will analyze what happens when closed string states are given a non-zero VEV.  In that situation, this simple argument does not apply, and we need to analyze the BRST transformation more carefully.  An argument almost as simple applies in this situation to show that the states $\J[m,n]$ and $\Jt[m,n]$ remain BRST closed.  Any term in the BRST transformation that replaces two $\psi$'s by $\p_z\c$ produces a closed string state, uncharged under the $\mf{sp}(K)$ flavour symmetry.  The states $\J$ and $\Jt$ are charged under the adjoint representation of $\mf{sp}(K)$, which is irreducible.  

For the states $\M[m,n]$, this argument does not apply,  as these states live in the exterior square of the fundamental of $\mf{sp}(K)$, which contains a copy of the trivial representation.  Indeed, 
\be
\M_{IJ} [0,0] = \psi_I \p_{w_1} \c \psi_J 
\ee
has a term in the BRST transformation which produces
\be
\Omega_{JI}\Tr( \p_z \c \p_{w_1} \c ).
\ee
There is a similar expression for the double-contraction term in the BRST transformation of $\M_{IJ}[m,n]$ and $\til{\M}_{IJ}[m,n]$.  We find that, before we quotient by the closed string states, $\Omega^{JI} \M_{IJ}[m,n]$ and $\Omega^{JI}\til \M_{IJ}[m,n]$ and the towers of states
\be
\Tr\big(\p_z\c\p_{w_i}\c\phi^{\done(m}\phi^{\dtwo\,n)}\big) 
\ee
cancel in BRST cohomology.  This observation will not be significant for us, because in our later analysis which includes closed string states, we will still set these particular states to zero.  (It is important to note that since this cancellation is between states on the $\CP^1$ defect and those in the bulk of $Y$, only the constant modes on the fibres of $Y$ are removed from the latter.)

Though double Wick contractions will not alter the BRST representative we choose, it will correct the computation in a subtle way. For example, considering double contraction in $Q_\mrm{BRST}(\psi_I\b\psi_J)$ will further modify the ADHM equations computed in \ref{eq:BRST-full}. This will be important when we compute the OPEs that comes from double Wick contractions. However, these OPEs are typically constrained by other methods, so we will not analyze them further in this paper.\footnote{There is a further subtlety we are suppressing here: nilpotency of the combined bulk and defect BRST operator fails on account of the $\Sp(N)$ gauge anomaly on the defect. This can be compensated by switching on a counterterm on $Y$ proportional to $\int_Y\mc{G}\Tr(\mc{A}\p\mc{A})$ where $\mc{G}$ is the Bochner-Martinelli kernel obeying $\dbar\mc{G} = \delta_{\CP^1}$. This counterterm also modifies the action of the BRST operator.}


\section{Some Simple OPEs} \label{sec:simpleOPEs}

We have given enough information on the OPEs and the BRST transformations that we can now turn to compute some simple OPEs in the chiral algebra, and compare to what is known for collinear limits in gauge theory.

We will begin by analyzing \emph{planar} OPEs of elements of the chiral algebra, before thinking about the very simplest non-planar contributions. We will focus on those OPEs involving at least one positive or negative helicity gluon state.

It will be convenient to organize the operators $\J_{IJ}[m,n](z) = \psi_I\phi^{\done(r}\phi^{\dtwo\,s)}\psi_J(z)$ on the defect into a generating function. Writing
\be
\phi(\lt;z) = \sum_\da\phi^\da(z)\lt_\da
\ee
we define
\be
	\J^{(n)}_{IJ}(\lt;z) := \psi_I\phi(\lt)^n\psi_J(z) = \sum_{r+s=n} \frac{(r+s)!}{r!s!}\lt_\done^r\lt_\dtwo^s\J_{IJ}[r,s](z)\,.
\ee
We can also define the generating function 
\be
	\J_{IJ}(\lt;z) := \sum_{n\geq 0}\frac{1}{n!}\J^{(n)}_{IJ}(\lt;z) = (\psi_Ie^{\phi(\lt)}\psi_J)(z)\,.
\ee
In all of these expressions, terms between two $\psi$'s are treated as $2 N \times 2 N$ matrices, either symmetric or anti-symmetric depending on whether they are in the adjoint of $\mf{sp}(N)$, or in the exterior square of the fundamental.  Fields are concatenated by matrix multiplication. The final $\psi$ is treated as a vector in the fundamental of $\mf{sp}(N)$, and the initial $\psi$ as a covector, using the symplectic pairing.
 
Similarly, we can define the generating function
\be
	\Jt_{IJ}(\lt;z) := (\psi_Ie^{\phi(\lt)}\p_{w_1}\c\p_{w_2}\c\psi_J)(z)\,.
\ee
As before we can expand $\Jt(\lt;z)$ as a sum 
\be
	\Jt_{IJ}(\lt;z) = \sum_{n\geq0} \frac{1}{n!} \Jt^{(n)}_{IJ} (\lt; z) = \sum_{r,s \geq 0}\frac{(r+s)!}{r!s!}\Jt_{IJ}[r,s](z) (\lt^\done)^r(\lt^\dtwo)^s\,.
\ee
The expressions $\J(\lt;z)$ and $\Jt(\lt;z)$
 correspond to positive and negative helicity gluon states, with the momenta encoded in the pair of spinors $\lt$ and $\lambda = (1,z)$.  The expansion into $\J^{(n)}$ and $\Jt^{(n)}$ corresponds to expanding states into soft modes.

 A planar contraction of (say) of $\J^{(n)}$ with $\J^{(m)}$ must be given by a Feynman diagram of the type described in Lemma \ref{lem:Feyn_rule}, with two additional features. First, it must be possible to draw it in the plane, with no propagators crossing and no ``twisted'' propagators. (Recall that because we have an $\mf{sp}(N)$ gauge theory, the double line notation for Wick contractions corresponds to un-oriented surfaces, and the propagator splits into two terms: one oriented, and one with a twist).

The second constraint is that the result of the Wick contraction must again produce a single-string operator, or the identity operator.  This is because we should think of the Feynman diagram as part of a diagram contributing to a planar two or three-point function of two or three single-string operators. Under these two constraints, there are only a very small number of allowed Wick contractions, as we will see.

Consider the OPEs of $\J_{IJ}[p,q]$ with $\J_{KL}[r,s]$.  Here, there are many possible contractions, but most of them are not planar. There are only three kinds of planar contractions:
\begin{enumerate}
\item  We can contract exactly one pair of $\psi$ fields and no other fields.
\item In the case $p=q=r=s=0$, so that there are no $\phi$ fields in the operators we are considering, we can contract both pairs of $\psi$ fields.
\item We can contract one pair of $\psi$ fields, and one pair of $\phi$ fields, which are adjacent to the $\psi$ fields being contracted.
\end{enumerate}
The three kinds of planar contraction are illustrated in Fig. \ref{fig:planar}.

\begin{figure}[ht]
\centering
    \subfloat{%
    \includegraphics{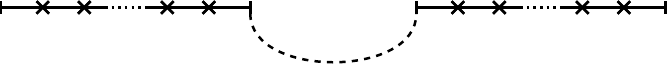}%
    } \\
    \subfloat{%
    \includegraphics{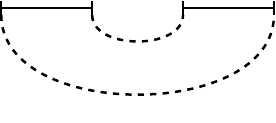}%
    } \\
    \subfloat{%
    \includegraphics{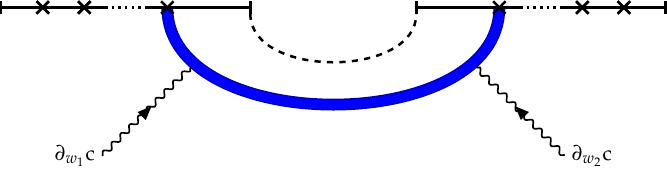}%
    }
    \caption{The three kinds of planar OPEs.  The first one is independent of $N$ and the second two are of order $N$.}
    \label{fig:planar}
\end{figure}

To see that these are the only allowed planar contractions we argue as follows. The contraction of two $\psi$ fields leads to closed string operators, which we are discarding, unless there are no $\phi$ fields in the operators, in which case it leads to the identity. If we do not contract any $\psi$ fields, the result of the OPE is the product of two single-string operators, and therefore is not planar.  Therefore, unless there are no $\phi$ fields in the operators, we must contract precisely one pair of $\psi$ fields.

Next let us consider the contraction of $\phi$ fields. If we contract two or more $\phi$ fields in a planar diagram, then we necessarily produce closed string fields, which we are discarding.  This is because the contraction of a pair of $\phi$ fields produces $\p_{w_1}\c\p_{w_2}\c$, which will be part of the closed string operator generated. 

Now let us turn to the computation of the planar OPEs.   It will be important to note that at the planar level, we can freely commute $\phi$ fields past each other in any word.  This is because, by the BRST relation \eqref{eq:BRST-full}, commuting two $\phi$'s past each other produces bilinears of $\psi$ or of $\gamma$. These split a single-string operator into two, and so correspond to non-planar contributions. 

The second and third kinds of contractions form a closed face, leading to a trace of the identity and generating a factor of $2N$. Both of these represent contributions from the axion background which can be removed by making the replacement $\fsp(N)\to\mathfrak{osp}(2N|N)$. In the first kind no faces are formed; these correspond to tree collinear splitting in flat space. We will address the flat space and axion background pieces separately.


\subsection{Tree OPEs in Flat Space} \label{subsec:TreeOPE}

In this subsection, we compute planar contractions of the first kind, corresponding to tree-level collinear splitting amplitudes on flat space.
  
\paragraph{Single contraction of $\psi$ field}

\bea \label{eq:JJ->J-tree1}
	&\J_{IJ}(\lt_1;z_1)\J_{KL}(\lt_2;z_2) \sim \wick{\c \psi_Ie^{\phi(\lt_1)}\psi_J(z_1) \psi_K e^{\phi(\lt_2)} \c\psi_L(z_2)} +  \wick{\c \psi_I e^{\phi(\lt_1)} \psi_J(z_1) \c\psi_K e^{\phi(\lt_2)}\psi_L(z_2)} \\
	&+ \wick{\psi_I e^{\phi(\lt_1)}\c \psi_J(z_1) \c\psi_K e^{\phi(\lt_2)}\psi_L(z_2)} + \wick{\psi_I e^{\phi(\lt_1)}\c \psi_J(z_1) \psi_K e^{\phi(\lt_2)}\c\psi_L(z_2)} \\
    & = \frac{4}{z_{12}}\Omega_{(J(K}\psi_{L)}e^{\phi(\lt_1)}e^{\phi(\lt_2)}\psi_{I)}(z_1)\,.
\eea

Using the fact that we can commute $\phi$ fields passed each other in planar OPEs, we find that a single contraction of $\psi$ fields leads to
\be \label{eq:JJ->J}
	\J_{IJ}(\lt_1;z_1)\J_{KL}(\lt_2;z_2) \sim \frac{4}{z_{12}}\Omega_{(J(K}\J_{L)I)}(\lt_1+\lt_2;z_1)\,.
\ee
This contraction is the first illustrated in Fig. \ref{fig:planar}. This is exactly the tree-level splitting function of Yang-Mills theory which we reviewed in Sect. \ref{subsec:ccas}.

This OPE is the Kac-Moody algebra at level zero for the double current algebra $\mf{sp}(K)[v^\done,v^\dtwo]$ (often called the $S$-algebra in reference to \cite{Strominger:2021mtt}).

We also define the generating functions for the matter part of the algebra
\bea
    &\M_{IJ}(\lt;z) = \psi_Ie^{\phi(\lt)}\p_{w_1}\c\psi_J(z)\,,\qquad \Mt_{IJ}(\lt;z) = \psi_I e^{\phi(\lt)}\p_{w_2}\c\psi_J(z)\,, \\
    &\M_{f,I}(\lt;z) = \gamma_fe^{\phi(\lt)}\psi_J(z)\,,\qquad  \Mt_I^f(\lt;z) = \widetilde\gamma^f e^{\phi(\lt)}\psi_J(z)\,.
\eea
Similar calculations using the $\psi-\psi$ Wick contraction give us 
\bea \label{eq:JO->O}
	&\J_{IJ}(\lt_1;z_1)\Jt_{KL}(\lt_2;z_2) \sim \frac{4}{z_{12}}\Omega_{(J(K}\Jt_{L)I)}(\lt_1+\lt_2;z_1)\,, \\
	&\J_{IJ}(\lt_1;z_1)\M_{KL}(\lt_2;z_2) \sim - \frac{4}{z_{12}}\Omega_{(J[K}\M_{L]I)}(\lt_1+\lt_2;z_1)\,, \\
	&\J_{IJ}(\lt_1;z_1)\M_{f,K}(\lt_2;z_2) \sim - \frac{2}{z_{12}}\Omega_{K(J}\M_{f,I)}(\lt_1+\lt_2;z_1)\,.
\eea
OPEs with the $\Mt_{IJ}(\lt;z),\Mt_I^f(\lt;z)$ fields take the same form as above.

\paragraph{Single contraction of $\gamma,\til\gamma$ fields}
For the OPE between $\M_f$ and $\Mt^g$, we consider a single contraction between the $\gamma$ and $\til{\gamma}$ fields. We have
\be \label{eq:MMt->Jt}
	\M_{I,f}(\lt_1;z_1)\Mt_J^g(\lt_2;z_2)\sim\wick{\c1\gamma_f e^{\phi(\lt_1)}\psi_I(z_1) \c1 {\til{\gamma}^g} e^{\phi(\lt_2)}\psi_J(z_2)} = - \frac{1}{4\pi^2z_{12}}\delta_f^{~\,g}\Jt_{IJ}(\lt_1+\lt_2;z_1)\,.
\ee
This differs from the standard normalization of the $\M,\Mt$ OPE as appears in \cite{Costello:2023vyy} (which we have reproduced in Appendix \ref{app:simplifyOPE} for convenience) by a factor of $4\pi^2$. We can bring it into this standard from by rescaling $\M\mapsto\M/2\pi,\Mt\mapsto\Mt/2\pi$.


\subsection{Tree OPEs in the Axion Background} \label{subsec:axion}

Next we turn our attention to planar contractions of the second and third kind. In these cases a closed face is generated leading to a factor of $2N$.

\paragraph{Double contraction of two $\psi$ pairs}
Contractions of the second kind only arise in the OPEs of $\J_{IJ} = \J_{IJ}[0,0]$. We find that
\be \J_{IJ}(z_1)\J_{KL}(z_2) \sim \frac{2N}{z_{12}^2}(\Omega_{JK}\Omega_{LI} + \Omega_{JL}\Omega_{KI}) + \frac{4}{z_{12}}\Omega_{(J(K}\J_{L)I)}(z_1)\,. \ee
We learn that the $\J_{IJ}$ generate an $\mf{sp}(K)$ current algebra at level $N$. The level matches the two-point amplitude of gluons generated by the brane backreaction, as can readily be verified by substituting the linearized self-dual field strengths
\be F_{1\da\db} = \t_{IJ}\lt_{1\da}\lt_{1\db}e^{ix\cdot p_1}\,,\qquad F_{2\da\db} = \t_{KL}\lt_{2\da}\lt_{2\db}e^{ix\cdot p_2} \ee
into the axion vertex
\be - \frac{N}{2\pi^2}\int_{\R^4}\d^4x\,\log\|x\|\op{tr}(F_{1\da\db}F_2^{\da\db}) = \frac{2N}{\la12\ra^2}(\Omega_{JK}\Omega_{LI} + \Omega_{JL}\Omega_{KI})\,. \ee
Repeating this calculation having made the replacement $\mf{osp}(N)\to\mf{osp}(M|N)$ would shift the level to $N-M/2$. In particular, for $M=2N$ the level would vanish, consistent with the fact that there is no backreaction in this case.

\paragraph{Contraction of one $\phi$ pair and one $\psi$ pair}
When a pair of $\phi$ fields is contracted it always produces both a planar and a non-planar contribution. In this subsection, we will focus on the planar contribution.  When the $\phi$ contraction is adjacent to the $\psi$ contraction, the planar piece is linear in $N$. But when the $\phi$ contraction is not adjacent to the $\psi$ contraction, we produce a closed string state which is removed in our algebra.  The only planar contribution is therefore from the last diagram illustrated in Fig. \ref{fig:planar}.

The relevant Wick contraction of $\J^{(m)}(\lt_1;z_1)$ with $\J^{(n)}(\lt_2; z_2)$ is
\bea \label{eq:JJ->J-tree2}
	&\J^{(m)}_{IJ}(\lt_1;z_1)\J^{(n)}_{KL}(\lt_2;z_2) \\
    &\sim \wick{(\psi_I\phi(\lt_1)^{m-1}\c2 \phi(\lt_1)\c1\psi_J)(z_1)(\c1\psi_K\c2 \phi(\lt_2)\phi(\lt_2)^{n-1}\psi_L)(z_2)} + (IJ),(KL)\,.
\eea
This has both planar and non-planar terms, depending on whether the $\mf{sp}(N)$ propagator is a double line or twisted double line, and whether the $\c$ ghosts form part of the open string operator after the OPE, or part of a closed string operator. The only relevant term is that depicted in Fig. \ref{fig:planar}, where the $\c$ ghosts form part of the open string operator.  The leading pole in this expression is 
\be \label{eq:JJ->NJt}
    \frac{N}{4\pi^2}\frac{[12]}{\la12\ra^2}\Omega_{JK}\psi_I\phi(\lt_1)^{m-1}\phi(\lt_2)^{n-1}\p_{w_1}\c\p_{w_2}\c\psi_L(z_2) + (IJ),(KL)\,.
\ee
This is proportional to the sum of terms in $\Jt^{(m+n-2)}$ which are homogeneous of order $m-1$ in $\lt_1$ and $n-1$ in $\lt_2$. 
  
To write it more precisely we use the generating function $\J(\lt;z)$. We find the leading pole in the OPE is
\be \label{eq:axion-OPE}
\J_{IJ}(\lt_1;z_1)\J_{KL}(\lt_2;z_2) \sim \frac{N}{\pi^2}\frac{[12]}{\la12\ra^2}\Omega_{(J(K}\int_{s,t=0}^1 \d s\d t\,\Jt_{L)I)}\big(s\lt_1 + t\lt_2;(z_1+z_2)/2\big).
\ee
It turns out that this expression gives us the correct subleading (order $1/z_{12}$) pole also.  (On the right hand side, evaluating at $z$ or at $\half(z_1 + z_2)$ does not change the order $1/z_{12}^2$ pole, but it does change the order $1/z_{12}$ pole).

Calculating the subleading pole directly from the Wick contractions is slightly painful.  Instead we will verify this by symmetry and associativity.  Because we are considering a planar OPE, the subleading term must be a single-string operator.  Symmetry considerations immediately tell us that the subleading OPE must be given by \eqref{eq:axion-OPE}, plus some term which is of the form
\be
\J_\sfa[p,q](z_1)\J_\sfb[r,s](z_2) \sim \frac{(ps-qr)C(p+q,r+s)}{z_{12}} f_{\sfa\sfb}^{~\,~\sfc}\p_z\J_\sfc[p+r-1,q+s-1](z_1)\,.
\ee
(using a basis of the Lie algebra $\mf{sp}(K)$ with structure constants $f_{\sfa\sfb}^{~\,~\sfc}$). The function $C(p+q,r+s)$ must be odd in its two variables.   It is not hard to show, using associativity, that no such term can appear, so that $C(p+q,r+s) = 0$. 

The terms linear in $N$ capture corrections to collinear singularities from coupling to the background axion which is generated by backreaction. These arise as modifications to the perturbiner from the Feynman diagrams in Fig. \ref{fig:axion-perturbiner}, where the solid dots represent insertions of either the quadratic or cubic vertex in
\be \label{eq:background-axion-vertex} - \frac{N}{4\pi^2}\int_{\R^4}\log\|x\|\op{tr}(F^2)\,. \ee
Evaluating these diagrams in the holomorphic collinear limit would be somewhat tedious. Fortunately they coincide with the first order correction in $N$ to the celestial OPE in the $\mrm{WZW}_4$ model on Burns space as computed in \cite{Costello:2023hmi}.

\begin{figure}[ht]
\centering
    \subfloat{%
    \includegraphics{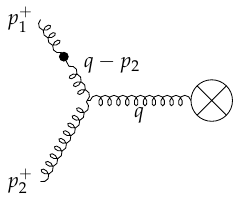}%
    }\hfil
    \subfloat{%
    \includegraphics{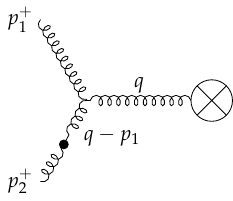}%
    }\hfil
    \subfloat{%
    \includegraphics{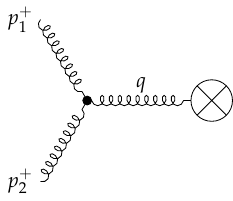}%
    } \\
    \subfloat{%
    \includegraphics{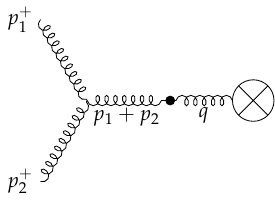}%
    }
    \caption{Corrections to the gluon perturbiner involving the background axion vertex.}
    \label{fig:axion-perturbiner}
\end{figure}

To see why this should be the case we perform a partial gauge-fixing of SDYM (ignoring the fermions for the moment) determined by a choice of K\"{a}hler structure on space-time. Such a gauge-fixing is certainly not valid in the full quantum theory, but for the moment we are interested in only tree-level effects. Writing $\R^4=\C^2$, anti-self-dual 2-forms decompose as
\be \Omega^2_-(\C^2) = \Omega^{2,0}(\C^2)\oplus\omega\Omega^0(\C^2)\oplus\Omega^{0,2}(\C^2) \ee
for $\omega$ the K\"{a}hler form. Integrating out the $(2,0)$ part of $B$ we learn that $F^{0,2}(A) = 0$, which is solved by
\be A^{0,1} = \sigma_0^{-1}\dbar\sigma_0\,. \ee
Similarly integrating out the $(0,2)$ part of $B$ we find that $F^{2,0}(A) = 0$, so that
\be A^{1,0} = - \p\sigma_\infty\sigma_\infty^{-1}\,. \ee
The $\Sp(K)$-valued combination $\sigma = \sigma_0\sigma_\infty$ is then gauge invariant, and $A$ is uniquely determined up to gauge by $\sigma$. This object is sometimes known as Yang's $J$-matrix \cite{Yang:1977zf}. Writing 
\be B^{1,1} = (-\i\omega)\sigma_\infty^{-1}\tau\sigma_\infty \ee
for $\tau\in\Omega^0(\C^2,\mf{sp}(K))$, the SDYM action then specializes to the Chalmers-Siegel form \cite{Chalmers:1996rq}
\be \label{eq:Chalmers-Siegel} \int_{\C^2}(-\i\omega)\,\op{tr}\big(\tau\p(\sigma^{-1}\dbar\sigma)\big)\,. \ee
An expression more convenient for perturbation theory can be obtained by writing $\sigma = \exp\chi$ where $\chi\in\Omega^0(\C^2,\mf{sp}(K))$, in terms of which the above is
\be \sum_{m=0}^\infty\frac{(-)^m}{(m+1)!}\int_{\C^2}(-\i\omega)\,\op{tr}\big(\tau\p(\mrm{ad}_\chi^m\dbar\chi)\big)\,. \ee
The vertices of this action coincide with those of the $\mrm{WZW}_4$ model; however, its propagators are directed so that its only tree diagrams have one external $\tau$ field.

Substituting the gauge-fixed form for $A$ into the background axion vertex \eqref{eq:background-axion-vertex} gives
\be - \frac{N}{4\pi^2}\int_{\C^2}\log\|x\|\op{tr}\big(\p(\sigma^{-1}\dbar\sigma)\p(\sigma^{-1}\dbar\sigma)\big)\,. \ee
We have the identity
\be \op{tr}(\p(\sigma^{-1}\dbar\sigma)^2) = -\frac{1}{2}\p\bar\p\bigg(\op{tr}(\sigma^{-1}\p\sigma\sigma^{-1}\dbar\sigma) - \frac{1}{3}\int_{[0,1]}\op{tr}((\tilde\sigma^{-1}\tilde\d\tilde\sigma)^3\big)\bigg)\,, \ee
where $\tilde\sigma$ is an extension of $\sigma$ to $\C^2\times[0,1]$ representing a homotopy to the identity map and $\tilde\d = \d_{\C^2} + \d_{[0,1]}$. This allows us to rewrite the vertex as
\be \label{eq:gauge-fixed-vertex} \frac{N}{8\pi^2}\int_{\C^2}\p\bar\p\log\|x\|\op{tr}\big(\sigma^{-1}\p\sigma\sigma^{-1}\dbar\sigma\big) - \frac{1}{48\pi^2}\int_{\C^2\times[0,1]}\p\bar\p\log\|x\|\op{tr}\big((\tilde\sigma^{-1}\tilde\d\tilde\sigma)^3\big)\,, \ee
which is precisely the coupling introduced in the $\mrm{WZW}_4$ model when the K\"{a}hler potential is shifted by
\be \frac{1}{2}\|x\|^2\mapsto\frac{1}{2}\|x\|^2 + N\log\|x\|\,. \ee
We recognize the right hand side as the K\"{a}hler potential for Burns space.

Forgetting the orientation of the propagators in Fig. \ref{fig:axion-perturbiner}, we can identify each diagram with its counterpart in the $\mrm{WZW}_4$ model on Burns space. They are the order $N$ corrections to the two incoming states, propagator and interaction respectively. Unlike on Burns space, the vertex \eqref{eq:gauge-fixed-vertex} can appear at most once in any (connected) Feynman diagram, since it acts as a sink for positive helicity states. These diagrams are therefore the only possible corrections to the $2\to1$ OPE at order $N$.

The order $N$ term in the $2\to1$ celestial OPE can therefore be deduced from the Burns OPE appearing in \cite{Costello:2023hmi}. The leading double pole is
\be \J_{IJ}(\lt_1;z_1)\J_{KL}(\lt_2;z_2) \sim \frac{N}{\pi^2}\frac{[12]}{\la12\ra^2}\Omega_{(J(K}\int_0^1\d s\int_0^1\d t\,\Jt_{L)I)}\big(s\lt_1+t\lt_2;\half(z_1+z_2)\big)\,, \ee
matching the dual computation \eqref{eq:axion-OPE}. The relative factor of $-1/4\pi^2$ in this OPE compared to that appearing in \cite{Costello:2023hmi} can be attributed to the unusual normalization of the $\mrm{WZW}_4$ kinetic term when contrasted with \eqref{eq:Chalmers-Siegel}.


\subsection{The Simplest One-Loop OPE} \label{subsec:NonPlanarOPE}

Thus far we've extensively studied the planar OPEs of the chiral algebra, and they match well-known tree-level results from collinear limits in $4d$ gauge theory. Let us now turn our attention to the simplest non-planar contributions to the OPEs. In particular, to the $2\to1$ OPE. There are of course many further non-planar contributions, some of which will be computed in Sect. \ref{sec:nonfactorizing}.

Under the constraint of Lemma \ref{lem:Feyn_rule}, only double Wick contraction of one $\phi$ pair and one $\psi$ pair have non-planar contributions to the 2 to 1 OPEs. We also need to be careful to keep track of the positions of the contracted $\phi$s. Both cases, when the $\phi$ contraction is adjacent to the $\psi$ contraction and when they are not adjacent, have non-planar contributions.

First, we recall the case when the $\phi$ contraction is adjacent to the $\psi$ contraction, which is analyzed in \eqref{eq:JJ->J-tree2}
\bea \label{eq:JJ->J-tree2-again}
	&\J^{(m)}_{IJ}(\lt_1;z_1)\J^{(n)}_{KL}(\lt_2;z_2) \\ &\sim \wick{(\psi_I\phi(\lt_1)^{m-1}\c2\phi(\lt_1)\c1\psi_J)(z_1)(\c1\psi_K\c2\phi(\lt_2)(\phi(\lt_2)^{n-1}\psi_L)}(z_2) + (IJ),(KL) \\		
	&= \frac{1}{2\pi^2}\frac{[12]}{\la12\ra^2}\Omega_{JK}(\psi_I\phi(\lt_1)^{m-1})_p(z_1)\omega_{qr}(\p_{w_1}\c^{[r[p}\p_{w_2}\c^{q]s]} \\
    &+ \omega^{[r[p}\p_{w_1}\c^{q]}_{~\,~t}\p_{w_2}\c^{s]t})(\phi(\lt_2)^{n-1}\psi_L)_s(z_1) + (IJ),(KL)\,.
\eea
The planar contribution from the above expression has been analyzed before; the non-planar contribution contributes a leading double pole
\be \label{eq:phipsi-np1}
	- \frac{1}{2\pi^2}\frac{[12]}{\la12\ra^2}\Omega_{JK}\psi_I\phi(\lt_1)^{m-1}\phi(\lt_2)^{n-1}\p_{w_1}\c\p_{w_2}\c\psi_L(z_1) + (IJ),(KL)\,.
\ee 

Next, we consider the contributions from the $\phi$ contractions that are not adjacent to the $\psi$. One such contribution is illustrated in Fig. \ref{fig:JJnonplanar}.

\begin{figure}[ht]
    \centering
        \includegraphics{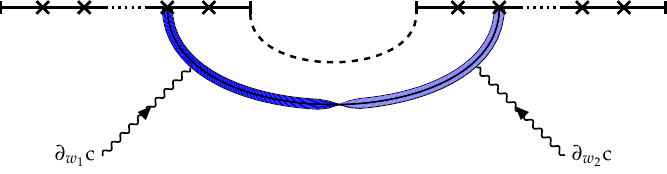}
    \caption{An example non-planar contraction of non-adjacent $\psi$ and $\phi$ fields.} \label{fig:JJnonplanar}
\end{figure}

In this case, the planar piece should be discarded as it generates closed string operators. The non-planar contribution produces the same expression as equation \eqref{eq:phipsi-np1} after we symmetrize the $\phi$s. Therefore, after summing over the $mn$ possible ways of performing the $\phi$ contractions, including the case when the $\phi$ contraction is adjacent to the $\psi$ contraction, we find the leading double pole
\be \label{eq:phipsi-np2}
	- \frac{1}{2\pi^2}\frac{[12]}{\la12\ra^2}mn\Omega_{JK}\psi_I\phi(\lt_1)^{m-1}\phi(\lt_2)^{n-1}\p_{w_1}\c\p_{w_2}\c\psi_L(z_1) + (IJ),(KL)\,.
\ee 
This expression can be also written as
\bea
	&\J_{IJ}(\lt_1;z_1)\J_{KL}(\lt_2;z_2) \sim - \frac{2}{\pi^2}\frac{[12]}{\la12\ra^2}\Omega_{(J(K}\J_{L)I)}\big(\lt_1+\lt_2;\half(z_1+z_2)\big)\,,
\eea
where here the subleading simple pole is fixed by associativity. This result matches the one-loop corrections to the celestial OPE found in \cite{Costello:2022upu}, corresponding to the one-loop QCD splitting amplitude. An example Feynman diagram contributing to this factorization is illustrated below (Fig. \ref{fig:1loopsplitting}).
\begin{figure}[ht]
    \centering
        \includegraphics{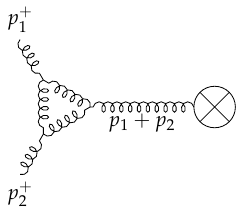}
    \caption{An example of a factorizing diagram contributing to the one-loop splitting amplitude. There are further contributions from bubble diagrams and with fermions running through the loop \cite{Bern:1995ix}.} \label{fig:1loopsplitting}
\end{figure}


\section{Large \texorpdfstring{$N$}{N} Chiral Algebras for Self-Dual Backgrounds} \label{sec:backgrounds}

So far we've taken pains to remove all but the trivial bulk operator from the $D5$ brane theory. This was achieved in Sect. \ref{subsec:SDQCD-string} by forming the quotient
\be \mc{A}^{N,K} = \mc{A}^{N,K}_\mrm{Defect}/I_{\vac}\,, \ee
where $I_{\vac}$ is the maximal ideal generated by bulk operators charged under any of the natural symmetries. In this section, we consider the result of instead quotienting by different ideals of bulk operators. We find that these correspond to switching on a variety of self-dual backgrounds, of which the simplest are gravitational and flavour backgrounds.

This analysis gives us a very large class of self-dual backgrounds with  a ``top-down'' version of the celestial holography, along the lines of that in \cite{Costello:2022jpg,Costello:2023hmi}.

Recall that the $D5$ worldvolume is the submanifold $Y \subset X$, which we have been giving coordinates $w_1,w_2,z$. In this section, to keep the notation compact, we will use coordinates $w_1,w_2,w_3$.  Recall that $Y$ is the bundle $\Oo(-1) \oplus \Oo(-3) \to \CP^1$, where $w_1,w_2$ are coordinates on the fibres and $z = w_3$ is a coordinate on the base.


\subsection{Ideals in the Bulk Algebra} \label{subsec:ideals}

As before, we let $\mc{A}^N_{D5}$ be the algebra of bulk operators on the $D5$ system.  Because this is a holomorphic theory in complex dimension $3$, the OPE has no singularities, making it a commutative algebra. Generators for the algebra are described in Table \ref{tab:bulk-states}.   Instead of setting all the generators $\mc{O}_i$ of the algebra to zero, we can instead implement a relation
\be\label{eq:general_vev}
\mc{O}_i = f_i (w_1,w_2,w_3) 
\ee
where if $\mc{O}_i$ is of spin $-s_i$, then $f_i$ is a section of the bundle $\Oo(2 s_i)$ on $Y$. It is important to note that $\{\mc{O}_i\}$ are independent gauge invariant generators of the commutative algebra $\mc{A}^N_{D5}$. To define ideals of the algebra $\mc{A}^{N,K}_\mrm{Defect}$, we can independently assign vevs to the composite operators $\{\mc{O}_i\}$ without specifying values for the fields $\phi,c,\dots$. More precisely, the relation \eqref{eq:general_vev} defines an ideal $I_{\Psi}$ of the algebra $\mc{A}^{N,K}_\mrm{Defect}$. Imposing these vevs corresponds to considering the quotient $\mc{A}^{N,K}_\Psi  = \mc{A}^{N,K}_\mrm{Defect} / I_\Psi$. 

For example, consider the operator $\gamma_f\til{\gamma}^g$.  This is of spin $-1$, so that we can impose a relation
\be \label{eq:gamma-quotient}
\gamma_f\til{\gamma}^g = M^g_f(w_i)
\ee
where $M^g_f$ is a section of the bundle $\mf{sl}(16)\otimes\Oo(2)$. 

Fields on the $D5$ worldvolume couple to bulk fields of the ten-dimensional string theory.  The simplest such couplings involve a single bulk field coupling to a single trace operator $\mc{O}_i\in\mc{A}^N_{D5}$.  On the worldsheet these correspond to disc amplitudes with a single closed string vertex operator inserted at the origin, and a number of open string vertex operators inserted on the boundary $S^1$.  The bulk fields which couple to particular single-trace operators $\mc{O}_i$ are listed in Table \ref{tab:bulk-states}.  The form of the open-closed couplings we use here, adapted to the target space description of the B-model, are given in \cite{Costello:2015xsa,Costello:2019jsy}.  There they were used to quantize Kodaira-Spencer gravity at all loop orders through a variant of the Green-Schwarz mechanism.

If we set $\mc{O}_i$ to some value, then we add a term to the Lagrangian of the $10d$ system whereby one of the fields is integrated over the worldvolume of the $D5$ system.  This can be interpreted as a source term, and solving yields a non-trivial background in the bulk.

\paragraph{Giving a VEV to the current for the $\mf{sl}(16)$ flavour symmetry}
Let us see how this works for the $\gamma_f\til{\gamma}^g$ operator.  This couples to the $\mf{sl}(16)$ matrices inside the $\mf{so}(32)$ ten-dimensional gauge theory by
\be
\frac{\i}{2\pi}\int_{D5}\gamma_f \til{\gamma}^g\,a^f_{0,g} 
\ee
(here, and elsewhere, the $\mf{sp}(N)$ indices have been contracted).  This term is to be added to the ten-dimensional holomorphic Chern-Simons Lagrangian, yielding, once we specify the value of $\gamma_f\til{\gamma}^g$, the Lagrangian
\be \label{eq:hCS-flavour}
- \bigg(\frac{\i}{2\pi}\bigg)^3\int_X\Omega_X\,\mrm{CS}(a_0) + \frac{\i}{2\pi}\int_{Y} M_f^g\,a_{0,g}^f
\ee
where as above $M^g_f$ is a section of $\Oo(2)$, and $\Omega_X$ is the holomorphic volume form on $X$.

It's worth pausing to explain the unusual normalization for the kinetic term of the holomorphic Chern-Simons theory in Equation \eqref{eq:hCS-flavour}. This can be traced to the Green-Schwarz mechanism for the $D9$ brane theory \cite{Green:1984sg,Costello:2019jsy}. The ten-dimensional holomorphic Chern-Simons theory suffers from a gauge anomaly arising from the variation of the hexagon diagram. Ignoring the background $\mf{so}(32)$ bundle for the moment, the corresponding anomaly cocycle can be expressed terms of the BV field $\mbf{a}_0$ as
\be \frac{1}{2\cdot5!}\bigg(\frac{\i}{2\pi}\bigg)^5\int_X\op{tr}_\mrm{Ad}\big(\mbf{a}_0(\p\mbf{a}_0)^4\big)\,. \ee
Here the trace if taken in the adjoint of $\mf{so}(32)$. This cocycle is independent of the normalization of the holomorphic Chern-Simons kinetic term. The anomaly is cancelled by an exchange of closed string fields coupling to $\mbf{a}_0$ by
\be \label{eq:flavour-couplings} \frac{C_\eta}{4!}\bigg(\frac{\i}{2\pi}\bigg)\int_X\eta\op{tr}_\mrm{F}\big(\mbf{a}_0(\p\mbf{a}_0)^3\big) + \frac{C_\mu}{2!}\bigg(\frac{\i}{2\pi}\bigg)\int_X\mu\op{tr}_\mrm{F}\big(\mbf{a}_0\p\mbf{a}_0\big)\,, \ee
where $\mrm{F}$ is the fundamental of $\mf{so}(32)$. Cancellation of the anomaly is possible on account of the trace identity
\be \op{tr}_\mrm{Ad}(X^6) = 15\op{tr}_\mrm{F}(X^2)\op{tr}_\mrm{F}(X^4)\,, \ee
requiring that $C_\eta C_\mu = 1/4(2\pi)^4$. Implicit in our normalization of the $D5$ and $D5^\prime$ worldvolume actions is the symmetric choice $C_\eta = C_\mu = 1/2(2\pi)^2$. But the field $\mu$ encodes a Beltrami differential on $X$ and the second term in equation \eqref{eq:flavour-couplings} is precisely the change in the holomorphic Chern-Simons action under this complex structure deformation. Noting that the trace in the fundamental of $\mf{so}(32)$ restricts to twice the trace in the fundamental of $\mf{sl}(16)$, the coefficient in Equation \eqref{eq:hCS-flavour} follows.

Varying $a_0$ gives us the source equation
\be
F^{0,2}(a_0)_g^f = 4\pi^2M_g^f\delta_Y\,.
\ee
This tells us that forming the quotient of the bulk algebra in equation \eqref{eq:gamma-quotient} has the effect of changing the backreacted $10d$ system, by inserting a source for the gauge field.  Note that this equation makes sense, because $\delta_Y$ is $(0,2)$-form valued in $\Oo(-2)$, whereas $M_g^f$ is a section of $\Oo(2)$. 

This construction works for all of the generators of the algebra $\mc{A}^{N}_{D5}$. Each bulk operator comes in a tower of states with an increasing number of $\phi$ fields; for simplicity, we only write down the fields associated to the lowest element of the tower, which is invariant under the $\mrm{SU}(2)$ rotating $\phi$.  The higher modes are given by the same expressions but with derivatives normal to the brane. 

We use the following notation for the field content of the type I string theory. The $\mf{so}(32)$ gauge field will be decomposed into the adjoint of $\mf{sl}(16)$, the exterior square of the fundamental, and the exterior square of the anti-fundamental.  The components are $a_{0,g}^f$, $a_{0,fg}$, $a_0^{fg}$ where $f,g$ are indices for the fundamental of $\mf{sl}(16)$. The bulk closed string fields, as before, consist of a cubic tensor $\eta \in \Pi\Omega^{0,\ast}(X,\wedge^3TX)$ and a Beltrami differential $\mu \in \Pi\Omega^{0,\ast}(X, TX)$.\footnote{We slightly abuse notation by calling $\mu$ a Beltrami differential, which strictly speaking only applies to the component of $\mu$ which is a $(0,1)$-form valued in $X$.} We will freely use the holomorphic volume form on $X$ to write $\mu$ as a $(4,\ast)$-form and $\eta$ as a $(2,\ast)$-form on $X$. We will go through the various bulk operators on the stack of $N$ $D5$ branes, and determine what ten-dimensional field it sources when we give it a VEV. 

\paragraph{Giving a VEV to $\gamma$ bilinears}
We will start with the open string operators $\gamma_f\gamma_g$ and $\til{\gamma}^f\til{\gamma}^g$.  We can set
\be \gamma_f\gamma_g = F_{fg}\,,\qquad\til{\gamma}^f\til{\gamma}^g = G^{fg} \ee
where $F$ is constant and $G$ is a section of $\Oo(4)$.   It is clear that, just as with $\gamma_f\til\gamma^g$, these provide a source for $a_0^{fg}$ and $a_{0,fg}$.  We note that on $X$, $a_{0,fg}$ is an element of $\Pi\Omega^{0,\ast}(X,\Oo(-2))$, and the source equation
\be
\dbar a_{0,fg} = 8\pi^2F_{fg} \delta_Y
\ee
makes sense because $\delta_Y$ is a section of $\Oo(-2)$. Similarly, $a_0^{fg}$ is an element of $\Pi\Omega^{0,\ast}(X,\Oo(2))$, and the equation
\be
\dbar a_0^{fg} = 8\pi^2G^{fg} \delta_Y
\ee
makes sense, because $G_{fg}$ is a section of $\Oo(4)$. 

\paragraph{Giving a VEV to the stress tensor}
The next operator is the stress tensor
\be
T^i = \Tr\bigg(\frac{1}{2}\phi_\db\p_{w_i}\phi^\db + \b\p_{w_i}\c + \dots\bigg).
\ee
This is of  different spin depending on the value of $i$. The coordinate independent way of saying this is that we can set $T^i = f^i$ where $f = f^i \d w_i \,\d^3 w$ is a section of the bundle $K_Y \otimes T^\ast Y$.

The stress tensor couples to the bulk Beltrami differential field by
\be 
\frac{\i}{2\pi}\int_{D5} \mu_i^{0,3}\,T^i
\ee
where $\mu_i^{0,3}$ is the component of $\mu$ which is the coefficient of $\d^3 \wbar\,\p_{w_i}$.

If we set the stress tensor to $T^i = f^i$, we find we have added a source term to the action for the twisted type I supergravity, so the action becomes
\be
\frac{\i}{2\pi}\int_{X} \mu\,\p^{-1} \dbar \eta + \frac{\i}{2\pi}\int_{D5} \mu_i^{0,3}\,f^i + \dots
\ee
where $\eta\in\Pi\Omega^{0,\ast}(X,\wedge^3 TX)$ and we have only included the closed string kinetic term and the source term.  Varying $\mu$ gives the equation
\be
\dbar \eta = - \p( \delta_{Y} f ) \label{eq:etasource}
\ee
so that giving a VEV to the stress tensor sources the trivector field $\eta \in \Omega^{0,1}(X, \wedge^3 T X)$.  

Let us explain how to interpret the right hand side.  Recall that $f^i \d w_i \,\d^3 w$ is a holomorphic section of the bundle $K_Y \otimes T^\ast Y$ on the the submanifold $Y$ supporting the stack of $D5$ branes. We can view $\delta_Y$ as a $(0,2)$-form valued in $K_Y^{-1}$ (using the holomorphic volume form on $X$), so that
\be
\delta_Y f^i \d w_i\,\d^3 w \in \br{\Omega}^{1,2}(X)\,.
\ee
(The symbol $\br{\Omega}^{\ast,\ast}$ indicates forms whose coefficients may be distributions instead of smooth functions). 

Therefore, $\p (\delta_Y f^i \d w_i)$ is naturally a distributional $(2,2)$-form on $X$.   Further, by contracting with the holomorphic volume form on $X$, we can view $\eta$ as a $(2,\ast)$-form on $X$, so that equation  \eqref{eq:etasource} gives a source for the $(2,1)$-form component of $\eta$.

\paragraph{Giving a VEV to $\p \c$ bilinears}
Next, we consider the bulk operator
\be
\op{Tr} \p_{w_i} \c \p_{w_j} \c\,. \label{eq:ccoperator}
\ee
The source equation for this operator is a little more complicated. Recall that the ten-dimensional trivector field $\eta$ is constrained to satisfy $\p \eta = 0$. Therefore there exists an element $\what\eta \in \Pi\Omega^{0,\ast}(X, \wedge^4 TX)$ so that $\p \what \eta = \eta$.  The operator $\op{Tr} \p_{w_i}\c\p_{w_j}\c$ couples by
\be
\frac{\i}{2\pi}\int_Y\eps^{\db\da}\what \eta_{\da\db ij} \op{Tr} \p_{w_i}\c\p_{w_j}\c\,.
\ee
The operator \eqref{eq:ccoperator} can be set equal to a $\p$-closed holomorphic $2$-form $\omega = \omega^{ij} \d w_i \d w_j$ on $Y$:
\be
\op{Tr} \p_{w_i} \c \p_{w_j} \c = \omega^{ij} \,.
\ee
Since the Beltrami differential couples to $\what \eta$ simply by $\int \mu \dbar \what \eta$, we find that the source equation for the operator $\eqref{eq:ccoperator}$ is simply  
\be
\dbar \mu = \delta_{Y} \omega.
\ee
This makes sense, because $\mu$ can be thought of as a $(4,\ast)$-form on $X$, $\delta_Y$ is a closed $(2,2)$-form, and $\omega$ is a closed $(2,0)$-form on $Y$. 

\paragraph{Giving a VEV to the remaining operator}
The final bulk operator is
\be
\frac{\eps^{ijk}}{3!2}\Tr\big((3\phi_\db\p_{w_i}\phi^\db - 2\b\p_{w_i}\c)\p_{w_j}\c\p_{w_k}\c + \dots\big)\,.
\ee
This operator  is of spin $-2$, and so can set to a section $F$ of $\Oo(4)$ on $Y$. That is, $F$ is a section of $K_Y^{2}$.  It couples to the $\p_{w_1} \p_{w_2} \p_{w_3}$ component of $\eta$, giving rise to the source equation
\be
\dbar \mu = \p ( \delta_Y F).
\ee
This equation  needs to be interpreted carefully. By using the holomorphic volume form on $X$, we can view $\delta_Y$ as a $(0,2)$-form valued in $K_Y^{-1}$. By hitting it with the section $F$ of $K_Y^2$, we view $\delta_Y F$ as a $(3,2)$-form on $X$. On the left hand side we view $\mu$ as a $(4,\ast)$-form.

We can sum up the fields on the $10d$ geometry $X$ sourced by giving VEVs to the bulk operators in the chiral algebra in Table \ref{tab:twistor-source}.
\begin{table}
\centering
    \begin{tabular}{c c c}
         \toprule 
         \thead{Operator} & \thead{Field sourced} & \thead{Source equation}  \\
         \midrule
         $\gamma_f\til{\gamma}^g $, $\gamma_f\gamma_g$, $\til{\gamma}^f \til{\gamma}^g$ & $\mrm{SO}(32)$ hCS gauge field $a_0$ & $\dbar a_0 = \delta_Y$ \\
         Stress tensor $T^i = \frac{1}{2}\Tr\big(\phi_\da\p_{w_i}\phi^\da + \dots\big)$ & $(3,1)$-form field $\eta$ & $\dbar \eta = \p\delta_Y$ \\
         $\Tr\p_{w_i}\c\p_{w_j}\c$ & Beltrami differential $\mu$ & $\dbar\mu = \delta_Y$ \\
         $\frac{1}{4}\eps^{ijk}\Tr\big(\phi_\da\p_{w_i}\phi^\da\p_{w_j}\c\p_{w_k}\c + \dots\big)$ & Beltrami differential $\mu $ & $\dbar \mu = \p \delta_Y$ \\
        \bottomrule
    \end{tabular}
\caption{Bulk operators and the fields they source on $X$}
\label{tab:twistor-source}
\end{table}


\subsection{Effect of Sourcing on Twistor Space} \label{subsec:twistor-sources}

We have described how working in a non-trivial vacuum of the $D5$ system has the effect of sourcing a field in the $10d$ geometry $X$, which we recall is the bundle $\Oo(-1) \oplus \Oo(-3)$ over twistor space.  Our ultimate interest, of course, is how this affects the dual system of the $D5$ branes wrapping twistor space, and then affects the self-dual gauge theory on $\R^4$.

To determine this, we proceed as follows.  The theory on the $D5^\prime$ branes has essentially the same field content as that on the $D5$ branes, except that we use an $\Sp(K)$ gauge group instead of $\Sp(N)$, and that the roles of the coordinates $v^\da$ and $w_1,w_2$ reversed. The coupling to the fields of the type I string theory is also the same.  

As on the $D5$ brane, operators on the $D5^\prime$ brane live in towers whose lowest lying elements have no or minimal dependence on the two scalar fields $\phi^\prime_1,\phi^\prime_2$. (These scalar fields are the bosonic ghosts associated to $(\alpha,\beta)\in\Pi\Omega^{0,1}(\PT,(\Oo(-1)\oplus\Oo(-3))\otimes \wedge^2_0\mrm{F}_K)$; the twistor uplifts of space-time fermions.) A field of the type I string theory which depends in some polynomial way on $w_1,w_2$ will couple to an operator on the $D5$ brane which will have similar polynomial dependence on the scalar fields $\phi_1^\prime, \phi_2^\prime$.  Fields which are independent of $w_1,w_2$ will couple to the  simplest operators.

In particular, this means that if the value of an operator on the $D5$ is a polynomial only in $z$, and constant in $w_1,w_2$, then the field it sources will couple to a very simple operator in the $D5^\prime$ system. Conversely, operators which are set to polynomials in $w_1,w_2$ will source fields that couple to complicated operators on the $D5^\prime$ system.  It will be difficult to determine how such fields affect the self-dual gauge theory on $\R^4$.  For this reason, we will focus on vacua where the operators are set equal to some polynomial only in $z$.

\paragraph{Giving a VEV to the current for the $\mf{sl}(16)$ flavour symmetry}
Let us now describe explicitly the affect of the various non-trivial vacua of the $D5$ system has on the $D5^\prime$ system.  We will start with the vacuum where $\gamma_f\til{\gamma}^g$ has a VEV. This involves setting  
\be
\gamma_f\til{\gamma}^g = M_f^g(z)
\ee
where $M(z)$ is a polynomial of order two in $z$ valued in $\mf{sl}(16)$ matrices. Then the $\mrm{SL}(16)$ gauge field we source satisfies
\be
F^{0,2}(a_0) = 4\pi^2\delta_{v^\da = 0}  M(z)\,. \label{eqn:sourceflavour}
\ee
When restricted to twistor space, this gives a background gauge field for the $\mrm{SL}(16)$ flavour symmetry satisfying the same equation.  This flavour symmetry gauge fields couples only to fermion bilinears on twistor space.  If we had included $w_i$ dependence in the source equation \eqref{eqn:sourceflavour}, we would have sourced fields that also couple to open string words involving two fermions and a number of scalar fields.

When we apply the Penrose transform, a holomorphic bundle on twistor space satisfying \eqref{eqn:sourceflavour} will become an $\mrm{SL}(16)$ bundle on space-time which satisfies the SDYM equation with a source:
\be
F(A_0)_{\alpha \beta} = 4\pi^2M_{\alpha \beta} \delta_{x = 0} \label{eqn:sdym_source}
\ee
where we expand $M(z) = M_{11} + 2M_{12} z + M_{22} z^2$. 
\begin{lemma}
Up to gauge equivalence, there is a unique solution to equation \eqref{eqn:sdym_source} that extends to $S^4$. 
\end{lemma}
\begin{proof}
We can check this on the twistor space $\CP^3$ of $S^4$.  There, the equation is of the form
\be
F^{0,2}(a_0) = 4\pi^2\delta_{\CP^1}\,.
\ee
The $\delta$-function on a $\CP^1$ inside $\CP^3$ is naturally a $(2,2)$-form.  To make it a $(0,2)$-form, we need to contract it with a section of the exterior square of the normal bundle to the $\CP^1$, of which there are $3$, because this exterior square is $\Oo(2)$.  

The obstruction to solving an equation like this is in $H^{0,2}(\CP^3)$, which vanishes.  The moduli of solutions (up to gauge equivalence) lie in $H^{0,1}(\CP^3)$, which also vanishes. 
\end{proof}

Note that this bundle is obtained by inserting the operator $\exp(\op{tr}(B_{0,\alpha\beta}(0)M^{\alpha\beta}))$ to the path integral of SDYM theory for the group $\mrm{SL}(16)$.

This leads to the conjecture:
\begin{conjecture}
Consider the $\op{Sp}(K)$ self-dual gauge theory, with massless fermions in $\mrm{F}_K\otimes\C^{16}\oplus\wedge^2_0\mrm{F}_K$, in the presence of the background axion field \eqref{eq:theta} and the background gauge field \eqref{eqn:sdym_source} for flavour symmetry.  Then, collinear singularities in this theory are given by the quotient of the chiral algebra $\mc{A}^{N,K}_\mrm{Defect}$ of the coupled system, by the relation which sets all bulk operators to zero except setting $\gamma\til{\gamma}$ to $M(z)$.   
\end{conjecture}
We perform explicit checks of this shortly.

\paragraph{Giving a VEV to $\gamma$ bilinears}
We can make a similar conjecture for the other vacua.  If we set
\be
\gamma_f\gamma_g = N_{fg}
\ee
where $N$ is a constant tensor, then we source a background field for the components of the ten-dimensional $\mf{so}(32)$ gauge field which live in $\Omega^{0,1}(X,\wedge^2\C^{16}\otimes\Oo(-2))$.  This satisfies
\be
\dbar a_{0,fg} = 8\pi^2\delta_Y N_{fg}
\ee
using the fact that the $\delta$-function on $Y$ is twisted by $\Oo(-2)$.  Restricting to $\PT$, we find a similar equation.  On space-time, the fields $a_{0,fg}$ become background scalars $\varphi_{fg}$ which couple to the fermions $\psi^f_\da$ on space-time by
\be
\varphi_{fg}\eps^{\db\da}\psi^f_{\da}\psi^g_{\db}\,.
\ee
The source equation is
\be
\triangle\varphi_{fg} = 8\pi^2N_{fg}\delta_{x=0}
\ee
so that 
\be
\varphi_{fg} = - 2N_{fg}\frac{1}{\norm{x}^2}\,.
\ee
Thus, $\phi$ is a background pion-like field with a Yukawa interaction which involves only positive helicity fermions.

A similar (but somewhat more involved) analysis will determine the effect of turning on a VEV for the operator $\til{\gamma}^f\til{\gamma}^g$.  This can be set to an order $4$ polynomial of $z$, which we can encode in a four-index tensor
\be
T^{fg}_{\alpha_1\alpha_2\alpha_3\alpha_4}
\ee
which is symmetric in the spinor indices.

The background introduces a coupling to the fermions of the form
\be
- 2\int_{\R^4} \frac{\d^4x}{\norm{x}^2} T^{fg}_{\alpha_1\alpha_2\alpha_3\alpha_4}\eps_{\db_1\db_3}\p_{x_{\alpha_1\db_1}}\psi_f^{\alpha_2}\p_{x_{\alpha_3\db_3}}\psi_g^{\alpha_4}\,.
\ee
Thus, the background coming from giving $\til\gamma^f\til\gamma^g$ a VEV is not as natural as that coming from $\gamma_f\til{\gamma}^g$ or from $\gamma_f\gamma_g$. 

\paragraph{Giving a VEV to $\p_{w_i}\c$ bilinears and Burns space}
Recall that we can set
\be
\Tr \p_{w_i}\c\p_{w_j}\c = \omega^{ij} 
\ee
where $\omega^{ij}\d w_i\d w_j$ is a $\p$-closed holomorphic $(2,0)$-form on $Y$. This gives rise to a source for the Beltrami differential on $X$ (viewed as a $(4,1)$-form):
\be
\dbar \mu = \delta_{Y} \omega\,. \label{eq:beltramisource}
\ee
Let us take 
\be
\omega = \d w_1 \d w_2 z^2.
\ee
Then, the source equation is precisely the equation that sources the twistor space of the Burns metric \cite{Costello:2023hmi}. We are led to the following conjecture, which we will investigate in more detail shortly: 
\begin{conjecture}
Collinear singularities in scattering amplitudes in the integrable $\op{Sp}(K)$ gauge theory on Burns space are controlled by the quotient of the chiral algebra $\mc{A}^{N,K}_\mrm{Defect}$ of the coupled system, by the relation which sets
\be
\op{Tr} (\p_{w_1}\c \p_{w_2} \c) = z^2\,.
\ee
\end{conjecture}
It is worth noting that Burns space is conformally equivalent to $\mbb{CP}^2$, so one could also talk about scattering on $\mbb{CP}^2$.  Since Burns space is asymptotically flat, it seems more natural to discuss scattering in that geometry.

There are other possible ways to specialize the operators $\Tr(\p_{w_i}\c\p_{w_j}\c)$ which lead to interesting backgrounds on twistor space. For example, if we set
\be
\Tr(\p_{w_1}\c\p_z\c) = w_1\,,
\ee
we source a Beltrami differential satisfying
\be
\dbar \mu = \delta_{Y} w_1 \d w_1 \d z 
\ee
(again viewing $\mu$ as a $(4,1)$-form on $X$).\footnote{We saw in Sect. \ref{sec:how-to-do-OPEs} that the constant mode in $w_1,w_2$ of the state $\Tr(\p_{w_1}\c\p_z\c)$ is exact. This is compatible with the above since we are assigning a VEV which is linear in $w_1$.} The effect of this Beltrami differential on self-dual theory on $\R^4$  is to give a mass to the fermions $\psi_\da$ which transform in the representation $\wedge^2_0\mrm{F}_K$. More precisely, we get a term
\be
-\frac{1}{4\pi^2}\int_{\R^4}\frac{\d^4x}{\norm{x}^2} \eps^{\db\da}\op{tr}(\psi_\da\psi_\db). 
\ee
\paragraph{Giving a VEV to the stress tensor gives a background axion}
The next background to consider is that sourced by giving a VEV to the stress tensor 
\be
T^i = \Tr\bigg(\frac{1}{2}\phi_\da\p_{w_i}\phi^\da - \b\p_{w_i}\c + \dots\bigg)
\ee
(again using the notation where $i$ runs from $1$ to $3$, corresponding to coordinates $w_1,w_2, w_3 = z$).  We can set $T^i = f^i$ where $f^i \d w_i\,\d^3 w$ is holomorphic $1$-form on $Y$ twisted by $K_Y$.  This gives rise to a source for the field $\eta$ as in equation \eqref{eq:etasource}.

Giving a VEV to the stress tensor will correspond, when we pass to $\R^4$, to turning on a background axion field which is different from the background $-\tfrac{N}{4\pi^2}\log\norm{x}$ sourced by the brane.  For example, let us look at the case when we set
\be
T_z = \Tr\bigg(\frac{1}{2}\phi_\da\p_z\phi^\da - \b\p_z\c + \dots\bigg) = 1.
\ee
In that case, the equation for $\eta$ (which we view as a $(2,\ast)$-form) reads
\be
\dbar\eta = \eps^{\db\da} \p \left( \iota_{\p_{v^{\da} } } \iota_{\p_{v^{\db} } } \d z \delta_{v^\dc=0} \right) \,.
\ee
Manipulations using the Cartan homotopy formula allow us to rewrite this equation as
\be
\dbar\eta = \eps^{\db\da} \mc{L}_{z\p_{v^\da}} \mc{L}_{\p_{v^\db}}  \delta_{v^\dc = 0} \label{eq:etasource2}
\ee
where $\mc{L}$ indicates Lie derivative.  The equation on twistor space which sources the logarithmic axion field is
\be
\dbar\eta = \delta_{v^\dc = 0}\,.
\ee
The vector fields $\p_{v^\da}$ and $z\p_{v^\db}$ correspond to the vector fields $\p_{u^\da}$ and $\p_{{\hat u}^\db}$ on space-time.  This means that equation \eqref{eq:etasource2} sources an axion field on space-time which is
\be
\rho = - \frac{1}{2}\triangle\bigg(\frac{1}{8\pi^2}\log\norm{x}\bigg) = \frac{1}{8\pi^2\norm{x}^2}\,.
\ee

\paragraph{The remaining bulk operator} 
The final bulk operator we consider is 
\be
\frac{\eps^{ijk}}{3!}\Tr\big((3\phi_\da\p_{w_i}\phi^\da + 2\b\p_{w_i}\c)\p_{w_j}\c\p_{w_k}\c\big) \,.
\ee
We can set this operator to a polynomial $F(z)$ of order $4$ in $z$. The (linearized) source equation is
\be
\dbar\mu  = \eps^{\db\da}\p(\iota_{v_{\da}} \iota_{v_{\db}}\delta_Y\d z\d w_1\d w_2 F(z) ) \label{eq:beltramisource2}
\ee
where the Beltrami differential $\mu$ is viewed as a $(4,1)$-form on $X$. 

Viewing the Beltrami differential again as a $(0,1)$-form valued in $TX$, we note that the only terms sourced by this equation are those proportional to $\p_{v^{\da}}$ and $\p_{v^{\db}}$.  The Beltrami differential is naturally divergence free. This means that, under Penrose's non-linear graviton construction \cite{Penrose:1976js}, the Beltrami differential sourced in equation \eqref{eq:beltramisource2} is given by a self-dual solution to the Einstein equations on space-time.

To determine explicitly which solution, let us introduce the twistor uplift of the self-dual GR, described by fields $h \in \Omega^{0,1}(\PT, \Oo(2))$ and $\til{h} \in \Omega^{0,1}(\PT, \Oo(-6))$.  The Lagrangian is \cite{Mason:2007ct}
\be
\frac{\i}{2\pi}\int_\PT \d v^\done \d v^\dtwo \d z\,\til{h}\bigg(\dbar h + \frac{1}{2}\eps^{\db\da}\p_{v^\da} h\p_{v^\db} h\bigg)\,.
\ee
We can write $\mu$ in terms of $h$ by
\be
\mu^{\da} = \eps^{\da\db}\p_{v^\db}h\,.
\ee
Then, the linearized source equation \eqref{eq:beltramisource2}, when written in terms of $h$, is
\be
\dbar h = F(z) \delta_{v^\dc=0}\,.
\ee
The full non-linear source equation is 
\be
\dbar h + \frac{1}{2}\eps^{\db\da}\p_{v^\da}h\p_{v^\db}h = F(z) \delta_{v^\dc=0}\,.
\label{eqn:hsource}
\ee
This is field sourced by the operator
\be
\frac{\i}{2\pi}\int_{\CP^1}\d z\,\til{h} F(z)\,.
\ee
In four-dimensions, this expression becomes an operator in self-dual GR built from the Lagrange multiplier field.

If  $F(z) = z^2$, this is exactly the expression found in \cite{Bittleston:2023bzp} which sources the Eguchi-Hanson metric. There is an important difference, however. The twistor space of Eguchi-Hanson involves a $\Z/2$ orbifold followed by a resolution of singularities.  Here, we do not take any $\Z/2$ quotient, and instead solve the sourced equation with asymptotically Euclidean boundary conditions.  This means that the resulting metric is \emph{singular}: indeed the only non-singular asymptotically Euclidean self-dual Einstein manifold is flat space \cite{Kronheimer:1989zs}.  This analysis does tell us that, away from the origin $x = 0$, the metric we source is isometric to the double-cover of the open subset of Eguchi-Hanson space where an $S^2$ is removed.

It will be helpful to work out the explicit form of the metric to first order. We will do this by gauge-fixing the metric, and writing it as the second derivative of a scalar field, following \cite{Plebanski:1975wn,Siegel:1992wd}.

To do this, let us solve the equation on twistor space, not in terms of $h \in \Omega^{0,1}(\PT, \Oo(2))$, but in terms of a field
\be
\gamma \in \Omega^{0,1}(\PT, \Oo(-2)) 
\ee
related to $h$ by
\be
h = F(z) \gamma\,.
\ee
A simple argument with exact sequences tells us that on-shell, every $h$ is gauge equivalent to one of this form. 

This tells us that every solution $g_{\alpha \da \beta \db}$ of the linearized self-dual Einstein equations can be described in terms of a scalar field $\varphi$.  The expression is
\be
g_{\da\alpha\db\beta} = \eps_{\da\db}\eps_{\alpha\beta} + F_{\alpha\beta}^{~\,~\,\gamma\eta} \p_{x^{\da\gamma}}\p_{x^{\db\eta}}\varphi\,,
\ee
where we have written the section $F(z)$ of $\Oo(4)$ as a symmetric tensor.

\begin{table}
\centering
    \begin{tabular}{c c c}
    \toprule
    \thead{Operator} & \thead{Background field type} & \thead{Field sourced} \\
    \midrule
    $\gamma_f\gamma_g$ & Pion-like field $\varphi_{fg}$ & $\tr\varphi_{fg} = \delta_{x=0}$ \\
    $\gamma_f\til\gamma^g$ & \makecell{Background field $A_0$ \\ for flavour symmetry} & $F(A_0)_- = \delta_{x=0}$ \\
    $\Tr\p_{w_1}\c\p_{w_2}\c$ & Metric $g$ & Burns metric \\
    Stress tensor $T_z = \frac{1}{2}\Tr\big(\phi_\da\p_z\phi^\da + \dots\big)$ & Axion field $\rho$ & $\tr\rho = \delta_{x=0}$ \\
    $\frac{1}{4}\eps^{ijk}\Tr\big(\phi_\da\p_{w_i}\phi^\da\p_{w_j}\c\p_{w_k}\c + \dots\big)$ & Metric $g$ & \makecell{Double cover \\ of Eguchi-Hanson}  \\
    \bottomrule
    \end{tabular}
    \caption{Background fields associated to giving a VEV to various operators}
    \label{tab:space-time-source}
\end{table}

In terms of the section $\gamma \in \Omega^{0,1}(\PT, \Oo(-2))$, the source equation is the simpler equation
\be
\dbar \gamma = \delta_{v^\da = 0} 
\ee
which, in terms of the scalar field $\varphi$, becomes
\be
\tr\varphi = \delta_{x = 0}\,.
\ee
We conclude that, up to some factors of $\pi$, and to leading order in the backreaction, the metric sourced is
\be
g_{\da\alpha\db\beta} = \eps_{\da\db}\eps_{\alpha\beta} - \frac{1}{4\pi^2}F_{\alpha\beta}^{~\,~\,\gamma\eta}\p_{x^{\da\gamma}}\p_{x^{\db\eta}}\bigg(\frac{1}{\norm{x}^2}\bigg)\,. 
\ee
It is not hard to see that (again to leading order in the second term) this expression gives a self-dual Einstein manifold away from $x = 0$. Indeed, the curvature tensor is a second derivative of the second term in $g$, and we have for instance
\be
W_{\alpha\beta\gamma\eta} = - \frac{1}{24\pi^2}F_{\alpha \beta \gamma \eta}\tr^2\bigg(\frac{1}{\norm{x}^2}\bigg) = \frac{1}{6}F_{\alpha\beta\gamma\eta}\triangle\delta_{x=0}
\ee
so that the anti-self-dual Weyl curvature is $F_{\alpha\beta\gamma\eta}$ times the Laplacian of the $\delta$-function at $x=0$.

In this work we will only perform the simplest possible check of the duality on the Eguchi-Hanson double-cover by matching the leading contribution to the two-point amplitude.  For these purposes it's enough to interpret the background as a graviton source, a perspective explained and utilized in Sect. \ref{sec:formfactors}, avoiding the need to perform a detailed analysis of scattering solutions on this singular geometry.  Nevertheless, it would be very interesting to perform a systematic analysis, and ultimately to recover the algebra deformation observed in \cite{Bittleston:2023bzp}.  We have summarized the four-dimensional background fields sourced by various bulk operators in Table \ref{tab:space-time-source}.


\section{Matching the Simplest Planar OPEs with Tree-Level Collinear Limits in Self-Dual Backgrounds} \label{sec:simple-background-OPEs}

We have argued that chiral algebras controlling self-dual gauge theory in certain backgrounds can be obtained as the large $N$ limit of quotients of the chiral algebra $\mc{A}^{N,K}_\mrm{Defect}$.  In this section, we will perform explicit checks of this to leading order. We will verify that tree-level scattering of two states (gluons or matter fields) on space-time, in the appropriate self-dual background, matches with the planar two-point function of the chiral algebra.  In this section, we will only work to leading order in the background fields.  Later, we will focus on particular background fields and show how to work to all orders.

Recall that each bulk operator of spin $-s$ was specified to be some polynomial $F(z)$ of order $2s$ (recalling that all such are of spin $\le 0$). We can also view $F(z)$ as a symmetric tensor of rank $2s$ in undotted spinor indices $F^{\alpha_1 \dots \alpha_{2s}}$. Similarly, the corresponding background field for self-dual gauge theory will not be Lorentz invariant, but will depend on the symmetric tensor $F$. Correlators when the bulk operator is given a VEV, and the scattering amplitudes in the presence of a background field, will depend on $F$.  We need some notation to indicate this dependence.

Suppose we are scattering some number states $1,\dots,n$  (or similarly taking correlators of $n$ operators).  These will be labelled by spinors $\lambda_{i,\alpha}$, $i = 1,\dots, n$, where as usual if we are computing correlators
\be
\lambda_i = (1,z_i)
\ee
and $z_i$ is the insertion point of the operator.  Given $2s$ states $i_1,\dots,i_{2s}$ (which are not necessarily distinct) we set 
\be
\ip{ i_1 \dots i_{2s}\mid F } = \lambda_{i_1,\alpha_{1}} \dots \lambda_{i_{2s},\alpha_{2s}} F^{\alpha_1 \dots \alpha_{2s} }\,.
\ee
In terms of the variables $z_i$, if we normalize $F$ so that it is monic and has roots $u_1,\dots,u_{2s}$, we have
\be
\ip{i_1 \dots i_{2s}\mid F} = \sum_{\sigma \in S_{2s} } \frac{1}{(2s)! } (z_{i_1} - u_{\sigma(1)} ) \dots (z_{i_{2s}} - u_{\sigma(2s)} )\,.
\ee

\paragraph{Background $\mf{sl}(16)$ flavour symmetry}
We have seen that when we set
\be
\gamma_f\til{\gamma}^g = M_f^g(z)
\ee
where $M(z)$ is a polynomial of order $2$, then on space-time we are working in a background $\mf{sl}(16)$ bundle which satisfies
\be
F(A_0)_{\alpha\beta} = 4\pi^2M_{\alpha\beta}\delta_{x=0}\,.
\ee
To leading order in $M$, the unique solution to this equation (up to gauge equivalence) is given by
\be
A_{0,\da\alpha} = - 2M_\alpha^{~\,\beta}\p_{x^{\da\beta}}\bigg(\frac{1}{\norm{x}^2}\bigg)\,.
\ee
To check that this solves the equation, recalling that we are working to leading order in $M$, we note that the field strength of the gauge field is
\bea
&F(A_0)_{\alpha\beta} = \frac{1}{2}\eps^{\db\da}\big(\p_{x^{\da\alpha}}A_{0,\db\beta} - \p_{x^{\db\beta}}A_{0,\da\alpha}\big) \\
&= \eps^{\db\da}\big(M_\alpha^{~\,\gamma}\p_{x^{\da\gamma}}\p_{x^{\db\beta}} +  M_\beta^{~\,\gamma}\p_{x^{\da\gamma}}\p_{x^{\db\alpha}}\big)\frac{1}{\|x\|^2}\,.
\eea
Using the fact that
\be
-\eps^{\db\da}\p_{x^{\da\alpha}}\p_{x^{\db\beta}}\bigg(\frac{1}{2\pi^2\norm{x}^2}\bigg) = \eps_{\alpha\beta}\delta_{x = 0}
\ee
the result holds. 

Now let us compute the leading order correction of this background to the two-point function.  Since we have a background gauge field for flavour symmetry, the scattering of the fermions $\psi^\da_g$, $\psi^{\alpha f}$ charged under this flavour symmetry will be modified.  Suppressing $\fsp(K)$ gauge indices, these fields of course couple to the background gauge field by
\be
\psi^\da_gA^g_{0,\da\alpha f}\psi^{\alpha f}\,.
\ee
To leading order in the matrix $M$, this means we add on the vertex
\be
2\psi^\da_gM^g_{\alpha\beta f}\psi^{\alpha f}\p_{x^\da_{~\,\beta}}\bigg(\frac{1}{\norm{x}^2}\bigg)\,.
\ee
Integrating by parts and using the Dirac equation for $\psi^\da_g$ brings this to
\be
-\frac{2}{\|x\|^2}\psi^\da_gM^g_{\alpha\beta f}\p_{x^\da_{~\,\beta}}\psi^{\alpha f}\,.
\ee
Let us compute the scattering in the presence of this two-point vertex, where $\psi^{\alpha f}$ has momentum $p_1$ and $\psi^\da_g$ has momentum $p_2$. As usual, we will encode each on-shell momenta into a pair of spinors $\lambda_{i,\alpha}$ and $\lt_{i,\da}$, for $i = 1,2$.  The two-component spinor $\lambda_i$ is $(1,z_i)$.  The scattering matrix can be read off the vertex; the factor of $1/2\pi^2\norm{x}^2$ is the Fourier transform of $2/\|p\|^2$ and so contributes $1/p_1 \cdot p_2 = 1/\ip{12}[12]$.  The expression
\be
\psi^\da_gM^g_{\alpha\beta f}\p_{x^\da_{~\,\beta}}\psi^{\alpha f}
\ee
contributes 
\be
[12]\lambda_1^\alpha\lambda_1^\beta M^g_{\alpha\beta f} = [12] \ip{11\mid M^g_f}\,.
\ee
The final answer is
\be
4\pi^2\Omega\frac{\ip{11 \mid M^g_f}}{\ip{12}}\,.
\ee
It's convenient to rewrite this in terms of $z_1,z_2$ instead of $\lambda_1,\lambda_2$. To do this we also rewrite $M^g_f$ as a quadratic polynomial in $z$.  In these terms, the answer is
\be \label{eq:flavour-2pt}
4\pi^2\Omega\frac{M^g_f(z_1)}{z_{12}}\,.
\ee

Let us compare this to the answer we get from the chiral algebra. The chiral algebra states corresponding to $\psi^{\alpha f}$ and $\psi^\da_g$ are  
\be
\M_f(\lt;z) = \gamma_fe^{\phi(\lt;z)}\psi\,,\qquad\til{\M}^g(\lt;z) = \til{\gamma}^g e^{\phi(\lt;z)}\psi\,.
\ee
We are interested in the leading-order OPE, as a function of the matrix $M$.  The leading order OPE will produce one copy of $\gamma_f\til{\gamma}^g$, giving
\be
\M_f(\lt_2;z_2) \til{\M}^g (\lt_1;z_1) \sim \frac{\Omega}{z_{12}}\gamma_f\til{\gamma}^g + \dots
\ee
from contraction of the $\psi$ fields.  The terms in the ellipsis are those involving the $\phi$ fields between two $\gamma$'s, or else terms resulting from Wick contraction of $\phi$ or $\gamma$ fields. These terms are all zero to the order we are working.

We have given a VEV to $\gamma_f\til{\gamma}^g$ in terms of $M$.  This gives the OPE
\be
\M_f(\lt_2;z_2)\til{\M}^g(\lt_1; z_1 ) \sim \Omega \frac{M^g_f(z_1)}{z_{12}} + \dots\,.
\ee
On the right hand side, we could evaluate $M^g_f$ at either $z_1$ or $z_2$, or at some point on the line between them. These choices differ by a regular term, not determined by the OPE.  However, the full correlation function must by invariant under the $\mrm{PSL}(2)$ conformal symmetry of the $z$ plane. Because of $\gamma_f\psi$ has spin $1/2$ whereas $\til{\gamma}^g\psi$ is of spin $-1/2$, the correlation function must use $z_1$ and not $z_2$, giving the expression above. Recalling that we should rescale the states $\M_f\mapsto\M_f/2\pi,\Mt^g\mapsto\Mt^g/2\pi$ in order to match the tree splitting functions, we find agreement with the space-time computation \eqref{eq:flavour-2pt}.

\paragraph{OPEs in the presence of a pion field mass}
We have seen that giving a VEV to the fermion bilinears 
\be
\gamma_f\gamma_g = N_{fg} 
\ee
where $N_{fg}$ is a constant tensor, amounts to introducing a pion-like background field
\be
\varphi_{fg} = - 2N_{fg}\frac{1}{\norm{x}^2} 
\ee
which couples to the fermions by
\be
- 2\int_{\R^4}\d^4x\,\eps^{\db\da}\psi_\da^f\psi_\db^g\frac{N_{fg}}{\norm{x}^2}\,. 
\ee
Following the discussion above, one can easily show that the two-point scattering in the presence of this vertex is
\be
- 4\pi^2\Omega\frac{N_{fg}}{\la12\ra}\,.
\ee
This is also what we find from the chiral algebra. The OPE of $\M_f(\lt_1,z_1)$ with $\M_g(\lt_2,z_2)$ is of the form
\be
\M_f(\lt_1,z_1)\M_g(\lt_2,z_2) \sim - \frac{\Omega}{\ip{12}} \gamma_f\gamma_g + \dots
\ee
coming from the contraction of a single pair of $\psi$ fields. The terms in ellipses do not contribute.  Since we have given a VEV to $\gamma_f\gamma_g$, we find the two-point function
\be
\ip{\M_f(\lt_1,z_1)\M_g(\lt_2,z_2)} = - \Omega\frac{N_{fg}}{\la12\ra}
\ee
matching the space-time computation upon making the replacements $\M_f\mapsto\M_f/2\pi,\Mt^g\mapsto\Mt^g/2\pi$.

We will not present details on the similar computation that is required to match the space-time and chiral algebra OPEs when we give a VEV to $\til{\gamma}^f \til{\gamma}^g$. 

\paragraph{Burns space two-point function and $\p_{w_i}\c$ bilinears}
Next we will check that if we give a VEV to $\Tr\p_{w_1}\c\p_{w_2}\c$, the two-point function in the chiral algebra matches that on Burns spaces.  In this section, we will work to leading order in the deformation away from the flat geometry.  We will discuss the all-order version of this computation in Sect. \ref{sec:Burns-all}. The $4d$ version of this computation was performed in \cite{Costello:2023hmi}. We redo it here for completeness. 

We will work somewhat more generally than just Burns space, and consider setting
\be
\Tr(\p_{w_1}\c\p_{w_2}\c) = H(z)
\ee
where $H(z)$ is a polynomial of order $4$ in $z$, corresponding to a four-index symmetric tensor $H_{\alpha\beta\gamma\delta}$. 

We have seen that when $H(z) = z^2$, this corresponds to deforming to Burns space.  The metric for generic $H$ is related by a coordinate change to the Burns space metric, and is given (to leading order in $H$) by
\be
g_{\da\alpha\db\beta} = \eps_{\da\db}\eps_{\alpha\beta} - \frac{1}{8\pi^2}H_{\alpha\beta}^{~\,~\,\gamma\eta}\p_{x^{\da\gamma}}\p_{x^{\db\eta}}\log\norm{x}\,.  
\ee
This couples to the stress tensor of the four-dimensional gauge theory
\be
2\op{tr}(B_{\alpha\beta}F(A)_{\da\db})
\ee
giving the vertex
\be
\frac{1}{4\pi^2}\int_{\R^4}\d^4x\,H^{\alpha\beta\gamma\eta}\op{tr}\big(B_{\alpha\beta}F(A)^{\da\db}\big)\p_{x^{\da\gamma}}\p_{x^{\db\eta}}\log\norm{x}\,.  
\ee
(There are also terms coupling to the matter field components of the stress tensor, but here we will only analyze the gauge fields.)

Integration by parts, together with the linearized Bianchi identity $\p_{x^{\db\alpha}} F_\da^{~\,\db} = 0$, brings this to the form
\be
\frac{1}{4\pi^2}\int_{\R^4}\d^4x\,H^{\alpha\beta\gamma\eta}\op{tr}\big(\p_{x^{\da\gamma}}\p_{x^{\db\eta}}B_{\alpha\beta}F(A)^{\da\db}\big)\log\norm{x}\,.
\ee
Since $-\tfrac{1}{4\pi^2}\log\norm{x}$ is the Fourier transform of $2/\|p\|^4$, this factor contributes 
\be
\frac{1}{2\ip{12}^2[12]^2}\,.
\ee
The factor
\be
H^{\alpha\beta\gamma\eta}\op{tr}\big(\p_{x^{\da\gamma}}\p_{x^{\db\eta}}B_{\alpha\beta}F(A)^{\da\db}\big)
\ee
contributes
\be
4\Omega_{J(K}\Omega_{L)I}\ip{1^4\mid H}[12]^2 
\ee
so that the two-point function is
\be
2\Omega_{J(K}\Omega_{L)I}\frac{\ip{1^4 \mid H}}{\ip{12}^2}\,.
\ee
Rewriting this in terms of the function $H(z)$ and using $\lambda_1 = (1,z_1)$ we get 
\be
2\Omega_{J(K}\Omega_{L)I}\frac{H(z_1)}{\ip{12}^2}\,.
\ee

From the chiral algebra perspective, it is easy to reproduce this result.  The OPE of $\Jt_{IJ}[0] = \psi_I\p_{w_1}\c\p_{w_2}\c\psi_J$ with $\J_{KL}[0] = \psi_K\psi_L$ which contracts both $\psi$ fields is
\be
\frac{(\Omega_{JK}\Omega_{LI} + \Omega_{JL}\Omega_{KI})}{\ip{12}^2}\Tr(\p_{w_1}\c\p_{w_2}\c\big)(z_1)\,.
\ee
Since we have specified a VEV of $\Tr(\p_{w_1}\c\p_{w_2}\c)(z) = H(z)$, we find the chiral algebra OPE gives
\be
2\Omega_{J(K}\Omega_{L)I}\frac{H(z_1)}{\ip{12}^2} 
\ee
just as in the space-time computation.  As in the case of the background gauge field for flavour symmetry, there is some ambiguity about where $H(z)$ is evaluated, as some different possibilities can differ by regular terms. However, the spin of the fields $\Jt_{IJ}[0]$ and $\J_{KL}[0]$ tell us that in the two-point function, we must evaluate at $z_1$.

\paragraph{Matching OPEs for the stress tensor and background axion}
We have seen that when we give a VEV to the stress tensor
\be
T_z = \Tr\bigg(\frac{1}{2}\phi_\da\p_z\phi^\da - \b\p_z\c + \dots\bigg)
\ee
corresponds to turning on a background axion of the form $\rho = 1/8\pi^2\norm{x}^2$.  We need to match the two-point gluon scattering in the presence of this background axion with the chiral algebra two-point function.

What we will find is that
\be
\J^{(1)}_{IJ}(\lt_1;z_1)\J^{(1)}_{KL}(\lt_2;z_2) \sim 2\Omega_{J(K}\Omega_{L)I}\frac{[12]}{\la12\ra} T_z + \dots
\ee
where the terms in ellipses are other open or closed strong states which do not get a VEV. This will tell us that, when we give a VEV to $T_z$, we get the two-point function
\be \label{eq:axion-derivative-2pt}
\ip{\J^{(1)}(\lt_1;z_1)\J^{(1)}(\lt_2;z_2)} = 2\Omega_{J(K}\Omega_{L)I}\frac{[12]}{\la12\ra}\,.
\ee
This we can then compare with the scattering of gluons in the presence of an axion. 

At first sight, this looks impossible, because the field $\b$ does not appear in any of our open string operators, and so apparently can not appear in the OPEs.

This difficulty is solved by a careful analysis of contributions from the loop correction to the BRST operator.  Let us consider the coefficient of $1/z$ in the planar OPE of $\J_{IJ}[1,0](z_1) = \psi_I \phi_\done \psi_J$ with $\J_{KL}[0,1](z_2) = \psi_K \phi_\dtwo \psi_L$.  The OPEs which involve $\phi-\phi$ contractions do not contribute, because they are either non-planar or they are planar but produce two faces, which does not contribute to first order in the back reaction.

Thus, there are terms with a single $\psi-\psi$ contraction, or with two $\psi-\psi$ contractions.  The term with a single $\psi-\psi$ contraction is 
\be
\frac{1}{z_{12}}\left(\Omega_{JK}\psi_I  \phi_\done \phi_\dtwo\psi_L +  
\Omega_{IK}\psi_J  \phi_\done \phi_\dtwo\psi_L +  \Omega_{JL}\psi_I  \phi_\done \phi_\dtwo\psi_K + \Omega_{IL}\psi_J \phi_\done \phi_\dtwo\psi_K \right)\,. \label{eqn:notsymmetrized}
\ee
The term with two $\psi$ contractions is, at order $1/z_{12}$,
\be
\frac{1}{z_{12}}\left(\Omega_{JK}\Omega_{LI} + \Omega_{JL}\Omega_{KI}\right)\Tr (\p_z\phi_\done \phi_\dtwo)\,. \label{eqn:closedstring}
\ee
We would like to write this order $1/z_{12}$ OPE as a sum of four terms:
\begin{enumerate}
\item A closed string state.
\item An open string state.
\item A BRST exact state.
\item Double-trace terms.
\end{enumerate}
Double-trace terms can then be dropped, as we work in the planar limit.  

The expression \eqref{eqn:notsymmetrized} is not a BRST closed open string state, because the $\phi$'s in the string are not symmetrized.  Similarly, the closed string state \eqref{eqn:closedstring} is not BRST closed.  Indeed, the BRST variation is easily seen to be
\be
\frac{1}{z_{12}}\left(\Omega_{JK}\Omega_{LI} + \Omega_{JL}\Omega_{KI} \right)\Tr([\p_z\c,\phi_\done]\phi_\dtwo)\,.
\ee
The only BRST closed state with the same quantum numbers as \eqref{eqn:closedstring} is the stress tensor.

The BRST variation of \eqref{eqn:notsymmetrized} and \eqref{eqn:closedstring} cancel.  The relevant term in the BRST variation of \eqref{eqn:notsymmetrized} comes from two Wick contractions of the BRST current $\psi_M\c\psi^M/2$ with \eqref{eqn:notsymmetrized}, yielding a closed string operator.  

To determine the decomposition into a sum of an open and closed string operator, we will apply the BRST variation of $\psi_M\b\psi_N$.  We will compute the BRST variation
\be
\frac{1}{2}Q_\mrm{BRST}\left( \Omega_{JK}\psi_I \b \psi_L +  
\Omega_{IK}\psi_J \b \psi_L +  \Omega_{JL}\psi_I \b \psi_K + \Omega_{IL}\psi_J \b \psi_K \right)\,.  \label{eqn:BRSTvariation}
\ee
There are several terms, the simplest is that associated to the BRST variation
\be
Q_\mrm{BRST}\b = [\phi_\done,\phi_\dtwo] + \gamma_f\til\gamma^f\,.
\ee
The terms like $\psi_M\gamma_f\til\gamma^f\psi_N$ can be dropped as they are double-trace and so sub-planar.  There is an additional contribution from the double Wick contraction of the BRST current with $\psi_M\b\psi_N$, so that the BRST variation \eqref{eqn:BRSTvariation} is
\bea
&- \frac{1}{2}\left(\Omega_{JK}\psi_I[\phi_\done,\phi_\dtwo]\psi_L +  
\Omega_{IK}\psi_J [\phi_\done,\phi_\dtwo]\psi_L + \Omega_{JL}\psi_I[\phi_\done,\phi_\dtwo]\psi_K + \Omega_{IL}\psi_J[\phi_\done,\phi_\dtwo]\psi_K\right) \\
&- \left(\Omega_{JK}\Omega_{LI} + \Omega_{JL}\Omega_{KI}\right)\Tr(\b\p_z\c)\,.
\eea
From this we conclude that we can write the order $1/z_{12}$ OPE of $\J[1,0]$ with $\J[0,1]$ as a sum of an open string state, a double-trace state, a BRST exact state, and 
\be
2\Omega_{J(K}\Omega_{L)I}\Tr\left(\p_z\phi_\done \phi_\dtwo - \b\p_z\c\right)\,.
\ee
This is the stress tensor plus $\p_z\Tr(\phi_\done\phi_\dtwo)/2$.  This total derivative does not contribute as we are not giving a VEV to $\Tr(\phi_\done\phi_\dtwo)$.

It's straightforward to compute the amplitude in the corresponding axion background. The relevant vertex is
\be \frac{1}{8\pi^2}\int_{\R^4}\frac{\d^4x}{\|x\|^2}\,\op{tr}(F_{\da\db}F^{\da\db})\,, \ee
and plugging momentum eigenstates into the quadratic piece gives
\be 2\Omega_{J(K}\Omega_{L)I}\frac{[12]}{\la12\ra}\,, \ee
agreeing with equation \eqref{eq:axion-derivative-2pt}.

\paragraph{Eguchi-Hanson and the remaining closed string state}
The final closed string state we will analyze is
\be
\mc{O} = \frac{\eps^{ijk}}{12}\Tr\big( (3\phi_\db\p_{w_i}\phi^\db - 2\b\p_{w_i}\c)\p_{w_j}\c\p_{w_k}\c + \dots\big)  \label{eq:finaloperator}
\ee
(which is written in coordinates $w_1,w_2,w_3 = z$).  Since this operator is of spin $-2$,  we are setting this equal to $F(z)$, a polynomial of order $4$ in $z$. 

As in the discussion of the stress tensor, it seems at first sight that it is impossible to produce this operator from the OPE of two gluon states, because this operator includes $\b$.  However, we can perform an analysis similar to the one used in the case of the stress tensor.  The order $1/z$ pole in the OPE of $\J^{(1)}$ with $\Jt^{(1)}$ produces $\Tr(\phi_\da\p_z\phi^\da\p_{w_1}\c\p_{w_2}\c)$ plus open string states which are not BRST closed.  By using the relations coming from the BRST transformation of $\psi_I\b\p_{w_1}\c\p_{w_2}\c\psi_J$, we can write the order $1/z$ pole as a sum of double-trace terms, BRST closed open string terms, and the operator $\mc{O}$ of \eqref{eq:finaloperator}.

We conclude that
\be
\Jt^{(1)}_{IJ}(\lt_1;z_1)\J^{(1)}_{KL}(\lt_2;z_2)\sim4\Omega_{J(K}\Omega_{L)I}\frac{[12]}{\ip{12}}\mc{O} + \dots 
\ee
where the ellipses indicate terms that are zero when we give a VEV to $\mc{O}$.   
To compute the two-point function, we use the OPE together with conformal invariance to conclude that 
\be
\ip {\Jt^{(1)}(\lt_1;z_1)\J^{(1)}(\lt_2;z_2)} = 4\Omega_{J(K}\Omega_{L)I}\frac{[12]\ip{1^4\mid F}}{\ip{12}}\,.
\ee

Let us check that the same result holds from the space-time perspective. The background metric, as a first order variation of the flat metric, is 
\be
g_{\da\alpha\db\beta} = \eps_{\da\db}\eps_{\alpha\beta} - F_{\alpha\beta}^{~\,~\,\gamma\eta}\p_{x^{\da\gamma}}\p_{x^{\db\eta}}\bigg(\frac{1}{4\pi^2\norm{x}^2}\bigg)\,.
\ee
This is identical to the case of Burns space, except the potential is $- 1/4\pi^2\norm{x}^2$ instead of $- \log\norm{x}/8\pi^2$. This couples to the stress tensor of the four-dimensional gauge theory giving the vertex
\be
\frac{1}{2\pi^2}\int_{\R^4}F^{\alpha\beta\gamma\eta}\op{tr}\big(B_{\alpha\beta}F(A)^{\da\db}\big)\p_{x^{\da\gamma}}\p_{x^{\db\eta}}\bigg(\frac{1}{\|x\|^2}\bigg)\,.
\ee
As in the case of Burns space, we can integrate by parts and use the Bianchi identity to move all derivatives to act on $B$.  Then, $1/2\pi^2\norm{x}^2$ contributes $1/\la12\ra[12]$, and the remaining terms give us $\ip{1^4\mid F}[12]^2$, so the amplitude is
\be 4\Omega_{J(K}\Omega_{L)I}\frac{[12]\ip{1^4\mid F}}{\la12\ra}\,, \ee
matching the chiral algebra computation. 


\section{From Collinear Singularities to Amplitudes in Self-Dual Backgrounds} \label{sec:amplitude-bootstrap}

So far we have described how to compute collinear singularities in certain self-dual backgrounds in terms of the large $N$ chiral algebra.  For the simplest two-point amplitudes in the previous section, we have seen that symmetry considerations imply that the amplitude is determined by its singularities.  In this section, we will investigate more generally under which circumstances we can write down amplitudes, and not just their collinear singularities.

The self-dual backgrounds we are considering have point sources.  Tree-level amplitudes in such backgrounds have ultra-violet singularities of the same kind as loop-level Feynman diagrams.  Indeed, tree-level Feynman diagrams in such backgrounds are simply the same as loop-level Feynman diagrams in a setting where the background field is treated as dynamical, and we insert an appropriate operator to source the background.  

Because of these UV singularities, we would expect there to be renormalization ambiguities in defining these amplitudes, coming from the ambiguity in defining the finite part of any counter-terms needed to remove the UV singularities.  The counter-terms are local operators of the self-dual gauge theory inserted at the location of the source for the background field.   

In this section, we will analyze this question from the point of view of the chiral algebra, and find that from the chiral algebra perspective exactly the same renormalization ambiguities appear.

The analysis of counter-terms in the background field context is a little different in one respect to the ordinary analysis of counter-terms on flat space.  The background we are working on depends on some parameter with non-trivial dimension and space-time indices.  For instance, in the case of a background field for the flavour symmetry, the background field satisfies
\be
F(A_0)_{\alpha\beta f}^g = 4\pi^2M_{\alpha \beta f}^g\delta_{x = 0}
\ee 
where $M_{\alpha \beta f}^g$ is the parameter for the background field.

Since the $\delta$-function has dimension $4$, and $F(A_0)_{\alpha\beta}$ has dimension $2$, the parameter $M_{\alpha\beta f}^g$ has dimension $-2$, and of course has spinor and flavour symmetry indices.  When we study counter-terms in the presence of the background field, at order $k$ in the parameter $M$, the counter-term is an expression like
\be
\mc{O}^{\alpha_1\beta_1\dots\alpha_k\beta_k  g_1 \dots g_k}_{f_1 \dots f_k} M_{\alpha_1\beta_1f_1}^{g_1} \dots M_{\alpha_k\beta_kf_k}^{g_k}
\ee
which is dimensionless, and as indicated is uncharged under rotations and flavour symmetry.  This means in particular that $\mc{O}$ is of dimension $2k$. 

The fact that $\mc{O}$ is of dimension $2k$ at order $k$ in $M$ means that there is much more freedom to choose counter-terms in this context than there is when one studies counter-terms for ordinary renormalizable theories, where counter-terms are typically dimensionless.  Indeed, the problem of constructing counter-terms in the presence of a background field has more in common with the study of non-renormalizable Lagrangians, such as Yang-Mills theory in dimension greater than $4$.

However, with a judicious choice of the background field parameter, we will show that counter-terms can be fixed uniquely by symmetry.  We will see that this happens, for instance, in the case of a background field for flavour symmetry of the form
\be M_{\alpha \beta f}^g = \xi_\alpha\xi_\beta D_f^g \ee
for a left-handed reference spinor $\xi$ and matrix $D$.  


\subsection{Amplitudes and Conformal Blocks}
\label{sec:uni_confblock}
From the chiral algebra perspective, we need to determine ways of giving consistent correlation functions for the chiral algebra deformed by the self-dual background.  In general, a way of giving correlation functions to a chiral algebra compatible with the OPEs is called a \emph{conformal block}. 

At infinite $N$, when we turn off all the background fields and work in the planar limit, the chiral algebra is that describing the collinear singularities in tree-level self-dual gauge theory, with gauge group $\mf{sp}(K)$ and matter as discussed before.  In \cite{Costello:2022wso} it was shown that the conformal blocks for this chiral algebra are precisely the \emph{gauge invariant local operators} of self-dual gauge theory.  Under this identification, correlators in a conformal block are form factors for the corresponding local operator. 

It is important for this analysis that we are working with the derived (or homological) space of conformal blocks \cite{beilinson2004chiral}; this arises as the homology of a natural cochain complex.  As we change the structure constants of the chiral algebra, the size of this cochain complex does not change, but its differential does.  Thus, there is a spectral sequence which computes the conformal blocks of the \emph{deformed, quantized} chiral algebra in terms of the those tree-level undeformed chiral algebra.  

Since there are no local operators for self-dual gauge theory in ghost number $1$ (or $-1$), the cohomology of this spectral sequence in  ghost number $0$ has no corrections.  We conclude that any ambiguity in defining a conformal block for the deformed algebra is given by adding a local operator.

This exactly matches the ambiguity we found by thinking about the question from a purely $4d$ gauge theory perspective.

\paragraph{Uniqueness of conformal blocks for the flavour background}
The flavour background involves identifying
\be
\gamma_f\til{\gamma}^g = M_f^g(z)
\ee
where $M$ is a matrix-valued polynomial of order $2$ in $z$.
From the gauge theory perspective, we are working in a background flavour symmetry gauge field whose field strength fails to be self-dual at the origin, 
$F(A_0)_{\alpha \beta f}^g = 4\pi^2M_{\alpha\beta f}^g\delta_{x^=0}$. 

Amplitudes which have a negative helicity gluon vanish in a flavour background. This is because gluons can only interact with the background field through the fermions, and in the self-dual theory only positive helicity gluons couple to the fermions.

We can see this from the chiral algebra perspective as well.  The flavour background involves giving a VEV to $\gamma \til\gamma$ and setting all other bulk fields to zero. Negative helicity gluon states involve the $\p_{w_i}\c$ field, which can not be Wick contracted away.  OPEs involving a negative helicity gluon, where all $\psi$ fields are contracted, necessarily give bulk operators other than $\gamma\til\gamma$ which are set to zero.

Because of this, it does not make sense to allow counter-terms for the flavour background involving the negative helicity gluon operator $B_{\alpha\beta}$.  This restriction will make it easier to show uniqueness of counter-terms; however, we need some other constraints.

The chiral algebra has a grading which corresponds on space-time to the dimension of operators/states.  This grading survives the deformation of the chiral algebra by the flavour background as long as we give $M_f^g$ dimension $-2$. This is because the flavour background is sourced by the operator $\op{tr}(B_0M)$, where $B_0$ is the anti-self-dual component of the flavour symmetry field strength, which is of dimension $2$.  If we would like the deformed algebra to be graded by dimension, it is necessary that $M$ is of dimension $-2$.

Since $M_{\alpha\beta f}^g$ is the parameter for the deformed algebra, the ambiguity in defining the conformal block will be an $M$-dependent local operator in self-dual gauge theory,
\be
\mc{O}^{\alpha\beta g}_fM_{\alpha\beta f}^g + \dots
\ee
where the ellipses indicate local operators with quadratic and higher dependence on $M$.  As mentioned above, this expression should be dimensionless, meaning that the coefficient of $M^k$ is of dimension $2k$. We further make the natural requirement that the whole expression is invariant under the action of the $\mf{sl}(16)$ flavour symmetry.

Let us see what are the possible operators at each order in the entries of $M_{\alpha\beta f}^g$.  To first order, we need some operator of dimension $2$. However, the lightest gauge invariant operator in self-dual gauge theory is of dimension $3$. 

To second order in $M_{\alpha\beta f}^g$, we need an operator of dimension $4$, with $4$ spinor indices, $2$ pairs of $\mf{sl}(16)$ indices, and which does not involve the negative helicity gluon operator $B_0$.  There are essentially two operators,\footnote{If we take into account the ways to separate the indices of these operators into two groups (like $M_{\alpha\beta f}^gM_{\gamma\eta d}^e$) each of these expressions corresponds to a small number of counter-terms.} which are 
\be \eps^{\beta\alpha}\delta^f_g\psi^\dc_eD_{x^{\dc\gamma}}\psi_{\eta d}\,,\qquad
\eps^{\beta\alpha}\eps^{\eta\gamma}\delta^f_g\delta^d_e\op{tr}(F_{\da\db}F^{\da\db})\,. \ee
(Here $D_{x^{\dc\gamma}} = \p_{x^{\dc\gamma}} + A_{\dc\gamma}$ is the gauge covariant derivative.) On-shell both of these expressions are total derivatives. The second term is the integrand for the second Chern class, and the first term can be seen to be a total derivative using the Dirac equation for $\psi^\dc$.  This means that if we are to integrate over the position of the source for the background gauge field, then this counter-term ambiguity will not contribute.  

At higher order in $M_{\alpha\beta f}^g$, there are more complicated expressions like this. However, all of them involve a contraction of at least one pair of undotted spinor indices with an $\eps$-tensor. Indeed, at order $k$ in $M$, a potential counter-term which involved no $\eps_{\alpha\beta}$ tensors would be of dimension $2k$ and have $2k$ undotted spinor indices.  The operator $B_{\alpha\beta}$ has two spinor indices and dimension $2$, $\psi_\alpha$ has dimension $3/2$ and one spinor index, and $D_{x^{\da\alpha}}$ has dimension $1$ and one spinor index.  All other terms in an operator have no undotted spinor indices and dimension at least $2$.  From this we see that the only way to build an operator of dimension $2k$ with $2k$ undotted spinor indices is to have $l$ copies of $B_{\alpha\beta}$ and $2(k-l)$ derivatives.  However, we have already seen that counter-terms involving $B$ are irrelevant.

We can restrict to operators without $\eps_{\alpha\beta}$ tensors by making a further specialization on $M_{\alpha\beta f}^g$.  We can ask that there exists a left-handed reference spinor $\xi_\alpha$ and a matrix $D_f^g$ such that 
\be
M_{\alpha\beta f}^g = \xi_\alpha\xi_\beta D_f^g\,.
\ee
Generically, we can take $D_f^g$ to be in a copy of $\mf{gl}(1) \subset \mf{sl}(16)$.  Viewing $M$ as a function of $z$, this constraint means $M(z) = f(z)^2 D$ where $f(z)$ is an order $1$ polynomial in $z$.  With this restriction on the form of the matrix $M$, there are no possible counter-terms. 

\paragraph{Uniqueness of conformal blocks for $\Tr\p_{w_1}\c\p_{w_2}\c$ and for $\Tr(\phi_\da\p_z\phi^\da\p_{w_1}\c\p_{w_2}\c + \dots)$}
Next, let us analyze the possible counter-terms that can arise for the two gravitational operators. The operator $\Tr\p_{w_1}\c\p_{w_2}\c$ can be set to a four-index tensor $H_{\alpha_1 \dots \alpha_4}$, which is of dimension $2$.  The operator $\Tr(\phi_\da\p_z\phi^\da\p_{w_1}\c\p_{w_2}\c + \dots)$ can be set to a four-index tensor $H_{\alpha_1 \dots \alpha_4}$ of dimension $4$.

At tree-level, since these operators both couple to the stress tensor, they involve one copy of $B$.  This implies that counter-terms when we give a VEV to these operators can involve at most one $B$ for each $H$. 

If we make the same assumption as before, that 
\be
H_{\alpha\beta\gamma\eta} = \xi_\alpha\xi_\beta\xi_\gamma\xi_\eta
\ee
for some spinor $\xi$, then it is not hard to show that conformal blocks are unique.  Indeed, in the case of the operator $\Tr\p_{w_1}\c\p_{w_2}\c$, the ambiguity in defining the conformal block at order $n$ in $H$ is a local operator in the $4d$ gauge theory of dimension $2n$ which is totally symmetric in $4n$ undotted spinor indices.  There are no such operators.

Similarly, for the operator $\Tr(\phi_\da\p_z\phi^\da\p_{w_1}\c\p_{w_2}\c + \dots)$, the ambiguity will be a local operator in the $4d$ gauge theory of dimension $4n$ which is totally symmetric in $4n$ undotted spinor indices, and with less than $n$ copies of $B$.  Again, there are no such operators.

Without this assumption on the form of $H$, there are ambiguities in defining conformal blocks. Conformal blocks do exist, but the choice of conformal block is the choice of a scheme for renormalizing in the presence of the background fields.

For example, for the Burns metric we give a VEV to $\Tr\p_{w_1}\c\p_{w_2}\c$ associated to $H = \xi^2\nu^2$, for spinors $\xi,\nu$.  The conformal block for the Burns metric, and other metrics with generic choices of $H$, are ambiguous.  Ambiguities in defining the conformal block to order $n$ in $H$ correspond to operators of dimension $2n$ with all dotted spinor indices contracted, and undotted spinor indices contracted with $n$ copies of $H$.  An example of such an operator with $n = 3$ is
\be
\op{tr}\big(F^{\da\db}D_{x^{\da\alpha}}D_{x^{\db\beta}}B_{\gamma\eta}\big)  \ip{H\mid H} H^{\alpha\beta\gamma\eta}\,. 
\ee
The two-point form factor of this operator is
\be
\ip{\Jt^{(2)}(\lt_1;z_1)\J^{(2)}(\lt_2;z_2)} \propto [12]^2\la1^4\mid H\ra\ip{H\mid H}.
\ee
This tells us that the correlators of the chiral algebra when we give a generic VEV to $\Tr\p_{w_1}\c\p_{w_2}\c$ have an ambiguity at cubic order in $H$ of this form.

For the particular VEV $H = \nu^2\xi^2$, corresponding to Burns space, it's curious to compare this perspective with that of \cite{Costello:2022jpg,Costello:2023hmi}.  There the authors also observed ambiguities in conformal blocks which they identified with the freedom to insert a local operator at the apparently singular origin of the backreacted geometry.  They also found a distinguished conformal block in which the manifest $\mrm{SU}(2)\times\mrm{U}(1)$ symmetry was enhanced to $\mrm{SO}(4)$, corresponding to completing the backreacted geometry to Burns by adjoining an exceptional $\CP^1$ at the origin. It'd be to interesting to leverage their results to fix the ambiguities in the Burns background $H = \nu^2\xi^2$ in self-dual gauge theory.


\section{Amplitudes in Burns-Type Backgrounds} \label{sec:Burns-all}

In this section, we analyze amplitudes in the background with VEV
\be
\Tr(\p_{w_1}\c\p_{w_2}\c) = H(z)\,. \label{eqn:burnstype_metric} 
\ee
We have seen that the special case when $H(z) = z^2$ corresponds to the Burns space metric.  In the spinor language, this corresponds to the case when 
\be
H_{\alpha\beta\gamma\eta} = \frac{1}{6}\left( \xi_{\alpha} \xi_{\beta} \nu_{\gamma} \nu_{\eta} + \text{permutations} \right)
\ee
where the permutations make the right hand side symmetric in the four indices.  We will write this as $H = \xi^2\nu^2$, where $\xi,\nu$ are spinors.

We have also seen that beyond leading order in $H$, there is not in general a unique conformal block for the algebra, so that correlation functions are not well defined.  This issue does not arise in the case when $H = \xi^4$. Only in this case are we able to compute amplitudes without any ambiguity.

The background metric corresponding to the VEV \eqref{eqn:burnstype_metric} is to leading order in $H$
\be
g_{\da\alpha\db\beta} = \eps_{\da\db}\eps_{\alpha\beta} - \frac{1}{8\pi^2}H_{\alpha \beta}^{~\,~\,\gamma\eta}\p_{x^{\da\gamma}}\p_{x^{\db\eta}}\log\norm{x}\,. \label{eqn:burnsmetric}
\ee
In the case when $H = \xi^2\nu^2$, this expression is known to be exact \cite{burns1986twistors} to all orders in $H$.  It follows that it remains exact when we take the limit $\nu\to\xi$.  


\subsection{The Two-Point Amplitude}

Let's begin by focusing on the two-point function calculation, since in \cite{Costello:2023hmi} it was shown that a chiral algebra computation agrees with the direct computations from space-time (the latter was first computed by Hawking et al. \cite{Hawking:1979hw}).  Our goal is to reproduce this calculation in our context.

The result of \cite{Costello:2023hmi} concerned, not self-dual gauge theory, but the $\mrm{WZW}_4$ model. This is a kind of gauge-fixed version of self-dual gauge theory which only has positive helicity fields $\J(\lt;z)$, corresponding to a single scalar field.  To match with the context of this paper, we need to modify the space-time calculation to deal with gauge fields and not a scalar field.

The second change we have to make compared to \cite{Costello:2023hmi} is that we are dealing here with a degeneration of the Burns space metric, which has $H(z) = z^2$.  This will modify the space-time computation of the two-point function, but in a fairly simple way.

Let us describe the space-time two-point function, as computed in \cite{Hawking:1979hw}.  We will write the two-point function, as computed on space-time, in terms of the scalar fields $\varphi(\lt;z)$, normalized so that in the flat space limit the state $\varphi$ becomes a plane wave. (This differs by a factor of $z$ from the normalization used in \cite{Costello:2023hmi}).  

The two-point function was found to be
\be
\ip{\varphi(\lt_1;z_1)\varphi(\lt_2;z_2)} = \frac{Nz_1 z_2}{z_{12}^2}J_0\left(\sqrt{\frac{4Nz_1z_2[12]}{z_{12}}}\right)
\ee
where $J_0$ is a Bessel function of the first kind and $N$ is the number of $D1$ branes.  Let us rewrite this when we view $H(z) = z^2$ as determined by a pair of spinors $\xi,\nu$.  In that language, the two-point function is
\be
\ip{\phi(\lt_1;z_1)\phi(\lt_2;z_2)} = \frac{N\ip{1\xi}\ip{1\nu}\ip{2\xi}\ip{2\nu}}{\la12\ra^2}J_0\left(\sqrt{\frac{4N\ip{1\xi}\ip{1 \nu} \ip{2\xi} \ip{2\nu}[12]}{\la 12 \ra}}\right)\,.
\ee

We would like to translate this into the two-point function for self-dual gauge theory. This was previously done by \cite{warner1982scattering}, but we will rederive it directly from the scalar field case. Since we are only considering tree-level two-point amplitudes, we may as well consider Abelian gauge theory.  In the gauge theory case, the cleanest interpretation of the two-point scattering amplitude is as an infinite sum of Feynman diagrams, where each diagram corresponds to a gluon in the presence of $n$ identical (self-dual conformal) gravitons, as depicted in Fig. \ref{fig:gluonscattering}.\footnote{The anti-self-dual Weyl curvature of these gravitons is a $\delta$-function supported at the origin $W_{\alpha\beta\gamma\delta} = \tfrac{1}{6}H_{\alpha\beta\gamma\delta}\delta_{x=0}$; however, they have non-vanishing Ricci curvature and self-dual Weyl curvature away from the origin.} This interpretation also applies to the scalar field computation.

\begin{figure}[ht]
    \centering
        \includegraphics{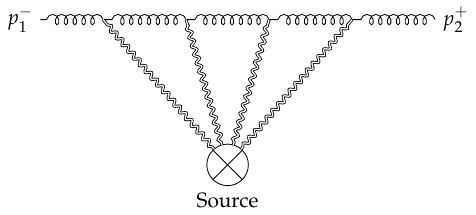}
    \caption{The scattering of gluons in the presence of a graviton with a point source. \label{fig:gluonscattering}}
\end{figure}

On flat space, we can gauge-fix SDYM theory and write it in terms of two scalars $\varphi,\til\varphi$.  We have
\be
B_{\alpha\beta} = \frac{2\xi_{(\alpha}\nu_{\beta)}}{\la\xi\nu\ra}\til\varphi\,,\qquad A_{\da\alpha} = \frac{2\xi_{(\alpha}\nu_{\beta)}}{\la\xi\nu\ra}\eps^{\gamma\beta}\p_{x^{\da\gamma}}\varphi\,.
\ee
The Abelian SDYM action becomes that of a free scalar $\til\phi\tr_g\phi$ for a background metric of the form
\be
g_{\da\alpha\db\beta} = \eps_{\da\db}\eps_{\alpha\beta} - \frac{1}{8\pi^2}H_{\alpha\beta}^{~\,~\,\gamma\eta}\p_{x^{\da\gamma}}\p_{x^{\db\eta}}\log\|x\|\,.
\ee
The metric here coincides with that considered in \cite{Costello:2023hmi} if we take $H = \xi^2\nu^2$ and make the replacement $N\to-1/8\pi^2$.

The coupling to the gauge-fixed SDYM and to the pair of scalars $\til\varphi,\varphi$ is then almost the same, it differs by a factor of
\be
\frac{\la1\xi\ra\la1\nu\ra}{\la2\xi\ra\la2\nu\ra}\,.
\ee
Let us consider how this factor changes the scalar field answer.  At the $n$\textsuperscript{th} order in the expansion, we have a single Feynman diagram where $n$ gravitons are connected to a gluon line, at vertices labelled $1,\dots,n$.  The momentum between the $i$\textsuperscript{th} and $(i+1)$\textsuperscript{th} vertex is $q_i$, so that the external momenta are $q_0=p_1$ and $q_n=p_2$, where $p_1$ is taken to be negative helicity and $p_n$ positive helicity. At the $i$\textsuperscript{th} vertex, we have a negative helicity gluon $(i-1)^-$ and a positive helicity gluon $i^+$. Compared to the computation in the scalar field, the factor at the $i$\textsuperscript{th} vertex is changed by $\la(i-1)\ra\la(i-1)\nu\ra/\la i\xi\ra\la i\nu\ra$.  These factors cancel out between adjacent vertices, except for a contribution of $\la1\xi\ra\la1\ra$ from the initial vertex and $1/\la2\xi\ra\la2\nu\ra$ from the final vertex.

We conclude that the scattering between a positive helicity and negative helicity gluon on Burns space is
\be \label{eq:gaugetheory-Besselfn}
\mc{A}(1^-,2^+) = \frac{\ip{1\xi}^2\ip{1\nu}^2}{2\la 12\ra^2}J_0\left(\sqrt{- \frac{\ip{1\xi}\ip{1\nu}\ip{2\xi}\ip{2\nu}[12]}{2\pi^2\la12\ra}}\right)\,.
\ee
Here the amplitude has be rescaled by a factor of $-1/4\pi^2$ to compensate for the unconventional normalization of the $\mrm{WZW}_4$ kinetic term.  We are interested in the limit $\nu\to\xi$.  In this limit \footnote{There is a slight subtlety here: the Burns space computation assumes that $\ip{\xi\nu} = 1$, so it could happen that equation \eqref{eq:gaugetheory-Besselfn} secretly includes various powers of $\ip{\xi\nu}$.  However, if one counts the powers of $H$ that appear -- one at each order in the backreaction -- one sees that this is not possible.} we get 
\be
\mc{A}(1^-,2^+) = \ip{\Jt(\lt_1;z_1)\J(\lt_2;z_2)} = \frac{\ip{1\xi}^4}{2\la 12\ra^2}J_0\left(\ip{1\xi}\ip{2\xi}\sqrt{-\frac{[12]}{2\pi^2\la12\ra}}\right)\,.
\ee

We would like to compare this with the chiral algebra computation.  The Feynman diagrams for the chiral algebra computation are indicated in Fig. \ref{fig:Burns2pt}.  In this figure, the bottom and top lines correspond to the states whose two-point functions we are computing. 

\begin{figure}[ht]
    \centering
        \includegraphics{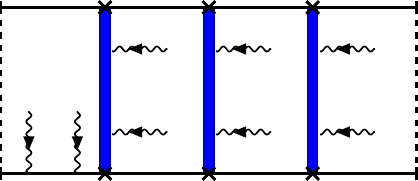}
    \caption{The two-point function on Burns space.  Here, the bottom and top lines indicate the states $\Jt^{(3)}(\lt_1;z_1)$ and $\J^{(3)}(\lt_2,z_2)$ respectively. \label{fig:Burns2pt}}
\end{figure}

Consider a diagram as in Fig. \ref{fig:Burns2pt}, where we are considering the two-point function of $\Jt^{(n)}(\lt_1;z_1)$ with $\J^{(n)}(\lt_2;z_2)$. (The diagram illustrated is the case $n = 3$.)  The two $\psi$ propagators contribute a factor of $1/\ip{12}$ and the $n$ bulk propagators contribute $- [12]/8\pi^2\la12\ra$.  So there is an overall factor of
\be
\bigg(\frac{1}{8\pi^2}\bigg)^n\frac{[12]^n}{\la12\ra^{n+2}}\,.
\ee
The number of copies of $H$ coincides with the number of faces.  The ambiguity in defining the chiral algebra correlator is the ambiguity of how one contracts the copies of $H$ with the spinors $\lambda_1,\lambda_2$ and with itself.  In general, the numerator will be an expression like
\be
\la1^4\mid H\ra^r \la1^22^2\mid H\ra^s \la H \mid H\ra^t \dots 
\ee
consisting of various possible ways of contracting the spinor indices.  If we take $H = \xi^4$ to be built from a single spinor, the numerator is fixed by the weights to be
\be
\ip{1\xi}^{2n + 4}\ip{2\xi}^{2n}\,.
\ee
The overall coefficient of the numerator is fixed by the operator product expansion to be $1$.

If we recall that the generating function $\Jt(\lt_1;z_1)$, $\J(\lt_2;z_2)$ are series in $\Jt[k,l](z_1)$ and $\J[k,l](z_2)$ with coefficients $1/k!l!$, we see that the two-point function computed by the chiral algebra is
\be\label{eq:Burns_2pt}
\frac{1}{2}\sum_{n\geq0}\frac{1}{(n!)^2}\frac{\ip{1\xi}^{2n+4}\ip{2\xi}^{2n}[12]^n}{(8\pi^2)^n\la12\ra^{n+2}}\,.
\ee
This is exactly the Bessel function that we found when computing this amplitude from the space-time perspective, as desired. (The overall factor of $1/2$ arises from the normalization of the invariant bilinear on $\fsp(K)$, which we've stripped.)


\subsection{The Three-Point Amplitude}

It is rather easy to generalize this calculation to compute the scattering of two positive helicity and one negative helicity gluon, as in Fig. \ref{fig:Burns3pt}.  The Feynman diagrams that contribute to the chiral algebra computation are illustrated in Fig. \ref{fig:3gluonscattering}.

\begin{figure}[ht]
    \centering
        \includegraphics{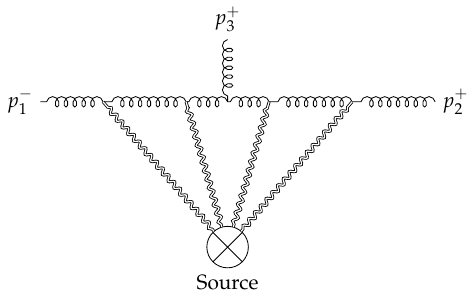}
    \caption{Two positive helicity and one negative helicity gluon scattering in the presence of a graviton with a point source. Although not depicted in this Feynman diagram, the source can also couple to the $3^+$ leg. \label{fig:3gluonscattering}}
\end{figure}

Let us compute the amplitudes of these diagrams where the first gluon is of negative helicity, and the second two of positive helicity. In any such diagram there are $k_{ij}$ $\phi-\phi$ propagators connecting the $i$\textsuperscript{th} to the $j$\textsuperscript{th} edge, for $i,j = 1,2,3$, $i \neq j$.  Any such edge contributes  $[12]/8\pi^2\ip{12}$.  The three $\psi-\psi$ propagators contribute the Parke-Taylor factor $1/\ip{12}\ip{23}\ip{31}$.  The fact that the states $\J^{(m)}(\lt;z)$ come with a factor of $1/m!$ means that the Wick contraction of all the edges contributes
\be
\frac{[12]^{k_{12}}[13]^{k_{13}}[23]^{k_{23}}} {(8\pi^2)^{k_{12}+k_{13}+k_{23}}\ip{12}^{k_{12} + 1}\ip{13}^{k_{13} + 1}\ip{23}^{k_{23} + 1}(k_{12}+k_{23})!(k_{23}+k_{13})!(k_{13} + k_{12})!}\,.
\ee
As above we assume that the Burns-type metric is associated to a tensor $H_{\alpha\beta\gamma\eta}$ of the form $H = \xi^4$ for a spinor $\xi$.  The amplitude is accompanied by an overall factor of
\be
\ip{1\xi}^{2k_{12} + 2k_{13} + 4} \ip{2\xi}^{2k_{12} + 2k_{23}}  \ip{3\xi}^{2k_{13} + 2k_{23}} 
\ee
where the exponents are determined by $\mrm{PSL}(2)$ symmetry.

\begin{figure}[ht]
    \centering
        \includegraphics{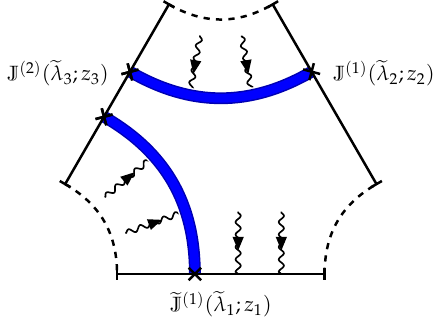}
    \caption{A simple chiral algebra Feynman diagram corresponding to scattering two positive helicity and one negative helicity gluon in the Burns-type background.  We are working at tree-level, so only planar diagrams contribute.  Each face of the diagram corresponds to a background graviton.}
    \label{fig:Burns3pt}
\end{figure}

Summing over all diagrams we find the total amplitude is
\bea\label{eq:Burns_3pt}
&\mc{A}_\mrm{tree}(1^-,2^+,3^+) = \frac{1}{2}\sum_{k_{12}, k_{13}, k_{23} \ge 0} \frac{\ip{1\xi}^{2k_{12} + 2k_{13} + 4}\ip{2\xi}^{2k_{12} + 2k_{23}}\ip{3\xi}^{2 k_{13} + 2 k_{23}}}{(8\pi^2)^{k_{12}+k_{13}+k_{23}}} \\ &\frac{[12]^{k_{12}} [13]^{k_{13}} [23]^{k_{23}}} {\ip{12}^{k_{12} + 1} \ip{13}^{k_{13} + 1} \ip{23}^{k_{23} + 1}(k_{12}+k_{23})!(k_{23}+k_{13})!(k_{13} + k_{12})!}\,.
\eea


\subsection{Differential Equations for Tree-Level \texorpdfstring{$n$}{n}-Point Amplitudes} \label{subsec:laplacetransform}

In this section, we will determine the tree-level amplitudes on the Burns-type geometry recursively, by deriving a differential equation for their Laplace transform.  The Burns-type geometry we use is that associated to the fourth power of a single spinor $\xi$.  

Consider the Burns-type amplitudes at tree-level,
\be
\mc{A}_\text{tree}(1^-,2^+, \dots,n^+;\xi)\,.
\ee
The Burns-type metric is encoded in a spinor $\xi$, so this amplitude is a function of $\xi$ and of the momenta.  

The differential equation is expressed, not in terms of the amplitude itself, but in terms of its Laplace transform
\be
\mc{L}\mc{A}_\text{tree}(1^-,2^+, \dots, n^+;\xi) = \int_{t_i=0}^\infty \d t_1 \dots \d t_n\,e^{-\sum t_i }\prod_{i = 1}^n t_i^{\lt_i \dpa{\lt_i}}   \mc{A}_\text{tree}(1^-,2^+,\dots,n^+;\xi)\,.
\ee 
The Laplace transform acts on the amplitude in the following way.  It replaces any occurrence of
\be
\prod_{i < j} [ij]^{a_{ij} }
\ee
by
\be
\prod_{i < j} [ij]^{a_{ij} } \prod_{k = 1}^n\Big(\sum_{l \neq k} a_{lk}\Big)!\,.
\ee
Amplitudes in Burns space become much simpler after the Laplace transform. For example, the two-point amplitude becomes
\be
\frac{\ip{1\xi}^4}{2\ip{12}^2}\left(1 - \frac{\ip{1\xi}^2\ip{2\xi}^2[12]}{8\pi^2\ip{12}}\right)^{-1}\,.
\ee
The three-point amplitude then admits a closed expression
\be \frac{\la1\xi\ra^4}{2\la12\ra\la23\ra\la13\ra}\prod_{1\leq i<j\leq3}\left(1 - \frac{\la i\xi\ra^2\la j\xi\ra^2[ij]}{8\pi^2\la ij\ra}\right)^{-1}\,. \ee

We can also view the Laplace transform as acting on the states we scatter (or the chiral algebra elements whose correlators we are taking) by replacing the exponential state
\be
\J_{IJ}(\lt;z) = \psi_I e^{\phi^\da\lt_\da} \psi_J 
\ee
by the Laplace transformed state
\be
\mc{L}\J_{IJ}(\lt;z) = \psi_I\bigg(\frac{1}{1 - \phi^\da\lt_\da}\bigg)\psi_J\,.
\ee
Similarly, for states of negative helicity, we have
\be
\mc{L}\Jt_{IJ}(\lt;z) = \psi_I\p_{w_1}\c\p_{w_2}\c\bigg(\frac{1}{1 - \phi^\da\lt_\da}\bigg)\psi_J\,.
\ee
We can apply the Penrose transform to write these Laplace transformed states in four-dimensional language. For a scalar field, the Laplace transform replaces a plane wave state $e^{\i p\cdot x}$ for a scalar field by $1/(1 - \i p\cdot x)$.  

For gluon states, the Laplace transform is slightly more complicated. A positive helicity gluon, expressed in terms of spinors $\lt$, $\lambda$, has field strength 
\be
 F = \frac{1}{2}\eps_{\alpha\beta}\d x^{\da\alpha}\,\d x^{\db\beta}\lt_\da\lt_\db e^{x^{\dc\gamma}\lt_\dc\lambda_\gamma}\,.
 \ee
 This positive helicity state is sent to is sent to
\be
\mc{L}F = \frac{1}{2}\eps_{\alpha\beta}\d x^{\da\alpha}\,\d x^{\db\beta}\lt_\da\lt_\db \frac{1}{(1 - x^{\dc\gamma}\lt_\dc\lambda_\gamma)^3}\,.
\ee

For negative helicity states, a plane wave is
\be
B = \frac{1}{2}\eps_{\da\db}\d x^{\da\alpha}\,\d x^{\db\beta}\lambda_\alpha\lambda_\beta e^{x^{\dc\gamma}\lt_\dc\lambda_\gamma}\,.
\ee

The Laplace transform sends this to the state
\be
\mc{L} B = \frac{1}{2}\eps_{\da\db}\d x^{\da\alpha}\,\d x^{\db\beta}\lambda_\alpha\lambda_\beta \frac{1}{(1 - x^{\dc\gamma}\lt_\dc\lambda_\gamma)}\,.
\ee

\paragraph{Differential equations for the Laplace transformed amplitude}
It is convenient to normalize the Laplace transformed amplitude slightly differently, and let
\be
\widehat{\mc{L}\mc{A}}_\text{tree}(1,\dots,n;\xi) = \frac{1}{\ip{1 \xi}^4}\mc{L}\mc{A}_\text{tree}(1^-,2^+\dots,n^+;\xi)\,. 
\ee
Graphically, $\widehat{\mc{L}\mc{A}}_\text{tree}$ is given by a sum over planar diagrams. The planar diagram can be viewed as a polygon with $n$ sides and some number of chords (that is, a $\phi-\phi$ propagator) stretched between the sides.  Each $\phi-\phi$ propagator contributes $\la i\xi\ra^2\la j \xi\ra^2/8\pi^2\la ij\ra$, and there is an overall Parke-Taylor factor $1/\ip{12} \dots \ip{n1}$.  

With the normalization we are using, $\widehat{\mc{L}\mc{A}}_\text{tree}$ is cyclically symmetric in the $n$ momenta, so we will no longer distinguish between positive and negative helicity states.  

We let 
\be
\mc{D}^{(i)} = \lt_i^\da \dpa{\lt_i^\da} 
\ee
be the dilatation operator on the $i$\textsuperscript{th} state.  

Then,
\be
\mc{D}^{(i)}\widehat{\mc{L}\mc{A}}_\text{tree}(1,\dots,n;\xi)  
\ee
is a sum over such planar diagrams, where we have also selected a $\phi-\phi$ starting at the $i$\textsuperscript{th} state.  By considering the effect of cutting along the diagram along this propagator, we find a differential equation for $\widehat{\mc{L}\mc{A}}_\text{tree}$.  

A $\phi-\phi$ propagator starting at the $i$\textsuperscript{th} side must end at some side $j\neq i$.  If we cut the planar diagram along this propagator, we turn the $\phi-\phi$ propagator into two $\psi-\psi$ propagators.  This gives us a disconnected diagram, leading to the differential equation
\bea \label{eq:diff_Burns_tree}
&\mc{D}^{(i)}\widehat{\mc{L}\mc{A}}_\text{tree}(1,\dots,n; \xi) \\
&= \sum_{j\neq i}\frac{[ij]\ip{ij}\ip{i\xi}^2\ip{j\xi}^2}{4\pi^2}\widehat{\mc{L}\mc{A}}_\text{tree}(i,i+1,\dots,j-1,j;\xi)\widehat{\mc{L}\mc{A}}_\text{tree}(j,j+1,\dots,i-1,i;\xi)\,.
\eea 
The factors on the right hand side arise because each the two $\psi-\psi$ propagator contributes a $1/\ip{ij}^2$, whereas the single $\phi-\phi$ propagator contributes 
\be
c_{ij} := \frac{[ij]\ip{i \xi}^2\ip{j \xi}^2}{8\pi^2\ip{ij}}\,.
\ee
The factor of $2$ again arises from our normalization of the invariant bilinear. We solve this differential equation in Appendix \ref{app:rec_Burns}, which leads to a recursion relation, expressing $n$-point tree level amplitudes in terms of $2, 3, \dots, n-1$-point tree level amplitudes:
\be
	\widehat{\mc{L}\mc{A}}_\text{tree}(1,\dots,n) = \left(\int_{0}^1ds\frac{(1-c_{12}s)(1-c_{1n}s)}{(1-c_{12})(1-c_{1n})}\mathcal{K}_{n-1}(s) + \frac{\ip{2n}}{\ip{12}\ip{1n}}\widehat{\mc{L}\mc{A}}_\text{tree}(2,\dots,n)\right)\,,
\ee
where
\begin{equation}
	\mathcal{K}_{n-1}(s) = \sum_{j = 3}^{n-1}2\ip{1j}^2c_{1j}\widehat{\mc{L}\mc{A}}_\text{tree}(s\widetilde{\lambda}_1,\widetilde{\lambda}_2,\dots,\widetilde{\lambda}_j)\widehat{\mc{L}\mc{A}}_\text{tree}(s\widetilde{\lambda}_1,\widetilde{\lambda}_j,\dots,\widetilde{\lambda}_n)\,.
\end{equation}


\subsection{One-Loop Amplitudes}

In SDYM theory in any background metric, there are no amplitudes at two-loops and higher. This is for simple combinatorial reasons: the theory has a $++-$ vertex and a $-+$ propagator, and there is no way to build a two-loop Feynman diagram with such ingredients.  (Of course, if we introduce another vertex, it can be possible to build higher-loop diagrams, but these would compute form factors).

The differential equation describing the Laplace transform of tree-level amplitudes can be extended to the Laplace transform of one-loop amplitudes. Before we do this, we should first check that the chiral algebra does not predict amplitudes at higher than one-loop.

\paragraph{Vanishing of amplitudes beyond one-loop}
From the perspective of the dual chiral algebra, loop number is computed using the standard large $N$ combinatorics. We will focus on single-trace correlators for simplicity.  Any single-trace diagram consists of some number of $\J$ and $\til \J$ states arranged in a cycle connected by $\psi-\psi$ propagators. This cycle is viewed as the boundary of a polygon, with some number $e$ of $\phi-\phi$ propagators which form chords on the polygon.  By viewing the $\phi-\phi$ propagators as double lines, we find a Riemann surface with some number $F$ of faces, including the ``outer'' face where the $\J$ and $\Jt$ states are placed.  Each face, except the outer face, has two lines emerging from it, labelled by $\p_{w_i}\c$. Each such face corresponds to an external graviton.

Let $\Sigma_\text{open}$ denote this Riemann surface, and $\Sigma_\text{closed}$ denote the closed surface obtained by filling in all of the faces.  

Each $\phi-\phi$ propagator gives rise to two $\p_{w_i}\c$s, as does each $\Jt$ state we insert.  There are two $\p_{w_i}\c$s for every face (except the outer face).  Also, the Euler characteristic of $\Sigma_\text{open}$ is minus the number of $\phi-\phi$ propagators. We conclude that
\begin{equation}
   - \chi ( \Sigma_{\text{open}} )  + \# \Jt = F - 1. 
\end{equation}
Using the fact that $\chi (\Sigma_{\text{closed}} ) = \chi(\Sigma_{\text{open}} ) + F$, we find
\begin{equation}
    \chi(\Sigma_{\text{closed}}) = \# \Jt + 1.
\end{equation}
Every closed connected Riemann surface has Euler characteristic $\le 2$.  The only surface with Euler characteristic $2$ is the sphere, which gives us tree-level amplitudes.  In this case, we must have only one insertion of a $\Jt$ state.

The only surface with Euler characteristic $1$ is $\mathbb{RP}^2$, corresponding to one-loop amplitudes.  Single-trace chiral algebra Feynman diagrams in this case have the topology of $\mbb{RP}^2$ with at least two discs removed, and all $\J$ states placed on the same boundary of the resulting surface (Fig. \ref{fig:BurnsRP2}). 

\begin{figure}[ht]
    \centering
        \includegraphics{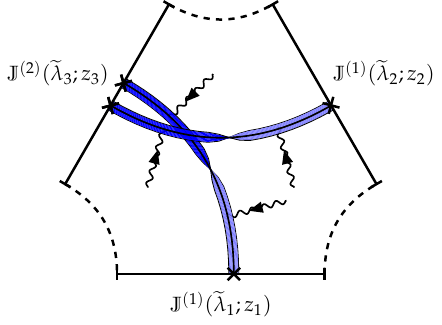}
    \caption{A one-loop three-point diagram.  This diagram has three faces, where the two internal faces each have two $\p_{w_i}\c$ fields.}
    \label{fig:BurnsRP2}
\end{figure}

\paragraph{Differential equations for one-loop amplitudes}
Now that we have shown that only one-loop amplitudes can occur, let us compute them.  
Just like for tree-level amplitudes, one-loop amplitudes are easiest to compute after a Laplace transform.   As with tree-level amplitudes, we write
\begin{equation}
\mc{L}\mc{A}_\text{one-loop}(1^+,2^+, \dots, n^+;\xi) = \int_{t_i = 0}^\infty \d t_1 \dots \d t_n\,e^{-\sum t_i} \prod_{i = 1}^{n} t_i^{ \lt_i \dpa{\lt_i}}\mc{A}_\text{one-loop}(1^+,2^+,\dots,n^+;\xi) 
\end{equation}

As in the tree-level case, we will find a differential equation for this Laplace transformed amplitude.  Before explaining it, we need to analyze the combinatorics of one-loop diagrams a little more carefully.  As we explained, every such diagram corresponds to a topological surface $\Sigma_\text{open}$ that has the topology of $\RP^2$ with several discs removed.  Since $\RP^2$ with a single disc removed is a M\"obius band, we can view $\Sigma_\text{open}$ as a M\"obius band with several discs removed.  The outer boundary component -- where the states $\J$ are placed -- will be taken to be the boundary of the M\"obius band.

A $\phi-\phi$ propagator in this picture gives a simple curve on the M\"obius band, stretching from the boundary to itself.  Such a simple curve can be oriented or unoriented, depending on whether we can choose an orientation of the boundary of the M\"obius band and of a neighbourhood of the curve in a compatible way.  In terms of the oriented double cover of the M\"obius band, which is a cylinder, oriented curves lift to curves on the cylinder starting and ending on the same boundary, whereas unoriented curves lift to curves going from one boundary to the other.

Untwisted propagators give rise to oriented curves on the M\"obius band, whereas twisted propagators give unoriented curves.  

The differential equation will involve cutting along a $\phi-\phi$ propagator. Cutting a M\"obius band along an oriented curve results in a disconnected surface, which is a union of a disc and a M\"obius band. (One can see this by considering the oriented double cover). Cutting along an unoriented curve results in a connected surface with the topology of a disc.   From this we see that the differential equation will have two terms, depending on whether we cut along a twisted or untwisted propagator.

 As in the tree-level case, the Laplace transformed amplitude $\mc{LA}_{\text{one-loop}}$ is a sum over Feynman diagrams. Applying the operator
 \be 
 \mc{D}^{(i)} = \lt_i \dpa{\lt_i} 
 \ee
 yields a similar sum over Feynman diagrams, but where we have selected a single $\phi-\phi$ propagator emerging from the $i$\textsuperscript{th} operator insertion. This $\phi-\phi$ propagator must end at some other operator insertion $j$.

As in the tree-level case, we replace this single $\phi-\phi$ propagator by two $\psi-\psi$ propagators.  Bearing in mind the fact that we may have a twisted or untwisted propagator, we find the following differential equation:
\bea\label{eq:diff_Burns_loop}
&\mc{D}^{(i)}\mc{L}\mc{A}_\text{one-loop}(1,\dots,n;\xi) \\
&= \sum_{j\neq i}\frac{[ij]\ip{ij}\ip{i\xi}^2\ip{j\xi}^2}{4\pi^2}\Big(\widehat{\mc{L}\mc{A}}_\text{tree}(i,i+1,\dots,j-1,j;\xi)\mc{L}\mc{A}_\text{one-loop}(j,j+1,\dots,i-1,i;\xi) \\
&+ \mc{L}\mc{A}_\text{one-loop}(i,i+1,\dots,j-1,j;\xi)\what{\mc{L}\mc{A}}_\text{tree}(j,j+1,\dots,i-1,i;\xi) \\
&+ \frac{1}{2}\what{\mc{L}\mc{A}}_\text{tree}(i, j , j-1, \dots, i, j, j+1,\dots, i-1 ;\xi)\Big)\,.
\eea 
The last term on the right hand side comes from cutting a twisted propagator, yielding a disc diagram.  The fact that the propagator we have cut is twisted is what leads to the unusual ordering of the external states. We also solve this differential equation in Appendix \ref{app:rec_Burns}. This leads to a recursion relation \eqref{eq:rec_Burns_loop}, expressing the $n$-point one-loop amplitudes in terms of $2, 3, \dots, n-1$-point one-loop amplitudes and $2, 3, \dots, n+1$-point tree-level amplitudes.


\subsection{Computing Form Factors of the Average Null Energy Operator}

Given a null vector $v$ in Minkowski space, the average null energy (ANE) operator is
\be
\mc{E}_v = \int_\R\d t\,T_{vv}(t v)\,.
\ee
It is the average of the $v-v$ component of the stress tensor along the line $v$.  The ANE operator has played a prominent role in many recent works, see for example \cite{Hofman:2008ar}.   

It turns out that the Burns-type amplitudes we have computed, when working to first order in the perturbation around the flat metric, encode the form factors of the ANE operator.  In this section, we will explain this observation and use it to compute the tree-level and one-loop form factors of the ANE operator in the self-dual gauge theory we are considering.  

This computation gives a partial answer for the ANE form factor for non-self-dual gauge theory: at tree-level we are computing the form factor with one state of negative helicity and an arbitrary number of positive helicity states, and at one-loop we are computing the all-plus form factor.

Consider the theory in the Burns background, where we give a VEV to $\partial_{w_1}\c\partial_{w_2}\c$. Let us set
\be
\ip{ \partial_{w_1} \c \partial_{w_2} \c (z) } = H(z)  
\ee 
where $H(z)$ is a polynomial of order $4$, or equivalently a four spinor index symmetric tensor $H_{\alpha_1 \dots \alpha_4}$.  

We will compute amplitudes to first order in the deformation away from the flat geometry, i.e., to first order in $H$.  These will be translated into information about the stress tensor form factor.

As before, we will assume that $H = \xi^4$ for some spinor $\xi$.  When working to first order in $H$, this is not necessary,  because there is no ambiguity in defining the scattering amplitudes.  The ambiguity arises at second order and higher in $H$, and are associated with terms where some indices of one copy of $H$ are contracted with indices of another copy of $H$.   However, it turns out that the ANE operator is encoded in the amplitude when $H = \xi^4$.

A simple computation shows that, to first order in $H$, the tree-level and one-loop single-trace amplitude for $n$ positive helicity gluons in the general Burns space background are given by the expressions 
\bea
&\ip{\Jt(\lt_1;z_1) \dots \J(\lt_n;z_n)}_\text{tree} 
= - \frac{1}{2\ip{12} \dots \ip{n1}}\ip{1^4 \mid H}\,, \\
&\ip{\J(\lt_1;z_1) \dots \J(\lt_n;z_n) }_\text{one-loop}
= - \frac{1}{4\pi^2\ip{12} \dots \ip{n1} } \sum_{1 \le i < j \le n} \frac{[ij] \ip{i^2 j^2 \mid H}}  { \ip{ij} }\,.
\eea
The tree-level amplitude is computed in the dual chiral algebra by a planar diagram with no $\phi-\phi$ propagator.  The one-loop amplitude is given by a diagram as in Fig. \ref{fig:one-twisted}, which is planar except for a single twisted $\phi-\phi$ propagator (and no other $\phi-\phi$ propagators). This propagator contributes the factor 
\be
  \frac{[ij]\ip{i^2j^2\mid H}}{2\pi^2\ip{ij}}\,.
\ee

\begin{figure}[ht]
    \centering
        \includegraphics{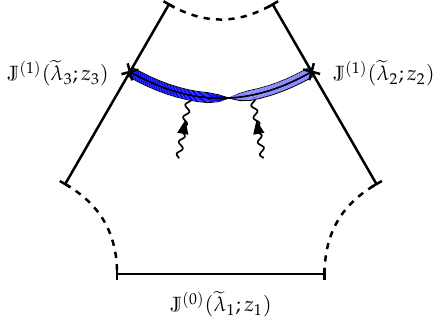}
    \caption{This figure illustrates one of the Feynman diagrams that contribute to the three-point function in the Burns background at one-loop, to leading order in perturbation around the flat metric.}
    \label{fig:one-twisted}
\end{figure}

This is the amplitude in a background metric which deforms the flat metric by
\be
\delta g_{\alpha \da \beta \db}(x) = - \frac{1}{8\pi^2}H_{\alpha\beta}^{~\,~\,\gamma\eta}\p_{x^{\da\gamma}}\p_{x^{\db\eta}}\log\norm{x}\,.
\ee
Inserting such a background metric can of course be realized as the form factor for an operator built from the stress tensor, namely
\be
\frac{1}{8\pi^2} \int_{\R^4}\d^4x\,T^{\da\alpha\db\beta}(x)\p_{x^{\da\gamma}}\p_{x^{\db\eta}}\log\norm{x}\,.
\ee
From the equation for scattering in the Burns background we conclude that
\bea
&\frac{1}{8\pi^2}H_{\alpha\beta}^{~\,~\,\gamma\eta} \int_{\R^4}\d^4x\,\p_{x^{\da\gamma}}\p_{x^{\db\eta}}\log\norm{x}\ip{T^{\da\alpha\db\beta}(x)\mid \Jt(\lt_1;z_1) \dots \J(\lt_n;z_n)}_\text{tree} \\
&= - \frac{1}{2\ip{12} \dots \ip{n1}}\ip{1^4 \mid H}\,, \\
&\frac{1}{8\pi^2}H_{\alpha\beta}^{~\,~\,\gamma\eta} \int_{\R^4}\d^4x\,\p_{x^{\da\gamma}}\p_{x^{\db\eta}}\log\norm{x}\ip{T^{\alpha\da\beta\db}(x) \mid \J(\lt_1;z_1) \dots \J(\lt_n;z_n)}_\text{one-loop} \\
&= - \frac{1}{4\pi^2\ip{12} \dots \ip{n1} } \sum_{1 \le i < j \le n} \frac{[ij] \ip{i^2 j^2 \mid H}}{\ip{ij}}\,.
\eea
Let $P = \sum p_i$ denote the sum of the $n$ external momenta.  Placing the stress tensor at $x$ is the same as placing the stress tensor at the origin, up to a factor of $e^{\i P\cdot x}$.  Therefore,
\bea
&\frac{1}{8\pi^2}H_{\alpha\beta}^{~\,~\,\gamma\eta} \int_{\R^4}\d^4x\,\p_{x^{\da\gamma}}\p_{x^{\db\eta}}\log\norm{x}e^{\i P\cdot x}\ip{T^{\da\alpha\db\beta}(0)\mid \Jt(\lt_1;z_1) \dots \J(\lt_n;z_n)}_\text{tree} \\
&= - \frac{1}{2\ip{12}\dots\ip{n1}}\ip{1^4\mid H}\,, \\ 
&\frac{1}{8\pi^2}H_{\alpha\beta}^{~\,~\,\gamma\eta} \int_{\R^4}\d^4x\,\p_{x^{\da\gamma}}\p_{x^{\db\eta}}\log\norm{x}e^{\i P\cdot x}\ip{T^{\alpha\da\beta\db}(0) \mid \J(\lt_1;z_1) \dots \J(\lt_n;z_n)}_\text{one-loop} \\
&= - \frac{1}{4\pi^2\ip{12} \dots \ip{n1} } \sum_{1 \le i < j \le n} \frac{[ij] \ip{i^2 j^2 \mid H}}{\ip{ij}}\,.
\eea
Performing the integral over $x$, we find the following expression for the stress tensor form factor at tree-level and at one-loop:
\bea
&\frac{P_{\da\gamma}P_{\db\eta}H_{\alpha\beta}^{~\,~\,\gamma\eta}}{(P\cdot P)^2}\ip{T^{\da\alpha\db\beta}\mid\Jt(\lt_1;z_1)\dots\J(\lt_n;z_n)}_\text{tree} = - \frac{1}{2\ip{12} \dots \ip{n1}}\ip{1^4 \mid H}\,, \\  
&\frac{P_{\da\gamma}P_{\db\eta}H_{\alpha\beta}^{~\,~\,\gamma\eta}}{(P\cdot P)^2}\ip{T^{\da\alpha\db\beta} \mid \J(\lt_1;z_1)\dots\J(\lt_n;z_n)}_\text{one-loop} \\
&= - \frac{1}{4\pi^2\ip{12}\dots\ip{n1}}\sum_{1 \le i < j \le n} \frac{[ij]\ip{i^2j^2 \mid H}}{\ip{ij}}\,.
\eea
If we assume, as above, that $H = \xi^4$, then we can simplify the expression.  

We let $h = \xi^2$. Then, $h_{\alpha\beta}$ is an element of $\mf{sl}(2)$ which is nilpotent, meaning $h_{\alpha\beta} h_{\gamma \delta} \eps^{\beta \gamma} = 0$.  Up to a sign we can determine $\xi$ from $h$, and it is more convenient to write everything in terms of $h$.

We have
\be
H_{\alpha \beta \gamma \delta} = h_{\alpha \beta} h_{\gamma \delta}\,.
\ee
Note that this is symmetric, because the right hand side vanishes whenever we contract any two spinor indices.  Then, 
\be
\ip{i^2 j^2 \mid H} = \ip{ij \mid h}^2\,.
\ee
Also
\be
P_{\da\gamma}P_{\db\eta}H_{\alpha\beta}^{~\,~\,\gamma\eta} = (h \circ P)_{\da \alpha } (h \circ P)_{\db \beta} 
\ee
where $h \circ P$ indicates the action of $h$, viewed as an element of $\mf{so}(4)$, on $P$.

We conclude that
\bea
&\frac{(h \circ P)_{\da\alpha}(h \circ P)_{\db\beta}}{(P\cdot P)^2}\ip{T^{\da\alpha\db\beta} \mid \Jt(\lt_1;z_1) \dots \J(\lt_n;z_n)}_\text{tree} = - \frac{1}{2\ip{12}\dots\ip{n1}}\ip{1^4\mid H}\,, \\
&\frac{(h \circ P)_{\da\alpha}(h \circ P)_{\db\beta}}{(P\cdot P)^2}\ip{T^{\da\alpha\db\beta} \mid \J(\lt_1;z_1) \dots \J(\lt_n;z_n)}_{\text{one-loop}} \\
&= - \frac{1}{4\pi^2\ip{12} \dots \ip{n1} }\sum_{1\le i < j \le n} \frac{ [ij] \ip{i j \mid h }^2 } {\ip{ij}}\,.
\eea
It is convenient to parametrize everything in terms of the vector $v = h \circ P$ instead of $h$.  Note that $v \cdot P = 0$ and $v \cdot v = 0$ (the latter follows from the fact that $h^2 = 0$). 

We can recover $h$ from $v$ as follows.  The map which sends an element $h$ of $\mf{sl}(2)$ to the vector $h \circ P$ is an isomorphism from $\mf{sl}(2)$ to the orthogonal complement of $P$.  This isomorphism sends nilpotent elements of $\mf{sl}(2)$ to null vectors in $P^\perp$.  We can explicitly invert this isomorphism to write $h$ in terms of $v$ as follows.  

First, we note that 
\be
\eps^{\db\da}P_{\da\alpha}P_{\db\beta} = \frac{1}{2}\eps_{\alpha\beta}(P \cdot P)
\ee
and, since $P \cdot v = 0$, $\eps^{\db\da}P_{\da\alpha}v_{\db\beta}$ is symmetric. Then,
\bea
\frac{2\eps^{\db\da}P_{\alpha\da}v_{\beta\db}}{P\cdot P} =  
\frac{2 \eps^{\db\da}\eps^{\eta\gamma}P_{\da\alpha}h_{\beta\gamma}P_{\db\eta}} { P \cdot P} = h_{\alpha\beta}\,.
\eea
Therefore,
\bea
&\ip{ij \mid h} = - \frac{2}{P\cdot P}\ip{i|P v|j} = - \frac{2}{P\cdot P}\sum_{k = 1}^n\ip{ik}[k|v|j\rangle\,, \\
&\ip{1^2\mid h} = - \frac{2}{P\cdot P}\sum_{k = 1}^n \ip{1k}[k|v|1\rangle\,. \\
\eea
We have re-expressed everything in terms of a null vector $v$ orthogonal to $P$, giving us the expression
\bea
&\ip{v^rv^sT_{rs} \mid \Jt(\lt_1;z_1) \dots \J(\lt_n;z_n)}_\text{tree}
= - \delta_{v\cdot P = 0}\frac{2}{\ip{12} \dots \ip{n1}} \Big(\sum_{k=1}^n\ip{1k}[k|v|1\rangle\Big)^2\,, \\
&\ip{v^rv^sT_{rs} \mid \J(\lt_1;z_1)\dots\J(\lt_n;z_n)}_\text{one-loop} \\
&= - \delta_{v \cdot P = 0} \frac{1}{\pi^2\ip{12} \dots \ip{n1}}\sum_{1 \le i < j \le n}\frac{[ij]}{\ip{ij}}\Big(\sum_{k=1}^n\ip{ik}[k|v|j\rangle\Big)^2\,.
\eea
The $\delta$-function forcing $v$ to be orthogonal to the overall momentum means that this expression is the form factor of the stress tensor which has been integrated along the line $v$.   We conclude that the form factor we have computed is the average of $T_{vv}$ along the line $v$, which is the ANE operator. 

Rewriting the previous expression, we find the tree-level and one-loop ANE form factors are given by  
\bea
&\ip{\int_{t\in\R}\d t\,T_{vv}(tv) \mid \Jt(\lt_1;z_1) \dots \J(\lt_n;z_n)}_\text{tree} = - \frac{2}{\ip{12} \dots \ip{n1}}\Big(\sum_{k=1}^n\ip{1k}[k|v|1\rangle\Big)^2 \\
&\ip{\int_{t\in\R}\d t\,T_{vv}(tv) \mid \J(\lt_1;z_1) \dots \J(\lt_n;z_n)}_\text{one-loop} \\
&= - \frac{1}{\pi^2\ip{12}\dots\ip{n1}}\sum_{1\le i < j \le n}\frac{[ij]}{\ip{ij}}\Big(\sum_{k=1}^n\ip{ik}[k|v|j\ra\Big)^2\,.
\eea


\subsection{Rephrasing in Terms of a Background Metric}

Form factors of the average null energy operator are the same as scattering amplitudes in a background metric that in light cone coordinates $x_+,x_-$, $x_i$, deforms the flat metric to first order by
\begin{equation}
(\d x_+)^2 \delta_{x_+ = 0, x_i = 0}.
\end{equation}
Performing a translation so that the $\delta$-function is at $x$ instead of at the origin amounts to multiplying the scattering amplitude by the phase $e^{\i P\cdot  x}$, where $P = \sum p_i$ is the total momentum of the $n$ incoming particles.  

By integrating over $x$ we can find the scattering amplitude in  a metric which deforms the flat metric by
\begin{equation}
H(x_+,x_i) (\d x_+)^2\,.
\end{equation}
where $\partial H/\partial x_- = 0$.  Metrics of this type are called \emph{pp-wave} metrics.  Note we do not assume that the Einstein equation is satisfied.

The result will be an integral over $x$ of the amplitudes written above, times $e^{\i P \cdot x} H(x)$.  This integral over $x$ will of course yield the Fourier transform $\what{H}(P)$ of $H(x)$. 

Explicitly, we find the scattering amplitude in such a background is
\bea
&\mc{A}^\text{tree}( 1^-, 2^+, \dots,n^+ \mid H(x) (\d x_+)^2 ) \\
&= - \frac{2}{\ip{12} \dots \ip{n1}}\Big(\sum_{k=1}^n\ip{1k}[k|\d x_+ |1\rangle\Big)^2 \delta_{P\cdot\d x_+ = 0}  \what{H} (P)\,, \\
&\mc{A}^\text{one-loop}( 1^+, 2^+, \dots,n^+ \mid H(x) (\d x_+)^2 ) \\
&= - \frac{1}{\pi^2\ip{12}\dots\ip{n1}}\sum_{1\le i < j \le n}\frac{[ij]}{\ip{ij}}\Big(\sum_{k=1}^n\ip{i k }[k | \d x_+ |j\ra\Big)^2 \delta_{P \cdot \d x_+ = 0} \what{H}(P) \,.
\eea


\section{Amplitudes in a Flavour Symmetry Background} \label{sec:flavour-amps}

The next amplitude we will compute is that in the flavour symmetry background. This turns out to be much more complicated than the Burns space case discussed above.

We can not treat an arbitrary flavour symmetry background, because there are ambiguities in defining the correlation functions for the chiral algebra.  These correspond to counter-term ambiguities when we study the QFT in the flavour background.

However, as we have discussed in Sect. \ref{sec:uni_confblock}, if we assume that $M_f^g(z) = f(z)^2D_f^g$, for some diagonal matrix $D$ and linear polynomial $f$, then there is a unique conformal block (with correct dimension and symmetry properties) for the chiral algebra deformed by giving a VEV to $\gamma_f\til\gamma^g$.  From the perspective of the $4d$ gauge theory, this means that computing amplitudes in the presence of the background flavour symmetry gauge field does not suffer from counter-term ambiguities.

In this section, we will compute these amplitudes.  We will find that the computation relies on some non-trivial combinatorial algebra.


\subsection{Two-Point Amplitude of Fundamental Fermions}

Switching on this flavour symmetry background will give the fermionic states $\M_f,\Mt^g$ a two-point amplitude, as depicted in Fig. \ref{fig:fermion2pt}.  When evaluating this, it's useful to note that any states involving $\p_{w_i}\c$ vanish, and so we can ignore any contributions from bulk propagators. Only $\gamma,\til\gamma$s can contribute.  Aside from those appearing explicitly in the states, the only way to generate these is through the ADHM constraint. Upon discarding all terms involving the ghost, it reads
\be [\phi^\done,\phi^\dtwo] + \til\gamma^f\gamma_f = 0\,. \ee
We can then readily compute the two-point term in the OPE
\be \label{eq:MMt-2pt} \M_f[m,0](z_1)\Mt^g[0,m](z_2) \sim - \frac{1}{z_{12}}\gamma_f(\phi^\done)^m(\phi^\dtwo)^m\til\gamma^g(z_2)\,. \ee
As above, we've chosen to suppress $\fsp(K)$ colour indices and stripped the associated factor of $\Omega$.

\begin{figure}[ht]
    \centering
        \includegraphics{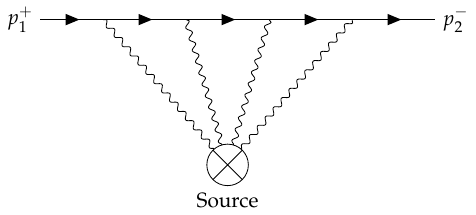}
    \caption{A fermion moving in the flavour symmetry background which has a point source.}
    \label{fig:fermion2pt}
\end{figure}

On the right hand side of this equation is a bulk field.  Our task is to determine what the expectation value of this bulk field is, once we have specified that $\ip{\gamma_f\til\gamma^g} = f(z)^2D_f^g$.  When we specify that this operator has a VEV, we also mean that the VEVs of other operators remain zero, so that in particular
\be
\ip{\gamma_f(\phi^\done)^{(m}(\phi^\dtwo)^{n)}\til\gamma^g} = 0
\ee
whenever $m+n>0$.  It's essential that in this expression, the $\phi$ states are symmetrized.

Using the BRST relation, we can write the product $(\phi^\done)^m(\phi^\dtwo)^m$ uniquely as a finite sum of products of expressions like $\gamma_f (\phi^\done)^{(r}(\phi^\dtwo)^{s)}\til\gamma^g$.  When we take expectation values of these operators, only the term which involves no $\phi$ fields will contribute. 

Determining the coefficient of this term is a problem in pure combinatorics which we discuss extensively in Appendix \ref{app:flavour-combinatorics}.  Let us phrase this problem slightly more abstractly, using the terminology used in the appendix. 

Consider a non-commutative algebra freely generated by two elements $X^\done,X^\dtwo$, denoted  $\C\la X^\done,X^\dtwo\ra$. We can define symmetric polynomials depending on a choice of reference spinor $\lt_\da$
\be X^m(\lt) = (X^\da\lt_\da)^m\,. \ee
This is a generating function of the symmetrized products of $X^\done$ and $X^\dtwo$. 

It's convenient to extend this to symmetric polynomials in two reference spinors
\be 
    X^{m,n}(\lt_1,\lt_2) = \frac{(\lt_2\cdot\p_{\lt_1})^n}{[m+n]_n} 
    X^{m+n}(\lt_1) = \frac{(\lt_1\cdot\p_{\lt_2})^m}{[m+n]_m}X^{m+n}(\lt_2)
\ee
where $[x]_n$ is the descending factorial $x (x-1) \dots (x-n+1)$.  This expression is again a generating function for the symmetrized products of $X^\done$ and $X^\dtwo$. However, since $\lt_1,\lt_2$ are independent reference spinors, with four components in total, this is a redundant generating functional: the coefficients of the four variables are not independent.  

We also use the shorthand notation
\be
C = [X^\done, X^\dtwo]\,.
\ee
Consider the expressions
\be
X^{m_1}(\lt_1) C X^{m_2}(\lt_2) C \dots C X^{m_n} (\lt_n)\,.
\label{eq:noncommutativealgebra}
\ee
It is not hard to show (and we will derive this carefully in the appendix) that the Taylor coefficients of these expressions as we expand in the variables $\lt_i$  form a basis of the non-commutative algebra $\C\la X^\done, X^\dtwo\ra $. (Of course, we must also vary $n$ and $m_1,\dots, m_n$.)

We would like to find the structure constants $R_{\{k_i\},\{l_j\}}$ obeying
\bea \label{eq:structure-constant-def}
&X^m(\lt_1)X^n(\lt_2) \\
&= \sum_{a=0}^{\min{m,n}} [12]^a \sum_{\substack{k_1+\dots+k_{a+1} = m-a \\ l_1+\dots+l_{a+1} = n-a}} R_{\{k_i\},\{l_j\}} X^{k_1,l_1}(\lt_1,\lt_2)C\dots CX^{k_{a+1},l_{a+1}}(\lt_1,\lt_2)\,,
\eea
where $C = [X^\done,X^\dtwo]$.  We can view $a,m,n$ as being determined by the sets $\{k_i\},\{l_j\}$ (constrained to have the same cardinality).  Note that $R_{\{m\},\{n\}} = 1$.  In the appendix we develop recursive formulae for these coefficients.

We find the following closed form expressions for some of the coefficients $R$:
\bea \label{eq:combinatorial-formulae}
&R_{\{k,\tilde k\},\{l,\tilde l\}} = \binom{k+\tilde k+l+\tilde l+2}{k+\tilde k+1}^{-1}\sum_{p=0}^k\binom{k+l-p}{l}\binom{\tilde k+\tilde l+1+p}{\tilde l}\,, \\
&R_{\{k,0^a\},\{l,0^a\}} = \frac{1}{[k+l+a+1]_a}\binom{k+a}{a}\,,\quad R_{\{0^a,\tilde k\},\{0^a,\tilde l\}} = \frac{1}{[\tilde k+\tilde l+a+1]_a}\binom{\tilde l+a}{a}\,.
\eea
In particular, from either of the last expressions in this equation, we see that
\be
X^m(\lt_1) X^m(\lt_2) = \frac{1}{(m+1)!}[12]^mC^m + \dots  
\ee

Returning to the chiral algebra, as a consequence of this formula and the BRST relation we have 
\bea \label{eq:comb-gaphiga}
&\gamma_f(\phi^\done)^m(\phi^\dtwo)^m\til\gamma^g = \frac{(-)^m}{(m+1)!}\gamma_f(\til\gamma^h\gamma_h)^m\til\gamma^g \\
&+ \text{terms involving }\gamma(\phi^\done)^{(r}(\phi^\dtwo)^{s)}\til\gamma \text{ for }r+s>0\,.
\eea
This implies  that \eqref{eq:MMt-2pt} evaluates to
\bea
&\M_f[m,0](z_1)\Mt^g[0,m](z_2) \sim \frac{1}{z_{12}}\frac{(-)^{m+1}}{(m+1)!}\gamma_f(\til\gamma^h\gamma_h)^m\til\gamma^g(z_2) \\
&= \frac{1}{z_{12}}\frac{(-)^{m+1}}{(m+1)!}\big(M(z_2)^{m+1}\big)_f^g\,.
\eea
For general matrices $M(z)$, the regular part is ambiguous. From the chiral algebra perspective, this is because of the lack of uniqueness of conformal blocks, and from the gauge theory perspective, it's because of counter-term ambiguities.

However, if we assume that $M(z) = f(z)^2D$ for a fixed matrix $D$ (which we assume, without loss of generality, to be diagonal), then counter-terms/conformal blocks are unique.  The two-point function is then uniquely determined by its singularities, and we have
\be \ip{\M_f[m,0](z_1)\Mt^g[0,m](z_2)} = \frac{1}{z_{12}}\frac{(-)^{m+1}}{(m+1)!}f(z_1)^mf(z_2)^{m+2}(D^{m+1})_f^g\,. \ee
This can be resummed to give the two-point function for hard generators
\bea
&\ip{\M_f(\lt_1;z_1)\Mt^g(\lt_2;z_2)} = \frac{1}{z_{12}}\sum_{m=0}^\infty\frac{(-)^{m+1}}{(m+1)!(m!)^2}f(z_1)^mf(z_2)^{m+2}[12]^m(D^{m+1})_f^g \\
&= - \frac{f(z_2)^2}{z_{12}} D_f^h\,{}_0F_2\big(1,2;-f(z_1)f(z_2)[12]D\big)^g_h\,,
\eea
where ${}_0F_2(1,2;x)$ is a generalized hypergeometric function.


\subsection{Three-Point Amplitude of Fundamental Fermions and a Gluon}

The amplitude for $n$ positive helicity gluons and two fundamental fermions is, in the chiral algebra language, the correlator 
\be
\ip{\M_f(\lt_1;z_1)\J(\lt_2;z_2)\dots\J(\lt_{n+1};z_{n+1})\Mt^g(\lt_{n+2};z_{n+2})}\,. 
\ee
Here we're suppressing the colour structure for clarity. An example Feynman diagram contributing to this amplitude is depicted in Fig. \ref{fig:fermiongluon}.

\begin{figure}
    \centering
        \includegraphics{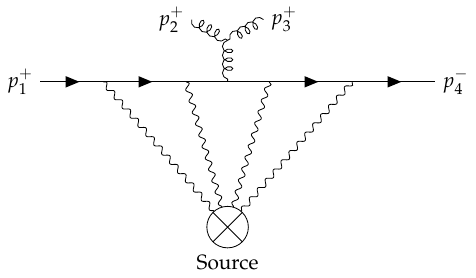}
    \caption{Gluons scattering off fermions in the flavour background.}
    \label{fig:fermiongluon}
\end{figure}

Expanding this in soft modes, we need to compute the correlators
\be
\ip{\M_f[k_1,l_1](z_1)\J[k_2,l_2](z_2)\dots\J[k_{n+1},l_{n+1}](z_{n+1})\til \M^g[k_{n+2},l_{n+2}]}\,. \label{eq:ngluoncorrelator}
\ee
In principle, this can be computed in a similar way to the computation of the correlator of just two fermions, however the combinatorics is more challenging.  First, the only Wick contractions that can appear are $\psi-\psi$ Wick contractions, because $\phi-\phi$ and $\gamma-\til\gamma$ propagators introduce $\p_{w_i}\c$ fields which are sent to zero in this background.  Further, we will focus on the colour-ordered correlator (assuming $K$ is not too small compared to $n$).  This means that we consider only situations where Wick contractions are between two adjacent states in the expression \eqref{eq:ngluoncorrelator}. 

Because of our assumption that the matrix $M(z)$ is of the form $f(z)^2 D$, there is no ambiguity in determining the correlation function \eqref{eq:ngluoncorrelator} from its singular part.  When all Wick contractions have been performed, the expression \eqref{eq:ngluoncorrelator} becomes
\be
\frac{1}{\la12\ra\la23\ra\dots\la(n+1)(n+2)\ra}\gamma_f(\phi^\done)^{(k_1} (\phi^\dtwo)^{l_1)}(z_1)\dots(\phi^\done)^{(k_{n+2}}(\phi^\dtwo)^{l_{n+2})}\til\gamma^g(z_{n+2})\,.
\ee
To determine the correlation function, we need to re-order the $\phi$s as before.

This amounts to the following problem in combinatorial algebra.  As before, consider the free associative algebra generated by non-commuting variables $X^\done,X^\dtwo$.  We have seen that there is a basis for this algebra consisting of the expressions in \eqref{eq:noncommutativealgebra}.  We can therefore write
\bea
X^{m_1}(\lt_1)X^{m_2}(\lt_2)\dots X^{m_n}(\lt_n) = \sum_a S_a\prod_{i < j} [ij]^{a_{ij}}C^{\half\sum_im_i}  + \text{ other basis elements}\,.
\eea
The sum is over strictly upper triangular $n \times n$ matrices $a_{ij}$ with non-negative integer entries, such that $m_i = \sum_{i<j}a_{ij} + \sum_{j<i}a_{ji}$.  The coefficient $S_a$ is some combinatorial constant to be determined and which will immediately give the formula for the scattering amplitude with $2$ fermions and $n-2$ gluons. 

The case of a single gluon can be solved using only the recursion coefficient $R_{(0^a), (0^a)}$  we used in the case with two fermions.  In that case, we have 
\bea
&X^{m_1}(\lt_1)X^{m_2}(\lt_2) X^{m_3}(\lt_3) \\
&= c_{m_1,m_2,m_3}[12]^{\half(m_1 + m_2 - m_3)} [13]^{\half(m_1 + m_3 - m_2)} [23]^{\half(m_2 + m_3 - m_1)}C^{\half(m_1 + m_2 + m_3)} \\
&+ \text{ other basis elements}
\eea
and we need to determine the constants $c_{m_1,m_2,m_3}$. Obviously they vanish unless $m_1 + m_2 + m_3$ is divisible by $2$ and $m_3\leq m_1 + m_2$ (and permutations).

 Set $m_i = k_i + l_i$.  The constants $c_{m_1,m_2,m_3}$ are the coefficients of an $\mrm{SL}(2)$ invariant linear map
\be
\op{Sym}^{m_1}(\C^2)\otimes\op{Sym}^{m_2}(\C^2)\otimes\Sym^{m_3}(\C^2) \to \C
\ee
from the tensor product of the representations of spin $m_i/2$ to the trivial representation.  (Here the non-commuting variables $X^\done,X^\dtwo$ form an $\mrm{SL}(2)$ doublet.)  Up to normalization, there is at most one such linear map, which is present if and only if $\abs{m_1 - m_3} \le m_2 \le m_1 + m_3$ and $m_1 + m_2 + m_3$ is even. By $\mrm{SL}(2)$ invariance, it suffices to project
\be (X^\done)^{m_1} (X^\done)^{(k_2}(X^\dtwo)^{l_2)} (X^\dtwo)^{m_3}\quad\text{onto}\quad[X^\done,X^\dtwo]^{\half(m_1 + m_2 + m_3)} \ee
in the basis we have discussed above.  Here, the middle two terms $(X^\done)^{(k_2} (X^\dtwo)^{l_2)}$ are symmetrized.  If these terms were not symmetrized, the problem would be equivalent to that discussed in the case of the amplitude with only fermions.

It turns out, however, that we find the same answer if we use the symmetrized expression $(X^\done)^{(k_2}(X^\dtwo)^{l_2)}$ or the unsymmetrized one $(X^\done)^{k_2}(X^\dtwo)^{l_2}$.  To see this, note that 
\be
(X^\done)^{(k_2} (X^\dtwo)^{l_2)} = (X^\done)^{k_2} (X^\dtwo)^{l_2} + f(X^\done,X^\dtwo) [X^\done,X^\dtwo] g(X^\done,X^\dtwo)
\ee
where $f,g$ are non-commutative polynomials in $X^\done,X^\dtwo$, each of which is of order less than $k_2$ in $X^\done$ and less than $l_2$ in $X^\dtwo$. 

We need to show that the projection of 
\be
(X^\done)^{m_1} f(X^\done,X^\dtwo) [X^\done,X^\dtwo] g(X^\done,X^\dtwo) (X^\dtwo)^{m_3} \label{eq:symmetrization}
\ee
onto $[X^\done,X^\dtwo]^{\half (m_1 + m_2 + m_3)}$ is zero.  Recall that we are using the basis of the algebra of non-commutative polynomials whose elements consist of symmetric expressions in $X^\done,X^\dtwo$ sandwiched between commutators (with initial and final symmetric expressions). To write \eqref{eq:symmetrization} in this basis, we need to write $(X^\done)^{m_1} f(X^\done,X^\dtwo)$ and $g(X^\done,X^\dtwo) (X^\dtwo)^{m_3}$ in the basis and then concatenate them.  

However, the coefficient of $[X^\done,X^\dtwo]^m$ in the expansion of $(X^\done)^{m_1} f(X^\done,X^\dtwo)$ in this basis is zero.  This is because $f(X^\done,X^\dtwo)$ is of order less than $l_2 = \half (m_1 + m_2 - m_3)$ in $X^\dtwo$, and $\half (m_1 + m_2 - m_3) \le m_1$.  Thus, there are more $X^\done$s than $X^\dtwo$s in $(X^\done)^{m_1} f(X^\done,X^\dtwo)$.  

From this, we see that
\be
(X^\done)^{m_1} (X^\done)^{(k_2} (X^\dtwo)^{l_2)} (X^\dtwo)^{m_3} 
= \frac{C^{\half (m_1 + m_2 + m_3) }}{\big(\half(m_1 + m_2 + m_3) + 1\big)!} + \text{ other basis elements}
\ee
where as above $k_2 = \half (m_2 + m_3 - m_1)$ and $l_2 = \half(m_1 + m_2 - m_3)$. This implies that
\be
c_{m_1,m_2,m_3} = \frac{1}{\big(\half(m_1+m_2+m_3)+1\big)!}\binom{m_2}{\half(m_2+m_3-m_1)}\,,
\ee
immediately giving us the amplitude with one gluon:
\bea
&\ip{\M_f(\lt_1;z_1)\J(\lt_2;z_2)\Mt^g(\lt_3;z_3)} \\
&= \frac{1}{\la12\ra\la23\ra}\sum_{a,b,c \ge 0}\frac{(-)^{a+b+c}f(z_1)^{a+b} f(z_2)^{a+c} f(z_3)^{b+c+2} }{ a! (a+b)!  (b+c)! c!  (a + b + c + 1)!  } [12]^{a} [13]^b[23]^c \big(D^{a + b + c + 1}\big)_f^g\,. 
\eea
Here $a=\half(m_1+m_2-m_3)$, $b=\half(m_1+m_3-m_2)$ and $c=\half(m_2+m_3-m_1)$. The conditions on $m_1,m_2,m_3$ are equivalent to $a,b,c\in\Z_{\geq0}$.


\subsection{Recursion Relations for \texorpdfstring{$n$}{n}-Point Amplitudes}

With two fermions and $n$ gluons, there is not such a simple expression for the amplitude. Nevertheless, there is a recursion relation which we will state.

As we mentioned in \eqref{eq:combinatorial-formulae} we have
\be
 R_{\{0^a,\tilde k\},\{0^a,\tilde l\}} = \frac{1}{[\tilde k+\tilde l+a+1]_a}\binom{\tilde l+a}{a}\,.
\ee
This tells us that
\be
X^{m} (\lt_{1} ) X^{n} (\lt_{2}) = \sum_{a = 0}^{\min{m, n}}  C^a [12]^a  X^{m - a, n - a}  (\lt_{1},\lt_2)   \frac{1}{[m + n-a +1]_a}\binom{n}{a} + \text{ other terms} \label{eq:recursion}
\ee
where the other terms, which have some $X$ to the left of a commutator, will not contribute to our final calculation.    

Manipulation of this formula will give the desired recursion relation.  Recall that we aim to determine the coefficient of $C^{\half\sum_im_i}$ in $X^{m_1}(\lt_1) \dots X^{m_n}(\lt_n)$.  We do this by first determining those terms in $X^{m_1}(\lt_1) X^{m_2}(\lt_2)$ with all the $C$'s to the left, and then repeating this process. 

Since we are interested ultimately in the scattering of hard modes, it will be helpful to write the recursion relation \eqref{eq:recursion} in those terms.  All the following manipulations will drop any terms where an $X$ is to the left of a $C$.  We have
\be \label{eq:hard-recursion}
e^{X(\lt_1)} e^{X(\lt_2)} = \sum_{m,n\geq0} \sum_{a=0}^{\min{m,n}} \frac{1}{m! n!}  C^a [12]^a  X^{m - a, n - a}  (\lt_{1},\lt_2)   \frac{1}{[m + n-a +1]_a}\binom{n}{a} 
\ee
where we recall that
\be X^{m,n}(\lt_1,\lt_2) = \frac{(\lt_2\cdot\p_{\lt_1})^n}{[m+n]_n}X^{m+n}(\lt_1) = \frac{(\lt_1\cdot\p_{\lt_2})^m}{[m+n]_m}X^{m+n}(\lt_2) 
 = S( X^m(\lt_1) X^n(\lt_2) ) \ee
where $S$ indicates symmetrization.
 
We can reorganize the sum by replacing $m$ by $m-a$ and $n$ by $n-a$ so that the three indices now range over all of $\Z_{\geq0}$. Equation \eqref{eq:hard-recursion} simplifies to
\be
e^{X(\lt_1)} e^{X(\lt_2)} = \sum_{m,n,a\geq0} \frac{C^a[12]^a}{(m+a)!n!a![m + n + a + 1]_a} X^{m,n}(\lt_1,\lt_2)\,.
\ee
Focusing on the terms where we fix both $a$ and $m+n = k$, we have the following identity:
\bea
&\sum_{m = 0}^k\frac{1}{[k+a+1]_a(m+a)!(k-m)!}X^{m,k-m}(\lt_1,\lt_2) \\
&= \frac{1}{k!}\int_{0<t_1<\dots<t_a<1}\d^a\mathbf{t}\int_{0<s_1<\dots<s_a<1}\d^a\mathbf{s}\,s_1X^k(t_1 s_1\lt_1 + s_1\lt_2)\,.
\eea
To check this, we note that
\be
X^k(ts\lt_1 + s\lt_2) = \sum_{m = 0}^k \binom{k}{m}t^m s^k X^{m,k-m}( \lt_1,\lt_2)  \,.
\ee
It is straightforward to apply the iterated integral to this expression giving the desired combinatorial factors.

This brings us to 
\bea
&e^{X(\lt_1)} e^{X(\lt_2)} = \sum_{k,a\geq0}\frac{C^a[12]^a}{a!k!} \int_{0 < t_1 < \dots < t_a < 1}\d^a\mbf{t}\int_{0 < s_1 < \dots < s_a < 1}\d^a\mbf{s}\,s_1 X^k(t_1s_1\lt_1 + s_1\lt_2) \\
&= \sum_{a\geq} \frac{C^a[12]^a}{a!}\int_{0 < t_1 < \dots < t_a < 1}\d^a\mbf{t}\int_{0 < s_1 < \dots < s_a < 1}\d^a\mbf{s}\,s_1 e^{X(t_1 s_1 \lt_1 + s_1 \lt_2)}\,.
\eea
This gives us the recursion relation for colour-ordered amplitudes, expressing the amplitude for $2$ fermions and $n$ gluons in terms of that for $2$ fermions and $n-1$ gluons:
\bea \label{eq:flavour-recursion}
&\ip{\M_f(\lt_1;z_1)\J(\lt_2;z_2) \dots \J(\lt_{n+1};z_{n+1}) \til \M^g(\lt_{n+2};z_{n+2})}  \\
&= \frac{1}{\ip{12}}\sum_{a\geq0}\frac{(-)^af(z_1)^af(z_2)^a[12]^a}{a!}(D^a)_f^h\int_{0 < t_1 < \dots < t_a < 1}\d^a\mbf{t}\int_{0 < s_1 < \dots < s_a < 1}\d^a\mbf{s}\,s_1 \\
&\Big\la\M_h(t_1s_1\lt_1 + s_1\lt_2;z_2)\J(\lt_3;z_3)\dots\J(\lt_{n+1};z_{n+1})\Mt^g(\lt_{n+2};z_{n+2})\Big\rangle\,.
\eea
It's worth noting that the solution to this recursion relation is of the form
\bea
&\ip{\M_f(\lt_1,z_1)\J(\lt_2,z_2)\dots\J(\lt_{n+1};z_{n+1})\Mt^g(\lt_{n+2};z_{n+2})} \\
&= \frac{1}{\la12\ra\dots\la(n+1)(n+2)\ra} F_f^g(\lt_1,\dots,\lt_{n+2}, D) 
\eea
where $F$ is an entire analytic function of $n+3$ complex variables, into one of whose entries we insert the matrix $D$.  To see that $F$ is analytic, the key point is to note that if $f(w)$ is an entire analytic function of one variable, then the expression
\be
\sum_{a\geq0}\frac{1}{a!}\int_{0 < t_1 < \dots < t_a < 1}\d^a\mbf{t}\int_{0 < s_1 < \dots < s_a < 1}\d^a\mbf{s}\,s_1f(s_1t_1\lt_1 + s_1\lt_2)
\ee
is an entire analytic function of the two variables $\lt_1$, $\lt_2$.  
Analyticity is seen by noting that in each term in the sum, the integral is bounded above in absolute value by the supremum of $f(w)$ on the triangle in $\C$ with vertices $0$, $\lt_2$, $\lt_1 + \lt_2$.  This makes it clear that the sum converges locally uniformly absolutely in $\lt_1,\lt_2$ (and also uniformly in the additional variables $z_i$, $\lt_i$), implying analyticity.


\subsection{Celestial Chiral Algebra in a Flavour Symmetry Background}

The recursion relation derived in the previous subsection follows from a deformation of the celestial chiral algebra on the flavour symmetry background.

Indeed, equation \eqref{eq:flavour-recursion} can be recovered from the following planar deformation of the $\J,\M\to\M$ OPE
\bea
&\J_{IJ}(\lt_1;z_1)\M^f_K(\lt_2;z_2) \sim - \frac{2}{\la12\ra}\sum_{a\ge0}\frac{1}{a!}[12]^a\big(M(z_2)^a\big)^f_h\\
&\times \int_{0 < s_1 < \dots  < s_a < 1}\d^a\mbf{s}\,s_1\int_{0 < t_1 < \dots  < t_a < 1}\d^k\mbf{t}\,\Omega_{K(I}\M^h_{J)}(s_1\lt_1+s_1t_1\lt_2;z_2)\,.
\eea
Since the celestial chiral algebra is sensitive only to the singular part of the OPE this holds for any $M(z)$, not only when $M(z) = f(z)^2D$.

Similar arguments lead to the following expression for the planar part of the $\J,\til\M\to\til\M$ OPE:
\bea \label{eq:JMt-flavour-int}
&\J_{IJ}(\lt_1;z_1)\Mt_{Kg}(\lt_2;z_2) \sim\frac{2}{\la12\ra}\sum_{a\ge0}\frac{(-)^a}{a!}[12]^a(M(z_2)^a)^g_h\\
&\times \int_{0 < s_1 < \dots  <s_a < 1}\d^a\mbf{s}\,s_1\int_{0 < t_1 < \dots < t_a < 1}\d^k\mbf{t}\,\Omega_{K(I}\Mt_{J)}^h(s_1\lt_1+s_1t_1\lt_2;z_2)\,.
\eea

The final modified $2\to1$ planar OPE in the flavour background is $\M,\Mt\to\Jt$. This is generated by the contraction of a pair of $\phi$ fields between the operators $\M_f$ and $\Mt^g$.

First we consider the case when the $\phi$ contraction is adjacent to the $\gamma_f$ and $\til{\gamma}^g$ fields.
\be
\wick{\gamma_f\c\phi(\lt_1)\phi(\lt_1)^{m-1}\psi(z_1)\til{\gamma}^g\c \phi(\lt_2)\phi(\lt_2)^{n-1}\psi(z_2)}\,.
\ee
We extract the term that includes $\gamma_f\til{\gamma}^g$.
\be
\M^{(m)}_{If}(\lt_1;z_1)\Mt^{(n),g}_J(\lt_2;z_2) \sim \frac{[12]}{8\pi^2\la12\ra}\gamma_f\til{\gamma}^g\psi_I\phi(\lt_1)^{m-1}\phi(\lt_2)^{n-1}\p_{w_1}\c\p_{w_2}\c\psi_J\,.
\ee
We can verify that this OPE can be written as
\be
\M_{If}(\lt_1;z_1)\Mt^{g}_J(\lt_2;z_2) \sim \frac{[12]}{8\pi^2\la12\ra}M(z_2)_f^g\int_0^1\d s\int_0^1\d t\,\Jt_{IJ}(s\lt_1 + t\lt_2;z_2)\,.
\ee

More generally, we consider $\phi$ contraction that is not adjacent to the $\gamma$ and $\til{\gamma}$ fields
\be
\wick{\gamma_f\phi(\lt_1)^a\c\phi(\lt_1)\phi(\lt_1)^{m-a-1}\psi_I(z_1) \til{\gamma}^g \phi(\lt_2)^a\c\phi(\lt_2)\phi(\lt_2)^{n-a-1}\psi_J(z_2)}
\ee
generating $\gamma_f\phi(\lt_1)^a\phi(\lt_2)^a\til{\gamma}^g$.
\bea \label{eq:OPE_MtM_fl_0}
&\M^{(m)}_{If}(\lt_1;z_1)\Mt^{(n),g}_J(\lt_2;z_2) \\
&\sim \frac{[12]}{8\pi^2\la12\ra }\gamma_f\phi(\lt_1)^a\phi(\lt_2)^a\til{\gamma}^g\psi_I\phi(\lt_1)^{m-a-1}\phi(\lt_2)^{n-a-1}\p_{w_1}\c\p_{w_2}\c\psi_J\,.
\eea
There are various ways to express this OPE. First, we can observe that, up to symmetrization of the $\phi$ field,
\bea
&\psi_I\phi(\lt_1)^{m-a-1}\phi(\lt_2)^{n-a-1}\p_{w_1}\c\p_{w_2}\c\psi_J(z_2) \\
&= \frac{(n-a-1)!}{(m+n-2a-2)!}(\lt_1\cdot\p_{\lt_2})^{n-a-1} \Jt_{IJ}^{(m+n-2a-2)}(\lt_2;z_2)\,.
\eea
where $\lt_1 \cdot \p_{\lt_2} = \lt_1^\da\p/\p\lt_2^\da$. Therefore, using the combinatorial formula \eqref{eq:comb-gaphiga} which implies 
\be \label{eq:sandwich-flavour}
\gamma_f\phi(\lt_1)^a\phi(\lt_2)^a\til{\gamma}^g(z) \sim \frac{(-)^a}{(a+1)!} [12]^a\big(M(z)^{a+1}\big)_f^g\,.
\ee
we can express \eqref{eq:OPE_MtM_fl_0} as follows
\bea \label{eq:OPE_MtM_fl_1}
&\M^{(m)}_{If}(\lt_1;z_1)\Mt^{(n),g}_J(\lt_2;z_2) \sim \sum_{a\ge0}\frac{[12]^{a+1}}{8\pi^2\la12\ra}\frac{(-)^a}{(a+1)!}(M(z_2)^{a+1})_f^g \\
&\frac{(n-a-1)!}{(m+n-2a-2)!}(\lt_1 \cdot \p_{\lt_2})^{n-a-1} \Jt_{IJ}^{(m+n-2a-2)}(\lt_2;z_2)\,.
\eea
Equivalently we can write this OPE in integral form as
\bea
&\M_{If}(\lt_1;z_1)\Mt^g_J(\lt_2;z_2) \sim \sum_{a\ge0}\frac{[12]^{a+1}}{8\pi^2\la12\ra}\frac{(-)^a}{(a+1)!}(M(z_2)^{a+1})_f^g\\
&\int_{0\leq s_1\leq\dots\leq s_{a+1}\leq1}\d^{a+1}\mbf{s}\int_{0\leq t_1\leq\dots\leq t_{a+1}\leq1}\d^{a+1}\mbf{t}\,\Jt_{IJ}(s_1\lt_1 + t_1\lt_2;z_2)\,.
\eea
Upon making the replacement $\M_f\mapsto\M_f/2\pi,\Mt^g\mapsto\Mt^g/2\pi$ to match the tree OPEs on flat space the factors of $\pi^2$ on the right hand side drop out.

Note that we do not expect these OPEs to associate on the nose. In a background where the celestial chiral algebra as a non-trivial two-point amplitude the failure in the associativity of the $2\to1$ OPEs can be compensated by the $2\to2$ OPE followed by the $2\to0$ OPE. 

 
\section{Amplitude and Form Factor Computations from Background Field Scattering} \label{sec:formfactors}

We have seen that the chiral algebra, where bulk operators are given a VEV, controls scattering amplitudes in certain self-dual backgrounds.  In this section, we will show how these computations can be rewritten as computations of form factors in Yang-Mills theory, in a context where the background field has been made dynamical.  In special cases, these form factors are the same as amplitudes in either Yang-Mills theory or Einstein-Yang-Mills theory.

This technique will allow us to generate new formulae for amplitudes in gauge theory, at up to two-loops, and more generally to prove rationality of a wide variety of form factors. 

Let us start in the simplest case, when the background field is a self-dual gauge field for flavour symmetry.  This is achieved when we set
\be
\gamma_f\til\gamma^g = M_f^g(z)
\ee
where $M_f^g(z)$ is a polynomial of order $2$ in $z$ valued in $\mf{sl}(16)$.  We can write $M$ instead as an expression symmetric in two spinor indices, $M_{\alpha\beta f}^g$.    

The background field satisfies the source equation 
\be
F(A_0)_{\alpha\beta f}^g = M_{\alpha\beta f}^g\delta_{x=0}\,. \label{eq:sourced-field}
\ee
Studying $\mf{sp}(K)$ self-dual gauge theory in this flavour background is very similar to studying $\mf{sp}(K)\oplus\mf{sl}(16)$ gauge theory where we have inserted the operator
\be
\exp{\left(\op{tr}(M^{\alpha\beta}B_{0,\alpha\beta}) \right)}\,. \label{eq:operator-insertion}
\ee
The $B_0$ in this equation is that of the $\mf{sl}(16)$ self-dual gauge theory.  

The goal of this section is to precisely formulate and prove the relationship between amplitudes in the background field \eqref{eq:sourced-field} and form factors for the operator \eqref{eq:operator-insertion}. 

The insertion of the operator \eqref{eq:operator-insertion} sources the field \eqref{eq:sourced-field}. However, this is problematic because this operator is not gauge invariant.  In order to match background field computations and form factor computations, we need to modify the background field construction to impose gauge invariance. 

Gauge invariance here means the following.  The background gauge field is parameterized by the matrices $M_{\alpha\beta}$, and so the scattering amplitudes in the presence of the background field are functions of the matrices $M_{\alpha\beta}$, as well as of the momenta of the external fields.  The gauge invariant part of the background field amplitude is obtained by projecting to functions of the $M_{\alpha\beta}$ which are invariant under the simultaneous conjugation $M_{\alpha\beta} \to g M_{\alpha\beta} g^{-1}$ of the three matrices.


\subsection{Background Field Amplitudes as Form Factors for Maxwell-Yang-Mills Theory}

We will use two slightly different ways to extract a gauge invariant part of the background field amplitude.  The first method is to choose a cocharacter of $\mrm{SL}(16)$, giving a copy of $\mf{gl}(1)$ inside $\mf{sl}(16)$.  We can then ask that the three matrices $M_{\alpha\beta}$ all lie in this copy of $\mf{gl}(1)$.  If we do this, then the background field amplitude is automatically a gauge invariant quantity, for gauge transformations in $\mf{gl}(1)$.  In this case, the operator \eqref{eq:operator-insertion} is also a gauge invariant operator for the Abelian gauge theory.  

This construction provides an exact equivalence between
\begin{enumerate}
\item Scattering amplitudes of the $\mf{sp}(K)$ SDYM in the flavour symmetry background given by \eqref{eq:sourced-field}.
\item Form factors for the operator \eqref{eq:operator-insertion}
in the $\mf{sp}(K)\oplus\mf{gl}(1)$ gauge theory with matter as before, where none of the external fields are photons of either helicity.
\end{enumerate}
It is very easy to see that these two quantities are the same, because the Feynman diagrams that compute them are identical (we will draw the diagrams below, when we consider the more complicated non-Abelian case).  The key point is that self-dual Maxwell theory has a $-+$ propagator, and only the positive helicity gluon interacts with matter fields.  The only Feynman diagrams that contribute to the form factor are diagrams where some number of photon propagators join the local operator to the fermion line. These are exactly the diagrams that appear in the Feynman diagram expansion of the background field amplitude, where now the same photons are sourced by the operator.  

It is important that we do not allow external photons. If there are external photons, then there would be additional diagrams in the form factor computation which involve fermion loops. 


\subsection{Background Field Amplitudes as Form Factors of \texorpdfstring{$\mf{sp}(K)\oplus\mf{sl}(8)$}{sp(K)+sl(8)} Gauge Theory}

We will also show a non-abelian version of this statement.   An extra subtlety here is that the bifundamental matter $\mrm{F}_K\otimes\C^{16}$ (where $\mrm{F}_K$ denotes the fundamental of $\mf{sp}(K)$) is chiral when viewed as a representation of $\mf{sl}(16)$.  It's therefore not possible to make the $\mf{sl}(16)$ gauge fields dynamical at the quantum level because of the chiral anomaly.  Instead, we must focus on some subalgebra of $\mf{sl}(16)$ for which the chiral anomaly vanishes.  

One natural way to do this is to choose a copy of $\mf{sl}(8)$ inside $\mf{sl}(16)$ under which the fundamental representation $\C^{16}$ decomposes as the fundamental plus the anti-fundamental of $\mf{sl}(8)$.  One can alternatively choose some subalgebra $\mf{sl}(R)\subset\mf{sl}(8)$ to make dynamical.

Let us do this, and ask that our matrices $M_{\alpha\beta}$ lie in $\mf{sl}(R)\subset\mf{sl}(8)$.  We will denote the amplitude depending on the  background field $M_{\alpha\beta}$ by an expression like
\be
\ip{ M_{\alpha\beta} \mid \J(\lt_1;z_1)\dots\Jt(\lt_n;z_n)}
\ee
where the states that are scattered can be fermions or gluons of any helicity.  We can project to background field amplitudes that are invariant under the adjoint action by performing a finite-dimensional integral over the space of matrices $M_{\alpha\beta}\in\mf{sl}(R)\otimes\op{Sym}^2(S_-)$, weighted by some integral kernel invariant under the adjoint action of $\mrm{SL}(R)$ on $\mf{sl}(R)$. We will choose to use average against
\be e^{-\tfrac{c}{2}\op{tr}(M^{\alpha\beta}M_{\alpha\beta})}P(M_{\alpha \beta}) \ee
for some $\mrm{SL}(R)$ invariant polynomial $P(M_{\alpha \beta})$ of the matrices $M_{\alpha\beta}$, giving us 
\be
\int_{M_{\alpha\beta}\in\mf{sl}(R)\otimes\op{Sym}^2(S_-)}e^{-\frac{c}{2}\op{tr}(M^{\alpha\beta}M_{\alpha\beta})} P(M_{\alpha\beta}) \ip{M_{\alpha\beta}\mid\J(\lt_1;z_1)\dots\Jt(\lt_n;z_n)}\,. \label{eq:averaged-amplitude}
\ee
(One should choose a contour so that the integral converges.)

We will show that there is an equivalence between:
\begin{enumerate}
\item Averages of single-trace background field amplitudes of the $\mf{sp}(K)$ SDYM theory in the flavour symmetry background given by \eqref{eq:averaged-amplitude}.
\item Form factors for the operator 
\be
\int_{M_{\alpha\beta}\in\mf{sl}(R)\otimes\op{Sym}^2(S_-)} e^{-\tfrac{c}{2}\op{tr}(M^{\alpha\beta}M_{\alpha\beta}) + \op{tr} (M^{\alpha\beta}B_{0,\alpha\beta})} P(M_{\alpha\beta})  \label{eq:averaged-operator}
\ee
in the $\mf{sp}(K)\oplus\mf{sl}(R)$ gauge theory, where none of the external fields are $\mf{sl}(R)$ gluons of either helicity.  Here $B_0$ is the anti-self-dual part of the field strength for the $\mf{sl}(R)$ self-dual gauge theory.  The Gaussian integral should be performed over a slice chosen to make the integral converge.
\end{enumerate}
Clearly, the Gaussian integral yields an operator of the form
\be
e^{\tfrac{1}{2c}\op{tr}(B_0^{\alpha\beta}B_{0,\alpha\beta})}\what{P}(B_{0,\alpha \beta})
\ee
where, if $P$ is a homogeneous polynomial of order $k$, then $\what{P}(B_{0,\alpha\beta})$ is $P(B_{0,\alpha\beta}/c)$ plus lower order terms arising from Wick contractions.  It's not hard to see that we can build any gauge invariant operator which only depends on $B_0$ (and not its derivatives) in this way.\footnote{It's possible to vary the construction to include derivatives of $B_0$ by giving a VEV to operators like $\gamma_f\exp(\phi^\da\lt_\da)\til \gamma^g$. We will not pursue this generalization.} We conclude that all such form factors must be rational functions with poles determined by a chiral algebra.\footnote{At least as long as we use an IR regulator that plays well with self-duality in $4d$. See \cite{Dixon:2024mzh} for further details.}

The proof of this equivalence relies on the fact that we are working with self-dual gauge theory, not full Yang-Mills.  In full Yang-Mills, the Feynman diagram expansion of the form factor would have complicated loop diagrams involving $\mf{sl}(R)$ gluons, not present in the background field computation.  In self-dual gauge theory, these loop diagrams don't contribute and the Feynman diagram expansion of the form factor is precisely the same as that of amplitudes in the presence of the background field.  

Form factors of self-dual gauge theory compute form factors of full Yang-Mills at particular loop numbers and helicity configurations, so that this method does allow us to access these special Yang-Mills form factors. 


\subsection{Feynman Diagrams for Background Field Amplitudes}

Let us describe the Feynman diagram interpretation of the background field amplitudes (these are illustrated in Fig. \ref{fig:backgroundloopdiagrams}). Working with background fields for the $\mf{sl}(R)$ flavour symmetry means treating the $\mf{sl}(R)$ sector classically, that is, at tree-level.  The relevant Feynman diagrams in the $\mf{sl}(R)$ sector are unions of rooted trees, where the leaves are placed at the source and the roots connect to the matter fields.

Diagrams like this are simply the Feynman diagram interpretation of building a solution to the $\mf{sl}(R)$ Yang-Mills equations with a source.

\begin{figure}[ht]
    \centering
        \subfloat[A diagram that can be interpreted as the scattering of two fermions to cubic order in the background field in $\mf{sl}(R)$, or equivalently as the form factor for the operator $\op{tr}(B_0^3)$ in the $\mf{sl}(R)$ gauge theory.]
        {
        \includegraphics{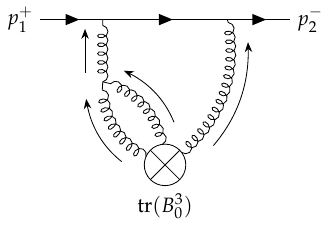}%
        } \hfil
        \subfloat[A diagram that is part of the form factor computation for the operator $\op{tr}(B_0^2)$, with two external positive helicity gluons in $\mf{sl}(R)$ and two external fermions.]
        {
        \includegraphics{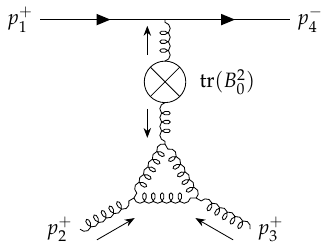}%
        }
        \caption{Two examples of Feynman diagrams contributing to form factors in $\mf{sl}(R)$ self-dual gauge theory. In both arrows indicate the flow of helicity from negative to positive. The crossed dot represents a local operator insertion; removing this vertex in the first diagram eliminates all $\mf{sl}(R)$ gluon loops, whereas in the second diagram a loop remains. Only the first can be obtained from a background field computation. \label{fig:backgroundloopdiagrams}}
\end{figure}

In the form factor computation, these diagrams also occur, but one can worry that there might also be other diagrams.  Recall that SDYM theory has a $-+$ propagator and a $++-$ vertex.  Diagrams can be viewed as in Fig. \ref{fig:backgroundloopdiagrams} as oriented, where the arrow goes from $-$ to $+$, and gluon vertices have two incoming and one outgoing.  The operator sourcing the background field will be something like $\op{tr}(B_0^k)$, where the spinor indices are contracted in some unspecified way.  This corresponds to a vertex attached to $k$ negative helicity gluons, which acts as a source.  

In the background field computation, the only diagrams that can contribute are those where the $\mf{sl}(R)$ gluons form unions of rooted trees, whose leaves attach to the vertex associated to the operator \eqref{eq:averaged-operator},  and whose roots connect to matter fields.  In these diagrams, two gluons that emerge from the local operator can join together, but a single gluon cannot split in two.

In the form factor computation, in principle one can have other diagrams, such as those illustrated in Fig. \ref{fig:backgroundloopdiagrams}.  These diagrams have the feature that a gluon emerging from the local operator splits into two. 

However, the combinatorics of SDYM Feynman diagrams tells us immediately that this second kind of diagram cannot occur, as long as none of our external fields are $\mf{sl}(R)$ gluons. In this situation, all of the $\mf{sl}(R)$ gluons must  be absorbed by the fermion line, which only couples to the positive helicity gluon.  Any time a gluon emerges from the local operator and splits into two, one of these two (after applying the $-+$ propagator) is negative helicity and can not be absorbed by the fermion line.

What this shows is that form factors for the non-gauge invariant operator $\exp \left( \op{tr} (M^{\alpha\beta}B_{0,\alpha\beta}) \right)$ and background field amplitudes in the background determined by $M^{\alpha \beta}$ are the same.  To obtain form factors of the gauge invariant amplitude, we perform an integral over $M$.  

Ultimately, $\mf{sl}(K)$ gauge theories are more interesting than $\mf{sp}(K)$ gauge theories.  If we restrict our external states to lie in the subalgebra $\mf{sl}(K) \subset \mf{sp}(K)$, the form factor we find will be that of the $\mf{sl}(K)$ gauge theory. This is simply because the Feynman diagrams that contribute to these form factors do not have any loops of $\mf{sp}(K)$ gluons.

This leads us to the following result.  Consider $\mf{sl}(K) \oplus \mf{sl}(R)$ self-dual gauge theory, with $R \le 8$ and with bifundamental matter.  As above let $B_0$ be the Lagrange multiplier field for the $\mf{sl}(R)$ theory.  Then, form factors for any trace of powers of $B_0$, with external gluons living only in $\mf{sl}(K)$, are rational and are given by correlators of a chiral algebra.  As mentioned before, this rationality is not expected to hold in dimensional regularization; instead one should use purely four-dimensional regulator such as a mass regulator.


\section{Two-Loop Amplitude Computations on Flat Space}

Let us use this technique to compute amplitudes of gauge theory with bifundamental matter.  The idea is the following.  As above, we consider the $\mf{sp}(K)$ self-dual gauge theory coupled to $\mf{sl}(R)$ self-dual gauge theory, where $\mf{sl}(R)\subset\mf{sl}(8)\subset\mf{sl}(16)$ and $\mf{sl}(8)\subset\mf{sl}(16)$ is a subalgebra under which the fundamental of $\sl(16)$ decomposes as the fundamental plus anti-fundamental of $\mf{sl}(8)$.  

If $\mrm{F}_K$ denotes the fundamental of $\mf{sp}(K)$ and $\mrm{F}_R$ that of $\mf{sl}(R)$, the relevant matter content consists of bifundamental matter in $\mrm{F}_K\otimes(\mrm{F}_R\oplus\mrm{F}_R^\vee)$.  (Matter uncharged under $\mf{sl}(R)$ will not be relevant for this computation).  From the point of view of the $\mf{sl}(R)$ gauge theory, this is self-dual gauge theory with $N_f = 2K$ (because $\mrm{F}_K$ is of dimension $2K$). 

We will compute all of the the two-loop all $+$ amplitudes in the $\mf{sp}(K)\oplus\mf{sl}(R)$ gauge theory with the following features:
\begin{enumerate}
    \item All the external gluons are in $\mf{sp}(K)$.
    \item There is a single $\mf{sl}(R)$ gauge theory propagator.  
\end{enumerate}
At four-points, there is only one Feynman diagram contributing to the amplitude; it is depicted in Fig. \ref{fig:twoloop}.  All diagrams contributing to this amplitude are planar.   

The $n$-point trace-ordered partial amplitude is
\be 
- \frac{R^2-1}{4R(4\pi)^4}\sum_{1 \le i < j < k < l \le n}\frac{\la ij\ra[jk]\la kl\ra[li] + [ij]\la jk\ra[kl]\la li\ra}{\la 12\ra\dots\la n1\ra}
\ee
Note that, up to a factor, this is the `parity-even part' of the one-loop all $+$ amplitude \cite{Mahlon:1993si,Bern:1993sx,Bern:1993qk}.  (Strictly only the numerator is parity even.)  The full amplitude is obtained by multiplying this by $\op{tr}_{\mrm{F}_K}(\t_1 \dots \t_n)$ and then summing over permutations.

It is straightforward to modify this to compute amplitudes where the gauge groups are $\mf{sl}(L) \subset \mf{sp}(K)$ and $\mf{sl}(R)$.  Because the $\mf{sp}(K)$ gluons only form trees in the Feynman diagrams of the amplitudes we are computing, we can compute the amplitudes the $\mf{sl}(L)$ gauge theory simply by restricting the external gluons to live in this subalgebra.  The fundamental of $\mf{sp}(K)$ becomes the fundamental plus anti-fundamental of $\mf{sl}(K)$, and we are considering $\mf{sl}(K) \oplus \mf{sl}(R)$ gauge theory with bifundamental matter.

If we do this, we find the 
trace-ordered amplitude becomes 
\be \label{eqn:twoloop} 
- \frac{R^2-1}{2R(4\pi)^4}\sum_{1 \le i < j < k < l \le n}\frac{\la ij\ra[jk]\la kl\ra[li] + [ij]\la jk\ra[kl]\la li\ra}{\la 12\ra\dots\la n1\ra}\,.
\ee
We get an extra factor of $2$ here because the trace ordered amplitude for $\mf{sl}(L)$ is multiplied by trace in the fundamental $F_L$, and not by the trace in $F_L \oplus F_L^\vee$.

At four-points Equation \eqref{eqn:twoloop} matches the $\eps\to0$ limit of the partial amplitude evaluated using dimensional regularization in \cite{Dixon:2024mzh,Bern:2002zk}. (The relevant formula is given in Equation (4.22) of the first reference.) The QCD five-point two-loop all-plus amplitude has been computed in \cite{Agarwal:2023suw}. We expect the same methods can be used to compute the partial amplitude considered here, and it would certainly be interesting to compare the results.


\subsection{Preliminaries}

It's important to note that this amplitude is essentially insensitive to the value of $R$.  Indeed, as we see from the Fig. \ref{fig:twoloop}, the interaction vertex for the $\mf{sl}(R)$ gluons never appears.  We find exactly the same answer by replacing the gauge algebra $\mf{sl}(R)$ by a copy of $\mf{gl}(1)\subset\mf{sl}(16)$.  We will do this in what follows, and at the end restore the colour factor.  If we do this then $\mrm{F}_K$ will be of some charge under $\mf{gl}(1)$. 

\begin{figure}[ht]
    \centering
        \includegraphics{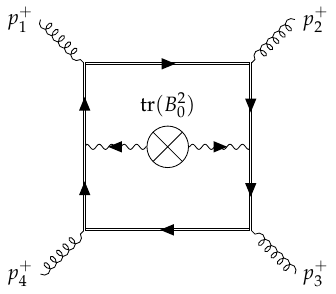}
    \caption{This figure depicts a two-loop diagram with external gluons in $\mf{sp}(K)$ and a gluon exchange in $\mf{sl}(R)$. Here the $\mf{sl}(R)$ gluons are represented by photon propagators with arrows indicating the flow of helicity. We argue above that the overall amplitude is unchanged, modulo a colour factor, if we replace $\mf{sl}(R)$ by an Abelian subalgebra. \label{fig:twoloop}}
\end{figure}

As before, we let $B_0$ denote the Lagrange multiplier field of the $\mf{gl}(1)$ gauge theory.  We would like to compute the form factor of the operator $\op{tr}(B_0^2)$, at two-loops, where the external fields are all $\mf{sp}(K)$ positive helicity gluons.  We have seen in Sect. \ref{sec:amplitude-bootstrap} that the counter-term ambiguities in this form factor are all total derivatives.  Therefore, while the form factor itself is ambiguous, its integral -- which is an amplitude -- is not ambiguous and can be computed from the chiral algebra.

To compute the amplitude from the chiral algebra, we should take some number of external gluons $\J^{(m)}(\lt;z) = \psi(\phi^\da\lt_\da)^m\psi$, contract all of the $\psi$ fields, and re-order the $\phi$ fields using the ADHM equation.  After reordering, we can be left with some terms like $(\gamma\til\gamma)^2$ (and with no $\phi$s).  Since the background field amounts to giving a VEV to $\gamma\til\gamma$, and we are working to quadratic order in the background field, only these terms will contribute to the form factor.  This makes it clear (and it is easy to see on other grounds) that the external fields must contain a total of exactly four $\phi$s.  

We also note that no $\phi-\phi$ or $\gamma-\til\gamma$ propagators can appear in the Feynman diagrams which contribute to this form factor.  Any such propagator would give rise to an operator containing $\p_{w_1}\c\p_{w_2}\c$ operator, and we are not giving a VEV to any such operator. 

We will use the freedom to add on counter-terms to simplify the form factor.  The tree-level form factor of $\op{tr}(F^2)$ has a two-point function
\be
\ip{\op{tr}(F^2)\mid\J^{(2)}(\lt_1;z_1)\J^{(2)}(\lt_2;z_2)} \propto [12]^2\,.
\ee
This operator is of course a total derivative, and does not affect the integrated form factor.  By adding this counter-term, we can assume that the two-loop two-point form factor of $\op{tr}(B_0^2)$ is zero.

This also implies that
\be \ip{\op{tr}(B_0^2) \mid \J^{(2)}(\lt_1;z_1)\J^{(1)}(\lt_2;z_2)\J^{(1)}(\lt_3;z_3)} = 0\,. \ee
Symmetry under the $\mrm{PSL}_2(\C)$ acting by conformal transformations on the $z$ plane tells us that this expression is proportional to $[12][13]/\la23\ra$.  However, the pole at $z_2 = z_3$ must be equal to the two-point function of $\J^{(2)}$ with $\J^{(2)}$, which we have already seen is zero.

The first non-zero trace-ordered amplitude is
\be
\ip{\op{tr}(B_0^2) \mid \J^{(1)}(\lt_1;z_1)\J^{(1)}(\lt_2;z_2)\J^{(1)}(\lt_3;z_3)\J^{(1)}(\lt_4;z_4)}\,,
\ee
which we will show evaluates to
\be \label{eqn:twoloop-again} - \frac{1}{2(4\pi)^4}\bigg(\frac{[12][34]}{\la12\ra\la34\ra} + \frac{[23][41]}{\la23\ra\la41\ra}\bigg)\,. \ee
This matches the $\eps\to0$ limit of this partial amplitude evaluated in dimensional regularization at four-points in \cite{Dixon:2024mzh,Bern:2002zk}. To recover the expression for $\mf{sl}(R)$ instead of $\mf{gl}(1)$ we simply dress the amplitude by a factor of
\be \op{tr}_{\mrm{F}_R}(\mrm{Cas}) = \frac{R^2-1}{2R} \ee
where $\mrm{Cas}$ is the Casimir of $\mf{sl}(R)$.


\subsection{Computing the Four-Point Function}

Let us compute this four-point function by studying the OPE of $\J^{(1)}_{IJ} = \psi_I\phi^\da\lt_\da\psi_J$ with itself.  Only $\psi-\psi$ Wick contractions play a role in this computation.  Using only $\psi-\psi$ contractions we have
\be
\J^{(1)}_{IJ}(\lt_1;z_1)\J^{(1)}_{KL}(\lt_2;z_2) = 
 - \frac{4\Omega_{(K(J}}{\la12\ra}\bigg(\J_{I)L)}^{(2)} (\lt_1+\lt_2;z_1) + \frac{1}{2}[12]\psi_{I)}[\phi^\done,\phi^\dtwo]\psi_{L)}(z_1)\bigg)\,.
\ee
The term involving $\J^{(2)}$ (where the $\phi$s are symmetrized) does not contribute.  We conclude that the pole in the four-point function at $z_1 = z_2$ is
\bea
&\ip{\op{tr}(B_0^2) \mid \J^{(1)}_{IJ}(\lt_1;z_1)\J^{(1)}_{KL}(\lt_2;z_2) \J^{(1)}_{MN}(\lt_3;z_3) \J^{(1)}_{PQ}(\lt_4;z_4)} \\
&\sim - \frac{2[12]}{\la12\ra}\Omega_{(K(J}\ip{\op{tr}(B_0^2)\mid \J_{I)L)}[C](z_1) \J^{(1)}_{MN}(\lt_3;z_3) \J^{(1)}_{PQ}(\lt_4;z_4) }
\eea
where $\J_{IL}[C] = - \J_{LI}[C]$ denotes $\psi_I[\phi^\done,\phi^\dtwo]\psi_L$.  

To compute the four-point function, it remains to compute the three-point function
\be
 \ip{ \op{tr}(B_0^2) \mid \J_{IL}[C](z_1) \J_{MN}^{(1)}(\lt_3;z_3) \J^{(1)}_{PQ}(\lt_4;z_4) }\,.
\ee
This can only have a pole at $z_3 = z_4$, because the residue of the pole at $z_1 = z_3$ or $z_3 = z_4$ would be the two-point function between operators of different spins.

The pole at $z_3 = z_4$ yields $\J[C](z_3)/2$ as well as $\J^{(2)}(\lt_3;z_3)$, but we have seen that $\J^{(2)}$ does not contribute.  The next step of the computation of the four-point function is to compute the two-point function
\be
\ip{ \op{tr}(B_0^2) \mid \J_{IJ}[C](z_1)\J_{KL}[C](z_3) }\,.
\ee
By the BRST relation, this is equivalent to computing
\be
\ip{ \op{tr}(B_0^2) \mid (\psi_I\til\gamma^f\gamma_f\psi_J)(z_1) (\psi_K\til\gamma^g\gamma_g\psi_L)(z_3) }\,.
\ee
Recall that the form factor we are computing is obtained by expanding to second order in a background field corresponding to giving a VEV 
\be
\ip{\gamma_f\til\gamma^g(z)} = \frac{1}{4\pi^2}M_f^g(z)
\ee
Here $M(z)$ is a matrix-valued polynomial of order $2$ in $z$, and we've normalized the right hand side so that on space-time the favour background sourced obeys
\be F_{\alpha\beta}(A_0) = M_{\alpha\beta}\delta_{x=0} \ee
As we discussed earlier, there is no loss of generality assuming that $M(z)$ is in a copy of $\mf{gl}(1)\subset\mf{sl}(16)$, so that $M(z) = f(z)D$ where $D$ is a single matrix with no $z$ dependence, and $f$ is a quadratic function of $z$. In spinor notation, we can write $M_{\alpha\beta} = \eta_{\alpha\beta}D$ for some two spinor index symmetric tensor $\eta$. 

We will write background field correlation functions as
\be
\ip{ \eta D \mid \mc{O}(z_1) \dots \mc{O}(z_n) }\,.
\ee  
Expanding in $\eta D$, the term quadratic in $D$ is $\op{tr}(D^2)$ multiplied by a linear combination of the two possible expressions built from $\eta$: 
\be \eta_{\alpha \beta} \eta^{\alpha \beta}\,,\qquad\eta_{(\alpha \beta}\eta_{\gamma\delta)}\,. \ee
The first term is in the trivial representation of $\mrm{SU}(2)$, and the second term is symmetrized in the spinor indices so that it lives in the spin $2$ representation.  The term proportional to $\ip{\eta \mid \eta}$ is exactly the form factor $\op{tr}(B_0^2)$.  
 
The next step in the computation is to to compute the expectation value  
\be \label{eq:twoptfunction} \ip{\eta D \mid (\psi_I\til\gamma^f\gamma_f \psi_J) (z_1) (\psi_K \til\gamma^g\gamma_g \psi_L) (z_3)} \ee
in this background, to quadratic order in $\eta D$. 

To do this, first we note that we are implicitly using the normally ordered product when we refer to $\psi_I\til\gamma^f\gamma_f\psi_J$. This operator is defined by the contour integral
\be \psi_I\til\gamma^f\gamma_f\psi_J(z) = \frac{1}{2\pi\i}\oint_{S^1}\frac{\d w}{w}\,(\psi_I\til\gamma^f)(z)(\gamma_f\psi_J)(z+w)\,, \ee
where $S^1$ is a small circle around $w=0$. Further, any $\gamma-\til\gamma$ contractions that appear in the normal ordering procedure will not matter, as they will lead to states involving $\p_{w_i}\c$ which vanish in our chosen background.

Using this expression for the normal ordering, the two-point function \eqref{eq:twoptfunction} can be written in terms of the four-point function
\be
\ip{ \eta D \mid (\psi_I\til\gamma^f)(z_1)(\gamma_f\psi_J)(z_2) (\psi_K\til\gamma^g)(z_3)(\gamma_g\psi_L)(z_4)}\,, \label{eq:fourptfunction}
\ee
by performing a contour integral which extracts the non-singular parts as $z_2\to z_1$, $z_4\to z_3$.  

It turns out that this four-point function has no singularities in this limit.  A $\til\gamma-\gamma$ OPE produces $\p_{w_i}\c$ fields, which can not contribute; and the $\psi(z_1)-\psi(z_2)$ OPE produces the bulk state $\gamma_f\til\gamma^f$.  Since we are assuming that the VEV for the $\gamma_f\til\gamma^g$ state is trace free, this also does not contribute.

The only singularities in the four-point function \eqref{eq:fourptfunction} are at $z_1 = z_4$ and $z_2 = z_3$ (as we do not give a VEV to the $\gamma_f\gamma_g$ or $\til\gamma^f\til\gamma^g$ bulk operators).  The pole at $z_1 = z_4$ and $z_2 = z_3$ is
\be \label{eq:polarpartfourpt}
\frac{\Omega_{JK}\Omega_{LI}}{\la14\ra\la23\ra}\ip{ \eta D \mid (\gamma_g\til\gamma^f)(z_1) (\gamma_f\til\gamma^g)(z_3) } = \frac{1}{(2\pi)^4}\frac{\Omega_{JK}\Omega_{LI}}{\la14\ra\la23\ra} \op{tr}(D^2) \ip{\eta \mid 1^2 }\ip{\eta \mid 3^2}\,.
\ee
Using the Schouten identity we can write 
\be
\ip{\eta \mid 1^2} \ip{\eta \mid 3^2} = \frac{1}{3}\ip{\eta \mid \eta} \ip{13}^2 + \ip{\eta^2 \mid 1^2 3^2 }\,,
\ee
where the second term is totally symmetric in the four spinor indices of $\eta^2$. This term can be discarded as, to compute the form factor of the $\mrm{SO}(4)$ invariant operator $\op{tr}(B_0^2)$, we need to expand to second order in $\eta D$ and then project onto invariant expressions.

Finally, we find that
\be \ip{ \op{tr}(B_0^2) \mid (\gamma_g\til\gamma^f)(z_1) (\gamma_f\til\gamma^g)(z_3) } = \frac{1}{(2\pi)^4}\ip{13}^2\,. \ee 
Inserting this back into our previous computations gives us
\be \ip{ \op{tr}(B_0^2) \mid (\psi_I\til\gamma^f)(z_1)(\gamma_f\psi_J)(z_2)(\psi_K\til\gamma^g)(z_3)(\gamma_g\psi_L)(z_4) } = \frac{\Omega_{JK}\Omega_{LI}}{(2\pi)^4}\frac{\ip{13}^2}{\la14\ra\la23\ra}\,. \ee
Indeed, as we have seen, this correlator only has poles at $z_1 = z_4$ and $z_2 = z_3$, and symmetry under $\mrm{PSL}(2)$ transformations in the $z$ plane determines it in full. Sending $z_2\to z_1$ and $z_4\to z_3$ gives us 
\be \ip{ \op{tr}(B_0^2) \mid (\psi_I\til\gamma^f\gamma_f\psi_J)(z_1)(\psi_K\til\gamma^g\gamma_g\psi_L)(z_3)} = \frac{\Omega_{JK}\Omega_{LI}}{(2\pi)^4}\,, \ee
or equivalently,
\be
\ip{ \op{tr}(B_0^2) \mid \J_{IJ}[C](z_1)\J_{KL}[C](z_3) } = \frac{1}{(2\pi)^4}\Omega_{J[K}\Omega_{L]I} = - \frac{1}{4(2\pi)^4}\op{tr}_{\mrm{F}_K}(\mrm{a}_{IJ}\mrm{a}_{KL})\,.
\ee
where $(\mrm{a}_{IJ})^K_{~\,L} = 2\delta^K_{~\,[I}\Omega_{J]L}$.

The singularity at $z_3 = z_4$ and the conformal weights determine the three-point function
\be
\ip{ \op{tr}(B_0^2) \mid \J_{IL}[C](z_1) \J_{MN}^{(1)}(\lt_3;z_3)\J^{(1)}_{PQ}(\lt_4;z_4)} = - \frac{1}{8(2\pi)^2}\frac{[34]}{\la34\ra}\op{tr}_{\mrm{F}_K}(\mrm{a}_{IL}\{\t_{MN},\t_{PQ}\})\,.
\ee
Here we've exploited the fact that in the fundamental representation  $\{\t_{MN},\t_{PQ}\} = -4\Omega_{(P(N}\mrm{a}_{M)Q}$.  From this we find that the four-point function is
\bea
&\ip{ \op{tr}(B_0^2) \mid \J_{IJ}^{(1)}(\lt_1;z_1)\J_{KL}^{(1)}(\lt_2;z_2)\J_{MN}^{(1)}(\lt_3;z_3)\J_{PQ}^{(1)}(\lt_4;z_4) } \\
&= - \frac{1}{(4\pi)^4}\frac{[12][34]}{\la12\ra\la34\ra}\op{tr}_{\mrm{F}_K}(\{\t_{IJ},\t_{KL}\}\{\t_{MN},\t_{PQ}\}) + \text{cyclic permutations of $\{2,3,4\}$}
\eea
Here when we cyclically permute $\{2,3,4\}$ we also permute the corresponding colour indices. Decomposing the anticommutators so that we have a sum over fundamental traces, we get a total of twelve terms. There are six possible orderings of the trace $\op{tr}(\t_1\t_{\sigma(2)}\t_{\sigma(3)}\t_{\sigma(4)})$ where $\sigma\in S(\{2,3,4\})$. The coefficient of the trivial permutation gives the trace-ordered four-point function
\bea \label{eqn:4pt-trace-orderd}
&\ip{ \op{tr}(B_0^2) \mid \J^{(1)}(\lt_1;z_1)\J^{(1)}(\lt_2;z_2)\J^{(1)}(\lt_3;z_3)\J^{(1)}(\lt_4;z_4) } \\
&= - \frac{1}{(4\pi)^4}\bigg(\frac{[12][34]}{\la12\ra\la34\ra} + \frac{[23][41]}{\la23\ra\la41\ra}\bigg)\,.
\eea
Finally, in order to deform sdQCD to the corresponding QCD we should insert $\op{tr}(B_0^2)/2$, leading to the quoted expression \eqref{eqn:twoloop-again}.

The $n$-point function is determined relatively easily by induction, using the tree-level OPE -- see \cite{Costello:2023vyy} for similar computations.\footnote{We are very grateful to Anthony Morales for identifying an error in the original version of this formula. This was caused by omitting the second term in equation \eqref{eqn:4pt-trace-orderd}. Whilst on the support of momentum conservation both terms are equal, as form factors they differ leading to an incorrect inductive $n$-point formula.} We find that the trace-ordered $n$-point amplitude is
\bea
&\ip{ \op{tr}(B_0^2) \mid \J(\lt_1;z_1) \dots \J(\lt_n;z_n) } \\
&= - \frac{1}{2(4\pi)^4}\sum_{1 \le i < j < k < l \le n} \frac{\la ij\ra[jk]\la kl\ra[li] + [ij]\la jk\ra[kl]\la li\ra} {\ip{12} \dots \ip{n1}}\,.
\eea
This is proportional to the `parity-even part' of the one-loop all $+$ amplitude. (Strictly only the numerator, denoted $E_n$ in \cite{Bern:1993qk}, is parity-even.)

As we mentioned, the $\mf{sl}(R)$ version of this calculation will be the same, where $\mf{sl}(R)$ gauge indices are contracted in the fundamental representation.


\section{Non-Factorizing OPEs} \label{sec:nonfactorizing}

We have analyzed some one-loop corrections to the OPE in Sect. \ref{subsec:NonPlanarOPE}, which are known to correspond to the one-loop QCD splitting amplitude. However, as shown in \cite{Costello:2022upu,Bittleston:2022jeq}, these $2\to 1$ OPEs alone fail to satisfy the associativity of the chiral algebra, and some $2\to 2$ OPEs are required to restore associativity. In this section, we derive these $2\to 2$ OPEs from the large $N$ chiral algebra.

In order to do so, we first review necessary ingredients for the OPE calculations. The OPEs of the defect fermions are given by
\be \psi_I^p(z_1)\psi_J^q(z_2) \sim \frac{1}{\la12\ra}\Omega_{IJ}\omega^{qp}\,. \ee
Here $\Omega$ and $\omega$ are the symplectic forms for $\mf{sp}(K)$ and $\mf{sp}(N)$ respectively. We will also need the bulk $\phi^\done-\phi^\dtwo$, $\gamma-\til\gamma$ and $\c-\b$ OPEs computed in Appendix \ref{app:diagram}
\bea
&\phi^{\done pq}(z_1)\phi^{\dtwo rs}(z_2) \sim \frac{1}{4\pi^2\la12\ra}\eps_{ij}\big(\omega^{[r[p}\p_{w_i}\c^{q]}_{~~t}\p_{w_j}\c^{s]t} + \p_{w^i}\c^{[r[p}\p_{w^j}\c^{q]s]}\big)(0)\,, \\
&\gamma_f^p(z_1)\til\gamma^{gq}(z_2) \sim \frac{1}{8\pi^2\la12\ra}\delta_f^{~\,g}\eps_{ij}\p_{w_i}\c^p_{~t}\p_{w_j}\c^{qt}(0)\,, \\
&\c^{pq}(z_1)\b^{rs}(z_2) \sim \frac{1}{4\pi^2\la12\ra}\eps_{ij}\big(\omega^{(r(p}\p_{w_i}\c^{q)}_{~\,t}\p_{w_j}\c^{s)t} - \p_{w_i}\c^{(r(p}\p_{w_j}\c^{q)s)}\big)(0)\,.
\eea

An important part of the computation is to rearrange the OPE into our chosen BRST representatives using ADHM constraint. The ADHM constraint receives a quantum correction arising from the defect: The BRST operator acts on the defect fermions via
\be Q\psi^p_I = \c^p_{~\,q}\psi^q_I\,, \ee
and so takes the form
\be Q = Q_\mrm{bulk} + \frac{1}{4\pi\i}\oint\d z\,\psi_I\c\psi^I\,. \ee
This leads to a modification of the ADHM constraint
\be Q_\mrm{bulk}\b = \dbar\b + [\c,\b] + [\phi^\done,\phi^\dtwo] + \frac{1}{2}\{\til\gamma^f,\gamma_f\}\,, \ee
so that it now reads
\bea
&Q\b^{rs} = \dbar\b^{rs} + [\c,\b]^{rs} + [\phi^\done,\phi^\dtwo]^{rs} + \frac{1}{2}\{\til\gamma^f,\gamma_f\}^{rs} \\
&+ \frac{1}{8\pi^2}\eps_{ij}\big((\p_{w_i}\c\psi_O)^{(r}(\p_{w_j}\c\psi^O)^{s)} + \psi_O^{(r}(\p_{w_i}\c\p_{w_j}\c\psi^O)^{s)}\big)\,.
\eea

We are now in a position to recover the non-factorizing part of the one-loop splitting function from the chiral algebra.

\begin{figure}[ht]
\centering
    \subfloat{%
    \includegraphics{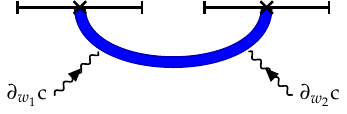}%
    } \hfil
    \subfloat{%
    \includegraphics{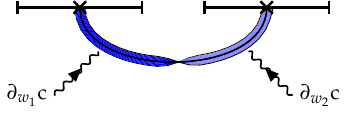}%
    }
    \caption{Contributions to the $2\to2$ OPE from a single $\phi$ contraction. In the first diagram the external $\p_{w_1}\c,\p_{w_2}\c$ appear on the same side of the propagator, and so it generates $\normord{\J_{(J(K}\Jt_{L)I)}}$ terms. In the second $\p_{w_1}\c$ is on a different side to $\p_{w_2}\c$, and it therefore generates $\normord{\M_{(J(K}\Mt_{L)I)}}$.} \label{fig:single-phi}
\end{figure}

As a warm-up, we first consider the OPE between $\J_{IJ}[1,0]$ and $\J_{KL}[0,1]$. There are two kinds of Wick contractions that can lead to a $2 \to 2$ OPE: a single contraction of a $\psi$ pair or a single contraction of a $\phi$ pair. We first consider the case of a single $\psi$ contraction.
\bea
&\J_{IJ}[1,0](z_1)\J_{KL}[0,1](z_2) \sim \wick{\psi_I \phi^\done\c\psi_J(z_1) \c\psi_K\phi^\dtwo\psi_L(z_2)} + (I\leftrightarrow J),(K\leftrightarrow L) \\
&\sim - \frac{2}{\la 12 \ra}\Omega_{(K(J}\psi_{I)}\{\phi^\done,\phi^\dtwo\}\psi_{L)}(z_2) - \frac{2}{\la 12  \ra}\Omega_{(K(J}\psi_{I)}[\phi^\done,\phi^\dtwo]\psi_{L)}(z_2)\,.
\eea
The first term is analyzed in the previous section and produces the planar OPE. In this section, we focus on the second term. Here we use the full ADHM equation derived from the modified BRST transformation \ref{eq:BRST-full} in the presence of the defect
\be \label{eq:ADHM-full}
[\phi^\done,\phi^\dtwo]^{rs} + \frac{1}{2}\{\til\gamma^f,\gamma_f\}^{rs} + \frac{1}{8\pi^2}\eps_{ij}\big((\p_{w_i}\c\psi_O)^{(r}(\p_{w_j}\c\psi^O)^{s)} + \psi_O^{(r}(\p_{w_i}\c\p_{w_j}\c\psi^O)^{s)}\big) = 0\,.
\ee
We find that 
\bea
&\psi_I[\phi^\done,\phi^\dtwo]\psi_L = - \frac{1}{8\pi^2}\big(\psi_I\psi_O\psi^O\p_{w_1}\c\p_{w_2}\c\psi_L + \psi_I\p_{w_1}\c\p_{w_2}\c\psi_O\psi^O\psi_L \\
&+ \psi_I\p_{w_1}\c\psi_O\psi^O\p_{w_2}\c\psi_L -  \psi_I\p_{w_2}\c\psi_O\psi^O\p_{w_1}\c\psi_L\big) - \frac{1}{2}\big(\psi_I\til\gamma^f\gamma_f\psi_L + \psi_I\gamma_f\til\gamma^f\psi_L\big)\,.
\eea
As a result, a single $\psi$ contraction lead to the following $2\to 2$ OPE:
\bea \label{eq:OPE2->2simple1}
&\J_{IJ}[1,0](z_1)\J_{KL}[0,1](z_2) \sim - \frac{1}{\la12\ra}\big(\Omega_{(J(K}\M_{L),f}\Mt_{I)}^f(z_2) + \Omega_{(K(J}\M_{I),f}\Mt_{L)}^f(z_2)\big) \\
&- \frac{1}{4\pi^2\la 12\ra}\big(\Omega_{(J(K}\J_{L)}^{~\,~O}\Jt_{I)O}(z_2) + \Omega_{(K(J}\J_{I)}^{~\,~O}\Jt_{L)O}(z_2) \\
&+  \Omega_{(J(K}\M_{L)}^{~\,~O}\Mt_{I)O}(z_2) + \Omega_{(K(J}\M_{I)}^{~\,~O}\Mt_{L)O}(z_2)\big)\,.
\eea
Next, we consider the case of a single $\phi$ contraction (Fig. \ref{fig:single-phi}).
\bea \label{eq:OPE2->2simple2}
&\J_{IJ}[1,0](z_1)\J_{KL}[0,1](z_2) \sim  \wick{ \psi_{Ip} \c\phi^{\done pq} \psi_{Jq}(z_1) \psi_{Kr} \c \phi^{\dtwo rs} \psi_{Ls}(z_2)} \\
&= \frac{1}{4\pi^2\la12\ra} \psi_{Ip} \psi_{Jq} \psi_{Kr} \psi_{Ls} \epsilon_{ij}(\omega^{[r[p} \p_{w_i}\c^{q]}_t \p_{w_j}\c^{s]t} + \p_{w_i}\c^{[r[p} \p_{w_j}\c^{q]s]})(z_2) \\
& = - \frac{1}{2\pi^2\la12\ra}\big(\J_{(J(K}\Jt_{L)I)}(z_2) - \M_{(J(K}\Mt_{L)I)}(z_2)\big)\,.
\eea
Upon rescaling $\M_f\mapsto\M_f/2\pi,\Mt^g\mapsto\Mt^g/2\pi$ to fix the standard form for the tree OPEs, Equations \eqref{eq:OPE2->2simple1} and \eqref{eq:OPE2->2simple2} match the results obtained from associativity in \cite{Costello:2022upu}.

\begin{figure}[ht]
    \centering
        \includegraphics{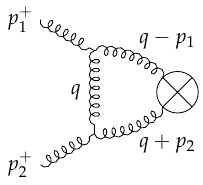}
    \caption{Non-factorizing diagram leading to a subleading collinear singularity at one-loop. There are further contributions with fermions running through the loop.} \label{fig:non-factorizing}
\end{figure}

These terms capture the subleading collinear singularity generated by the loop integration in the non-factorizing diagram illustrated in Fig. \ref{fig:non-factorizing}.

\paragraph{A more general case}
As a more nontrivial example, we compute the $2\to2$ OPE between $\J^{(1)}_{IJ}(\lt_1;z_1)$ and $\J^{(n)}_{KL}(\lt_2;z_2)$ for general $n$. 
	
We first consider the case of a single $\psi$ contraction. 
\be
\J^{(1)}_{IJ}(\lt_1;z_1)\J^{(n)}_{KL}(\lt_2;z_2) \sim - \frac{4}{\la12\ra} \Omega_{(K(J}\psi_{I)}\phi(\lt_1)\phi(\lt_2)^n\psi_{L)}(z_2)\,.
\ee
We have the following identity
\be
\sum_{i=0}^{n-1}(i+1)\phi(\lt_2)^{n-1-i}[\phi(\lt_1),\phi(\lt_2)]\phi(\lt_2)^i = (n+1)\phi(\lt_1)\phi(\lt_2)^n - (\lt_1\cdot\p_{\lt_2})\phi(\lt_2)^{n+1}\,.
\ee
Therefore,
\bea
&\J^{(1)}_{IJ}(\lt_1;z_1)\J^{(n)}_{KL}(\lt_2;z_2) \sim - \frac{4}{\la 12 \ra} \frac{1}{n+1} \Omega_{(K(J}(\lt_1\cdot\p_{\lt_2}) \psi_{I)} \phi(\lt_2)^{n+1}\psi_{L)}(z_2) \\
&- \frac{4}{\la 12 \ra}\sum_{i = 0}^{n-1}\frac{i+1}{n+1}\Omega_{(K(J}  \psi_{I)}\phi(\lt_2)^{n-1-i}[\phi(\lt_1),\phi(\lt_2)]\phi(\lt_2)^i\psi_{L)}(z_2)\,.
\eea
The first term gives a $2 \to 1$ OPE and is analyzed in the last section. We focus on the terms in the second line. Using the ADHM equation \ref{eq:ADHM-full} we find the following
\bea
&\psi_I\phi(\lt_2)^{n-1-i}[\phi(\lt_1),\phi(\lt_2)]\phi(\lt_2)^i\psi_L(z_2) = - \frac{[12]}{8\pi^2}\big(\normord{\J^{(n-1-i)}_{IO}(\lt_2;z_2)\Jt^{(i),O}_L(\lt_2;z_2)} \\
&+ \normord{\M^{(n-1-i)}_{IO}(\lt_2;z_2)\Mt^{(i),O}_L(\lt_2;z_2)} - (I \leftrightarrow L,i \leftrightarrow n-1-i)\big) \\
&+ \frac{1}{2}\big(\normord{\M_{I,f}^{(n-1-i)}(\lt_2;z_2)\Mt_L^{(i),f}(\lt_2;z_2)} - (I \leftrightarrow L,i \leftrightarrow n-1-i)\big)\,.
\eea
Therefore, a single $\psi$ contraction gives us the following $2 \to 2$ OPE
\bea \label{eq:psi2->2}
&\J^{(1)}_{IJ}(\lt_1;z_1)\J^{(n)}_{KL}(\lt_2;z_2) \sim - \frac{[12]}{2\pi^2\la12\ra}\Omega_{(K(J}\sum_{i = 0}^{n-1}\frac{n-i}{n+1}\Big(\normord{\J^{(i),O}_{I)}(\lt_2;z_2)\Jt^{(n-1-i)}_{L)O}(\lt_2;z_2)} \\
&+ \normord{\M^{(i),O}_{I)}(\lt_2;z_2)\Mt^{(n-1-i)}_{L)O}(\lt_2;z_2)} - (I\leftrightarrow L,i \leftrightarrow n-1-i)\Big) \\
&- 2\Omega_{(K(J}\sum_{i = 0}^{n-1}\frac{n-i}{n+1}\Big(\normord{\M^{(i)}_{I)f}(\lt_2;z_2)\Mt^{(n-1-i),f}_{L)}(\lt_2;z_2)} - (I\leftrightarrow L,i \leftrightarrow n-1-i)\Big)\,.
\eea
A single $\phi$ contraction gives us the following
\bea \label{eq:phi2->2}
&\J^{(1)}_{IJ}(\lt_1;z_1)\J^{(n)}_{KL}(\lt_2;z_2) \\
&\sim \sum_{i = 0}^{n-1}\wick{\psi_{Ip}\c\phi^{pq}(\lt_1) \psi_{Jq}(z_1) (\psi_K\phi(\lt_2)^i)_r\c\phi^{rs}(\lt_2) (\phi(\lt_2)^{n-1-i}\psi_L)_s(z_2)} \\
&= \frac{[12]}{4\pi^2\la12\ra}\sum_{i = 0}^{n-1}\psi_{Ip}\psi_{Jq}(\psi_{K}\phi(\lt_2)^i)_r(\phi(\lt_2)^{n-1-i}\psi_L)_s\eps_{ij}(\omega^{[r[p}\p_{w_i}\c^{q]}_t\p_{w_j}\c^{s]t} + \p_{w_i}\c^{[r[p}\p_{w_j}\c^{q]s]}) \\
&= - \frac{[12]}{2\pi^2\la12\ra}\sum_{i = 0}^{n-1}\big(\J^{(i)}_{(J(K}(\lt_2;z_2)\Jt^{(n-1-i)}_{L)I)}(\lt_2;z_2) - \M^{(i)}_{(J(K}(\lt_2;z_2)\Mt^{(n-1-i)}_{L)I)}(\lt_2;z_2)\big)\,.
\eea

The full $2\to2$ OPE can be expressed using the structure constants $R_{\{k,\tilde k\},\{l,\tilde l\}}$ evaluated in Appendix \ref{app:flavour-combinatorics}. We find that
\bea \label{eq:2->2soft}
&\J_{IJ}[m,0](z_1)\J_{KL}[0,n](z_2) \\
&\sim - \frac{1}{2\pi^2\la12\ra}\sum_{k+\tilde k=m-1}\sum_{l+\tilde l=n-1}\Big(\normord{\J_{(J(K}[k,l]\Jt_{L)I)}[\tilde k,\tilde l]} - \normord{\M_{(J(K}[k,l]\Mt_{L)I)}[\tilde k,\tilde l]} \\
&+ R_{\{k,\tilde k\},\{l,\tilde l\}}\Omega_{(K(J}(\normord{\J_{I)}^{~\,~O}[k,l]\Jt_{L)O}[\tilde k,\tilde l]} + \normord{\M_{J)}^{~~O}[k,l]\Mt_{K)O}[\tilde k,\tilde l]}) \\
&+ R_{\{l,\tilde l\},\{k,\tilde k\}}\Omega_{(J(K}(\normord{\J_{L)}^{~~\,O}[k,l]\Jt_{I)O}[\tilde k,\tilde l]} + \normord{\M_{L)}^{~~O}[k,l]\Mt_{I)O}[\tilde k,\tilde l]})\Big)(z_2) \\
&- \frac{2}{\la12\ra}\sum_{k+\tilde k=m-1}\sum_{l+\tilde l=n-1}\big(R_{\{k,\tilde k\},\{l,\tilde l\}}\Omega_{(K(J}\normord{\M_{I)f}[k,l]\Mt_{L)}^f[\tilde k,\tilde l]} \\
&+ R_{\{l,\tilde l\},\{k,\tilde k\}}\Omega_{(J(K}\normord{\M_{L)f}[k,l]\Mt_{I)}^f[\tilde k,\tilde l]}\big)(z_2)
\eea
This takes a considerably simpler form when expressed using hard generators. Using the integral expression for the structure constants $R_{\{k,\tilde k\},\{l,\tilde l\}}$ given in Appendix \ref{app:flavour-combinatorics}
\bea \label{eq:2->2hard}
&\J_{IJ}(\lt_1;z_1)\J_{KL}(\lt_2;z_2) \sim - \frac{[12]}{2\pi^2\la12\ra}\int_{0\leq s_i\leq1}\d^2\mbf{s}\int_{0\leq t_j\leq1}\d^2\mbf{t}\,\delta(1-s_1-s_2)\delta(1-t_1-t_2) \\
&\big(\normord{\J_{(J(K}(s_1\lt_1+t_1\lt_2;z_2)\Jt_{L)I)}(s_2\lt_1+t_2\lt_2;z_2)} \\
&+ \Theta(t_1\leq s_1)\Omega_{(K(J}\normord{\J_{I)}^{~~\,O}(s_1\lt_1+t_1\lt_2;z_2)\Jt_{L)O}(s_2\lt_1+t_2\lt_2;z_2)} \\
&+ \Theta(s_1\leq t_1)\Omega_{(J(K}\normord{\J_{L)}^{~~\,O}(s_1\lt_1+t_1\lt_2;z_2)\Jt_{I)O}(s_2\lt_1+t_2\lt_2;z_2)}\big)
+ \text{fermions}\,.
\eea
Here $\Theta(t)$ is the indicator function. The fermion terms can be inferred from equation \eqref{eq:2->2soft}.

We show in Appendix \ref{app:simplifyOPE} the these terms reproduce the full subleading one-loop collinear singularity coming from the non-factorizing illustrated in Fig. \ref{fig:non-factorizing}. For a general gauge algebra $\g$ and anomaly free representation $R$ these take the form
\bea
&\J_\sfa(\lt_1;z_1)\J_{\sfb}(\lt_2;z_2) \sim \frac{[12]}{8\pi^2\la12\ra}\int_{0\leq s_i\leq1}\d^2\mbf{s}\int_{0\leq t_j\leq1}\d^2\mbf{t}\,\delta(1-s_1-s_2)\delta(1-t_1-t_2) \\
&\Theta(t_1\leq s_1)f_{\sfa\sfe}^{~\,~\sfc} f_{\sfb\sfd}^{~\,~\sfe}\normord{\J_\sfc(s_1\lt_1+t_1\lt_2)\Jt^\sfd(s_2\lt_1+t_2\lt_2)}(z_2) \\
&+ \Theta(s_1\leq t_2)f_{\sfb\sfe}^{~\,~\sfc}f_{\sfa\sfd}^{~\,~\sfe}\normord{\J_\sfc(s_1\lt_1+t_1\lt_2)\Jt^\sfd(s_2\lt_1+t_2\lt_2)}(z_2) + \mrm{fermions}\big)\,.
\eea
Unlike the leading double pole discussed in Sect. \ref{subsec:NonPlanarOPE}, which only appears in the $\J-\J$ OPE, non-factorizing terms of the above type can appear in the OPEs of many states in the chiral algebra. It's not hard to see that these are also reproduced in the dual \cite{Zeng:2023qqp}.


\subsubsection*{Acknowledgements}  We would like to thank Lance Dixon, Davide Gaiotto, Anthony Morales, Natalie Paquette, Atul Sharma, David Skinner and Yehao Zhou for helpful comments and conversations. We are furthermore particularly grateful to Anthony Morales for pointing out a computational error in the first version of this paper. This work was supported by the Simons Collaboration on Celestial Holography. Research at Perimeter Institute is supported in part by the Government of Canada through the Department of Innovation, Science and Economic Development and by the Province of Ontario through the Ministry of Colleges and Universities. KZ is also supported by Harvard University CMSA.

\begin{appendix}


\section{Tree-Level Large \texorpdfstring{$N$}{N} BRST Cohomology Computations} \label{app:tree-cohomology}

In this section, we will compute the large $N$ BRST cohomology at tree-level, using standard homological algebra techniques. 

Given any Lie algebra $\g$, and any symplectic representation $V$ of $\g$, there is a super-Lie algebra denoted by $\g_V$ in \cite{Costello:2018fnz} which consists of 
\be \g_V = \g\oplus\Pi V\oplus\g^\vee\,. \ee
The Lie bracket of an element of $\g$ with anything is given by the $\g$ action, and the only other non-trivial bracket is the map 
\be V\otimes V\to\g^\vee \ee
given by the moment map.

The Lie algebra describing the fields on the stack of $N$ $D5$ branes is
\be \mscr{L}_N = \mf{sp}(N)_{(\wedge^2_0\mrm{F}_N\otimes\C^2) \oplus (\C^{16}\otimes\mrm{F}_N) \oplus (\C^{16}\otimes\mrm{F}_N)}\llbracket w_1,w_2,z\rrbracket\,. \ee
That is, we form the Lie algebra $\g_V$ for our gauge algebra $\mf{sp}(N)$ and matter content $(\wedge^2_0\mrm{F}_N\otimes\C^2)\oplus(\C^{16}\otimes\mrm{F}_N)\oplus(\C^{16}\otimes\mrm{F}_N)$, and then adjoin series in the three variables $w_1,w_2,z$.  (Relative) Lie algebra cohomology of this Lie algebra describes classical local operators on the $D5$ brane system; this is by now a standard construction \cite{Eager:2018oww,Chang:2013fba,Costello:2018zrm,Budzik:2023xbr}.

For future reference, we note that the part of $\mscr{L}_N$ which does not contain any fundamental representations is the $\Z/2$ fixed points in the algebra
\be \sl(2N)\llbracket z,w_1,w_2,\eps_\done,\eps_\dtwo\rrbracket  \ee
where $\eps_\da$ are fermionic variables.  The $\Z/2$ action acts on $\sl(2N)$ by the symplectic transpose 
\bea
A \mapsto A^t\,,\qquad (A^t)_p^q = \omega_{pr}\omega^{qs}A^r_s\,.
\eea
whose fixed points are the Lie algebra $\mf{sp}(N)$.  The $\Z/2$ action also sends $\eps_\da\to-\eps_\da$. This means that on the fixed points, the coefficients of $\eps_\da$ live in $\wedge^2_0\mrm{F}_N$.     

To incorporate the fermions on the defect, we note that there is a Lie algebra homomorphism
\be \mscr{L}_N \to \mf{sp}(N) \llbracket z\rrbracket \label{eq:homomorphism} \ee
obtained by setting $w_i \to 0$ and sending the matter fields and $\mf{sp}(N)^\vee$ to zero.  The (Abelian) Lie algebra for the free fermions is
\be \C^{2K} \otimes F \llbracket z\rrbracket \,. \ee
These elements are given bosonic parity; the formulation of Lie algebra cohomology reverses parity, giving us fermionic operators. 

This is acted on by $\mscr{L}_N$ via the homomorphism \eqref{eq:homomorphism}.  Thus, we can form the semi-direct product Lie algebra
\be \mscr{L}_{N,K} = \mscr{L}_N \ltimes \big(\C^{2K} \otimes F\llbracket z\rrbracket\big)\,. \ee
The Lie algebra cohomology
\be H^\ast(\mscr{L}_{N,K},\mf{sp}(N))\,. \ee
relative to $\mf{sp}(N)$ is the space of classical local operators on the defect in the $D5$ system.

We are interested in the part of the Lie algebra cohomology, at large $N$, which does not come from the pure $D5$ system.  At large $N$, the Lie algebra cohomology is quite easy to compute, using the Loday-Quillen-Tsygan \cite{loday1984cyclic,tsygan1983homology} theorem and related results \cite{LODAY198893,Zeng:2023lox}.  In this limit, the Lie algebra cohomology is described by $\Sp(N)$ invariant words in the elements of the (dual of) the Lie algebra $\mscr{L}_{N,K}$.  There are closed string words, forming dihedral cohomology \cite{LODAY198893}; and open string words, starting and ending with an element in the fundamental representation. Open string words can be computed by Tor groups \cite{Zeng:2023lox}. 

Since we only care about the part of the Lie algebra cohomology that is only present on the defect, we are only interested in words which contain at least one $\psi$ field (where $\psi$ indicates the bifundamental fermion on the defect).  Since $\psi$ lives in the fundamental representation of $\Sp(N)$, any expression containing a $\psi$ is a product of open string words. There are three possibilities: open string words with a $\psi$ at each end, open string words with a $\psi$ at one end and a $\gamma$ at another end, and open string words with a $\psi$ at one end and a $\til{\gamma}$ at the other end.  

An elementary computation (using the ideas around the Loday-Quillen-Tsygan theorem) gives us the following:
\begin{proposition}
The single-string part of the Lie algebra cohomology of $\mscr{L}_{N,K}$, at large $N$, consisting of strings with a $\psi$ at one end and a $\gamma$ (or $\til{\gamma}$) at the other, is the dual of
\be \C^{2K} \otimes \C^{16} \otimes	\op{Tor}_{\C\llbracket z,w_1,w_2,\eps_\done,\eps_\dtwo\rrbracket }(\C\llbracket z,w_1,w_2\rrbracket ,\C\llbracket z\rrbracket )\,. \ee

Similarly, the single-string large $N$ Lie algebra cohomology consisting of strings with a $\psi$ at each end is the $\Z/2$ invariants in
\be \C^{2K} \otimes \C^{2K} \otimes \op{Tor}_{\C\llbracket z,w_1,w_2,\eps_\done,\eps_\dtwo\rrbracket} (\C\llbracket z\rrbracket ,\C\llbracket z\rrbracket) 
\ee
where the $\Z/2$ action is specified in the proof.
\end{proposition}
\begin{proof} 
	Let us first treat the case of the $\gamma-\psi$ strings.  Any such open string state is given by an expression which starts with $\gamma$ (or derivatives of $\gamma$), applies a number of matrix-valued fields $\b,\c,\phi^\da$ (or derivatives), and ends with $\psi$ (or derivatives).  The $\Sp(N)$ indices are contracted by matrix multiplication, and by contracting with a fundamental at the beginning and end.  At each step in the string, the data of which field we have ($\b,\c,\phi^\da$), together with all of the derivatives are expressed as an element of the dual of the algebra $A = \C\llbracket z,w_i,\eps_\da\rrbracket$. This is because, as we remarked above, the matrix-valued fields are elements of the $\Z/2$ invariants of $A \otimes \sl(2N)$.  The fields $\phi^\da$ are the coefficients of $\eps_\da$, $\c$ is the coefficient of $1$, and $\b$ is the coefficient of $\eps_\done \eps_\dtwo$.  Derivatives are counted by the dual of $\C\llbracket z,w_i\rrbracket$.  The possible derivatives of the initial $\gamma$ are also counted by the dual of $\C\llbracket z,w_i\rrbracket$ and of the final $\psi$ by the dual of $\C\llbracket z\rrbracket$. 

	Because we are dealing with an $\Sp(N)$ gauge theory, an open string state is the same as a state where the words are read in the reverse order.  This symmetry is broken by declaring that we start with $\gamma$ and end with $\psi$. 

	In this way, we see that the open string states in the Lie algebra cochain complex (which is the BRST complex) map identically onto the dual of the bar complex:
	\be \C^{16} \otimes \C^{2K} \otimes	\op{Bar}(\C\llbracket z,w_i\rrbracket,A,\C\llbracket z\rrbracket) = 	\C^{16} \otimes \C^{2K} \otimes	\oplus_{n\ge0} \C\llbracket z,w_i\rrbracket \otimes A^{\otimes n} \otimes \C\llbracket z\rrbracket\,.
	\ee
	The fact of $\C^{16} \otimes \C^{2K}$ comes from the flavour indices on the initial and final states. It is standard (and easy to see) that the differential on the Lie algebra cohomology, which is the same as the BRST differential, is the dual to the differential on the bar complex, which multiplies adjacent elements with a sign.   

	Since the homology of the bar complex is $\op{Tor}_A(\C\llbracket z,w_i\rrbracket,\C\llbracket z\rrbracket)$, this completes the proof.

	The statement for strings starting and ending with $\psi$ is more or less the same, except that we have to take into account the fact that a string read in reverse order gives the same state (up to a sign).  Before applying this $\Z/2$ action, we find the dual of  
	\be \C^{2K} \otimes \C^{2K}  \otimes \op{Tor}_{A} (\C\llbracket z\rrbracket,\C\llbracket z\rrbracket)\,. \ee
	The $\Z/2$ action permutes the two copies of $\C^{2K}$, reverse the order of a string in the bar complex, and acts on $A$ by the $\Z/2$ action described above with an extra factor of $-1$. 
\end{proof}
Now let us compute these Tor groups.  This is straightforward. The rules are the following. The Tor groups will be expressed as series in some number of bosonic and fermionic variables.  Any variable in $\C\llbracket z,w_i,\eps_\da\rrbracket$ which is present in both the left and right module is present in the Tor group. Any variable which is present in the left but not the right (or vice versa) drops out. Any variable that is absent on  both the left and the right changes parity.  Thus,
\be \op{Tor}_{\C\llbracket z,w_i,\eps_\da\rrbracket}(\C\llbracket z,w_i\rrbracket,\C\llbracket z\rrbracket) = \C\llbracket z,\lt_\da\rrbracket \ee
where $\lt_\da$ are bosonic.  Similarly,
\be \op{Tor}_{\C\llbracket z,w_i, \eps_\da\rrbracket}(\C\llbracket z\rrbracket,\C\llbracket z\rrbracket) = \C\llbracket z,\lt_\da,\delta^i\rrbracket \ee
where $\delta_i$ are fermionic.  

To complete the calculation, we need to describe the $\Z/2$ action on this Tor group. It turns out that it sends $\delta^i$ to $-\delta^i$ and fixes $\lt_\da$.  

Let us reintroduce the vector spaces arising from the flavour symmetry of the $\gamma$ and $\psi$ fields. We find that the single-string Lie algebra cohomology of $\mscr{L}^{N,K}$, containing at least one $\psi$ field, is the dual of 
\be \C^{16} \otimes \C^{2K} \otimes \C\llbracket z,\lt_\da\rrbracket \oplus \big(\C^{2K} \otimes \C^{2K} \otimes \C\llbracket z,\lt_\da,\delta^i\rrbracket\big)^{\Z/2}\,. \ee
This computation gives us a basis of the single-string part of BRST cohomology involving $\psi$ as presented in Table \ref{tab:celestial-ops} in the main text. 


\section{Direct calculation of the \texorpdfstring{$\phi-\phi$}{phi-phi} OPE} \label{app:diagram}

In this appendix we explicit compute the singular part of the OPE between $\phi^\done(0)$ and $\phi^\dtwo(z)$ in the worldvolume theory of the $N$ $D5$ branes. Lemma \ref{lem:Feyn_rule} tells us that the only diagram we need to evaluate is the one illustrated below.

\begin{figure}[ht]
    \centering
    \includegraphics{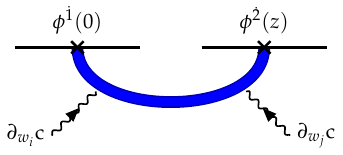}
    \caption{A Feynman diagram contributing to the $\phi-\phi$ OPE.  All other diagrams differ only in their colour structure.} \label{fig:phi-OPE}
\end{figure}

For completeness, let's compute this with $\phi^\dtwo$ at some arbitrary bulk point $w$. At the end of the computation we will take $w\to(0,z)$. In this appendix we write $z = w_3$, and let $i,j,k,\dots$ take values in the range $\{1,2,3\}$.

Stripping off the colour factor and working in a patch isomorphic to $\C^3$, the $\Phi-\Phi$ propagator coincides with that of the free holomorphic theory. A local computation of this type is not capable of determining the regular part of the OPE, but is sufficient for the singular term. Fixing the gauge by requiring the fields to lie in the kernel of the $\dbar^\ast$ operator defined using the flat metric, the propagator is proportional to the Bochner-Martinelli kernel
\be P(w) = \frac{1}{4\pi^2}\frac{\eps^{ijk}\bar w_i\d\bar w_j\d\bar w_k}{\|w\|^6}\,. \ee
Here we've normalized the action with an overall factor of $\i/2\pi$ so that the propagator obeys 
\be \dbar P = \frac{1}{4\pi^2}\bar\delta^{0,3} \ee
with
\be - \bigg(\frac{\i}{2\pi}\bigg)^3\int_{\C^3}\d^3w\,\bar\delta^{0,3}(w) = 1\,. \ee
The diagram illustrated in Fig. \ref{fig:phi-OPE} evaluates to
\be \bigg(\frac{\i}{2\pi}\bigg)^2\int_{w',w''}\d^3w'\d^3 w''\,P(w')\c(w')P(w'-w'')\c(w'')P(w''-w)\,. \ee
Concentrating on the antiholomorphic forms appearing in the numerators, we find
\bea
&\eps^{i_1j_1k_1}\bar w_{i_1}'\d\bar w_{j_1}'\d\bar w_{k_1}'\eps^{i_2j_2k_2}(\bar w' - \bar w'')_{i_2}\d(\bar w'-\bar w'')_{j_2}\d(\bar w'-\bar w'')_{k_2} \\
&\eps^{i_3j_3k_3}(\bar w'' - \bar w)_{i_3}\d\bar w_{j_3}''\d\bar w_{k_3}'' = 8\eps^{ijk}\bar w_i\bar w_j'\bar w_k''\d^3\bar w'\d^3\bar w''\,.
\eea
Therefore the integrand simplifies to
\be 8\bigg(\frac{\i}{2\pi}\bigg)^8\eps^{ijk}\bar w_i\int_{w',w''}\frac{\d^6w'\d^6w''\,\bar w_j'\bar w_k''\,\c(w')\c(w'')}{\|w'\|^6\|w'-w''\|^6\|w''-w\|^6}\,. \ee
We proceed by performing the integral over $w'$ using the standard Feynman parametrisation trick:
\be \int_{w'}\frac{\d^6w'\,w_j'\,\c(w')}{\|w'\|^6\|w'-w''\|^6} = \frac{5!}{2!^2}\int_0^1\d s\,s^2(1-s)^2\int_{w'}\frac{\d^6w'\,\bar w_j'\,\c(w')}{(\|w'-sw''\|^2+s(1-s)\|w''\|^2)^6}\,. \ee
Shifting the integration variable this becomes
\be \frac{5!}{2!^2}\int_0^1\d s\,s^2(1-s)^2\int_{w'}\frac{\d^6w'\, (\bar w_j'+s\bar w_j'')\,\c(w'+sw'')}{(\|w'\|^2+s(1-s)\|w''\|^2)^6}\,. \ee
Anti-holomorphic derivatives of the ghost vanish on-shell, so we need only expand the ghost field holomorphically in its arguments. The precise number of holomorphic derivatives that can appear are fixed by the requirement that the integrals over the phases of the $w_i'$ variables are non-vanishing. This gives
\be \frac{5!}{2!^2}\int_0^1\d s\,s^2(1-s)^2\int_{w'}\frac{\d^6w'\, (\bar w_j'w_l'\p_{w_l}\c(sw'')+s\bar w_j''\c(sw''))}{(\|w'\|^2+s(1-s)\|w''\|^2)^6}\,. \ee
We can discard the second term, as it will vanish upon contracting with $\eps^{ijk}\bar w_k''$. At this stage the dependence on $w'$ is completely fixed, and we can perform the integrals to get
\be \frac{(-2\pi\i)^3}{2!^2}\int_0^1\d s\,\frac{\delta_{jl}\p_{w_l}\c(sw'')}{\|w''\|^4}\,. \ee
Next we repeat these tricks in the $w''$ variables.
\bea
&2\bigg(\frac{\i}{2\pi}\bigg)^5\eps^{ijk}\bar w_i\int_0^1\d s\int_{w''}\frac{\d^6w''\,\bar w''_k\delta_{jl}\p_{w_l}\c(sw'')\c(w'')}{\|w''\|^4\|w''-w\|^6} \\
&= 4!\bigg(\frac{\i}{2\pi}\bigg)^5\eps^{ijk}\bar w_i\int_{[0,1]^2}\d s\d t\,t^2(1-t)\int_{w''}\frac{\d^6w''\,\bar w_k''\delta_{jl}\p_{w_l}\c(sw'')\c(w'')}{(\|w''-tw\|^2 + t(1-t)\|w\|^2)^5}\,.
\eea
Shifting the argument in $w''$, and noting that to get a non-vanishing result for the integral over the phases we must get an extra factor of $w_k''$ from the second ghost (the anti-symmetry of $\eps^{ijk}$ prevents the first ghost contributing) this is
\bea
&4!\bigg(\frac{\i}{2\pi}\bigg)^5\eps^{ijk}\bar w_i\int_{[0,1]^2}\d s\d t\,t^2(1-t) \\
&\int_{w''}\frac{\d^6w''\,\bar w_k''w_m''\delta_{jl}\p_{w_l}\c(sw''+st w)\p_{w_m}\c(w''+tw)}{(\|w''\|^2 + t(1-t)\|w\|^2)^5} \\
&= \bigg(\frac{\i}{2\pi}\bigg)^2\frac{\eps^{ijk}\bar w_i}{\|w\|^2}\int_{[0,1]^2}\d s\d t\,t\delta_{jl}\delta_{km}\p_{w_l}\c(stw)\p_{w_m}\c(tw)\,.
\eea
This can be written in a more symmetric form by rescaling $s$ by $t$ so that it now takes values in the range $[0,t]$.
\be \bigg(\frac{\i}{2\pi}\bigg)^2\frac{\bar w_i}{\|w\|^2}\int_{0\leq s\leq t\leq1}\d s\d t\,\eps^i_{~jk}\p_{w_j}\c(sw)\p_{w_k}\c(tw)\,. \ee
In the limit $w\to(0,z)$ this simplifies to
\be \bigg(\frac{\i}{2\pi}\bigg)^2\frac{1}{z}\int_{0\leq s\leq t\leq1}\d s\d t\,\eps_{jk}\p_{w_j}\c(0,sz)\p_{w_k}\c(0,tz)\,. \ee
As explained above, only the singular term matches that on the non-trivial complex manifold $Y$. It reads
\be - \frac{1}{4\pi^2}\frac{\eps_{ij}}{2z}\p_{w_i}\c\p_{w_j}\c(0)\,. \ee
Reintroducing colour indices
\be \label{eq:phiOPE-singular-again}
\phi^{\done pq}(0)\phi^{\dtwo rs}(z,0)\sim -\frac{1}{4\pi^2}\frac{\eps_{ij}}{z}\big(\omega^{[r[p}\p_{w_i}\c^{q]}_{~\,t}\p_{w_j}\c^{s]t} + \p_{w_i}\c^{[r[p}\p_{w_j}\c^{q]s]}\big)(0)\,.
\ee

Since $Q_\mrm{BRST}\p_{w_i}\p_{w_j}\c = [\c,\p_{w_i}\p_{w_j}\c] + [\p_{w_i}\c,\p_{w_j}\c]$ we can replace $\half\eps_{ij}\p_{w_i}\c\p_{w_j}\c = \p_{w_1}\c\p_{w_2}\c$ in simpler gauge invariant operators. Making this replacement recovers equation \eqref{eq:phiOPE-singular} from the bulk of the manuscript.

This same computation gives the $\gamma,\til\gamma$ and $\c,\b$ OPEs, which differ only in their colour structure. We have
\be
\gamma^p_f(0)\til\gamma^{qg}(z,0)\sim -\frac{1}{4\pi^2}\frac{\eps_{ij}}{2z}\delta_f^{~\,g}\p_{w_i}\c^p_{~t}\p_{w_j}\c^{qt}(0)
\ee
and
\be
\c^{pq}(0)\b^{rs}(z,0) \sim -\frac{1}{4\pi^2}\frac{\eps_{ij}}{z}\big(\omega^{(r(p}\p_{w_i}\c^{q)}_{~\,t}\p_{w_j}\c^{s)t} - \p_{w_i}\c^{(r(p}\p_{w_j}\c^{q)s)}\big)(0)\,.
\ee
The relative sign in the second term of the above OPE when compared to equation \eqref{eq:phiOPE-singular-again} is ultimately responsible for the relative sign in the coefficients of the $\J\Jt$ and $\M\Mt$ terms in the $2\to2$ OPE.


\section{Simplifying Expected OPEs} \label{app:simplifyOPE}

The first non-trivial one-loop corrected celestial OPEs for integrable self-dual QCD in the absence of a background axion have been given in \cite{Costello:2023vyy}. For a general Lie algebra $\g$ with basis $\{\t_\sfa\}$ and fermions in the representation $R$ with basis $\{e_i\}$ we write
\be [\t_\sfa,\t_\sfb] = f_{\sfa\sfb}^{~\,~\sfc}\t_\sfc\,\qquad \t_\sfa.e_i = e_jg^j_{~\,i\sfa}\,. \ee
Then the $\J_\sfa[1,0],\J_\sfb[0,1]$ OPE up to one-loop is
\bea \label{eq:J[1,0]J[0,1]}
&\J_\sfa[1,0](z_1)\J_\sfb[0,1](z_2) \sim \frac{1}{z_{12}}f_{\sfa\sfb}^{~\,~\sfc}\J_\sfc[1,1](z_1) \\
&+ \frac{1}{4\pi^2}\bigg(\frac{1}{24z_{12}^2}\op{tr}_{\g\oplus\Pi R}([\t_\sfa,\t_\sfb]\t_\sfc)\Jt^\sfc(z_1) - \frac{1}{48z_{12}}\op{tr}_{\g\oplus\Pi R}([\t_\sfa,\t_\sfb]\t_\sfc)\p_z\Jt^\sfc(z_1) \\
&+ \frac{1}{8z_{12}}\big((f_{\sfe\sfa}^{~\,~\sfc}f_{\sfd\sfb}^{~\,~\sfe} + f_{\sfd\sfa}^{~\,~\sfe}f_{\sfe\sfb}^{~\,~\sfc})\normord{\J_\sfc\Jt^\sfd}(z_1) + (g^i_{~\,k\sfa}g^k_{~\,j\sfb} + g^k_{~j\sfa}g^i_{~\,k\sfb})\normord{\M_i\Mt^j}(z_1)\big)\bigg)\,.
\eea
(The relative factors of $2\pi\i$ when compared to \cite{Costello:2022upu,Costello:2023vyy} can be attributed to the differing normalization of our action.) This is the OPE we'd like to reproduce with a dual computation. In our case $\g = \fsp(K)$, so we identify $\sfa = (IJ)$, and have a basis
\be (\t_{(IJ)})^K_{\,~\,L} = \delta^K_{\,~\,I}\Omega_{JL} + \delta^K_{\,~\,J}\Omega_{IL} = 2\delta^K_{~\,(I}\Omega_{J)L}\,. \ee
The structure constants are then determined by
\be [\t_{IJ},\t_{KL}] = \Omega_{IK}\t_{(JL)} + \Omega_{IL}\t_{(JK)} + \Omega_{JK}\t_{(IL)} + \Omega_{JL}\t_{(IK)} = - 4\Omega_{(K(I}\t_{J)L)}\,, \ee
so that
\be f^{(MN)}_{(IJ)(KL)} = - 4\Omega_{(K(I}\delta^{(M}_{~J)}\delta^{N)}_{~L)}\,. \ee
Next consider the fundamental representation, denoted $\mrm{F}_K$.  Writing $e_I$ for our basis we have
\be \t_{(IJ)}.e_L = e_K(\t_{(IJ)})^K_{\,~\,L}\,, \ee
so that
\be g^K_{\,~\,L(IJ)} = 2\delta^K_{~\,(I}\Omega_{J)L} = (\t_{(IJ)})^K_{\,~\,L}\,. \ee
Finally, for the traceless anti-symmetric square of the fundamental representation, denoted $\wedge^2_0\mrm{F}_K$. Writing $a_{[IJ]}$ for our basis
\be \t_{(IJ)}.a_{[KL]} = - 4\Omega_{[K(I}a_{J)L]}\,. \ee
so that
\be g^{[MN]}_{[KL](IJ)} = 4\Omega_{[K(I}\delta^{[M}_{~J)}\delta^{N]}_{~L]}\,. \ee
Now
\bea
&\op{tr}_{\mrm{F}_K}(\t_{(IJ)}\t_{(KL)}) = 4\Omega_{(I(L}\Omega_{K)J)}\,, \\
&\op{tr}_{\wedge^2_0\mrm{F}_K}(\t_{(IJ)}\t_{(KL)}) = 8(K-1)\Omega_{I(L}\Omega_{K)J} = 2(K-1)\op{tr}_\mrm{F}(\t_{(IJ)}\t_{(KL)})\,, \\
&\op{tr}_\mathrm{Ad}(\t_{(IJ)}\t_{(KL)}) = 8(K+1)\Omega_{I(L}\Omega_{K)J} = 2(K+1)\op{tr}_\mrm{F}(\t_{(IJ)}\t_{(KL)})\,.
\eea
From this we infer that
\be \op{tr}_{\Ad\oplus\Pi R}(\t_{(IJ)}\t_{(KL)}) = (2(K+1)-2(K-1) - 16)\op{tr}_{\mrm{F}_K}(\t_{(IJ)}\t_{(KL)}) = - 12\op{tr}_{\mrm{F}_K}(\t_{(IJ)}\t_{(KL)})\,. \ee
We can also readily compute
\bea
&\frac{1}{2}\big(f^{(MN)}_{(RS)(IJ)}f^{(RS)}_{(PQ)(KL)} + (IJ)\leftrightarrow(KL)\big)\J_{(MN)}\Jt^{(PQ)} \\
&= 4\big(\J_{(K(J}\Jt_{I)L)} + \Omega_{(K(J}\J_{I)}^{~\,~O}\Jt_{L)O} + (IJ)\leftrightarrow(KL)\big)\,,
\eea
where we note that there's no normal ordering ambiguity. Similarly
\bea
&\frac{1}{2}\big(g^{[MN]}_{[RS](IJ)}g^{[RS]}_{[PQ](KL)} + (IJ)\leftrightarrow(KL)\big)\M_{[MN]}\Mt^{[PQ]} \\
&= - 4\big(\M_{(J(K}\Mt_{L)I)} - \Omega_{(K(I}\M_{J)}^{~\,~O}\Mt_{L)O} + (IJ)\leftrightarrow(KL)\big)\,.
\eea
For $\g = \fsp(K)$ there is no normal ordering ambiguity in this expression either.

Specializing to the case of $\g = \fsp(K)$, $R = \wedge^2_0\mrm{F}_K\oplus16\mrm{F}_K$ SDQCD the OPE \eqref{eq:J[1,0]J[0,1]} reads
\bea \label{eq:target-OPE}
&\J_{(IJ)}[1,0](z_1)\J_{(KL)}[0,1](z_2) \sim \frac{4}{z_{12}}\Omega_{(J(K}\J_{L)I)}[1,1](z_1) - \frac{1}{\pi^2}\bigg(\frac{2}{z_{12}^2}\Omega_{(J(K}\Jt_{L)I)}(z_1) \\
&- \frac{1}{z_{12}}\Omega_{(J(K}\p_z\Jt_{L)I)}(z_1) + \frac{1}{2z_{12}}\big(\normord{\J_{(J(K}\Jt_{L)I)}}(z_1) - \normord{\M_{(J(K}\Mt_{L)I)}}\big)(z_1) \\
&+ \frac{1}{4z_{12}}\Omega_{(K(J}\big(\normord{\J_{I)}^{~\,~O}\Jt_{L)O}}(z_1) + \normord{\M_{I)}^{~\,~O}\Mt_{L)O}}(z_1) + \normord{\M_{I),f}\Mt_{L)}^f}(z_1)\big) \\
&+ \frac{1}{4z_{12}}\Omega_{(J(K}\big(\normord{\J_{L)}^{~\,~O}\Jt_{I)O}}(z_1) + \normord{\M_{L)}^{~\,~O}\Mt_{I)O}}(z_1) + \normord{\M_{L),f}\Mt_{I)}^f}(z_1)\big)\bigg)\,.
\eea
The remaining OPEs are easy to write down. One particularly important tree OPE for fixing the normalization of the dual chiral algebra states is
\be \M_i(\lt_1;z_1)\Mt^j(\lt_2;z_2) \sim \frac{1}{z_{12}}g^j_{~\,i\sfa}\Jt^\sfa(\lt_1+\lt_2;z_1)\,. \ee
Specializing to the SDQCD under consideration and concentrating on the states charged under the $\mrm{SL}(16)$ flavour symmetry we find that
\be \M_{I,f}(\lt_1;z_1)\til\M_J^g(\lt_2;z_2) \sim - \frac{1}{z_{12}}\delta_f^{~\,g}\Jt_{IJ}(\lt_1+\lt_2;z_1)\,. \ee


\section{Bulk States} \label{app:bulk-states}

In this section, we verify tree-level BRST invariance of the bulk states appearing in the Table \ref{tab:bulk-states}. All but two are evidently closed using the standard transformations of $\phi,\gamma,\til\gamma$ together with $Q_\mrm{BRST}\p_{w_i}\c = [\c,\p_{w_i}\c]$.

The first non-trivial state is
\be \label{eq:Tm-partial} T^{i\da_1\dots\da_m} = \frac{1}{2}\Tr S\big(\phi^{\da_1}\dots\phi^{\da_m}\phi_\db\p_{w_i}\phi^\db\big) + \dots\,. \ee
We begin by showing that, in the absence of a defect, there exist terms involving either the antifield $\b$ or fundamental bosons $\gamma,\til\gamma$ we can add in place of $+ \dots$ which render the above state BRST invariant. With an $\mrm{SL}_2(\C)$ rotation on $\da_k$ indices we can bring equation \eqref{eq:Tm-partial} into the form
\be T^i[m,0] = \frac{1}{2}\Tr S\big(\phi^{\done m}\phi_\db\p_{w_i}\phi^\db\big)  + \dots = \frac{1}{2(m+1)}\sum_{a+b=m}\Tr\big(\phi^{\done a}\phi_\db\phi^{\done b}\p_{w_i}\phi^\db\big) + \dots \ee
Acting with the bulk BRST operator
\bea
&Q_\mrm{BRST}\Tr\big(\phi^{\done a}\phi_\db\phi^{\done b}\p_{w_i}\phi^\db\big) = \Tr\big(\phi^{\done a}\phi^\dtwo\phi^{\done b}[\p_{w_i}\c,\phi^\done]\big) - \Tr\big(\phi^{\done(m+1)}[\p_{w_i}\c,\phi^\dtwo]\big) \\
&= \Tr\big(\phi^{\done b}[\p_{w_i}\c,\phi^\done]\phi^{\done a}\phi^\dtwo\big) + \sum_{c+d=m}\Tr\big(\phi^{\done c}[\p_{w_i}\c,\phi^\done]\phi^{\done d}\phi^\dtwo\big)\,.
\eea
This can be partially compensated with
\bea
&Q_\mrm{BRST}\Tr\big(\phi^{\done a}\b\phi^{\done b}\p_{w_i}\c\big) = \Tr\big(\phi^{\done a}[\phi^\done,\phi^\dtwo]\phi^{\done b}\p_{w_i}\c\big) + \frac{1}{2}\Tr\big(\phi^{\done a}\{\gamma_f,\til\gamma^f\}\phi^{\done b}\p_{w_i}\c\big) \\
&= \Tr\big(\phi^{\done b}[\p_{w_i}\c,\phi^\done]\phi^{\done a}\phi^\dtwo\big) + \frac{1}{2}\til\gamma^f\phi^{\done b}\p_{w_i}\c\phi^{\done a}\gamma_f + \frac{1}{2}\til\gamma^f\phi^{\done a}\p_{w_i}\c\phi^{\done b}\gamma_f\,.
\eea
The terms involving $\gamma,\til\gamma$ are BRST exact
\be Q_\mrm{BRST}\til\gamma^f\phi^{\done a}\p_{w_i}(\phi^{\done b}\gamma_f) = \til\gamma^f\phi^{\done a}\p_{w_i}\c\phi^{\done b}\gamma_f = - Q_\mrm{BRST}\p_{w_i}(\til\gamma^f\phi^{\done a})\phi^{\done b}\gamma_f\,. \ee
This is consistent since $\til\gamma^f\phi^{\done m}\gamma_f$ is BRST exact. We conclude that
\bea \label{eq:Tm}
&T^i[m,0] = \frac{1}{2}\Tr S\big(\phi^{\done m}\phi_\db\p_{w_i}\phi^\db\big) - \frac{m+2}{2}\Tr S\big(\phi^{\done m}\b\p_{w_i}\c\big) + \frac{m+2}{2}\til\gamma^fS\big(\phi^{\done m}\p_{w_i}\big)\gamma_f
\eea
is a BRST invariant bulk state. For $m=0$ it coincides with the stress tensor
\be T^i[0,0] = \frac{1}{2}\Tr(\phi_\db\p_{w_i}\phi^\db) - \Tr(\b\p_{w_i}\c) + \til\gamma^f\p_{w_i}\gamma_f\,. \ee
(The slightly unconventional sign in front of the second term if generated by anticommuting the components of the Beltrami differential $\mu = \mu_i\p_{w_i}$ past $\b$.) In \eqref{eq:Tm} we choose the derivative $\p_{w_i}$ to act on $\gamma_f$ rather than $\til\gamma^f$, consistent with the expected form of the stress tensor.

Coupling to the defect adds a new term to the BRST differential \eqref{eq:BRST-full}
\be \label{eq:Q-defect} Q_\mrm{BRST}^\text{defect}\b = \frac{1}{16\pi^2}\eps_{jk}\big(2\p_{w_j}\c\psi_I\psi^I\p_{w_k}\c + \psi_I\psi^I\p_{w_j}\c\p_{w_k}\c + \p_{w_j}\c\p_{w_k}\c\psi_I\psi^I\big)\,. \ee
Hence
\bea
&Q_\mrm{BRST}^\text{defect}\Tr(\phi^{\done a}\b\phi^{\done b}\p_{w_i}\c) \\
&= - \frac{1}{16\pi^2}\eps_{jk}\big(2\psi^I\p_{w_j}\c\phi^{\done b}\p_{w_i}\c\phi^{\done a}\p_{w_k}\c\psi_I + \psi^I\p_{w_j}\c\p_{w_k}\c\phi^{\done b}\p_{w_i}\c\phi^{\done a}\psi_I \\
&+ \psi^I\phi^{\done b}\p_{w_i}\c\phi^{\done a}\p_{w_j}\c\p_{w_k}\c\psi_I\big) = - \frac{1}{16\pi^2}\eps_{jk}\big(\psi^I\p_{w_j}\c[\p_{w_k}\c,\phi^{\done b}\p_{w_i}\c\phi^{\done a}]\psi_I + a\leftrightarrow b\big)\,.
\eea
In the usual way we can express a commutator with $\p_{w_k}\c$ as the BRST variation of a derivative
\be Q_\mrm{BRST}^\text{defect}\Tr(\phi^{\done a}\b\phi^{\done b}\p_{w_i}\c) = - \frac{1}{16\pi^2}\eps_{jk}Q_\mrm{BRST}\big(\psi^I\p_{w_j}\c\p_{w_k}(\phi^{\done b}\p_{w_i}\c\phi^{\done a})\psi_I + a\leftrightarrow b\big)\,, \ee
We therefore find that in full
\bea
&T^i[m,0] = \frac{1}{2}\Tr S\big(\phi^{\done m}\phi_\db\p_{w_i}\phi^\db\big) - \frac{m+2}{2}\Tr S\big(\phi^{\done m}\b\p_{w_i}\c\big) + \frac{m+2}{2}\til\gamma^fS(\phi^{\done m}\p_{w_i})\gamma_f \\
&+ \frac{1}{16\pi^2}(m+2)\eps_{jk}\psi_I\p_{w_j}\c\p_{w_k}\big(S\big(\phi^{\done m}\p_{w_i}\c\big)\big)\psi^I\,.
\eea

The second non-trivial state has the form
\be \cO^{\da_1\dots\da_m} = \frac{1}{2^2}\eps_{ijk}\Tr S\big(\phi^{\da_1}\dots\phi^{\da_m}\phi_\db\p_{w_i}\phi^\db\p_{w_j}\c\p_{w_k}\c\big) + \dots\,. \ee
Rotating so that all $\da_k=\done$, unpacking the symmetrization gives
\bea
&\frac{1}{2^2}\eps_{ijk}\Tr S\big(\phi^{\done m}\phi_\db\p_{w_i}\phi^\db\p_{w_j}\c\p_{w_k}\c\big) \\
&= \frac{1}{2[m+3]_3}\eps_{ijk}\Big(\sum_{a+b+c+d=m}\Tr\big(\phi^\dtwo\phi^{\done a}\p_{w_i}\phi^\done\phi^{\done b}\p_{w_j}\c\phi^{\done c}\p_{w_k}\c\phi^{\done d} + \text{cyclic permutations}\big) \\
&- (m+3)\sum_{a+b+c=m}\Tr(\p_{w_i}\phi^\dtwo\phi^{\done a}\p_{w_j}\c\phi^{\done b}\p_{w_k}\c\phi^{\done c})\Big)\,.
\eea
Here we mean appropriately graded cyclic permutations of $\{\p_{w_i}\phi^\done,\p_{w_j}\c,\p_{w_k}\c\}$. The BRST variation of the above is
\bea \label{eq:QH}
&\frac{m+4}{2[m+3]_3}\eps_{ijk}\sum_{a+b+c+d=m}\Tr\big(\phi^\dtwo\phi^{\done a}[\p_{w_i}\c,\phi^\done]\phi^{\done b}\p_{w_j}\c\phi^{\done c}\p_{w_k}\c\phi^{\done d} + \text{cyclic permutations}\big) \\
&= \frac{m+4}{2^2}\Tr S(\phi^{\done m}\phi^\dtwo[\p_{w_i}\c,\phi^\done]\p_{w_j}\c\p_{w_k}\c)\,.
\eea
Here we mean graded cyclic permutations of $\{[\phi^\done,\p_{w_i}\c],\p_{w_j}\c,\p_{w_k}\c\}$. In getting the above we've used $\eps_{ijk}[\p_{w_i}\c,\p_{w_j}\c] = 0$, since $[~,\,]$ is the graded commutator. We can partially compensate \eqref{eq:QH} with the BRST variation of
\be \eps_{ijk}\Tr S\big(\phi^{\done m}\b\p_{w_i}\c\p_{w_j}\c\p_{w_k}c\big) = \frac{3!}{[m+3]_3}\eps_{ijk}\sum_{a+b+c+d=m}\Tr\big(\b\phi^{\done a}\p_{w_i}\c\phi^{\done b}\p_{w_j}\c\phi^{\done c}\p_{w_k}\c\phi^{\done d}\big)\,. \ee
The contribution from $Q_\mrm{BRST}\b = [\phi^\done,\phi^\dtwo] + \dots$ is
\bea
&\frac{3!}{[m+3]_3}\eps_{ijk}\sum_{a+b+c+d=m}\Tr\big(\phi^\dtwo\phi^{\done a}[\p_{w_i}\c,\phi^\done]\phi^{\done b}\p_{w_j}\c\phi^{\done c}\p_{w_k}\c\phi^{\done d} + \text{cyclic permutations}\big) \\
&= 3\eps_{ijk}\Tr S\big(\phi^{\done m}\phi^\dtwo[\p_{w_i}\c,\phi^\done]\p_{w_j}\c\p_{w_k}\c\big)\,,
\eea
but there's also a term generated by $Q_\mrm{BRST}\b = \dots + \{\gamma_f,\til\gamma^f\}/2$
\bea
&\frac{3!}{[m+3]_3}\eps_{ijk}\sum_{a+b+c+d=m}\til\gamma_f\phi^{\done a}\p_{w_i}\c\phi^{\done b}\p_{w_j}\c\phi^{\done c}\p_{w_k}\c\phi^{\done d}\gamma^f = \eps_{ijk}\til\gamma_fS(\phi^{\done m}\p_{w_i}\c\p_{w_j}\c\p_{w_k}\c)\gamma^f \\
&= - \frac{1}{[m+3]_3}\eps_{ijk}Q_\mrm{BRST}\sum_{a+b+c+d=m}\til\gamma_f\phi^{\done a}\p_{w_i}\c\phi^{\done b}\p_{w_j}\c\phi^{\done c}\p_{w_k}(\phi^{\done d}\gamma^f)\,.
\eea
In the presence of the defect, our BRST differential acquires a new contribution \eqref{eq:Q-defect}. This will generate a word capped on both ends with $\psi$ and with five insertions of $\p_{w_l}\c$ in the middle. These are certainly exact.


\section{Combinatorics for Symmetrizing Free Products} \label{app:flavour-combinatorics}

At many stages in the bulk of this work we are required to decompose a product of symmetric polynomials in non-commuting variables into a sum of terms consisting of symmetric polynomials sandwiched between commutators.  More precisely, we consider the group ring of the free group on two generators, denoted $\C\la X^\done,X^\dtwo\ra$. We can define symmetric polynomials depending on a choice of reference spinor $\lt_\da$
\be X^m(\lt) = (X^\da\lt_\da)^m\,. \ee
Let $(X^\done)^{(r}(X^\dtwo)^{s)}$ denote the symmetrized product of $(X^\done)^r$ with $(X^\dtwo)^s$.  These are the coefficients in the expansion of $X^{r+s}(\lt)$ in its two components. 

Then, a basis for the non-commutative algebra is provided by expressions like 
\be
(X^\done)^{(r_1} (X^\dtwo)^{s_1)} [X^\done,X^\dtwo] (X^\done)^{(r_2} (X^\dtwo)^{s_2)} [X^\done,X^\dtwo] \dots [X^\done,X^\dtwo] (X^\done)^{(r_k} (X^\dtwo)^{s_k)}
\label{eq:noncommutativealgebra-again}
\ee
for $r_i, s_i,k \ge 0$.  To see that this is a basis, note that the algebra $\C\la X^\done,X^\dtwo\ra$ of non-commutative polynomials has a two-sided ideal $I$ generated by the commutator $[X^\done,X^\dtwo]$.  We can take the $I$-adic filtration, where $F^k\C\la X^\done,X^\dtwo\ra$ is the subspace spanned by products of $k$ elements of $I$.  Expressions like those in \eqref{eq:noncommutativealgebra-again} form a basis for the associated graded $F^k / F^{k+1}$.  It is obvious that these expressions span $F^k / F^{k+1}$, This is because for any expression involving $k$ commutators, words between any two commutators can be symmetrized at the price of introducing further commutators; but we are setting expressions with $k+1$ commutators to zero.   

It is not quite so obvious, but true, that these expressions form a basis.  To check this, we need only check that there are $2^n$ such elements which are sums of words of length $n$ in $X^\done$ and $X^\dtwo$.  Any word in $X^\done$ and $X^\dtwo$ can be written uniquely in the form
\be
(X^\done)^{r_1} (X^\dtwo)^{s_1} (X^\dtwo X^\done) (X^\done)^{r_2} (X^\dtwo)^{s_2} (X^\dtwo X^\done) \dots (X^\dtwo X^\done) (X^\done)^{r_k} (X^\dtwo)^{r_k}
\ee
where\footnote{To see this, note that any word has $k$ occurrences of $X^\dtwo X^\done$ for $k \ge 0$. Between any two occurrences of $X^\dtwo X^\done$, and before and after the initial and final occurrence, we can only have $(X^\done)^r (X^\dtwo)^s$ for $r,s \ge 0$.}  $r_i, s_i, k \ge 0$, and $\sum r_i + s_i + 2 k = n$.  Since the indices take the same range as those in the expression \eqref{eq:noncommutativealgebra-again} (when we restrict the latter to words of length $n$), we see that there are $2^n$ elements in the set \eqref{eq:noncommutativealgebra}.

We will now consider symmetric polynomials in two reference spinors
\be X^{m,n}(\lt_1,\lt_2) = \frac{(\lt_2\cdot\p_{\lt_1})^n}{[m+n]_n}X^{m+n}(\lt_1) = \frac{(\lt_1\cdot\p_{\lt_2})^m}{[m+n]_m}X^{m+n}(\lt_2)\,. \ee
We would like to find the structure constants $R_{\{k_i\},\{l_j\}}$ obeying
\bea \label{eq:structure-constant-app}
&X^m(\lt_1)X^n(\lt_2) \\
&= \sum_{a=0,\dots,\min\{m,n\}} [12]^a \sum_{\substack{k_1+\dots+k_{a+1} = m-a \\ l_1+\dots+l_{a+1} = n-a}} R_{\{k_i\},\{l_j\}} X^{k_1,l_1}(\lt_1,\lt_2)C\dots CX^{k_{a+1},l_{a+1}}(\lt_1,\lt_2)\,,
\eea
where $C = [X^\done,X^\dtwo]$.  The fact that these coefficients are well-defined follows from the fact that the expressions written above form a basis.  

We can view $a,m,n$ as being determined by the sets $\{k_i\},\{l_j\}$ (constrained to have the same cardinality). Note that $R_{\{m\},\{n\}} = 1$. We can obtain recursion relations for these structure constants by multiplying \eqref{eq:structure-constant-app} by $X(\lt_1)$. To evaluate the r.h.s. we write
\be X(\lt_1)X^{k+l}(\lt_2) = X^{1,k+l}(\lt_1,\lt_2) + [12]\sum_{\substack{p+q=k+l-1}}R_{\{0,0\},\{p,q\}}X^p(\lt_2)CX^q(\lt_2)\,. \ee
Acting with $(\lt_1\cdot\p_{\lt_2})^k/[k+l]_k$ gives
\bea
&X(\lt_1)X^{k,l}(\lt_1,\lt_2) =  X^{k+1,l}(\lt_1,\lt_2) \\
&+ [12]\sum_{\substack{k^\prime_1+k^\prime_2=k \\ l^\prime_1 + l^\prime_2=l-1}}\frac{\binom{k_1^\prime+l^\prime_1}{k_1^\prime}\binom{k^\prime_2+l^\prime_2}{k_2^\prime}}{\binom{k+l}{k}}R_{\{0,0\},\{k^\prime_1+l^\prime_1,k^\prime_2+l^\prime_2\}}X^{k_1^\prime,l_1^\prime}(\lt_1,\lt_2)CX^{k_2^\prime,l_2^\prime}(\lt_1,\lt_2)\,. 
\eea
The second term is not present if $l=0$. Therefore
\bea
&X^{m+1}(\lt_1)X^n(\lt_2) \\
&= \sum_{a = 0,\dots,\min\{m,n\}}[12]^a\sum_{\substack{k_1+\dots+k_{a+1}=m-a \\ l_1+\dots+l_{a+1}=n-a}}R_{\{k_i\},\{l_j\}}X^{k_1+1,l_1}(\lt_1,\lt_2)C\dots CX^{k_{a+1},l_{a+1}}(\lt_1,\lt_2) \\
&+ \sum_{a = 0,\dots,\min\{m,n\}}[12]^{a+1}\sum_{\substack{k_1+\dots+k_{a+1}=m-a \\ l_1+\dots+l_{a+1}=n-a}}\sum_{\substack{k^\prime_1+k^\prime_2 = k_1 \\ l^\prime_1+l^\prime_2=l_1-1}}\frac{\binom{k_1^\prime+l^\prime_1}{k_1^\prime}\binom{k^\prime_2+l^\prime_2}{k_2^\prime}}{\binom{k+l}{k}}R_{\{0,0\},\{k^\prime_1+l^\prime_1,k^\prime_2+l^\prime_2\}}R_{\{k_i\},\{l_j\}} \\
&X^{k_1^\prime,l_1^\prime}(\lt_1,\lt_2)CX^{k_2^\prime,l_2^\prime}(\lt_1,\lt_2)C\dots CX^{k_{a+1},l_{a+1}}(\lt_1,\lt_2) \\
&= \sum_{a=0,\dots,\min\{m+1,n\}}[12]^a\sum_{\substack{k_1+\dots+k_{a+1}=m+1-a \\ l_1+\dots+l_{a+1}=n-a}}\Bigg(R_{\{k_1-1,k_2,\dots,k_{a+1}\},\{l_j\}} + \frac{\binom{k_1+l_1}{k_1}\binom{k_2+l_2}{k_2}}{\binom{k_1+k_2+l_1+l_2+1}{k_1+k_2}} \\
&R_{\{0,0\},\{k_1+l_1,k_2+l_2\}}R_{\{k_1+k_2,k_3,\dots,k_{a+1}\},\{l_1+l_2+1,l_3,\dots,l_{a+1}\}}\Bigg)X^{k_1,l_1}(\lt_1,\lt_2)C\dots CX^{k_{a+1},l_{a+1}}(\lt_1,\lt_2)\,.
\eea
This gives the following recursion relation for the structure constants
\bea \label{eq:recursion-1}
&R_{\{k_i\}_{i=1}^{a+1},\{l_j\}_{j=1}^{a+1}} = R_{\{k_1-1,k_2,\dots,k_{a+1}\},\{l_j\}_{j=1}^{a+1}} \\
&+ \frac{\binom{k_1+l_1}{k_1}\binom{k_2+l_2}{k_2}}{\binom{k_1+k_2+l_1+l_2+1}{k_1+k_2}}R_{\{0,0\},\{k_1+l_1,k_2+l_2\}}R_{\{k_1+k_2,k_3,\dots,k_{a+1}\},\{l_1+l_2+1,l_3,\dots,l_{a+1}\}}\,.
\eea
When $k_1 = 0$, the first term is understood to vanish. Similarly when $a=0$ the second is taken to be zero. The missing ingredients are the structure constants $R_{\{0,0\},\{p,q\}}$, but these are straightforward to compute. Indeed we did so in Sect. \ref{sec:nonfactorizing}, finding
\be R_{\{0,0\},\{p,q\}} = \frac{q+1}{p+q+2}\,. \ee
There is, of course, an entirely analogous relation obtained by multiplying on the left of equation \eqref{eq:structure-constant-app} with $X(\lt_2)$. We find
\bea \label{eq:recursion-2}
&R_{\{k_i\}_{i=1}^{a+1},\{l_j\}_{j=1}^{a+1}} = R_{\{k_i\}_{i=1}^{a+1},\{l_1,\dots,l_a,l_{a+1}-1\}} \\
&+ \frac{\binom{k_a+l_a}{k_a}\binom{k_{a+1}+l_{a+1}}{k_{a+1}}}{\binom{k_1+k_2+l_1+l_2+1}{l_1+l_2}}R_{\{k_a+l_a,k_{a+1}+l_{a+1}\},\{0,0\}}R_{\{k_1,\dots,k_{a-1},k_a+k_{a+1}+1\},\{l_1,\dots,l_{a-1},l_a+l_{a+1}\}}\,.
\eea
As in equation \eqref{eq:recursion-1} we discard the first term if $l_{a+1}=0$ and the second if $a=0$. It's easy to check that
\be R_{\{p,q\},\{0,0\}} = \frac{p+1}{p+q+2}\,. \ee
As a very basic consistency check, we can easily verify that $R_{\{p,q\},\{0,0\}}$ obeys \eqref{eq:recursion-1} and $R_{\{0,0\},\{p,q\}}$ obeys \eqref{eq:recursion-2}. In the main text we make use of the following structure constants, for which we have explicit formulas
\bea \label{eq:explicit-structure-constants}
&R_{\{k,\tilde k\},\{l,\tilde l\}} = \binom{k+\tilde k+l+\tilde l+2}{k+\tilde k+1}^{-1}\sum_{p=0}^k\binom{k+l-p}{l}\binom{\tilde k+\tilde l+1+p}{\tilde l}\,, \\
&R_{\{k,0^a\},\{l,0^a\}} = \frac{1}{[k+l+a+1]_a}\binom{k+a}{a}\,,\quad R_{\{0^a,\tilde k\},\{0^a,\tilde l\}} = \frac{1}{[\tilde k+\tilde l+a+1]_a}\binom{\tilde l+a}{a}\,.
\eea


\section{Recursion Relation for Burns Space Amplitudes} \label{app:rec_Burns}

In this appendix, we solve the differential equation \eqref{eq:diff_Burns_tree} for the tree level Burns space amplitude. This will give us an recursion relation expressing $n$-point amplitude in terms of $2,\dots,n-1$-point amplitudes.
	
Recall that the differential equation takes the form:
\bea
&\mc{D}^{(i)}\widehat{\mc{L}\mc{A}}_\text{tree}(1,\dots,n; \xi) \\
&= \sum_{j\neq i}\frac{[ij]\ip{ij}\ip{i\xi}^2\ip{j\xi}^2}{4\pi^2}\widehat{\mc{L}\mc{A}}_\text{tree}(i,i+1,\dots,j-1,j;\xi)\widehat{\mc{L}\mc{A}}_\text{tree}(j,j+1,\dots,i-1,i;\xi)\,.
\eea 
Since the rescaled amplitude $\widehat{\mc{L}\mc{A}}_\text{tree}(1,\dots,n; \xi)$ is cyclically symmetric, we can fix $i = 1$. This equation can be separated into two parts: one containing the $n$-point amplitude $\widehat{\mc{L}\mc{A}}_\text{tree}(1,\dots,n; \xi)$, and the other containing amplitudes with fewer than $n$ points.
	 	\bea
&\mc{D}^{(1)}\widehat{\mc{L}\mc{A}}_\text{tree}(1,\dots,n) \\
& = \frac{[12]\ip{12}\ip{1\xi}^2\ip{2\xi}^2}{4\pi^2}\widehat{\mc{L}\mc{A}}_\text{tree}(1,2)\widehat{\mc{L}\mc{A}}_\text{tree}(1,\dots,n)\\
&+\frac{[1n]\ip{1n}\ip{1\xi}^2\ip{n\xi}^2}{4\pi^2}\widehat{\mc{L}\mc{A}}_\text{tree}(1,\dots,n)\widehat{\mc{L}\mc{A}}_\text{tree}(1,n)\\
&+ \sum_{j = 3}^{n-1}\frac{[1j]\ip{1j}\ip{1\xi}^2\ip{j\xi}^2}{4\pi^2}\widehat{\mc{L}\mc{A}}_\text{tree}(1,\dots,j)\widehat{\mc{L}\mc{A}}_\text{tree}(j,j+1,\dots,n,1)\,.
\eea 
 We can further rewrite this equation as an ordinary differential equation (ODE). We define
\begin{equation}
	\mathcal{F}_{(1,\dots,n)}(t) = t^{\mc{D}^{(1)}}\widehat{\mc{L}\mc{A}}_\text{tree}(1,\dots,n)\,.
\end{equation}
For convenience we also write $\mathcal{F}_n(t) = \mathcal{F}_{(1,\dots,n)}(t)$. This quantity can also be identified with the $n$-point amplitude
\begin{equation}\label{eq:Burnamp_t}
	\frac{1}{\ip{1\xi}^4}\ip{\mathcal{L}\Jt(t\lt_1;z_1) \mathcal{L}\J(\lt_2;z_2)\dots \mathcal{L}\J(\lt_n;z_n)}_\text{tree} \,.
\end{equation}
We can therefore rewrite the differential equation as
\bea\label{eq:diff_Burns_sep}
	&\frac{\d}{\d t}\mathcal{F}_n(t) = \frac{[12]\ip{12}\ip{1\xi}^2\ip{2\xi}^2}{4\pi^2}\mathcal{F}_2(t)\mathcal{F}_n(t) +\frac{[1n]\ip{1n}\ip{1\xi}^2\ip{n\xi}^2}{4\pi^2}\mathcal{F}_{(1,n)}(t)\mathcal{F}_n(t)\\
    &+ \sum_{j = 3}^{n-1}\frac{[1j]\ip{1j}\ip{1\xi}^2\ip{j\xi}^2}{4\pi^2}\mathcal{F}_j(t)\mathcal{F}_{(1,j,j+1,\dots,n)}(t)\,.
\eea 
This differential equation is subject to the initial value condition
\begin{equation}
	\mathcal{F}_n(0) = \frac{\ip{2n}}{\ip{12}\ip{1n}}\widehat{\mc{L}\mc{A}}_\text{tree}(2,\dots,n)\,,
\end{equation}
which can be derived from \eqref{eq:Burnamp_t}. From our two-point amplitude \eqref{eq:Burns_2pt}, it's easy to compute that
\begin{equation}
	\mathcal{F}_2(t) = \frac{1}{2\ip{12}^2(1 - c_{12}t)},\qquad \mathcal{F}_{(1,n)}(t) = \frac{1}{2\ip{1n}^2(1 - c_{1n}t)}\,,
\end{equation}
where we define
\begin{equation}
	c_{ij} = \frac{\ip{i\xi}^2\ip{j\xi}^2[ij]}{8\pi^2\ip{ij}}\,.
\end{equation}
Then the differential equation takes the standard form of a non-homogeneous first order linear ODE. We first solve the homogeneous equation
\be
    \frac{\d}{\d t}I_n(t) = \left(\frac{c_{12}}{1-c_{12}t} +\frac{c_{1n}}{1-c_{1n}t}\right)I_n(t)\,,
\ee
which has solution 
\begin{equation}
	I_n(t) = \frac{1}{(1-c_{12}t)(1-c_{1n}t)}\,.
\end{equation}
Then the non-homogeneous equation, subjects to the initial value, have solution
\begin{equation}\label{eq:rec_Burns_tree}
	\mathcal{F}_n(t) = I_n(t)\left(\int_{0}^tdsI_n(s)^{-1}\mathcal{K}_{n-1}(s) + \frac{\ip{2n}}{\ip{12}\ip{1n}}\widehat{\mc{L}\mc{A}}_\text{tree}(2,\dots,n)\right)\,,
\end{equation}
where
\begin{equation}
	\mathcal{K}_{n-1}(t) = \sum_{j = 3}^{n-1}2\ip{1j}^2c_{1j}\mathcal{F}_j(t)\mathcal{F}_{(1,j,j+1,\dots,n)}(t)\,.
\end{equation}
We can use this recursion relation to compute the first few amplitudes. For $n = 3$, $\mathcal{K}_{2}(t)$ vanishes. We have
\bea
	&\mc{L}\mc{A}_\text{tree}(1,2,3) = \frac{\ip{1\xi}^4}{(1-c_{12})(1-c_{13})}\frac{\ip{23}}{\ip{12}\ip{13}}\frac{1}{2\ip{23}^2(1-c_{23})}`\\
	&= \frac{\ip{1\xi}^4}{2\ip{12}\ip{23}\ip{13}(1 - c_{12})(1-c_{23})(1-c_{13})}\,.
\eea
This matches the result \ref{eq:Burns_3pt} computed directly using Wick contraction.

For $n = 4$, we have
\begin{equation}
	\mathcal{K}_{3}(t) = \frac{c_{13}}{2\ip{12}\ip{23}\ip{34}\ip{14}(1 - c_{12}t)(1-c_{23})(1-c_{13}t)^2(1-c_{34})(1-c_{14}t)}\,.
\end{equation}
After integration we find
\bea
	&\mc{L}\mc{A}_\text{tree}(1,2,3,4) \\
	&= \frac{\ip{1\xi}^4}{2(1-c_{12})(1-c_{23})(1 - c_{34})(1-c_{14})\ip{12}\ip{23}\ip{34}\ip{14}}\left(\frac{1}{1 - c_{13}} + \frac{1}{1 - c_{24}} - 1\right)\,.
\eea
	
We can also solve the one-loop differential equation in the same way. By rescaling $\widetilde{\lambda}_1$ by $t$, we obtain a differential equation that resemble \eqref{eq:diff_Burns_sep}. In fact, the homogeneous part of the equation is the same. Therefore we obtain the same factor $I_n(t)$. The solution becomes
\begin{equation}
	\mathcal{F}^{\text{loop}}_n(t) = I_n(t)\left(\int_{0}^tdsI_n(s)^{-1}\mathcal{K}^{\text{loop}}_{n-1}(t) + \frac{\ip{2n}}{\ip{12}\ip{1n}}\widehat{\mc{L}\mc{A}}_\text{loop}(2,\dots,n)\right)
\end{equation}
where
\bea\label{eq:rec_Burns_loop}
	&\mathcal{K}^{\text{loop}}_{n-1}(t) = \sum_{j = 3}^{n}2\ip{1j}^2c_{1j}\mathcal{F}_j(t)\mathcal{F}^{\text{loop}}_{(1,j,j+1,\dots,n)}(t) + \sum_{j = 2}^{n-1}2\ip{1j}^2c_{1j}\mathcal{F}^{\text{loop}}_j(t)\mathcal{F}_{(1,j,j+1,\dots,n)}(t)\\
	&+ \sum_{j = 2}^{n}\ip{1j}^2c_{1j}\mathcal{F}_{(1,j,j-1,\dots,1,j,j+1,\dots,n)}(t)\,.
\eea

\end{appendix}


\bibliographystyle{JHEP}
\bibliography{main}

\providecommand{\href}[2]{#2}\begingroup\raggedright\begin{thebibliography}{10}

\bibitem{Strominger:2017zoo}
A.~Strominger, \emph{{Lectures on the Infrared Structure of Gravity and Gauge
  Theory}},  \href{https://arxiv.org/abs/1703.05448}{{\ttfamily 1703.05448}}.

\bibitem{Costello:2023vyy}
K.J.~Costello, \emph{{Bootstrapping two-loop QCD amplitudes}},
  \href{https://arxiv.org/abs/2302.00770}{{\ttfamily 2302.00770}}.

\bibitem{Dixon:2024mzh}
L.J.~Dixon and A.~Morales, \emph{{On gauge amplitudes first appearing at two
  loops}},  \href{https://arxiv.org/abs/2407.13967}{{\ttfamily 2407.13967}}.

\bibitem{Costello:2022jpg}
K.~Costello, N.M.~Paquette and A.~Sharma, \emph{{Top-Down Holography in an
  Asymptotically Flat Spacetime}},
  \href{https://doi.org/10.1103/PhysRevLett.130.061602}{\emph{Phys. Rev. Lett.}
  {\bfseries 130} (2023) 061602}
  [\href{https://arxiv.org/abs/2208.14233}{{\ttfamily 2208.14233}}].

\bibitem{Costello:2023hmi}
K.~Costello, N.M.~Paquette and A.~Sharma, \emph{{Burns space and holography}},
  \href{https://doi.org/10.1007/JHEP10(2023)174}{\emph{JHEP} {\bfseries 10}
  (2023) 174} [\href{https://arxiv.org/abs/2306.00940}{{\ttfamily
  2306.00940}}].

\bibitem{Costello:2021bah}
K.J.~Costello, \emph{{Quantizing local holomorphic field theories on twistor
  space}},  \href{https://arxiv.org/abs/2111.08879}{{\ttfamily 2111.08879}}.

\bibitem{Ward:1977ta}
R.S.~Ward, \emph{{On Selfdual gauge fields}},
  \href{https://doi.org/10.1016/0375-9601(77)90842-8}{\emph{Phys. Lett. A}
  {\bfseries 61} (1977) 81}.

\bibitem{Mason:2005zm}
L.J.~Mason, \emph{{Twistor actions for non-self-dual fields: A Derivation of
  twistor-string theory}},
  \href{https://doi.org/10.1088/1126-6708/2005/10/009}{\emph{JHEP} {\bfseries
  10} (2005) 009} [\href{https://arxiv.org/abs/hep-th/0507269}{{\ttfamily
  hep-th/0507269}}].

\bibitem{Costello:2017dso}
K.~Costello, E.~Witten and M.~Yamazaki, \emph{{Gauge Theory and Integrability,
  I}}, \href{https://doi.org/10.4310/ICCM.2018.v6.n1.a6}{\emph{ICCM Not.}
  {\bfseries 06} (2018) 46} [\href{https://arxiv.org/abs/1709.09993}{{\ttfamily
  1709.09993}}].

\bibitem{Costello:2018gyb}
K.~Costello, E.~Witten and M.~Yamazaki, \emph{{Gauge Theory and Integrability,
  II}}, \href{https://doi.org/10.4310/ICCM.2018.v6.n1.a7}{\emph{ICCM Not.}
  {\bfseries 06} (2018) 120}
  [\href{https://arxiv.org/abs/1802.01579}{{\ttfamily 1802.01579}}].

\bibitem{Costello:2019tri}
K.~Costello and M.~Yamazaki, \emph{{Gauge Theory and Integrability, III}},
  \href{https://arxiv.org/abs/1908.02289}{{\ttfamily 1908.02289}}.

\bibitem{williams2020ren}
B.R.~Williams, \emph{Renormalization for holomorphic field theories},
  {\emph{Communications in Mathematical Physics} {\bfseries 374} (2020) 1693}.

\bibitem{Dixon:2024tsb}
L.J.~Dixon and A.~Morales, \emph{{Rational QCD loop amplitudes and quantum
  theories on twistor space}},
  \href{https://arxiv.org/abs/2411.10967}{{\ttfamily 2411.10967}}.

\bibitem{Bjerrum-Bohr:2011jrh}
N.E.J.~Bjerrum-Bohr, P.H.~Damgaard, H.~Johansson and T.~Sondergaard,
  \emph{{Monodromy--like Relations for Finite Loop Amplitudes}},
  \href{https://doi.org/10.1007/JHEP05(2011)039}{\emph{JHEP} {\bfseries 05}
  (2011) 039} [\href{https://arxiv.org/abs/1103.6190}{{\ttfamily 1103.6190}}].

\bibitem{Bardeen:1995gk}
W.A.~Bardeen, \emph{{Selfdual Yang-Mills theory, integrability and multiparton
  amplitudes}}, \href{https://doi.org/10.1143/PTPS.123.1}{\emph{Prog. Theor.
  Phys. Suppl.} {\bfseries 123} (1996) 1}.

\bibitem{Costello:2022wso}
K.~Costello and N.M.~Paquette, \emph{{Celestial holography meets twisted
  holography: 4d amplitudes from chiral correlators}},
  \href{https://doi.org/10.1007/JHEP10(2022)193}{\emph{JHEP} {\bfseries 10}
  (2022) 193} [\href{https://arxiv.org/abs/2201.02595}{{\ttfamily
  2201.02595}}].

\bibitem{Costello:2022upu}
K.~Costello and N.M.~Paquette, \emph{{Associativity of One-Loop Corrections to
  the Celestial Operator Product Expansion}},
  \href{https://doi.org/10.1103/PhysRevLett.129.231604}{\emph{Phys. Rev. Lett.}
  {\bfseries 129} (2022) 231604}
  [\href{https://arxiv.org/abs/2204.05301}{{\ttfamily 2204.05301}}].

\bibitem{Elliott:2020ecf}
C.~Elliott, P.~Safronov and B.R.~Williams, \emph{{A taxonomy of twists of
  supersymmetric Yang\textendash{}Mills theory}},
  \href{https://doi.org/10.1007/s00029-022-00786-y}{\emph{Selecta Math.}
  {\bfseries 28} (2022) 73} [\href{https://arxiv.org/abs/2002.10517}{{\ttfamily
  2002.10517}}].

\bibitem{Johansen:1994aw}
A.~Johansen, \emph{{Twisting of $N=1$ SUSY gauge theories and heterotic
  topological theories}},
  \href{https://doi.org/10.1142/S0217751X9500200X}{\emph{Int. J. Mod. Phys. A}
  {\bfseries 10} (1995) 4325}
  [\href{https://arxiv.org/abs/hep-th/9403017}{{\ttfamily hep-th/9403017}}].

\bibitem{Budzik:2023xbr}
K.~Budzik, D.~Gaiotto, J.~Kulp, B.R.~Williams, J.~Wu and M.~Yu,
  \emph{{Semi-chiral operators in 4d $ \mathcal{N} $ = 1 gauge theories}},
  \href{https://doi.org/10.1007/JHEP05(2024)245}{\emph{JHEP} {\bfseries 05}
  (2024) 245} [\href{https://arxiv.org/abs/2306.01039}{{\ttfamily
  2306.01039}}].

\bibitem{Green:1984bx}
M.B.~Green, J.H.~Schwarz and P.C.~West, \emph{{Anomaly Free Chiral Theories in
  Six-Dimensions}},
  \href{https://doi.org/10.1016/0550-3213(85)90222-6}{\emph{Nucl. Phys. B}
  {\bfseries 254} (1985) 327}.

\bibitem{Erler:1993zy}
J.~Erler, \emph{{Anomaly cancellation in six-dimensions}},
  \href{https://doi.org/10.1063/1.530885}{\emph{J. Math. Phys.} {\bfseries 35}
  (1994) 1819} [\href{https://arxiv.org/abs/hep-th/9304104}{{\ttfamily
  hep-th/9304104}}].

\bibitem{Hartogs}
F.~Hartogs, \emph{Einige {Folgerungen} aus der \emph{Cauchy}schen
  {Integralformel} bei {Funktionen} mehrerer {Ver{\"a}nderlichen}.},
  {\emph{M{\"u}nch. {Ber}.} {\bfseries 36} (1906) 223}.

\bibitem{Aharony:2015zea}
O.~Aharony, M.~Berkooz and S.-J.~Rey, \emph{{Rigid Holography and
  Six-Dimensional $\mathcal{N}=\left(2,0\right)$ theories on AdS$_{5}$
  {\texttimes} $\mathbb{S}^{1}$}},
  \href{https://doi.org/10.1007/JHEP03(2015)121}{\emph{JHEP} {\bfseries 03}
  (2015) 121} [\href{https://arxiv.org/abs/1501.02904}{{\ttfamily
  1501.02904}}].

\bibitem{Ishtiaque:2018str}
N.~Ishtiaque, S.~Faroogh~Moosavian and Y.~Zhou, \emph{{Topological holography:
  The example of the D2-D4 brane system}},
  \href{https://doi.org/10.21468/SciPostPhys.9.2.017}{\emph{SciPost Phys.}
  {\bfseries 9} (2020) 017} [\href{https://arxiv.org/abs/1809.00372}{{\ttfamily
  1809.00372}}].

\bibitem{Itzhaki:2005tu}
N.~Itzhaki, D.~Kutasov and N.~Seiberg, \emph{{I-brane dynamics}},
  \href{https://doi.org/10.1088/1126-6708/2006/01/119}{\emph{JHEP} {\bfseries
  01} (2006) 119} [\href{https://arxiv.org/abs/hep-th/0508025}{{\ttfamily
  hep-th/0508025}}].

\bibitem{Nunez:2023nnl}
C.~Nunez, M.~Oyarzo and R.~Stuardo, \emph{{Confinement in (1 + 1) dimensions: a
  holographic perspective from I-branes}},
  \href{https://doi.org/10.1007/JHEP09(2023)201}{\emph{JHEP} {\bfseries 09}
  (2023) 201} [\href{https://arxiv.org/abs/2307.04783}{{\ttfamily
  2307.04783}}].

\bibitem{Mikhaylov:2014aoa}
V.~Mikhaylov and E.~Witten, \emph{{Branes And Supergroups}},
  \href{https://doi.org/10.1007/s00220-015-2449-y}{\emph{Commun. Math. Phys.}
  {\bfseries 340} (2015) 699}
  [\href{https://arxiv.org/abs/1410.1175}{{\ttfamily 1410.1175}}].

\bibitem{burns1986twistors}
D.~Burns, \emph{Twistors and harmonic maps}, {\emph{Amer. Math. Soc. conference
  talk, Charlotte, NC} (1986) }.

\bibitem{Eguchi:1978xp}
T.~Eguchi and A.J.~Hanson, \emph{{Asymptotically flat selfdual solutions to
  Euclidean gravity}},
  \href{https://doi.org/10.1016/0370-2693(78)90566-X}{\emph{Phys. Lett. B}
  {\bfseries 74} (1978) 249}.

\bibitem{Douglas:1996uz}
M.R.~Douglas, \emph{{Gauge fields and D-branes}},
  \href{https://doi.org/10.1016/S0393-0440(97)00024-7}{\emph{J. Geom. Phys.}
  {\bfseries 28} (1998) 255}
  [\href{https://arxiv.org/abs/hep-th/9604198}{{\ttfamily hep-th/9604198}}].

\bibitem{Costello:2016mgj}
K.~Costello and S.~Li, \emph{{Twisted supergravity and its quantization}},
  \href{https://arxiv.org/abs/1606.00365}{{\ttfamily 1606.00365}}.

\bibitem{Raghavendran:2021qbh}
S.~Raghavendran, I.~Saberi and B.R.~Williams, \emph{{Twisted Eleven-Dimensional
  Supergravity}},
  \href{https://doi.org/10.1007/s00220-023-04745-2}{\emph{Commun. Math. Phys.}
  {\bfseries 402} (2023) 1103}
  [\href{https://arxiv.org/abs/2111.03049}{{\ttfamily 2111.03049}}].

\bibitem{Witten:1988ze}
E.~Witten, \emph{{Topological Quantum Field Theory}},
  \href{https://doi.org/10.1007/BF01223371}{\emph{Commun. Math. Phys.}
  {\bfseries 117} (1988) 353}.

\bibitem{Bershadsky:1993cx}
M.~Bershadsky, S.~Cecotti, H.~Ooguri and C.~Vafa, \emph{{Kodaira-Spencer theory
  of gravity and exact results for quantum string amplitudes}},
  \href{https://doi.org/10.1007/BF02099774}{\emph{Commun. Math. Phys.}
  {\bfseries 165} (1994) 311}
  [\href{https://arxiv.org/abs/hep-th/9309140}{{\ttfamily hep-th/9309140}}].

\bibitem{Baulieu:2010ch}
L.~Baulieu, \emph{{SU(5)-invariant decomposition of ten-dimensional Yang-Mills
  supersymmetry}},
  \href{https://doi.org/10.1016/j.physletb.2010.12.044}{\emph{Phys. Lett. B}
  {\bfseries 698} (2011) 63} [\href{https://arxiv.org/abs/1009.3893}{{\ttfamily
  1009.3893}}].

\bibitem{Saberi:2021weg}
I.~Saberi and B.R.~Williams, \emph{{Twisting pure spinor superfields, with
  applications to supergravity}},
  \href{https://doi.org/10.4310/PAMQ.2024.v20.n2.a2}{\emph{Pure Appl. Math.
  Quart.} {\bfseries 20} (2024) 645}
  [\href{https://arxiv.org/abs/2106.15639}{{\ttfamily 2106.15639}}].

\bibitem{Costello:2019jsy}
K.~Costello and S.~Li, \emph{{Anomaly cancellation in the topological string}},
  \href{https://doi.org/10.4310/ATMP.2020.v24.n7.a2}{\emph{Adv. Theor. Math.
  Phys.} {\bfseries 24} (2020) 1723}
  [\href{https://arxiv.org/abs/1905.09269}{{\ttfamily 1905.09269}}].

\bibitem{Costello:2021kiv}
K.~Costello and B.R.~Williams, \emph{{Twisted heterotic/type I duality}},
  \href{https://arxiv.org/abs/2110.14616}{{\ttfamily 2110.14616}}.

\bibitem{Witten:2003nn}
E.~Witten, \emph{{Perturbative gauge theory as a string theory in twistor
  space}}, \href{https://doi.org/10.1007/s00220-004-1187-3}{\emph{Commun. Math.
  Phys.} {\bfseries 252} (2004) 189}
  [\href{https://arxiv.org/abs/hep-th/0312171}{{\ttfamily hep-th/0312171}}].

\bibitem{Berkovits:2004hg}
N.~Berkovits, \emph{{An Alternative string theory in twistor space for N=4
  superYang-Mills}},
  \href{https://doi.org/10.1103/PhysRevLett.93.011601}{\emph{Phys. Rev. Lett.}
  {\bfseries 93} (2004) 011601}
  [\href{https://arxiv.org/abs/hep-th/0402045}{{\ttfamily hep-th/0402045}}].

\bibitem{Coates:2001ewh}
T.~Coates and A.~Givental, \emph{{Quantum Riemann\textendash{}Roch, Lefschetz
  and Serre}}, \href{https://doi.org/10.4007/annals.2007.165.15}{\emph{Ann.
  Math.} {\bfseries 165} (2007) 15}
  [\href{https://arxiv.org/abs/math/0110142}{{\ttfamily math/0110142}}].

\bibitem{beilinson1988koszul}
A.A.~Beilinson, V.A.~Ginsburg and V.V.~Schechtman, \emph{Koszul duality},
  {\emph{Journal of geometry and physics} {\bfseries 5} (1988) 317}.

\bibitem{Bern:2002zk}
Z.~Bern, A.~De~Freitas, L.J.~Dixon and H.L.~Wong, \emph{{Supersymmetric
  regularization, two loop QCD amplitudes and coupling shifts}},
  \href{https://doi.org/10.1103/PhysRevD.66.085002}{\emph{Phys. Rev. D}
  {\bfseries 66} (2002) 085002}
  [\href{https://arxiv.org/abs/hep-ph/0202271}{{\ttfamily hep-ph/0202271}}].

\bibitem{Costello:2018zrm}
K.~Costello and D.~Gaiotto, \emph{{Twisted Holography}},
  \href{https://arxiv.org/abs/1812.09257}{{\ttfamily 1812.09257}}.

\bibitem{Strominger:2021mtt}
A.~Strominger, \emph{{$w_{1+\infty}$ Algebra and the Celestial Sphere: Infinite
  Towers of Soft Graviton, Photon, and Gluon Symmetries}},
  \href{https://doi.org/10.1103/PhysRevLett.127.221601}{\emph{Phys. Rev. Lett.}
  {\bfseries 127} (2021) 221601}.

\bibitem{Yang:1977zf}
C.N.~Yang, \emph{{Condition of Self-duality for SU(2) Gauge Fields on Euclidean
  Four-Dimensional Space}},
  \href{https://doi.org/10.1103/PhysRevLett.38.1377}{\emph{Phys. Rev. Lett.}
  {\bfseries 38} (1977) 1377}.

\bibitem{Chalmers:1996rq}
G.~Chalmers and W.~Siegel, \emph{{The Selfdual sector of QCD amplitudes}},
  \href{https://doi.org/10.1103/PhysRevD.54.7628}{\emph{Phys. Rev. D}
  {\bfseries 54} (1996) 7628}
  [\href{https://arxiv.org/abs/hep-th/9606061}{{\ttfamily hep-th/9606061}}].

\bibitem{Bern:1995ix}
Z.~Bern and G.~Chalmers, \emph{{Factorization in one loop gauge theory}},
  \href{https://doi.org/10.1016/0550-3213(95)00226-I}{\emph{Nucl. Phys. B}
  {\bfseries 447} (1995) 465}
  [\href{https://arxiv.org/abs/hep-ph/9503236}{{\ttfamily hep-ph/9503236}}].

\bibitem{Costello:2015xsa}
K.~Costello and S.~Li, \emph{{Quantization of open-closed BCOV theory, I}},
  \href{https://arxiv.org/abs/1505.06703}{{\ttfamily 1505.06703}}.

\bibitem{Green:1984sg}
M.B.~Green and J.H.~Schwarz, \emph{{Anomaly Cancellation in Supersymmetric D=10
  Gauge Theory and Superstring Theory}},
  \href{https://doi.org/10.1016/0370-2693(84)91565-X}{\emph{Phys. Lett. B}
  {\bfseries 149} (1984) 117}.

\bibitem{Penrose:1976js}
R.~Penrose, \emph{{Nonlinear gravitons and curved twistor theory}},
  \href{https://doi.org/10.1007/BF00762011}{\emph{Gen. Rel. Grav.} {\bfseries
  7} (1976) 31}.

\bibitem{Mason:2007ct}
L.J.~Mason and M.~Wolf, \emph{{Twistor Actions for Self-Dual Supergravities}},
  \href{https://doi.org/10.1007/s00220-009-0732-5}{\emph{Commun. Math. Phys.}
  {\bfseries 288} (2009) 97} [\href{https://arxiv.org/abs/0706.1941}{{\ttfamily
  0706.1941}}].

\bibitem{Bittleston:2023bzp}
R.~Bittleston, S.~Heuveline and D.~Skinner, \emph{{The celestial chiral algebra
  of self-dual gravity on Eguchi-Hanson space}},
  \href{https://doi.org/10.1007/JHEP09(2023)008}{\emph{JHEP} {\bfseries 09}
  (2023) 008} [\href{https://arxiv.org/abs/2305.09451}{{\ttfamily
  2305.09451}}].

\bibitem{Kronheimer:1989zs}
P.B.~Kronheimer, \emph{{The Construction of ALE Spaces as Hyper-K{\"a}hler
  Quotients}}, {\emph{J. Diff. Geom.} {\bfseries 29} (1989) 665}.

\bibitem{Plebanski:1975wn}
J.F.~Plebanski, \emph{{Some solutions of complex Einstein equations}},
  \href{https://doi.org/10.1063/1.522505}{\emph{J. Math. Phys.} {\bfseries 16}
  (1975) 2395}.

\bibitem{Siegel:1992wd}
W.~Siegel, \emph{{Selfdual {$\mathcal{N}=8$} supergravity as closed
  {$\mathcal{N}=2$} (${\mathcal{N}=4}$) strings}},
  \href{https://doi.org/10.1103/PhysRevD.47.2504}{\emph{Phys. Rev. D}
  {\bfseries 47} (1993) 2504}
  [\href{https://arxiv.org/abs/hep-th/9207043}{{\ttfamily hep-th/9207043}}].

\bibitem{beilinson2004chiral}
A.~Beilinson and V.~Drinfeld, \emph{Chiral algebras}, vol.~51, American
  Mathematical Soc. (2004).

\bibitem{Hawking:1979hw}
S.W.~Hawking, D.N.~Page and C.N.~Pope, \emph{{The propagation of particles in
  space-time foam}},
  \href{https://doi.org/10.1016/0370-2693(79)90812-8}{\emph{Phys. Lett. B}
  {\bfseries 86} (1979) 175}.

\bibitem{warner1982scattering}
N.P.~Warner, \emph{The scattering of spin-1 particles by quantum gravitational
  bubbles}, {\emph{Communications in Mathematical Physics} {\bfseries 86}
  (1982) 419}.

\bibitem{Hofman:2008ar}
D.M.~Hofman and J.~Maldacena, \emph{{Conformal collider physics: Energy and
  charge correlations}},
  \href{https://doi.org/10.1088/1126-6708/2008/05/012}{\emph{JHEP} {\bfseries
  05} (2008) 012} [\href{https://arxiv.org/abs/0803.1467}{{\ttfamily
  0803.1467}}].

\bibitem{Mahlon:1993si}
G.~Mahlon, \emph{{Multi-gluon helicity amplitudes involving a quark loop}},
  \href{https://doi.org/10.1103/PhysRevD.49.4438}{\emph{Phys. Rev. D}
  {\bfseries 49} (1994) 4438}
  [\href{https://arxiv.org/abs/hep-ph/9312276}{{\ttfamily hep-ph/9312276}}].

\bibitem{Bern:1993sx}
Z.~Bern, L.J.~Dixon and D.A.~Kosower, \emph{{New QCD results from string
  theory}},  in \emph{{International Conference on Strings 93}}, 5, 1993
  [\href{https://arxiv.org/abs/hep-th/9311026}{{\ttfamily hep-th/9311026}}].

\bibitem{Bern:1993qk}
Z.~Bern, G.~Chalmers, L.J.~Dixon and D.A.~Kosower, \emph{{One loop N gluon
  amplitudes with maximal helicity violation via collinear limits}},
  \href{https://doi.org/10.1103/PhysRevLett.72.2134}{\emph{Phys. Rev. Lett.}
  {\bfseries 72} (1994) 2134}
  [\href{https://arxiv.org/abs/hep-ph/9312333}{{\ttfamily hep-ph/9312333}}].

\bibitem{Agarwal:2023suw}
B.~Agarwal, F.~Buccioni, F.~Devoto, G.~Gambuti, A.~von Manteuffel and
  L.~Tancredi, \emph{{Five-parton scattering in QCD at two loops}},
  \href{https://doi.org/10.1103/PhysRevD.109.094025}{\emph{Phys. Rev. D}
  {\bfseries 109} (2024) 094025}
  [\href{https://arxiv.org/abs/2311.09870}{{\ttfamily 2311.09870}}].

\bibitem{Bittleston:2022jeq}
R.~Bittleston, \emph{{On the associativity of 1-loop corrections to the
  celestial operator product in gravity}},
  \href{https://doi.org/10.1007/JHEP01(2023)018}{\emph{JHEP} {\bfseries 01}
  (2023) 018} [\href{https://arxiv.org/abs/2211.06417}{{\ttfamily
  2211.06417}}].

\bibitem{Zeng:2023qqp}
K.~Zeng, \emph{{Twisted Holography and Celestial Holography from Boundary
  Chiral Algebra}},  \href{https://arxiv.org/abs/2302.06693}{{\ttfamily
  2302.06693}}.

\bibitem{Costello:2018fnz}
K.~Costello and D.~Gaiotto, \emph{{Vertex Operator Algebras and 3d $
  \mathcal{N} $ = 4 gauge theories}},
  \href{https://doi.org/10.1007/JHEP05(2019)018}{\emph{JHEP} {\bfseries 05}
  (2019) 018} [\href{https://arxiv.org/abs/1804.06460}{{\ttfamily
  1804.06460}}].

\bibitem{Eager:2018oww}
R.~Eager and I.~Saberi, \emph{{Holomorphic field theories and
  Calabi\textendash{}Yau algebras}},
  \href{https://doi.org/10.1142/S0217751X19500714}{\emph{Int. J. Mod. Phys. A}
  {\bfseries 34} (2019) 1950071}
  [\href{https://arxiv.org/abs/1805.02084}{{\ttfamily 1805.02084}}].

\bibitem{Chang:2013fba}
C.-M.~Chang and X.~Yin, \emph{{1/16 BPS states in $\mathcal N=$ 4
  super-Yang-Mills theory}},
  \href{https://doi.org/10.1103/PhysRevD.88.106005}{\emph{Phys. Rev. D}
  {\bfseries 88} (2013) 106005}
  [\href{https://arxiv.org/abs/1305.6314}{{\ttfamily 1305.6314}}].

\bibitem{loday1984cyclic}
J.-L.~Loday and D.~Quillen, \emph{{Cyclic homology and the Lie algebra homology
  of matrices}}, {\emph{Commentarii Mathematici Helvetici} {\bfseries 59}
  (1984) 565}.

\bibitem{tsygan1983homology}
B.L.~Tsygan, \emph{{The homology of matrix Lie algebras over rings and the
  Hochschild homology}}, {\emph{Russian Mathematical Surveys} {\bfseries 38}
  (1983) 198}.

\bibitem{LODAY198893}
J.-L.~Loday and C.~Procesi, \emph{Homology of symplectic and orthogonal
  algebras},
  \href{https://doi.org/https://doi.org/10.1016/0001-8708(88)90061-8}{\emph{Advances
  in Mathematics} {\bfseries 69} (1988) 93}.

\bibitem{Zeng:2023lox}
K.~Zeng, \emph{{Loday-Quillen-Tsygan theorem on Quivers}},
  \href{https://arxiv.org/abs/2307.10545}{{\ttfamily 2307.10545}}.

\end{thebibliography}\endgroup


\end{document}